\documentclass[11pt]{article}
\usepackage{a4wide}
\usepackage[mathscr]{eucal}
\usepackage{tikz}
\usepackage{tikz-cd}
\usetikzlibrary{positioning,automata}
\usetikzlibrary{arrows,calc}
\usetikzlibrary{decorations.markings,plotmarks}
\usepackage{graphics}
\usepackage{color}
\usepackage{hyperref}
\usepackage{enumerate}
\usepackage{mathtools}
\usepackage{cite}
\usepackage{amsmath,amssymb,amsthm}
\usepackage{xcolor}
\usepackage{verbatim}

\usepackage{graphicx}
\usepackage{caption}
\usepackage{subcaption}
\usepackage{tablefootnote}

\usepackage{empheq}

\DeclareFontFamily{U}{mathc}{}
\DeclareFontShape{U}{mathc}{m}{it}%
{<->s*[1.03] mathc10}{}

\DeclareMathAlphabet{\myscr}{U}{mathc}{m}{it}

\newtheorem{theorem}{Theorem}
\newtheorem{proposition}[theorem]{Proposition}
\newtheorem{corollary}[theorem]{Corollary}
\newtheorem{lemma}[theorem]{Lemma}
\newtheorem{fact}[theorem]{Fact}

\theoremstyle{definition}
\newtheorem{definition}[theorem]{Definition}

\newtheorem{question}[theorem]{Question}

\newtheorem{remark}[theorem]{Remark}
\newtheorem*{notation}{Notation}
\newtheorem{claim}{Claim}

\newtheorem{observation}[theorem]{Observation}

\numberwithin{theorem}{section}
\numberwithin{claim}{theorem}

\DeclareMathOperator{\Tr}{Tr}

\DeclareMathOperator{\Ker}{Ker}

\DeclareMathOperator{\vspan}{span}
\DeclareMathOperator{\term}{term}
\DeclareMathOperator{\Diag}{Diag}
\DeclareMathOperator{\tw}{tw}
\DeclareMathOperator{\Ran}{Ran}
\DeclareMathOperator{\vect}{vec}
\DeclareMathOperator{\bnd}{bnd}
\DeclareMathOperator{\sign}{sign}
\DeclareMathOperator{\polylog}{polylog}

\newcommand{\cten}{{\vcenter{\hbox{\scalebox{0.75}{\copyright}}}}}

\newcommand{\wt}[1]{\widetilde{#1}}

\newcommand{\ol}[1]{\overline{#1}}

\newcommand{\mvec}[1]{#1_{\mathfrak{v}}}
\newcommand{\oa}[1]{\boldsymbol{\myscr{#1}}}

\definecolor{myback}{rgb}{0.96, 0.96, 0.96}

\usepackage[framemethod=tikz,xcolor=true]{mdframed}
\newmdenv[%
    leftmargin=0.1cm,
    backgroundcolor=myback,%
    roundcorner=5pt,%
    tikzsetting={draw=black, line width=2.0pt}%
    ]{Optbox}%

\DeclareMathOperator{\tr}{tr}

\newcommand{\norm}[1]{\left\lVert#1\right\rVert}

\title{Exact nuclear norm, completion and decomposition for random overcomplete tensors via degree-4 SOS}
\author{Bohdan Kivva \\ The University of Chicago \\ { bkivva@uchicago.edu} \and Aaron Potechin\\ The University of Chicago \\ { potechin@uchicago.edu}}

\begin{document}
\pagenumbering{roman}

\maketitle

\begin{abstract}

In this paper we show that simple semidefinite programs inspired by degree $4$ SOS can exactly solve the tensor nuclear norm, tensor decomposition, and tensor completion problems on tensors with random asymmetric components. More precisely, for tensor nuclear norm and tensor decomposition, we show that w.h.p. these semidefinite programs can exactly find the nuclear norm and components of an $(n\times n\times n)$-tensor $\mathcal{T}$ with $m\leq n^{3/2}/\polylog(n)$ random asymmetric components. Unlike most of the previous algorithms, our algorithm provides a certificate for the decomposition, does not require knowledge about the number of components in the decomposition and does not make any assumptions on the sizes of the coefficients in the decomposition.  As a byproduct, we show that w.h.p. the nuclear norm decomposition exactly coincides with the minimum rank decomposition for tensors with $m\leq n^{3/2}/\polylog(n)$ random asymmetric components.

For tensor completion, we show that w.h.p. the semidefinite program, introduced by Potechin \& Steurer (2017) for tensors with orthogonal components, can exactly recover an $(n\times n\times n)$-tensor $\mathcal{T}$ with $m$ random asymmetric components from only $n^{3/2}m\polylog(n)$ randomly observed entries. For non-orthogonal tensors this improves the dependence on $m$ of the number of entries needed for exact recovery over all previously known algorithms and provides the first theoretical guarantees for exact tensor completion in the overcomplete regime.
\end{abstract}

\vspace{3cm}

\section*{Acknowledgements} B.K. was partially supported by his advisor Laszlo Babai's  NSF Grant CCF 1718902. B.K. is grateful to Shuo Pang and Wenjun Cai for some helpful discussions. A.P. was partially supported by NSF grant CCF-2008920.

\pagebreak

\tableofcontents

\pagebreak

\pagenumbering{arabic}

\section{Introduction}

In this paper, we study several problems on tensors, namely nuclear norm minimization, tensor completion, and tensor decomposition. While these problems are known to be hard in the worst case \cite{Hillar-Lim}, we show that for tensors with random asymmetric components, semidefinite programs inspired by degree $4$ sum of squares can solve these problems exactly. 

For an order-3 tensor $\mathcal{T}\in \mathbb{R}^{n^3}$ with $m$ random independent asymmetric components:   
\begin{itemize}
\item We give the first polynomial time algorithm that w.h.p. exactly computes the nuclear norm of $\mathcal{T}$ if $m \leq n^{3/2}/\polylog(n)$ and provides a certificate for the nuclear norm.
\item We give a polynomial time algorithm that w.h.p. exactly completes $\mathcal{T}$ from $N \geq  mn^{3/2}\polylog(n)$ randomly selected entries.  This improves the dependence of $N$ on the rank $m$ over all previously known efficient algorithms for exact recovery (which addresses the open problem posed by~\cite[p.8]{jain-oh} about the dependence of the number of entries needed for exact recovery on the rank $m$) and matches bounds on $N$ of the algorithms for approximate recovery. In our algorithm $m$ can be as large as $n^{3/2}/\polylog(n)$, giving the first algorithm for provable exact recovery in the overcomplete regime.
\item We give a novel polynomial time algorithm that w.h.p. reconstructs the components of $\mathcal{T}$ exactly when $m$  is as large as $n^{3/2}/\polylog(n)$ even in the presence of missing entries.
\item We prove that for $m \leq n^{3/2}/\polylog(n)$ w.h.p. the nuclear norm decomposition coincides with the minimum rank decomposition. Despite being natural the question of when these decompositions are the same does not seem to have been addressed before. 
\end{itemize}
 All of these results hold for order-$d$ tensors if one replaces $n^{3/2}$ with $n^{d/2}$ in all the above bounds (see Appendix~\ref{sec:higher-order-tensors}). 
 
 We provide numerical experiments in support of our claims in Section~\ref{sec:numerical}.

For tensor completion we are using the semidefinite program proposed in~\cite{potechin-steurer-exact} for tensors with orthogonal components. We provide a new (more involved) analysis to show that this program succeeds w.h.p. in completing tensors with random asymmetric components.

We also note the following features of our results. First, our algorithms are robust to the magnitudes of the coefficients of the components. We can handle tensors where the coefficients of the components are exponentially large or small. Second, our proofs are specifically for tensors with random asymmetric components. While we believe the same results should be true for tensors with random symmetric components, this would require a separate analysis.

\subsection{Tensor Decompositions: Tensor Rank and Tensor Nuclear Norm}\label{tensorrankandnuclearnormintro}
A fundamental questions about matrices and tensors is as follows.

\begin{enumerate}
\item[(Q1)] \textit{Given a matrix $M$ or a tensor $T$, what is the best way to decompose it?} 
\end{enumerate}
In other words, given a matrix $M$, or an order $3$ tensor $\mathcal{T}$, what is the best way to write it as 
\begin{equation}
    M = \sum_{i=1}^{r}{{\sigma_i}{u_i}{v_i}^{\top}} \quad \text{or} \quad \mathcal{T} = \sum_{i=1}^{m}{{\lambda_i}u_i \otimes v_i \otimes w_i},
\end{equation}  
where the components $u_i$, $v_i$, and $w_i$ are unit vectors?

\textit{For matrices, there is a canonical decomposition}, the singular value decomposition (SVD). The singular value decomposition $M = \sum_{i=1}^{r}{{\sigma_i}{u_i}{v_i}^{\top}}$ minimizes both the number of components (for the SVD, $r$ is the rank of $M$) and the sum of the magnitudes of the coefficients (for the SVD, $\sum_{i=1}^{r}{|\sigma_i|}$  is the nuclear norm of $M$). 

\textit{For tensors, there is no canonical decomposition}. One natural choice is to minimize the number of components $m$ (the minimum possible $m$ is the \textit{tensor rank}). Such a decomposition is known as a minimum rank or CP decomposition. Another natural choice is to minimize the sum of the magnitudes of the coefficients. The minimum possible sum is the \textit{tensor nuclear norm}, so we call such a decomposition a nuclear norm decomposition. 

Unfortunately, both rank and nuclear norm decompositions are hard to compute.
Determining the rank of a tensor is known to be NP-hard~\cite{hastad}. Over finite fields, even approximating the rank of an order-3 tensor up to $1+1/1852-\varepsilon$ is NP-hard \cite{swernofsky}. Similarly, approximating the nuclear norm of a tensor is also known to be NP-hard \cite{Friedland-Lim-nuclear-approx}.

That said, there are several ways in which nuclear norm decompositions behave better than rank decompositions.
\begin{enumerate}

\item Tensor rank is sensitive to small perturbations. It is even possible that the tensor rank can be decreased by arbitrarily small perturbations~\cite{desilva-lim-rank-ill-posed}!

On the other hand, the tensor nuclear norm is a convex measure on tensors. Because of this, the tensor nuclear norm is robust to noise and there are semidefinite programming relaxations for the tensor nuclear norm. In a certain sense the nuclear norm is the best convex relaxation of the rank \cite{CRPW-convex-relaxation-rank}.
\item Even if we are given a rank decomposition, it may be very hard to verify that there is no decomposition which has a smaller number of components, especially in the overcomplete case where the rank of the tensor $m$ is larger than $n$. In fact, the best lower bound that we have for the tensor rank of an explicit order-3 tensor is $3n - o(n)$ \cite{shpilka}.

On the other hand, since the tensor nuclear norm is dual to the injective norm of a tensor, there is a dual certificate for proving lower bounds on the tensor nuclear norm (though it may be hard to find and verify this dual certificate, we discuss this in Section \ref{tensordecompositionintro}). This means that if we have both a nuclear norm decomposition and this dual certificate then we can certify the value of the tensor nuclear norm. 

\end{enumerate}

In practice, tensor problems are frequently cast as nuclear norm minimization problems (explicitly, or implicitly, using nuclear norm or its surrogates as a regularizer), while assuming that the underlying tensor has low rank~\cite{Yuan-Zhang, liu2014factor, visual-data, liu2014trace, chen2019low}. Hence, the following question naturally arises.

\begin{enumerate}
\item[(Q2)] \textit{When does tensor nuclear decomposition coincide with rank decomposition?} 
\end{enumerate}

To the best of our knowledge, this question was not explicitly studied before. In this paper we prove the following.
\begin{definition}
We say that $\mathcal{T} = \sum\limits_{i=1}^{m}\lambda_i a^{(1)}_i\otimes a^{(2)}_i\otimes \ldots \otimes a^{(d)}_i$ is an order-$d$ \emph{tensor with $m$ random asymmetric components}, if $\mathcal{V} = \{a_{i}^{(t)}\in \mathbb{S}^{n-1}\mid i\in [m], t\in [d] \}$ is a collection of $md$ independent random vectors sampled from a uniform distribution on a unit sphere.
\end{definition}

\begin{theorem}[Main I]\label{thm:main:nuclear-norm-rank-same}
If $\mathcal{T}$ is an order-$d$ tensor with $m\leq n^{d/2}/polylog(n)$ random asymmetric components then with high probability (w.h.p.) $\mathcal{T}$ has a unique nuclear norm decomposition and this decomposition is a rank decomposition as well. 
\end{theorem}

\subsection{Previous Work on Tensor Decomposition and Our Contribution}\label{tensordecompositionintro}

Most previous work on tensor decomposition has focused on the following variant of the tensor decomposition problem. Given a tensor $\mathcal{T} = \sum_{i=1}^{m}{{\lambda_i}u_i \otimes v_i \otimes w_i}$ where the coefficients $\lambda_i$ and components $u_i$, $v_i$ and $w_i$  are chosen in a certain way, can we recover these coefficients and components from $\mathcal{T}$? 

This question goes back to Hitchcock in 1927 \cite{hitchcock-1, hitchcock-2} and has applications in numerous fields of research: psychometrics, chemometrics, numerical linear algebra, computer vision, neuroscience, data mining, etc. For references, see page 2 of \cite{kolda-survey}.

In the undercomplete case (i.e., when $m\leq n$) with linearly independent components the problem can be solved by Jennrich's algorithm \cite{harshman}. The algorithm essentially reduces the problem to a matrix SVD by looking at a random slice of the tensor. Another approach which is frequently used in practice, alternating least square (ALS) minimization, was proposed by Carrol, Chang \cite{carrol-chang} and Harshman~\cite{harshman}. 

Despite decades of research, until recently, very little was known in the overcomplete regime, i.e. when the rank of the tensor is larger than the dimension of its components.

In 2014, Anandkumar, Ge, and Janzamin \cite{AGJ-tensor-decomp-alt-min} proved that for any $C>0$ an incoherent order-3 tensor with $m \leq C\cdot n$ random asymmetric components can be decomposed exactly using modified alternating least square minimization. Moreover, they proved local convergence guarantees for the algorithm as long as $m \leq n^{3/2}/\polylog(n)$.

Ma, Shi, Steurer \cite{MSS-decomp} proved that an order-3 tensor with random symmetric components (i.e. $\mathcal{T} = \sum_{i=1}^{m}{\lambda_i(a_i \otimes a_i \otimes a_i)}$ where each $a_i$ is a random unit vector) can be decomposed approximately in polynomial time when $m \leq n^{3/2}/\polylog(n)$. However, we note that this algorithm cannot be combined with~\cite{AGJ-tensor-decomp-alt-min} to achieve exact recovery, as for \cite{AGJ-tensor-decomp-alt-min} the components are assumed to be random asymmetric. Also, the algorithm of \cite{MSS-decomp} is very slow. To address this, Hopkins, Schramm, Shi and Steurer~\cite{HSSS} proposed a faster spectral algorithm  for approximate recovery of tensor components for random symmetric order-3 tensors of rank $m \leq n^{4/3}/\polylog(n)$.  

In the smoothed analysis setup of the tensor decomposition problem where an adversary chooses the tensor and then the components of the tensor are perturbed, ~\cite{bhaskara2014smoothed} proved that the components of an order-$d$ tensor can be recovered approximately if $m\leq n^{\lfloor\frac{d-1)}{2}\rfloor}/2$. \cite{MSS-decomp} provided the first robust algorithm that works in the smoothed analysis setup for overcomplete order-4 tensors. A spectral algorithm for order-4 tensors under the same smoothed analysis setup was proposed by Hopkins, Schramm and Shi~\cite{hopkins2019robust}. To the best of our knowledge there is no algorithm that has provable guarantees for overcomplete order-3 tensors in the smoothed analysis setup.

We discuss decomposition algorithms that work in the presence of missing entries in Sec.~\ref{sec:intro:completion}.

\subsubsection{Lower Bounds on Tensor Rank}
The following important question is left unresolved by the above results (mainly in the overcomplete regime). Assume a decomposition algorithm  found a decomposition of a tensor
\begin{equation}\label{eq:ten-decomp}
    \mathcal{T} = \sum_{i=1}^{m}{{\lambda_i}u_i \otimes v_i \otimes w_i}.\tag{$*$}
\end{equation}
Can we certify that this is a rank decomposition and/or a nuclear norm decomposition?

If the components $u_i$, $v_i$ and $w_i$ are chosen randomly, it can be shown using a dimension argument that for $m\ll n^2$ the decomposition \eqref{eq:ten-decomp} is almost surely a rank decomposition and this rank decomposition is unique (up to the signs of the components). Thus, in this case, finding the decomposition \eqref{eq:ten-decomp} is almost surely equivalent to finding a rank decomposition.

However, as noted in Section \ref{tensorrankandnuclearnormintro}, it may be hard to verify that \eqref{eq:ten-decomp} is a rank decomposition for a given tensor. In the undercomplete regime with linearly independent components, Jennrich's algorithm~\cite{harshman} certifies rank decomposition and shows its uniqueness. When $m \leq \frac{3n}{2} - 1$ (or $\frac{kn}{2} - 1$ for an order $k$ tensor), Kruskal's sufficient condition~\cite{kruskal1977three} (or its recent simplification by Lovitz and Petrov~\cite{lovitz-petrov}) can be used to certify tensor rank.
However, when $m \gg n$, there is no known way to certify the tensor rank. In fact, as shown by Strassen~\cite{strassen1973vermeidung}, if we had a rank lower bound of $\Omega(n^{1+\epsilon})$ on the rank of an explicit tensor for some $\epsilon > 0$ then this would imply strong circuit lower bounds.

\subsubsection{Lower Bounds for Tensor Nuclear Norm Via Dual Certificates}
Since the tensor nuclear norm $\norm{\cdot}_*$ is dual to the injective norm $\norm{\cdot}_{\sigma}$ (see e.g. \cite{Friedland-Lim}), we can prove lower bounds on the tensor nuclear norm using what we call \emph{dual certificates}. In fact, we are not aware of any other way to lower bound the tensor nuclear norm.  The dual certificates play central role in our completion and decomposition algorithms, so before describing our algorithms we discuss dual certificates. The key fact which we use is as follows.
\begin{lemma}\label{lem:dual-cert-cond}
\eqref{eq:ten-decomp} is a nuclear norm decomposition if and only if there exists an $\mathcal{A} \in \mathbb{R}^{n^3}$ s.t. 
\begin{enumerate}
    \item[]\text{(DC1).}\ \  $\langle \mathcal{A}, u_i\otimes v_i \otimes w_i \rangle = 1$ for all $i \in [m]$
    \item[]\text{(DC2).}\ \ $|\langle \mathcal{A}, x\otimes y\otimes z \rangle| \leq 1$ for all unit length $x, y, z\in \mathbb{R}^{n}$ (or, equivalently, $\Vert \mathcal{A} \Vert_{\sigma}\leq 1$).
\end{enumerate}
\end{lemma}
\begin{proof}[Proof of the if direction:]
Given such an $\mathcal{A}$, we have $\langle \mathcal{A}, \mathcal{T} \rangle = \sum_{i=1}^{m}{\lambda_i}$. At the same time, for any other decomposition $\mathcal{T} = \sum\limits_{i=1}^{m'} \lambda'_i (u'_i\otimes v'_i \otimes w'_i)$, we have  $\langle \mathcal{A}, \mathcal{T} \rangle\leq \sum_{i=1}^{m'}{|\lambda'_i|}$.

Thus, $\norm{\mathcal{T}}_* = \sum_{i=1}^{m}{\lambda_i}$ and $\mathcal{T} = \sum\limits_{i=1}^{m} \lambda_i (u_i\otimes v_i \otimes w_i)$ is a nuclear norm decomposition.
\end{proof}
For an explanation of why tensor nuclear norm and injective norm are dual to each other and a proof of the only if direction, see Appendix~\ref{dualityappendix}.
\begin{remark}
Note that whether or not \eqref{eq:ten-decomp} is a nuclear norm decomposition does not depend on the coefficients $\lambda_i$ (as long as they are all positive)! In particular, this means that studying approximate nuclear decompositions does not really reveal the number of components or their directions for nuclear decompositions, as in principle, most of the components may have very small coefficients in front of them.
\end{remark}

Based on this lemma, we make the following definition.
\begin{definition}
We say that $\mathcal{A}$ is a \emph{dual certificate} for a nuclear norm decomposition $\mathcal{T} = \sum\limits_{i=1}^{m} \lambda_i (u_i\otimes v_i \otimes w_i)$ if conditions (DC1) and (DC2) of Lemma~\ref{lem:dual-cert-cond} hold. We say that $\mathcal{A}$ is a strong dual certificate if $u_i\otimes v_i\otimes w_i$ are the only rank one tensors satisfying (DC1).
\end{definition}
While the dual certificates we construct are exact, we need the notion of an approximate dual certificate in order to discuss prior work and the technical challenges we overcome.
\begin{definition}
We say that $\mathcal{A}$ is an approximate dual certificate for \eqref{eq:ten-decomp} if
\begin{equation}
    \langle \mathcal{A}, u_i\otimes v_i \otimes w_i \rangle \geq 1 - o(1)\quad \forall i\in [m]\quad \text{and}\quad |\langle \mathcal{A}, x\otimes y\otimes z \rangle| \leq 1 + o(1)
\end{equation}
for all unit length $x, y, z\in \mathbb{R}^{n}$ (or equivalently, $\Vert \mathcal{A} \Vert_{\sigma} \leq 1 + o(1)$).
\end{definition}

A major difficulty with the dual certificates is that for a given $\mathcal{A}$ it may be very hard to check that $\Vert \mathcal{A} \Vert_{\sigma}\leq 1$. Indeed, in general, finding even an approximate dual certificate $\mathcal{A}$ is NP-hard as estimating the nuclear norm and estimating the injective norm are both NP-hard. That said, finding $\mathcal{A}$ and checking that $\Vert \mathcal{A} \Vert_{\sigma}\leq 1$ may be more feasible for special classes of tensors $\mathcal{T}$ such as when $\mathcal{T}$ has random components.

\subsubsection{Prior Work on Dual Certificates}
For the question of whether $\mathcal{T} = \sum_{i=1}^{m}{{\lambda_i}u_i \otimes v_i \otimes w_i}$ is a nuclear norm decomposition when the components $u_i$, $v_i$ and $w_i$ are random, it can be shown that w.h.p., 
$\mathcal{A} = \sum_{i=1}^{m}{u_i \otimes v_i \otimes w_i}$ is an approximate dual certificate, which is sufficient to show that $\norm{\mathcal{T}}_* \geq (1-o(1))\sum_{i=1}^{m}{\lambda_i}$. However, this leaves two questions.
\begin{enumerate}
\item[(Q3)] \textit{First, can we find an exact dual certificate? Second, can we prove (certify) that a given $\mathcal{A}$ is an exact or approximate dual certificate?} 
\end{enumerate}
For tensors with orthogonal components, it is easy to see that $\mathcal{A} = \sum_{i=1}^{m}{u_i \otimes v_i \otimes w_i}$ is an exact dual certificate. For tensors with non-orthogonal components much less is known.

For the first question, {\cite[Lemmas 1-3]{Li-et-al}} claimed to construct a dual certificate when $m < n^{\frac{17}{16}}/\polylog(n)$, but their proof has a serious flaw\footnote{On page 23 in \cite{Li-et-al} the inequality $\Vert C^{t}-C^{t-1}\Vert\leq \eta \Vert C^{t-1}-C^{t-2}\Vert$ was proved for $t$ large enough, but on page 24 it is used for all $t\in \mathbb{N}$.} when $m>n/\polylog(n)$. 

For the second question, ~\cite{Ge-Ma} proved that in the case when $\mathcal{T}$ has $\leq n^{3/2}/\polylog(n)$ symmetric random components (i.e. $\mathcal{T} = \sum_{i=1}^{m}{{\lambda_i}u_i \otimes u_i \otimes u_i}$), w.h.p. $\mathcal{A} = \sum_{i=1}^{m}{u_i \otimes u_i \otimes u_i}$ is an approximate dual certificate and this can be certified by degree-12 SoS.

\subsubsection{Our Contribution}
In this paper, we resolve questions (Q3) for tensors with up to $n^{3/2}/\polylog(n)$ asymmetric random components. More precisely, we prove the following theorem.
\begin{theorem}[Main II, informal]\label{thm:main:dual-certificates-alg}
Given a tensor $\mathcal{T} = \sum_{i=1}^{m}{{\lambda_i}u_i \otimes v_i \otimes w_i}$ where $m \leq n^{3/2}/\polylog(n)$ and the components $u_i$, $v_i$, and $w_i$ are random unit vectors, w.h.p. there exists a dual certificate $\mathcal{A}$ for this decomposition of $\mathcal{T}$. Moreover, w.h.p. degree-4 SoS can find $\mathcal{A}$ and certify that $\mathcal{A}$ is a dual certificate.
\end{theorem}

To the best of our knowledge, prior to our work it was not even known for which overcomplete tensor decompositions there exists a (nuclear norm) dual certificate. 

Using the dual certificate, we can immediately check whether $u \otimes v \otimes w$ is a potential component of a nuclear norm decomposition of $\mathcal{T}$ by checking whether $\langle \mathcal{A}, u \otimes v \otimes w \rangle = 1$. However, it is not immediately clear how to recover the components $u_i$, $v_i$, and $w_i$ from $\mathcal{A}$. In Sections~\ref{sec:ten-decomp-intro} and ~\ref{sec:tensor-decomposition-algorithm}, we describe a novel algorithm for tensor decomposition which is based on computing dual certificates for $O(\log(n)^2)$ tensors. 

\begin{theorem}[Main III, informal]\label{thm:nuclear-norm-intro}
For order $3$ tensors with $m \leq n^{3/2}/\polylog(n)$ random asymmetric components, degree $4$ SOS w.h.p. can \emph{find and certify} the tensor nuclear norm and nuclear decomposition components.
\end{theorem}

We note that the exact recovery guaranteed by our result does not seem to follow from any prior work. Most of the prior work for decomposition of overcomplete order-3 tensors~\cite{MSS-decomp, HSSS} assumes symmetric components. Even though applications typically deal with the symmetric tensors, their results do not seem to transfer to tensors with asymmetric components. Thus it appears that the best known theoretical guarantees for asymmetric order-3 tensor decomposition are by~\cite{AGJ-tensor-decomp-alt-min}. 

Our algorithm enjoys several nice properties: it does not need to know the number of components of the tensor in advance and it does not depend on the magnitudes of the coefficients. We mention that prior algorithms  either make the assumption that the ratio $\kappa = \left|\frac{\lambda_{max}}{\lambda_{min}}\right|$ between the largest and the smallest coefficients is close to 1~\cite{MSS-decomp, HSSS}, or they have at least polynomial dependence on $\kappa$~\cite{AGJ-tensor-decomp-alt-min} of the running time.

We note that our algorithm for finding a dual certificate requires solving a large semidefinite program, and hence does not scale well with $n$. The following natural question arises

\begin{enumerate}
\item[(Q4)] \textit{Is there an algorithm to compute a dual certificate of a tensor that scales to large $n$?} 
\end{enumerate}

\subsection{Tensor Completion}\label{sec:intro:completion}
The tensor completion problem is a higher order analog of the matrix completion problem. In the matrix completion problem, we are given some but not all of the entries of some ``low-complexity'' matrix $M$ and we are asked to fill in the remaining entries. A famous example of the matrix completion problem was the Netflix Prize competition. In this competition, participants were given users' preferences for some movies and they were asked to predict those users' preferences on other movies.

A typical assumption enforcing ``low complexity'' of the matrix is the assumption of having low rank. As described in a recent survey~\cite{matrix-compl-survey}, there are many methods for low rank matrix completion, including nuclear norm minimization~\cite{fazel, candes-recht, candes-tao}, singular value thresholding~\cite{matr-compl-svt}, iteratively reweighted least squares minimization~\cite{matr-compl-iter-rew-nuc-1, matr-compl-iter-rew-nuc-2}, greedy algorithms~\cite{matr-compl-greedy-1, matr-compl-greedy-2}, alternating minimization~\cite{matr-compl-als-1, matr-compl-als-2}, and optimization over smooth Reimannian Manifold~\cite{riemann-opt-1, riemann-opt-2}. Of these methods, nuclear norm minimization is the most accurate method and has the best theoretical guarantees. If the components of the matrix are incoherent then with high probability nuclear norm minimization can reconstruct an  $n_1 \times n_2$ matrix of rank $r$ exactly from only $O(rn\log(n)^2)$ entries where $n = \max{\{n_1,n_2\}}$ \cite{candes-tao, recht}. That said, nuclear norm minimization is also the most expensive technique in terms of time and memory.

Similarly, in the tensor completion problem, we are given some but not all of the entries of some tensor $\mathcal{T}$ and we are asked to fill in the remaining entries. However, since the rank and nuclear norm of a tensor are much harder to compute than the rank and nuclear norm of a matrix, tensor completion is less well understood than matrix completion.

As a baseline, one way to solve the tensor completion problem is to flatten the tensor into a matrix and then solve the matrix completion problem. If we flatten a third order tensor $\mathcal{T}\in \mathbb{R}^{n^3}$ of rank $m$, this gives us an $n^2 \times n$ matrix of rank at most $m$. Using the results on matrix completion, if the components of $\mathcal{T}$ are incoherent then with high probability $\mathcal{T}$ can be recovered exactly from $O(n^2 m\polylog(n))$ random entries. However, this approach neglects part of the structure of the tensor and thus demands many more tensor entries than needed~\cite{mu2014square}. 

In 2016, Barak and Moitra \cite{Barak-Moitra} showed that a noisy  tensor of rank $m\leq n^{3/2}/\polylog(n)$ can be completed approximately if the number of given entries is essentially $\geq mn^{3/2}$. To do this, they used degree 6 SOS and analyzed it using Rademacher complexity. They also gave evidence that tensor completion from less than $mn^{3/2}$ samples might be computationally hard. Later, Montanari and Sun \cite{montanari-sun} proposed a spectral algorithm which has essentially the same sample complexity but in practice scales to larger tensors. A question that is left open after this line of work is ``what can be said about exact recovery?''.

\begin{center}
\begin{table}[h]
\begin{center}
 \begin{tabular}{||c | c | c | c | c | c ||} 
 \hline
 Paper & Entries  & Rank & Recovery & Method & Components \\ [0.5ex] 
 \hline\hline
 Matrix completion & $n^{2}m$ & $n$ & exact & & incoherent\\ 
 \hline
 Jain and Oh \cite{jain-oh} & $n^{3/2}m^5$ & $n^{3/10}$ & exact & alt.min. & orthogonal\\ 
 \hline
 Barak and Moitra \cite{Barak-Moitra} & $\boldsymbol{n^{3/2}m}$ & $\boldsymbol{n^{3/2}}$ & approx. & deg-6 SOS & incoherent\\
 \hline
 Xia and Yuan \cite{xia-yuan}\tablefootnote{As pointed out on p.3 of \cite{liu-moitra}, this algorithm has not been proven to run in polynomial time because there is no bound for how many iterations are required for convergence.} & $n^{3/2}m^7$ & $n^{3/14}$ & exact & spectral+GD& incoherent \\
\hline
 Potechin, Steurer \cite{potechin-steurer-exact} & $\boldsymbol{n^{3/2}m}$ & $n$ & exact & deg-4 SOS& orthogonal\\ 
 \hline
 Montanari Sun \cite{montanari-sun} & $\boldsymbol{n^{3/2}m}$ & $\boldsymbol{n^{3/2}}$ & approx. & spectral & random\\ 
 \hline
 Cai, Li, Poor, Chen \cite{CLPC-non-convex-completion} & $n^{3/2}$ & $O(1)$ & exact & spectral+GD& incoherent\\
 \hline
 Liu and Moitra \cite{liu-moitra} & $n^{3/2}m^{O(1)}$ &  $n^{\varepsilon}$ & exact & \textbf{alt.min.+} & incoherent\\ 
 \hline
 \textbf{this paper} & $\boldsymbol{n^{3/2}m}$ & $\boldsymbol{n^{3/2}}$ & \textbf{exact} & deg-4 SOS & random \\[0.3ex] 
 \hline
\end{tabular} 
\caption{\small In this table we summarize some state of the art results on tensor completion emphasizing the number of entries required and the maximal rank for which algorithm works. Polylogarithmic factors are ignored.}\label{table-compl-algo}
\end{center}
\end{table} 
\end{center}
For exact tensor completion, a variety of approaches have been proposed. Roughly speaking, these approaches can be divided into three groups: alternating minimization~\cite{jain-oh, liu-moitra}, spectral methods~\cite{xia-yuan, CLPC-non-convex-completion}, and nuclear norm minimization~\cite{potechin-steurer-exact}. The algorithms based on alternating minimization and spectral methods are fast and have low memory requirements and thus better scale to large tensors. Another benefit is that in addition to recovering a tensor they also find its components.

However, for the known algorithms based on these approaches the number of samples needed to provably reconstruct a tensor has a poor dependence on the rank $m$. In particular, it seems that neither of the known  algorithms can provably reconstruct the tensor exactly when $m>n^{3/10}$. For some of algorithms (e.g.~\cite{jain-oh}) this may be an artifact of the analysis, for other algorithms (e.g.~\cite{liu-moitra}) suboptimal dependence on $m$ seems to be by design.

Even though in practice the tensor completion problem is primarily interesting when the rank $m$ is small, for many practical aplications $m$ is moderately large when compared to $n$ (for example, in  Fig.~\ref{fig:akiyo}, good recovery quality is achieved only when $m \approx \sqrt{n}$). Hence, we argue that \emph{dependence on $m$ cannot be ignored}. Therefore, it is of theoretical and practical importance to develop algorithms with optimal dependence not only on $n$, but on $m$ as well. In fact the problem of improving the dependence of the sample complexity on the tensor rank $m$ was explicitly emphasized as an important direction already in~\cite{jain-oh}.

\begin{figure}
\begin{subfigure}[b]{0.19\textwidth}
\begin{center}
\includegraphics[width = \textwidth]{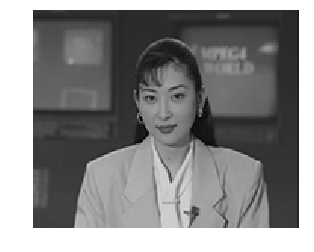}
\caption{Original frame}
\end{center}
\end{subfigure}
\begin{subfigure}[b]{0.19\textwidth}
\begin{center}
\includegraphics[width = \textwidth]{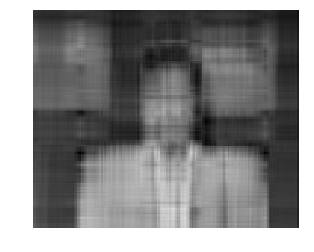}
\caption{Rank 5}
\end{center}
\end{subfigure}
\begin{subfigure}[b]{0.19\textwidth}
\begin{center}
\includegraphics[width = \textwidth]{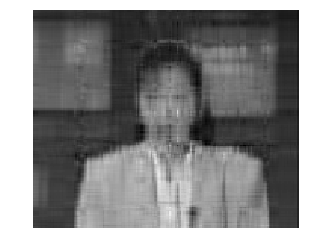}
\caption{Rank 10}
\end{center}
\end{subfigure}
\begin{subfigure}[b]{0.19\textwidth}
\begin{center}
\includegraphics[width = \textwidth]{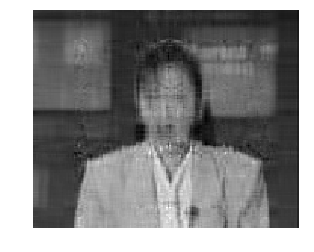}
\caption{Rank 15}
\end{center}
\end{subfigure}
\begin{subfigure}[b]{0.19\textwidth}
\begin{center}
\includegraphics[width = \textwidth]{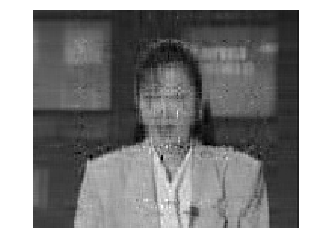}
\caption{Rank 20}
\end{center}
\end{subfigure}
\caption{\small Quality of the tensor completion, using ALS with 40\% of entries being known, depending on the rank of the solution. The original tensor is first 50 frames of ``Akiyo" video ($176\times 144$ pixels) \cite{YUV-website}. }\label{fig:akiyo}
\end{figure}

 Yuan and Zhang \cite{Yuan-Zhang} proved that nuclear norm minimization~\eqref{eq:nuclear-norm-minimization-intr} requires roughly  $m^2n+\sqrt{m}n^{3/2}$ random entries $\Omega$ to recover the tensor $\mathcal{T}$, and thus provides better sample complexity guarantees than algorithms in Table~\ref{table-compl-algo} for $m<\sqrt{n}$. Interestingly, when $m\leq n^{1/3}$, this gives $m^{1/2}$ dependence of the observed entries on $m$, which was experimentally observed by~\cite{jain-oh}. Unfortunately, for general tensors, nuclear norm minimization is known to be NP-hard. In the case of tensors with $m\leq n$ orthogonal components, Potechin and Steurer~\cite{potechin-steurer-exact} showed that degree-4 SOS can reconstruct the tensor from roughly $n^{3/2}m$ random entries.
 
\begin{Optbox}
 \begin{equation}\label{eq:nuclear-norm-minimization-intr}
  \text{Minimize:} \ \Vert X \Vert_{*}\qquad \text{Subject to:} \quad X_{\omega} = \mathcal{T}_{\omega}, \text{ for } \omega\in \Omega. 
  \end{equation}
 \end{Optbox}

In this paper, we show that when $\mathcal{T}$ has $m \leq n^{3/2}/\polylog(n)$ random asymmetric components, a simple SDP introduced by~\cite{potechin-steurer-exact} can complete $\mathcal{T}$ exactly via tensor nuclear norm minimization from roughly $n^{3/2}m$ observed entries. 

\begin{theorem}[Main IV, informal]\label{thm:tensor-completion-intro} For order $3$ tensors with $m \leq n^{3/2}/\polylog(n)$ random asymmetric components, degree $4$ SOS w.h.p. can reconstruct the tensor from $n^{3/2}m\polylog(n)$ randomly selected entries.
\end{theorem}
As a byproduct we get the first polynomial time algorithm that achieves exact recovery in the overcomplete regime. We emphasize that even though from a practical point of view the overcomplete regime may be not the most interesting regime to consider, from a theoretical point of view the fact that our algorithm is able to operate in the overcomplete regime required completely new techniques. All previous papers for exact tensor completion (except~\cite{potechin-steurer-exact}) analyzed the subspaces generated by the components in each mode $\{u_1, u_2, \ldots, u_m\}$, $\{v_1, v_2, \ldots v_m\}$, $\{w_1, w_2, \ldots, w_m\}$. In the overcomplete regime, these types of arguments are not available as all these subspaces w.h.p. coincide with $\mathbb{R}^n$.

Just as for tensor nuclear norm and tensor decomposition, our algorithm is robust to the sizes of the coefficients and provides a dual certificate that no completion with smaller nuclear norm is possible.

While we use the same algorithm as \cite{potechin-steurer-exact}, we significantly improve the analysis by making the connection to tensor nuclear norm explicit and showing that the algorithm also works for tensors with asymmetric random components, which is much less restricted than orthogonal components. As we discuss in Sections~\ref{sec:ideas},~\ref{sec:dual-certificates-sec2}, and~\ref{sec:Omega-dual-certificate} this extension is quite challenging.

We also note that our results are not covered by the state of the art results in the approximate recovery regime. While~\cite{Barak-Moitra} can handle incoherent components whereas our current analysis can only handle asymmetric random components, our algorithm makes no assumptions about the coefficients and uses degree-4 SOS instead of degree-6 (which gives a chance to use our algorithm for tensors of moderate size, e.g. with $n\approx 50\text{-}70$). Our result also does not seem to be covered by~\cite{montanari-sun} even in the approximate regime, as their algorithm works only up to rank $n^{3/4}$ for incoherent components, and for rank up to $n^{3/2}$ it assumes random symmetric components (here observe that after applying the reduction of~\cite{Barak-Moitra} from asymmetric to symmetric components one cannot assume randomness anymore).

Even though our algorithm works in polynomial time, it requires solving a semidefinite program (SDP) of size $(n^2+n)\times (n^2+n)$ with $O(n^4)$ sparse constraints, which is not feasible in practice for large $n$. We were able to run numerical experiments for $n = 70$, which compares to $n = 100$  for~\cite{CLPC-non-convex-completion} and $n =200$ (but scales to $n\sim 1000$) for very recent work of Liu and Moitra~\cite{liu-moitra}. This gives rise to the following natural question.

\begin{question} Does there exist a tensor completion algorithm that is able to recover a tensor exactly from roughly $n^{3/2}m$ entries and which is able to handle large-scale tensors? 
\end{question}

\subsection{Key ideas of the proofs and difficulties in constructing an exact dual certificate}\label{sec:ideas}

To prove Theorems~\ref{thm:main:nuclear-norm-rank-same},~\ref{thm:main:dual-certificates-alg},~\ref{thm:nuclear-norm-intro}, and~\ref{thm:tensor-completion-intro} we analyze a pair of dual semidefinite programs, designed in such a way that the primal program always has a feasible solution corresponding to the nuclear norm decomposition of the tensor. To prove our main theorems, it will be sufficient to construct a feasible solution to the dual program \eqref{eq:opt-constraints-id-intr} which certifies that this solution is optimal.

\subsubsection{Constructing a candidate dual certificate}
We recall that a dual certificate satisfies two conditions
\begin{equation}\label{eq:dual:certificate:keyideas}
\begin{split}
 & \langle \mathcal{A}, u_i\otimes v_i\otimes w_i \rangle = 1 \quad \text{for all } i\in [m]\\
 & \langle \myscr{A}, x\otimes y \otimes z\rangle \leq 1, \quad \text{for all} \quad  x, y, z \in S^{n-1}.
    \end{split}
\end{equation} 

A crucial distinction between the exact and approximate regimes is that in order to show that $\myscr{A}$ is an approximate dual certificate we just need to show that 
\begin{equation}\label{eq:dual:approxcertificate:keyideas}
    \langle\myscr{A}_{appr}, u_i\otimes v_i\otimes w_i\rangle\geq 1-\varepsilon\ \  \forall i \quad \text{and} \quad  \langle\myscr{A}_{appr}, x\otimes y\otimes z\rangle\leq 1+\varepsilon.
\end{equation}
For random components it follows from standard concentration inequalities that $\myscr{A}_{appr} = \sum\limits_{i=1}^{m} u_i\otimes v_i \otimes w_i$ satisfies these conditions. In~\cite{Ge-Ma} Ge and Ma proved that degree-12 SOS can certify these inequalities locally around $u_i\otimes u_i\otimes u_i$ in the case of random symmetric components. 

In contrast, the situation with exact dual certificates is more difficult, as we need to prove an inequality~\eqref{eq:dual:certificate:keyideas} which has $m$ equality points. It is not even clear that there should be a solution to the linear constraints in~\eqref{eq:dual:certificate:keyideas} that satisfies an approximate upper bound as in~\eqref{eq:dual:approxcertificate:keyideas}. By computing the gradient at these $m$ points we get the following linear constraints that $\mathcal{A}$ must satisfy for all $x, y, z\in S^{n-1}$ and all $i\in [m]$
\begin{equation}\label{eq:certificate-linear-constr-intro}
    \langle \mathcal{A}, x\otimes v_i\otimes w_i\rangle = \langle u_i, x\rangle,\quad  \langle \mathcal{A}, u_i\otimes y\otimes w_i\rangle = \langle v_i, y\rangle,\quad \langle \mathcal{A}, u_i\otimes v_i\otimes z\rangle = \langle w_i, z\rangle.
\end{equation}
In order to satisfy these linear constraints we search for a dual certificate of the following form
\begin{equation}\label{eq:certificate-3-form-intr}
 \oa{A}_{cand} = \sum\limits_{i=1}^{m} u_i\otimes v_i \otimes w_i+ \sum\limits_{i=1}^{m} \alpha_i\otimes v_i \otimes w_i+ u_i\otimes \beta_i \otimes w_i+ u_i\otimes v_i \otimes \gamma_i.
 \end{equation}
where we think of $\alpha_i$, $\beta_i$, and $\gamma_i$ as small corrections.
\begin{definition}\label{def:certificate-candidate} We say that $\oa{A}$ is a \textit{certificate candidate} for $\mathcal{V}$ if $\oa{A}$ satisfies conditions~\eqref{eq:certificate-linear-constr-intro} and $\oa{A}$ can be written as in Eq.~\eqref{eq:certificate-3-form-intr}.
\end{definition}

\begin{remark}
Potential dual certificates of this form were also considered by~\cite{Li-et-al}, we explain the reasons behind~\eqref{eq:certificate-3-form-intr} in Sec.~\ref{sec:dual-certificate-contsr2}. We note that having an explicit form~\eqref{eq:certificate-3-form-intr} for $\mathcal{A}$ is important for the further stages of the proofs. 
\end{remark}

\begin{remark}
One of the difficulties of constructing and analysing $\mathcal{A}_{cand}$ that satisfies constraints from Eq.~\eqref{eq:certificate-linear-constr-intro} is the fact that these constraints are linearly dependent! In particular, this complicates obtaining norm bounds for the correction terms $\alpha_i$, $\beta_i$ and $\gamma_i$. Because of this, we need to perform a careful analysis of the nullspace of the matrix that corresponds to the  linear constraints~\eqref{eq:certificate-linear-constr-intro}, as discussed in Sec.~\ref{sec:dual-certificate-contsr2} and~\ref{sec:corr-terms}. 
\end{remark}

\subsubsection{Certifying a candidate dual certificate via the dual program}
Assume that we have a candidate $\mathcal{A}$ for a dual certificate. It is straightforward to verify whether the equality constraints hold, so we only need to deal with certifying the inequality constraint. Observe that it is sufficient to replace it with the following inequality and to verify that it holds.
\begin{equation}\label{eq:nuclear-norm-squares-ineq}
     2\langle\myscr{A}, x\otimes y \otimes z\rangle \leq \Vert x\Vert^2+\Vert y\Vert^2\cdot \Vert z\Vert^2 , \quad \text{for all} \quad  x, y, z \in \mathbb{R}^{n}.
\end{equation}

To verify this inequality, we use the PSD constraint in the dual semidefinite program \eqref{eq:opt-constraints-id-intr}, which as we discuss below is stronger.

\begin{definition} We say that $Z\in \mathbb{R}^{n^2\times n^2}$ is a \textit{zero polynomial matrix} if for any $i, j, i', j'\in [n]$ 
\[Z_{(i, j)(i', j')}+Z_{(i', j)(i, j')}+Z_{(i, j')(i', j)}+Z_{(i', j')(i, j)} = 0\]
 We use the notation $Z\equiv_{poly} 0$. More generally, we say that $A\equiv_{poly} B$, if $A-B\equiv_{poly}  0$.
\end{definition}

 \begin{Optbox}
\begin{equation}\label{eq:opt-constraints-id-intr}
\begin{split}
& \text{Maximize:}\quad 2\langle \oa{A}, \mathcal{T} \rangle\\
& \text{Subject to:}\quad \left(\begin{matrix}
I_n & -A\\
-A^T & Z+I_{n^2}
\end{matrix}\right) \succeq 0, \quad Z\equiv_{poly} 0, \quad A_{i, (j, k)} =\oa{A}_{(i, j, k)}.
\end{split}
\end{equation}
\end{Optbox}

The following property explains the name ``zero-polynomial'' and why $Z$ appears in Eq.~\eqref{eq:opt-constraints-id-intr}.

\begin{proposition} $Z\in \mathbb{R}^{n^2}$ is a zero polynomial matrix if and only if for all $y, z\in \mathbb{R}^n$ the equality $(y\otimes z)^TZ(y\otimes z) = 0$ holds.
\end{proposition}

Using this proposition, if the constraints in ~\eqref{eq:opt-constraints-id-intr} are satisfied, then $\mathcal{A}$ satisfies the desired inequality~\eqref{eq:nuclear-norm-squares-ineq}.

\subsubsection{Constructing a feasible solution to the dual program}\label{sec:ideas-for-B}

It follows from the discussion above and Lemma~\ref{lem:dual-cert-cond} that if the  optimal solution of the optimization problem~\eqref{eq:opt-constraints-id-intr} is $2\norm{T}_*$ then $\mathcal{A}$ is a dual certificate for $\mathcal{T}$. Therefore, in order to prove Theorem~\ref{thm:main:dual-certificates-alg} it is sufficient to construct a feasible solution to the dual program~\eqref{eq:opt-constraints-id-intr} that achieves value $2\norm{T}_*$. To do this, we will take the candidate dual certificate $\mathcal{A}$ which we construct and find a matrix $Z$ that satisfies the PSD condition in Eq.~\eqref{eq:opt-constraints-id-intr}.

In order to prove that the matrix in Eq.~\eqref{eq:opt-constraints-id-intr} is PSD it is sufficient to find a matrix $B$ s.t. 
\begin{equation}\label{eq:conditions-on-B}
    B\preceq I_{n^2} \quad  \text{and} \quad  B\equiv_{poly} A^TA,
\end{equation}
as if we take $Z = A^TA - B$ then $Z \equiv_{poly} 0$ and
\begin{equation}\label{eq:reason-for-B}
    \left(\begin{matrix}
I_n & -A\\
-A^T & Z+I_{n^2}
\end{matrix}\right)  = \left(\begin{matrix}
I_n & -A\\
-A^T & A^TA
\end{matrix}\right) +\left(\begin{matrix}
0 & 0\\
0 & I_{n^2}-B
\end{matrix}\right) \succeq 0.
\end{equation}

However, finding a matrix $B$ which satisfies these conditions \eqref{eq:conditions-on-B} exactly is more difficult than it appears. An important reason for this is that if $\mathcal{A}$ is indeed a dual certificate then these conditions imply an additional condition on $B$. 
\begin{observation}\label{obs:kernel-of-psd}
If $\mathcal{A}$ is a dual certificate for $\mathcal{T}$ and $(\mathcal{A}, Z)$ satisfy constraints~\eqref{eq:opt-constraints-id-intr}, then 
\begin{equation}
    \left(\begin{matrix}
    u_i\\
    v_i\otimes w_i
    \end{matrix}\right)
    \in \Ker\left(\begin{matrix}
        I_n & -A\\
        -A^T & Z+I_{n^2}
        \end{matrix}\right)\quad \forall i\in [m].
\end{equation}
\end{observation}

This observation and Eq.~\eqref{eq:reason-for-B} imply that $B$ should satisfy
\begin{equation}\label{eq:equality-cond-on-B}
\mathcal{L} = \vspan\{v_i\otimes w_i \mid i\in [m]\}\ \subseteq\ \Ker(I_{n^2} - B)\ \subseteq\ \mathbb{R}^{n^2}.
\end{equation}

Now that we have described the conditions which $B$ needs to satisfy, we discuss how to construct $B$. 

We cannot take $B = A^{T}A$ because $\norm{A} \gg 1$, so $\norm{{A^T}A} \gg 1$ and Eq.~\eqref{eq:conditions-on-B} is not satisfied. Fortunately, there is a standard trick to construct a matrix that differs from a given one by a zero polynomial matrix and often has much smaller norm (see e.g.~\cite{Barak-Moitra,Ge-Ma, potechin-steurer-exact}).

\begin{definition}\label{def:tw2} For an $n^2\times n^2$ matrix $M$ define its \textit{twisted matrix} $\tw_2(M)$~as
\[ \left(\tw_2(M)\right)_{(i, j)(i', j')} = M_{(i, j')(i', j)} \quad \text{for}\quad  i, j, i', j'\in [n].\]
\end{definition}

If we apply this trick to the approximate dual certificate $\mathcal{A}_{appr} = \sum\limits_{i=1}^{m} u_i\otimes v_i \otimes w_i$ by taking $B_{appr} = \tw_2(\mathcal{A}_{appr}^T{\mathcal{A}_{appr}})$, it can be shown that with high probability $\norm{B_{appr}}$ is $1 + o(1)$ which is sufficient to certify that $\mathcal{A}_{appr}$ is an approximate dual certificate.

However, if we take $B_0 = \tw_2(A^TA)$, this doesn't quite work because condition~\eqref{eq:equality-cond-on-B} is not satisfied.
To fix this, we search for $B$ of the form $B = B_0+Z_0$, where
\begin{equation}
    B_0 = \tw_2(A^TA)\qquad \text{and}\qquad Z_0|_{\mathcal{L}} = I_{\mathcal{L}} - B_0|_{\mathcal{L}},\quad Z_0\equiv_{poly} 0.
\end{equation} 
To verify that the matrix $B$ constructed in this way satisfies conditions~\eqref{eq:conditions-on-B} we will check the following:
\begin{enumerate}
\item The desired $Z_0$ exists w.h.p.\quad  (see Section~\ref{sec:zero-poly-corr-exists}).
\item $\Vert B_0 - P_{\mathcal{L}} \Vert = \wt{O}\left(m/n^{3/2}\right)$, where $P_{\mathcal{L}}$ is a projector on $\mathcal{L}$ \quad (see Section~\ref{sec:B0}).
\item $\Vert Z_0\Vert = \wt{O}\left(m/n^{3/2}\right)$ \quad (see Section~\ref{sec:zero-poly-corr-norm}).
\end{enumerate}
Indeed, this implies that $B|_{\mathcal{L}} = I_{\mathcal{L}}$ and $\Vert B|_{\mathcal{L}^{\perp}}\Vert = \wt{O}\left(m/n^{3/2}\right)$. Hence, $B\preceq I_{n^2}$.

We note that proving an existence of $Z_0$ is a special case of a question that is interesting on its own right. Which linear maps can be realized by zero polynomial matrices?

\begin{enumerate}
\item[(Q5)] \textit{Given a a subspace $\mathcal{X}\in \mathbb{R}^{n^2}$ and a linear map $M: \mathcal{X} \rightarrow \mathbb{R}^{n^2}$, when does there exist a zero polynomial matrix $Z$ such that $Z|_{\mathcal{X}} = M|_{\mathcal{X}}$?} 
\end{enumerate}

Zero polynomial matrices are an important tool in the SOS hierarchy, so we expect that getting better understanding of the question above might help to attack other problems.

We answer this question in a special case that covers our needs.

\begin{theorem}\label{thm:zero-poly-correction-intro}
Assume $m\leq n^2/polylog(n)$. Let $D\in M_{n^2}(\mathbb{R})$ be a symmetric matrix s.t.
\begin{equation}\label{eq:local-zero-poly-cond-intr}
(x\otimes w_i)^T D (v_i\otimes w_i) = 0 \quad \text{and} \quad (v_i\otimes x)^T D (v_i\otimes w_i) = 0\quad \text{for all } x\in \mathbb{R}^n,\  i\in [m].
\end{equation}
Then w.h.p. over the randomness of $\mathcal{V}$, there exists a symmetric matrix $Z\in M_{n^2}(\mathbb{R})$ s.t. 
\[ Z \equiv_{poly} 0 \qquad \text{and} \qquad Z(v_i\otimes w_i) = D(v_i\otimes w_i) \quad \forall i\in [m]. \]
\end{theorem}

Our proof to this theorem is constructive which allows us to get desired norm bounds on $Z_0$, as discussed above.

\begin{definition}
For a dual certificate $\mathcal{A}$, we call a triple $(\mathcal{A}, B, Z)$ that satisfies conditions~\eqref{eq:conditions-on-B}-\eqref{eq:equality-cond-on-B} an \emph{SOS-dual certificate}.
\end{definition}

\subsubsection{Tensor completion}

The importance of the dual certificates for the tensor completion problem can be seen through the following statement used by both \cite{Yuan-Zhang} and \cite{potechin-steurer-exact}.

\begin{definition} We say that $\oa{A}_{\Omega}$ is an $\Omega$-restricted dual certificate for a tensor $\mathcal{T}$ (a collection $\mathcal{V}$), if $\oa{A}_{\Omega}$ is a dual certificate for $\mathcal{T}$ ($\mathcal{V}$, respectively) and $(\oa{A}_{\Omega})_{\omega} = 0$ for $\omega\notin \Omega$.
\end{definition}

\begin{observation}[\cite{Yuan-Zhang, potechin-steurer-exact}] $X= \mathcal{T}$ is a solution to problem~\eqref{eq:nuclear-norm-minimization-intr} if there exists an $\Omega$-restricted dual certificate $\oa{A}_{\Omega}$ for $\mathcal{T}$. Moreover, if the restricted vectors $\{(u_i\otimes v_i\otimes w_i)_{\Omega}\in \mathbb{R}^{\Omega} \mid i\in [m]\}$ are linearly independent, then $X = \mathcal{T}$ is the unique solution to problem~\eqref{eq:nuclear-norm-minimization-intr}.
\end{observation} 

The key observation made by Potechin and Steurer~\cite{potechin-steurer-exact} is that one can recover $\mathcal{T}$ from an optimal solution to a semidefinite program~\eqref{eq:tensor-compl-primal} (which we present in Section~\ref{sec:semidefinite-completion}) if and only if there exists a feasible solution ($\mathcal{A}_{\Omega}$, $Z_{\Omega}$) of the program~\eqref{eq:opt-constraints-id-intr}, with $\mathcal{A}_{\Omega}$ being an $\Omega$-restricted dual certificate. 

We can use an approach similar to  Section~\ref{sec:ideas-for-B} for constructing $Z_{\Omega}$ once we have $\mathcal{A}_{\Omega}$. To construct an $\Omega$-restricted dual certificate $\mathcal{A}_{\Omega}$, we use quantum golfing technique, which was also used by~\cite{Yuan-Zhang, potechin-steurer-exact}. In this technique one starts with a dual certificate $\mathcal{A}$ for $\mathcal{T}$ with no missing entries and alternatively projects $\mathcal{A}$ on subspaces where conditions~\eqref{eq:certificate-linear-constr-intro} are satisfied and where entries outside of $\Omega$ are zero (see Sec.~\ref{sec:constr-Omega-certificate-intro},~\ref{sec:Omega-dual-certificate}). 

While we are using the same approach as~\cite{potechin-steurer-exact}, apart from quantum golfing, our analysis is different because we have a different starting point. For tensors with orthogonal components, there is a simple (non-restricted) dual certificate $\mathcal{A} = \sum\limits_{i=1}^{n} u_i\otimes v_i\otimes w_i$. Similarly, for orthogonal tensors, the matrix $B$ from Eq.~\eqref{eq:conditions-on-B} for this dual certificate has simple form. In contrast, in our setup a dual certificate $\mathcal{A}$ and the corresponding matrix $B$, to which quantum golfing is applied, are already complicated. Thus we require a more intricate analysis for the later steps.

\subsubsection{Tensor decomposition algorithm}\label{sec:ten-decomp-intro}

We now sketch the main idea behind our tensor decomposition algorithm based on the dual certificates. We are going to extract components of the tensor $\mathcal{T}$ using Observation~\ref{obs:kernel-of-psd}. Let $m\leq n^{3/2}/\polylog(n)$. Assume that we have an access to the subspace
\begin{equation}
    \mathcal{S}_{uvw} = \vspan\{u_i\otimes v_i\otimes w_i \mid i\in [m]\}\subset \mathbb{R}^{n^3}
\end{equation}
and let $\mathcal{T} = \sum\limits_{i=1}^{n}\lambda_i u_i\otimes v_i\otimes w_i$ and $\mathcal{T}' = \sum\limits_{i=1}^{n}\lambda'_i u_i\otimes v_i\otimes w_i$ be a pair of tensors in $\mathcal{S}_{uvw}$. Assume that all of the $\lambda_i$ are positive and that $\mathcal{T}'$ is sampled randomly from $\mathcal{S}_{uvw}$. Then we expect that w.h.p. approximately half of the $\lambda_i'$ are positive. Moreover, by Theorem~\ref{thm:main:nuclear-norm-rank-same}, for random $u_i, v_i, w_i$, w.h.p. both $\mathcal{T}$ and $\mathcal{T}'$ are written in their nuclear norm decompositions.

Let $\mathcal{A}$ and $\mathcal{A}'$ be dual certificates for $\mathcal{T}$ and $\mathcal{T}'$ and let $Z, Z'$ be matrices such that the conditions from~\eqref{eq:opt-constraints-id-intr} are satisfied. Then
\begin{equation}\label{eq:decomp-subspace-split}
    \left(\begin{matrix}
    u_i\\
    v_i\otimes w_i
    \end{matrix}\right)
    \in \Ker\left(\begin{matrix}
        2I_n & -A-A'\\
        -A^T-(A')^T & Z'+Z+2I_{n^2}
        \end{matrix}\right)\quad \text{iff}\quad  \lambda'_i>0, \text{ and } \lambda_i>0.
\end{equation}
In Section~\ref{sec:tensor-decomposition-algorithm} we will argue that wlog we may assume that the kernel of this matrix is precisely the span of these vectors. This means that by sampling one random $\mathcal{T}'\in \mathcal{S}_{uvw}$ and considering the null space of the matrix in Eq.~\eqref{eq:decomp-subspace-split} we are able to split the space $\mathcal{L} = \vspan\{v_i\otimes w_i \mid i\in [m]\} $ into two subspaces of approximately equal dimension. Repeating this for $O(\log(m)^2)$ randomly sampled tensors from $\mathcal{S}_{uvw}$ w.h.p. we will be able to split $\mathcal{L}$ into the individual components $v_i \otimes w_i$, which allows us to recover the entire components $u_i \otimes v_i\otimes w_i$. 

We explain the details of the algorithm and how to find the subspace $\mathcal{S}_{uvw}$ in Section~\ref{sec:tensor-decomposition-algorithm}.

\subsubsection{Main technical challenges and our solution}

As one can see, our  constructions are sequential, namely the next object we construct depends on objects we constructed before. This can be summarized by the diagram below.

\begin{equation}\label{diagram:difficulty}
\begin{tikzcd}[arrows={-Stealth}]
 \oa{A} \text{ (Sec.~\ref{sec:corr-terms},~\ref{sec:dual-certificate})}\rar\dar & B_0 \text{ (Sec.~\ref{sec:B0})}\rar\dar & Z_0 \text{ (Sec.~\ref{sec:zero-poly-corr})}\dar \\%
\oa{A}_{\Omega} \text{ (Sec.~\ref{sec:Omega-dual-certificate})}\rar & B_{\Omega} \text{ (Sec.~\ref{sec:SOS-Omega-certificate})} \rar & Z_{\Omega} \text{ (Sec.~\ref{sec:SOS-Omega-certificate})}
\end{tikzcd}
\end{equation}

Typically, a matrix $X$ which is a part of the construction of a feasible solution will need to satisfy a combination of linear constraints and a spectral norm bound. To satisfy the linear constraints, our construction for $X$ will frequently use the following framework:
\begin{itemize}
\item Solve a linear equation $Ay = b$ for $y\in \mathbb{R}^{k\cdot t}$
\item Rearrange $y$ into a $(k\times t)$-dimensional matrix $Y$
\item Multiply $Y$ by some matrices or rearrange the entries in a certain nice way (for example, if $Y$ is an $n^2\times n^2$ matrix we may consider the matrix $Y'_{(a, b), (a', b')} = Y_{(a, b'), (a', b)}$). Apply these operations multiple times in any order.
\end{itemize}
We then need to bound the spectral norm of the resulting matrix $X$. For such operations it is usually not very hard to get good Frobenius norm bounds for $X$ in terms of the norm of $b$, since any entry rearrangement or matrix reshaping has no effect on the Frobenius norm.  

However, to the best of our knowledge there is no general technique which would allow one to get tight spectral norm bounds for matrix $X$ in terms of the initial vector $b$.  

Note that in contrast to approximate setup, or the case of orthogonal components, already $\mathcal{A}$ involves some nontrivial correction terms $\alpha_i, \beta_i, \gamma_i$. Every transition in the diagram~\eqref{diagram:difficulty} induces a new level of corrections that builds on top of previous ones. Therefore, on every step of our proof we need to prove sufficiently tight norm bounds and have sufficiently explicit understanding of the corrections constructed on the previous steps. At the final steps of the proof we need to keep track and prove norm bounds for $\sim 1000$ correction terms, which is completely infeasible without an appropriate technique.  

To overcome this technical challenge, we introduce a class of \textit{inner product (IP) graph matrices} inspired by the \textit{graph matrices} studied in~\cite{planted-clique-graph-matrices, graph-matrix-bounds,  cai-potechin}. We show that all the matrices involved in our constructions can be approximated by IP graph matrices up to an error with a small Frobenious norm. Each IP graph matrix corresponds to a colored diagram. We show that all the operations sketched above correspond to some simple transformations of the corresponding diagrams. Norm bounds using graph matrices and tensor networks usually require \emph{significant case analysis}, which is by far infeasible in our problem. The crucial benefit of IP graph matrices is the fact that simple combinatorial properties of the diagrams, such as color-connectivity, imply norm bounds for the corresponding matrices, and so \emph{almost no case analysis is needed}. Furthermore, instead of keeping track of $\sim 1000$ correction terms, we just need to \emph{keep track of simple combinatorial properties} of their diagrams and verify that applied transformations preserve these properties.

\section*{Notation}

We say that a sequence of random events $A_n$ happens \textit{with high probability} (w.h.p.) if for all $n$ we have  $\mathbb{P}(A_n) = 1 -  n^{-\omega(1)}$. Note that the the intersection of polynomially many events that happen with high probability itself happens with high probability.

We say that ``$A_n$ holds if $g(n)\ll f(n)$'' if there exists a polynomial $p(x)$ such that $A_n$ is true for any $g(n)<f(n)/p(\log(n))$. We use the  notation $f(n) = \wt{O}(g(n))$ to say that $f(n)\leq g(n)p(\log(n))$ with high probability for some polynomial $p(x)$.

For a subspace $\mathcal{X}$, we denote the projector onto subspace $\mathcal{X}$ by $P_{\mathcal{X}}$. In all other cases, for an operator $M$, we define $M_{\mathcal{X}}$ to be the restriction of $M$ to the subspace $\mathcal{X}$ (of the domain of $M$). We use $S^{n-1}$ to denote the $(n-1)$-dimensional sphere in $\mathbb{R}^n$  

For the notation $\tw_2(M)$, $X\cten Y$, $\mathcal{B}(\cdot, \cdot)$ see Def.~\ref{def:tw2},~\ref{def:cten} and~\ref{def:B0-construction}.

\section{Our algorithms and details omitted in Section~\ref{sec:ideas}}\label{sec:setup-overview}

In this section, we describe the primal and dual semidefinite programs for tensor nuclear norm minimization and tensor completion, describe in detail our tensor decomposition algorithm, and provide more details behind the constructions of a dual certificate $\mathcal{A}$ and $\Omega$-restricted dual certificate $\mathcal{A}_{\Omega}$. We point the reader to Appendix~\ref{sec:sumofsquares-view} for an explanation of the sum of squares hierarchy and how the semidefinite programs we consider naturally arise from applying the sum of squares hierarchy to a reformulation of the tensor nuclear norm problem.

From now on, we concentrate on order-3 tensors. For the discussion of tensors of higher order, see Appendix~\ref{sec:higher-order-tensors}.
\subsection{Semidefinite Program for Tensor Nuclear Norm}\label{sec:nuclear-norm-programs}


Like Potechin and Steurer \cite{potechin-steurer-exact} (see the alternative description of Algorithm 4.1 and the discussion near the top of p.17), we use the following primal and dual semidefinite programs.

 \begin{Optbox}
 \begin{enumerate}
\item[] \textbf{Primal:} Minimize\quad $\Tr\left(\begin{matrix}
M_{U} & T\\
T^T & M_{VW}
\end{matrix}\right)$

\quad\quad\quad\quad Subject to\quad\quad $\left(\begin{matrix}
M_{U} & T\\
T^T & M_{VW}
\end{matrix}\right) \succeq 0$, $\forall i,j,k \in [n], T_{i,(j,k)} = \mathcal{T}_{(i,j,k)}$
\[
\forall j,k,j',k' \in [n],\quad  (M_{VW})_{jkj'k'} =  (M_{VW})_{j'kjk'} = (M_{VW})_{jk'j'k} =  (M_{VW})_{j'k'jk}
\]
\item[] \textbf{Dual:}\quad  Maximize\quad $2\langle \oa{A}, \mathcal{T} \rangle$
\begin{equation}\label{eq:opt-constraints-id-sec2}
\text{Subject to:}\quad\quad\quad \left(\begin{matrix}
I_n & -A\\
-A^T & Z+I_{n^2}
\end{matrix}\right) \succeq 0, \quad Z\equiv_{poly} 0, \quad A_{i, (j, k)} =\oa{A}_{(i, j, k)}.
\end{equation}
\end{enumerate}
\end{Optbox}
We have already discussed the dual in Section~\ref{sec:ideas}. To see the connections between the primal program, the nuclear norm of a tensor, and the dual program we make the following observations.
\begin{observation}\label{obs:primal-optval-bound}
The optimal value of the primal program is at most $2\Vert \mathcal{T} \Vert_*$. 
\end{observation}
\begin{proof}
Let $\mathcal{T} = \sum_{i=1}^{m}{\lambda_i(u_i \otimes v_i \otimes w_i)}$ be a nuclear norm decomposition of $\mathcal{T}$ and observe that $\sum_{i=1}^{m}{\lambda_i\left(\begin{matrix}
u_i \\
v_i \otimes w_i
\end{matrix}\right)\left(\begin{matrix}
u_i^T & (v_i \otimes w_i)^T
\end{matrix}\right)}$ is a feasible solution to this semidefinite program with objective value $2\sum_{i=1}^{m}{\lambda_i} = 2\norm{\mathcal{T}}_{*}$.
\end{proof}

\begin{proposition}
For any feasible solution to the dual SDP, for the corresponding $\mathcal{A}$, for any unit vectors $x$, $y$, and $z$, $\langle \oa{A}, x \otimes y \otimes z \rangle \leq 1$
\end{proposition}
\begin{proof}
Observe that for any unit vectors $x$, $y$, and $z$, since ${(y \otimes z)^T}Z(y \otimes z) = 0$
\[
0 \leq \left(\begin{matrix}
x \\
y \otimes z 
\end{matrix}\right)^T\left(\begin{matrix}
I_n & -A\\
-A^T & Z+I_{n^2}
\end{matrix}\right)\left(\begin{matrix}
x \\
y \otimes z 
\end{matrix}\right) = 2 - 2\langle \oa{A}, x \otimes y \otimes z \rangle
\]
so $\langle \oa{A}, x \otimes y \otimes z \rangle \leq 1$, as needed.
\end{proof}
As discussed in Lemma~\ref{lem:dual-cert-cond}, since $\langle \oa{A}, x \otimes y \otimes z \rangle \leq 1$ for all unit vectors $x$, $y$, and $z$, we must have that $\norm{\mathcal{T}}_{*} \geq \langle \oa{A}, \mathcal{T} \rangle$. Thus, both the primal and dual semidefinite programs give a lower bound on $\norm{\mathcal{T}}_{*}$. However, these observations are not sufficient to show that these two semidefinite programs are indeed dual to each other. We now give a direct proof that weak duality holds (i.e. the value of the primal is at least as large as the value of the dual), which is the direction that we need. For an explanation for why these semidefinite programs are dual to each other and why strong duality holds (i.e. the primal and dual always have the same value), see Appendix~\ref{dualityappendix}.
\begin{definition}
For a pair of matrices $X, Y\in \mathbb{R}^{n_1\times n_2}$ define $X\bullet Y = \sum\limits_{i=1}^{n_1}\sum\limits_{j=1}^{n_2} X_{ij}Y_{ij}$.
\end{definition}

\begin{proposition}[Weak Duality]\label{prop:weak-duality-nuclear-norm}
For any feasible primal solution $\left(\begin{matrix}
M_{U} & T\\
T^T & M_{VW}
\end{matrix}\right)$ and any feasible dual solution $ \left(\begin{matrix}
I_n & -A\\
-A^T & Z+I_{n^2}
\end{matrix}\right)$, we have $\Tr\left(\begin{matrix}
M_{U} & T\\
T^T & M_{VW}
\end{matrix}\right) \geq 2\langle \oa{A}, \mathcal{T} \rangle$.
\end{proposition}
\begin{proof}
Observe that by Schur's inequality,
\begin{equation}\label{eq:primal-dual-ineq-schur}
\left(\begin{matrix}
M_{U} & T\\
T^T & M_{VW}
\end{matrix}\right) \bullet \left(\begin{matrix}
I_n & -A\\
-A^T & Z+I_{n^2}
\end{matrix}\right) = tr\left(\begin{matrix}
M_{U} & T\\
T^T & M_{VW}
\end{matrix}\right) - 2\langle \oa{A}, \mathcal{T} \rangle \geq 0.
\end{equation}
\end{proof}
Together, these observations imply that if there is a solution to the dual program with value $2\norm{\mathcal{T}}_{*}$ then both the primal and the dual program correctly compute the nuclear norm $\norm{\mathcal{T}}_{*}$. In this paper, we show that for tensors with asymmetric random components, this is indeed true with high probability.

\begin{theorem}\label{thm:main-nuclear-norm-minimizationver2}
Let $\mathcal{T} = \sum_{i=1}^{m}{\lambda_i(u_i \otimes v_i \otimes w_i)}$ be a tensor with $m \ll n^{3/2}$ random asymmetric components. Then w.h.p., the dual program \eqref{eq:opt-constraints-id-sec2}
has a solution which has value $2\norm{\mathcal{T}}_{*}$. Moreover, for this solution, $\mathcal{A}$ is a strong dual certificate and the nullspace of the matrix $\left(\begin{matrix}
I_n & -A\\
-A^T & Z+I_{n^2}
\end{matrix}\right)$ is $\vspan{\left\{\left(\begin{matrix}
u_i \\
v_i \otimes w_i
\end{matrix}\right)\right\}}$.
\end{theorem}

\subsection{Semidefinite program for exact tensor completion}\label{sec:semidefinite-completion}

Given a tensor where only some of the entries are known, a natural way to try and complete it is to fill in the missing entries so that the nuclear norm is minimized. More precisely, for a given tensor $\mathcal{T}$ and a random set $\Omega$ of known entries, we consider the following problem
 \begin{Optbox}
 \begin{equation}\label{eq:nuclear-norm-minimization}
  \text{Minimize:} \ \Vert X \Vert_{*}\qquad \text{Subject to:} \quad X_{\omega} = \mathcal{T}_{\omega}, \text{ for } \omega\in \Omega. 
  \end{equation}
 \end{Optbox}
  
Yuan and Zhang \cite{Yuan-Zhang} analyzed this method and proved that w.h.p. $X = \mathcal{T}$ is the unique solution to problem~\eqref{eq:nuclear-norm-minimization} if a tensor $\mathcal{T}$ with Tucker rank $m$ satisfies certain low-coherence assumptions and $|\Omega|\gg m^2n+\sqrt{m}n^{3/2}$. However, they did not provide any algorithm to solve the nuclear minimization problem and in general it is known to be NP-hard~\cite{Friedland-Lim-nuclear-approx}.

Like Potechin and Steurer \cite{potechin-steurer-exact}, we consider the following primal and dual semidefinite programs which are a relaxation of this problem \eqref{eq:nuclear-norm-minimization}. Note that these programs are the same as the primal and dual programs for tensor nuclear norm except that in the primal program we only require that $X_{i, (j, k)} = \mathcal{T}_{(i, j, k)}$ for $(i,j,k) \in \Omega$ and in the dual program we require that $\oa{A}_{(i, j, k)} = 0$ whenever $(i,j,k) \notin \Omega$.

\begin{Optbox}
\begin{equation}\label{eq:tensor-compl-primal}
\begin{split}
 &\text{Minimize:}\quad  {\Tr\left(\begin{matrix}
M_{U} & X\\
X^T & M_{VW}
\end{matrix}\right).} \\
 &\text{Subject to:}\quad  
 {\left(\begin{matrix}
M_{U} & X\\
X^T & M_{VW}
\end{matrix}\right) \succeq 0}, \quad 
 \forall (i,j,k) \in \Omega,\quad  X_{i, (j, k)} = \mathcal{T}_{(i, j, k)} \\
&\ \quad  \forall j,k,j',k' \in [n], (M_{VW})_{jkj'k'} =  (M_{VW})_{j'kjk'} = (M_{VW})_{jk'j'k} =  (M_{VW})_{j'k'jk}
 \end{split}
\end{equation}
\end{Optbox}


\begin{Optbox}
\vspace{-0.5cm}
\begin{flalign}\label{eq:opt-objective-compl}
\ \quad \text{Maximize }\quad 2\langle \oa{A}, \mathcal{T} \rangle &&
\end{flalign}
\vspace{-0.7cm}
\begin{equation}\label{eq:opt-constraints-compl}
\begin{gathered}
\text{Subject to:}\quad \left(\begin{matrix}
I_n & -A\\
-A^T & Z+B
\end{matrix}\right) \succeq 0, \quad Z\equiv_{poly} 0, \quad B\preceq I_{n^2}, \quad A_{i, (j, k)} = \oa{A}_{(i, j, k)}. \\
\oa{A}_{(i, j, k)} = 0 \quad \text{for all } (i, j, k)\notin \Omega
\end{gathered}
\end{equation}
\end{Optbox}

For tensors $\mathcal{T}$ that admit an orthogonal $\mu$-incoherent decomposition of rank $m$, Potechin and Steurer~\cite{potechin-steurer-exact} showed that if $|\Omega|\gg mn^{3/2}\mu^{O(1)}$, then w.h.p. there is a feasible solution to the dual program where $\mathcal{A}$ is a strong dual certificate and $Ker\left(\begin{matrix}
I_n & -A\\
-A^T & Z+B
\end{matrix}\right) = \vspan{\{(u_i,v_i \otimes w_i)\}}$. This implies\footnote{under the assumption that the vectors $\{(u_i \otimes v_j \otimes w_j): i,j \in [m]\}$ are linearly independent} that $X = \mathcal{T}$  is the unique optimal solution to the primal program \eqref{eq:opt-constraints-compl} and thus solving the primal program solves the problem~\eqref{eq:nuclear-norm-minimization} and recovers the tensor $\mathcal{T}$.

We prove an analogous result for tensors with asymmetric random components. 
In particular, in Sections ~\ref{sec:Omega-dual-certificate} and \ref{sec:SOS-Omega-certificate} we prove a stronger version of Theorem~\ref{thm:main-nuclear-norm-minimizationver2} by constructing a feasible solution where $\oa{A}$ is an $\Omega$-restricted strong dual certificate if $\Omega$ is a uniformly random set of entries of size $N\gg mn^{3/2}$.   

\begin{theorem}\label{thm:main-sos-dual-certificate} Let $m \ll n^{3/2}$ and $N\gg mn^{3/2}$. Then w.h.p. over the randomness of $\mathcal{V}$ and $\Omega$, there exists a solution $(\oa{A}_{\Omega}, B_{\Omega}, Z_{\Omega})$ to \eqref{eq:opt-objective-compl}-\eqref{eq:opt-constraints-compl}, where $\oa{A}_{\Omega}$ is an $\Omega$-restricted strong dual certificate for $\mathcal{V}$. 
\end{theorem} 

\begin{corollary} Let $m \ll n^{3/2}$ and $N\gg mn^{3/2}$. Assume that $\{(u_i\otimes v_i\otimes w_i)_{\Omega}\}$ restricted to entries in $\Omega$ are linearly independent. Then w.h.p. in any optimal solution to SDP~\eqref{eq:tensor-compl-primal},  $X = T$, where $T_{i, (j, k)} = \mathcal{T}_{(i, j, k)}$ for all $(i, j, k)\in [n]^3$.
\end{corollary}
\begin{proof}
The triple $(\oa{A}_{\Omega}, B_{\Omega}, Z_{\Omega})$ from Theorem~\ref{thm:main-sos-dual-certificate} is a feasible solution to~\eqref{eq:opt-objective-compl}-\eqref{eq:opt-constraints-compl}. Moreover, $\langle \oa{A}_{\Omega}, \mathcal{T}\rangle = \Vert \mathcal{T} \Vert_*$. Thus, by weak duality and by the argument in Observation~\ref{obs:primal-optval-bound}, the optimal value of~\eqref{eq:tensor-compl-primal} is $2 \Vert \mathcal{T} \Vert_*$. Moreover, since $\oa{A}_{\Omega}$ is a strong dual certificate, by Eq.~\eqref{eq:primal-dual-ineq-schur}, the value $2 \Vert \mathcal{T} \Vert_*$ is only achieved if $X\in \vspan\{u_i\otimes v_i\otimes w_i \mid i\in [m]\}$. Thus $X = T$. 
\end{proof}

\begin{definition} A triple $(\oa{A}, B, Z)$ which satisfies constraints \eqref{eq:opt-constraints-compl} is called an ($\Omega$-restricted) \textit{SOS dual certificate} for $\mathcal{V}$ if $\oa{A}$ is an ($\Omega$-restricted) dual certificate for $\mathcal{V}$.
\end{definition}

\subsection{Exact tensor decomposition algorithm}\label{sec:tensor-decomposition-algorithm}

Let $\mathcal{T}$ be a tensor with $m\ll n^{3/2}$ random asymmetric components. For our tensor decomposition algorithm, we assume that the null space of the optimal solution of~\eqref{eq:opt-constraints-id-sec2} satisfies
\begin{equation}\label{eq:kernel-assumption}
\Ker\left(\begin{matrix}
I_n & -A\\
-A^T & Z+I_{n^2}
\end{matrix}\right) = \vspan{\{(u_i,v_i \otimes w_i)\}}.
\end{equation} Experimentally, this assumption holds for an optimal solution of~\eqref{eq:opt-constraints-id-sec2}. Futhermore, as we explain in Appendix~\ref{sec:decomposition:missing} this condition may be ensured by running a variant of the nuclear norm SDP at most $O(\log(n))$ times.

As was mentioned in Section~\ref{sec:ten-decomp-intro}, our algorithm assumes access to the subspace $\mathcal{S}_{uvw} = \vspan\{u_i \otimes v_i \otimes w_i\mid i\in[m]\}$. Before explaining the actual decomposition algorithm, let us show how this space can be computed.

Observe that using Eq.~\eqref{eq:kernel-assumption}, by considering the last $n^2$ coordinates of the nullspace we can extract $\mathcal{L} = \vspan\{v_i\otimes w_i\mid i\in[m]\}$. In a similar way, by considering an analog of SDP~\eqref{eq:opt-constraints-id-sec2} that corresponds to reshaping $\mathcal{T}$ with respect to $w$ components instead of $u$ components we may compute $\vspan\{u_i\otimes v_i\mid i\in[m]\}$. Finally, in Corollary~\ref{cor:subspace-inter} we prove that w.h.p.
\begin{equation} \mathcal{S}_{uvw} = \left(I_n\otimes \vspan\{v_i\otimes w_i\mid i\in [m]\}\right) \cap   \left(\vspan\{u_i\otimes v_i\mid i\in [m]\}\otimes I_n\right).
\end{equation}
Therefore, w.h.p. we can compute $\mathcal{S}_{uvw}$.

Our tensor decomposition algorithm is as follows:
\begin{enumerate}
\item Find $\mathcal{S}_{u,v,w} = \vspan\{u_i \otimes v_i \otimes w_i\}$ as described above.
\item Take $k = O(\log (n)^2)$ random i.i.d. tensors $\mathcal{T}_1, \ldots, \mathcal{T}_k$ uniformly on a unit sphere in $\mathcal{S}_{uvw}$.
\item Solve the SDP~\eqref{eq:opt-constraints-id-sec2} for $\mathcal{T}_1, \ldots, \mathcal{T}_k$ to obtain $k$ SOS dual certificates $\left(\begin{matrix}
I_n & -A_k\\
-A_k^T & Z_k+I_{n^2}
\end{matrix}\right)$.
\item Use divide-and-conquer approach to find individual components $u_i\otimes v_i\otimes w_i$ in the following way.

Consider the nullspace of $\left(\left(\begin{matrix}
I_n & -A\\
-A^T & Z+I_{n^2}
\end{matrix}\right) + \sum_{j=1}^{k'}{\left(\begin{matrix}
I_n & -{b_j}A_j\\
-{b_j}A_j^T & Z_j+I_{n^2}
\end{matrix}\right)}\right)$ for some $k' \in \{0,1,\ldots,k\}$ and signs $b_1,\ldots,b_{k'} \in \{-1,1\}$. If this nullspace has dimension $1$, then it will be $\vspan\{u_i \otimes v_i \otimes w_i\}$ for some $i \in [m]$. We describe how to choose $k'$ and the signs $b_1,\ldots,b_{k'} \in \{-1,1\}$ below.
\end{enumerate}
To see why this works, we show the following lemma. Let 
\begin{equation}
    \mathcal{T}_j = \sum_{i=1}^{m}{c_{ij}(u_i \otimes v_i \otimes w_i)}\quad \text{for } j\in [k],\ c_{ij}\in \mathbb{R}
\end{equation}
\begin{lemma}
Let $\mathcal{T}$ be a tensor with $m\ll n^{3/2}$ asymmetric components. Let $\mathcal{T}_i$ be constructed as explained above. Then w.h.p. the nullspace  
\begin{equation}\label{eq:nullspace-scary}
\mathcal{X} = \Ker\left(\left(\begin{matrix}
I_n & -A\\
-A^T & Z+I_{n^2}
\end{matrix}\right) + \sum_{j=1}^{k'}{\left(\begin{matrix}
I_n & -{b_j}A_j\\
-{b_j}A_j^T & Z_j+I_{n^2}
\end{matrix}\right)}\right)
\end{equation}
is $\vspan\{(u_i,v_i \otimes w_i): \forall j \in [k'], {b_j}sign(c_{ij}) = 1\}$.
\end{lemma}
\begin{proof}
 As discussed above, w.h.p. the nullspace of the SOS dual certificate $\left(\begin{matrix}
I_n & -{b_j}A_j\\
-{b_j}A_j^T & Z_j+I_{n^2}
\end{matrix}\right)$ is equal to $\vspan{\{({b_j}sign(c_{ij})u_i,v_i \otimes w_i)\}}$. We now observe that for all $i$, 
\[
(-u_i, v_i \otimes w_i)^T\left(\begin{matrix}
I_n & -A\\
-A^T & Z+I_{n^2}
\end{matrix}\right)(-u_i, v_i \otimes w_i) = 4
\] 
and $\forall j \in [k'], \left(\begin{matrix}
I_n & -{b_j}A_j\\
-{b_j}A_j^T & Z_j+I_{n^2}
\end{matrix}\right) \succeq 0$ so the nullspace of $\mathcal{X}$
must be $span\{(u_i,v_i \otimes w_i): \forall j \in [k'], {b_j}sign(c_{ij}) = 1\}$.
\end{proof}
Using this, we can find the components as follows. If the nullspace $\mathcal{X}$ as in Eq.~\eqref{eq:nullspace-scary}
has dimension more than $1$, we increment $k'$ by $1$ and take both signs for $b_{k'+1}$. This splits the nullspace $span\{(u_i,v_i \otimes w_i): \forall j \in [k'], {b_j}sign(c_{ij}) = 1\}$ into two subspaces $span\{(u_i,v_i \otimes w_i): \forall j \in [k'], {b_j}sign(c_{ij}) = 1 \wedge sign(c_{i(k'+1)}) = 1\}$ and $span\{(u_i,v_i \otimes w_i): \forall j \in [k'], {b_j}sign(c_{ij}) = 1 \wedge sign(c_{i(k'+1)}) = -1\}$ and we can repeat this process until we obtain subspaces with dimension $1$. Since every $c_{ij}$ has positive and negative sign with probabilities close to $1/2$, we explain in Appendix~\ref{sec:decomposition:missing}, that with high probability after $k = O(\log(n)^2)$ splits all the nullsapces will have dimension 1. 

\begin{remark}
We observe that Eq.~\eqref{eq:kernel-assumption} also holds for tensor completion version of the SDP. Therefore, w.h.p. we can reconstruct $\mathcal{S}_{uvw}$ exactly from $mn^{3/2}\polylog(n)$ randomly observed entries of the tensor $\mathcal{T}$. In other words, our algorithm works with almost no changes for tensors with missing entries.
\end{remark}

\subsection{Construction of the candidate dual certificates}\label{sec:dual-certificates-sec2}

In this section we give more details about the construction of a dual certificate and $\Omega$-restricted dual certificate. The full constructions are provided in Section~\ref{sec:corr-terms} and~\ref{sec:Omega-dual-certificate}.

\subsubsection{Construction of a dual certificate}\label{sec:dual-certificate-contsr2}

Our construction of a dual certificate candidate is inspired by \cite{Li-et-al}. First, note that a dual certificate satisfies the following necessary condition.

\begin{equation}
    \max_{x, y, z\in S^{n-1}} \langle \mathcal{A}, x\otimes y\otimes z \rangle \quad \text{is achieved at} \quad u_i\otimes v_i\otimes w_i \ \text{for } i\in [m].
\end{equation}

By computing the gradient it imposes the following linear conditions on $\mathcal{A}$

\begin{lemma}[\cite{Li-et-al}]\label{lem:dual-certificate-necessary} Assume that $\oa{A}$ is a dual certificate for a collection $\mathcal{V} = \{(u_i, v_i, w_i)\mid i\in [m]\}$. Then for all unit vectors $x, y, z \in S^{n-1}$ and all $i\in [m]$,
\begin{equation}\label{eq:conditions-onA-2}
    \langle \mathcal{A}, x\otimes v_i\otimes w_i\rangle = \langle u_i, x\rangle,\quad  \langle \mathcal{A}, u_i\otimes y\otimes w_i\rangle = \langle v_i, y\rangle,\quad \langle \mathcal{A}, u_i\otimes v_i\otimes z\rangle = \langle w_i, z\rangle.
\end{equation}
\end{lemma}
\begin{proof} Follows from the gradient computation for $\langle \mathcal{A}, x\otimes y\otimes z\rangle$ at points $(u_i\otimes v_i\otimes w_i)$.
\end{proof}

\begin{definition}\label{def:certificate-candidate-2} We say that $\oa{A}$ is a \textit{certificate candidate} for $\mathcal{V}$ if $\oa{A}$ satisfies conditions~\eqref{eq:conditions-onA-2} and $\oa{A}$ can be written as
 \begin{equation}\label{eq:certificate-3-form-intr2}
 \oa{A} = \sum\limits_{i=1}^{m} u_i\otimes v_i \otimes w_i+ \sum\limits_{i=1}^{m} \alpha_i\otimes v_i \otimes w_i+ u_i\otimes \beta_i \otimes w_i+ u_i\otimes v_i \otimes \gamma_i.
 \end{equation}
\end{definition}

In Section~\ref{sec:corr-terms}, we prove that for $m\ll n^2$ w.h.p. there exists a certificate candidate and moreover matrices with columns $\alpha_i$, $\beta_i$ and $\gamma_i$, respectively, have norms bounded by $\wt{O}\left(m/n^{3/2}\right)$. A similar result is claimed\footnote{However, their proof has a flaw when $m>n$. On page 23 in \cite{Li-et-al} the inequality $\Vert C^{t}-C^{t-1}\Vert\leq \eta \Vert C^{t-1}-C^{t-2}\Vert$ was proved for $t$ large enough, but on page 24 it is used for all $t\in \mathbb{N}$.} in \cite{Li-et-al}.  

To prove this, we view the conditions~\eqref{eq:conditions-onA-2} as a system of linear equations $M\oa{A} = D$ for $M\in \mathbb{R}^{3mn\times n^3}$ (by taking $n$ orthonormal basis vectors in place of $x$, $y$ and $z$). Finding $\mathcal{A}$ of an explicit form~\eqref{eq:certificate-3-form-intr2} corresponds to solving a linear system $MM^TY = D$.  

As we discuss in Section~\ref{sec:corr-terms}, $MM^T$ has a non-trivial kernel! Therefore, to show that this equation has a solution we analyze the kernel of the random matrix $MM^T$, by guessing a simpler matrix that approximates $M$ well. To show the norm bounds for $\alpha_i$, $\beta_i$ and $\gamma_i$ we approximate $(MM^T)^{-1}$ with an explicit matrix up to a very high precision and then analyze the solution provided by this approximation.

In Section~\ref{sec:dual-certificate}, we justify that w.h.p. the certificate candidate is a dual certificate for $m\ll n^{3/2}$. For this we need to justify that
$\langle \oa{A}, x\otimes y\otimes z\rangle \leq 1$ for all $x, y, z\in S^{n-1}$. The analysis distinguishes two different cases: when $x\otimes y\otimes z$ is close to some $u_i\otimes v_i\otimes w_i$; and when it is far from all of them.
\begin{theorem}\label{thm:certificate-existence} Let $m \ll n^{3/2}$ and let $\mathcal{V} = \{a_i^t \mid i\in [m],\ t\in [3]\}$ be a collection of $3m$ independent random vectors uniformly distributed on $S^n$. Then w.h.p. there exists a strong dual certificate $\oa{A}$ for $\mathcal{V}$.
\end{theorem}

We observe that Theorem~\ref{thm:main:nuclear-norm-rank-same} immediately follows from this claim. Unfortunately, this is not sufficient for Theorem~\ref{thm:main:dual-certificates-alg}, as this proof cannot be captured with degree 4 SOS. We present a construction for the feasible solution to~\eqref{eq:opt-constraints-id-intr} in Sections~\ref{sec:B0} and~\ref{sec:zero-poly-corr}.

We would like to emphasize the following claim and a question related to our construction which is inspired by it. Consider the linear subspace
\begin{equation}\label{eq:S-def}
\mathcal{S} = \vspan\{x\otimes v_i\otimes w_i,\ u_i\otimes x\otimes w_i,\ u_i\otimes v_i\otimes x \mid i\in [m], x\in \mathbb{R}^n\}.
\end{equation}

Clearly, Eq.~\eqref{eq:certificate-3-form-intr2} is equivalent to saying that $\oa{A}\in \mathcal{S}$. In fact, it will be evident from the further discussion that the following statement holds.

\begin{lemma}\label{lem:unique-projection-on-s} Let $m\ll n^2$. W.h.p. there is a unique vector $\oa{A}$ in $\mathcal{S}$ which satisfies Eq.~\eqref{eq:conditions-onA-2}. If $\oa{A}'$ is a dual certificate for $\mathcal{V}$, then $P_{\mathcal{S}}\oa{A}' = \oa{A}$, where $P_{\mathcal{S}}$ is an orthogonal projector on $\mathcal{S}$.
\end{lemma}
\begin{proof} Follows from Theorem~\ref{thm:candidate-exists} and the discussion in Section~\ref{sec:M-system}.
\end{proof}

\begin{question} If $n^{3/2}\ll m\ll n^2$, is it possible that for a ``generic'' collection $\mathcal{V}$ there exists a dual certificate $\oa{A}'$, while $\oa{A}$ is not a dual certificate? Is $\oa{A}$ w.h.p. a dual certificate for $n^{3/2}\leq m\ll n^2$? 
\end{question}

\subsubsection{Construction of an {$\Omega$}-restricted dual certificate}\label{sec:constr-Omega-certificate-intro}

We consider a random set $\Omega$ which is obtained by including each element of $[n]^3$ independently with probability $N/n^3$. So that the expected number of entries in $\Omega$ is $N$. For $\Omega$ sampled in such way consider random diagonal $n^3\times n^3$ matrices
\begin{equation}\label{eq:ROmega-definition}
(R_{\Omega})_{\omega, \omega} = \begin{cases}
      N/n^3 & \text{if $\omega\in \Omega$} \\
      0 & \text{if $\omega \notin \Omega$}
    \end{cases}, \qquad \text{and}\qquad \ol{R}_{\Omega} = I_{n^3} - R_{\Omega}.
\end{equation}  

The $R_{\Omega}$ is a scaled version of the projector onto entries in $\Omega$, so that $\mathbb{E}_{\Omega}[R_{\Omega}] = I_{n^3}$.

Assume that $\oa{A}\in \mathcal{S}$ is a dual certificate for $\mathcal{V}$. In view of Lemma~\ref{lem:dual-certificate-necessary} and Lemma~\ref{lem:unique-projection-on-s} the necessary conditions for an $\Omega$-restricted dual certificate can be formulated as
\begin{equation}\label{eq:AOmega-necessary-cond}
 P_\mathcal{S}\oa{A}_{\Omega} = \oa{A}\qquad \text{and}\qquad (\oa{A}_{\Omega})_{\omega} = 0, \text{ for all } \omega\notin \Omega.
\end{equation}

To find a vector which satisfies both of these condition we use ``quantum golfing'' technique (which was used e.g. in \cite{potechin-steurer-exact, Yuan-Zhang}). The idea is to start with $\oa{A}$ and then alternatively project it onto subspaces where each of the two conditions hold. 

To implement this, we sample $k$ independent random sets $\Omega_i$ for $i\in [k]$ as described above so that each of them has expected size $N$. Define $\Omega = \bigcup_{i=1}^{k} \Omega_i$. Then, it is easy to see that each entry of $[n]^3$ is included into $\Omega$ independently with equal probability and the expected size of $\Omega$ is at most $kN$. 

 We will search for $\oa{A}_{\Omega}$ as a result of $k$ alternating projections by $R_{\Omega_i}$ and $P_{\mathcal{S}}$ for sufficiently large constant $k$ plus a small correction $\oa{A}_{\Omega, sm}$. That is, we consider $\oa{A}_{\Omega}$ of the form

\begin{equation}\label{eq:Aomega-constr-intr}
\oa{A}_{\Omega} = \sum\limits_{j=1}^{k} R_{\Omega_j}\left(\prod_{i=1}^{j-1}P_{\mathcal{S}}\ol{R}_{\Omega_{j-i}}\right)\oa{A}+\oa{A}_{\Omega, sm},
\end{equation}

In Section~\ref{sec:Omega-dual-certificate} we will show that such alternating projection indeed converges to a tensor that satisfies conditions~\eqref{eq:AOmega-necessary-cond}; and that once we are close enough to a target tensor, we can ``jump there'' in one step with a small correction term $\oa{A}_{\Omega, sm}$.

\section{Preliminaries}

In this section we present some classical norm bounds and concentration inequalities.
\subsection{Properties of random vectors on a unit sphere}

\begin{fact}\label{fact:inner-products} Let $u$ be a random vector on $S^{n-1}$ and $x, y$ be fixed unit vectors. Then 
\[ \mathbb{E} \langle u, x \rangle^2 = \dfrac{1}{n}, \qquad and \qquad \mathbb{E} \langle u, x \rangle \langle u, y \rangle = \dfrac{\langle x, y \rangle}{n}.\] 
\end{fact}

\begin{fact}\label{fact:inner-produc-hp-bound} Let $x\in S^{n-1}$ be a fixed vector, and $u$ be a random vector uniformly distributed on $S^{n-1}$. Then with probability at least $1 - \exp(-\log(n)^{10})$
\[ \vert \langle u, x \rangle \vert  = \wt{O}\left(\dfrac{1}{\sqrt{n}}\right).
\]
\end{fact}

\begin{definition}\label{def:cten}
Let $X$ and $Y$ be matrices with equal number of columns. Define the (columnwise) \textit{Khatri-Rao product} of $X$ and $Y$  to be the matrix $X\cten Y$ with $i$-th column $x_i\otimes y_i$, where $x_i$ and $y_i$ are $i$-th columns of $X$ and $Y$.
\end{definition}

Let $\mathcal{U} = \{u_i, v_i\in S^{n-1}\mid i\in [m]\}$ be a set of $2m$ independent uniformly distributed random vectors on the unit sphere. Let $U$ and $V$ be matrices with columns $u_i$ and $v_i$, respectively.

\begin{lemma}\label{lem:basic-norm-bounds} With high probability
\[ \Vert U\Vert = 1+\wt{O}\left(\dfrac{\sqrt{m}}{\sqrt{n}}\right) \quad \text{and} \quad \Vert U\cten V\Vert = 1+\wt{O}\left(\dfrac{\sqrt{m}}{n}\right)\]
\end{lemma}
\begin{corollary}\label{cor:basic-deg2-sum} With high probability, for each $x\in \mathbb{R}^n$
\[\sum\limits_{i=1}^{m} \langle u_i, x\rangle^2 \leq \Vert x\Vert^2\left(1+\wt{O}\left(\dfrac{m}{n}\right)\right).
\]
\end{corollary}

We also need the following bounds for the sum of 4-th powers.
\begin{theorem}[\cite{BBH12}]\label{thm:deg4weakbound} If $m\ll n^2$, then w.h.p., for each $x\in \mathbb{R}^n$ with $\Vert x \Vert =1$,
\begin{equation}
\sum\limits_{i=1}^{m}\langle a_i, x \rangle^4 \leq \widetilde{O}\left(1\right).
\end{equation}
\end{theorem}

\begin{theorem}[Ge, Ma {\cite[Lemma 5]{Ge-Ma}}]\label{thm:deg4bound} If $m\ll n^2$, then with high probability, for arbitrary $x\in \mathbb{R}^n$ with $\Vert x \Vert =1$,
\begin{equation}
\sum\limits_{i=1}^{m}\langle u_i, x \rangle^4 \leq 1+\widetilde{O}\left(\frac{1}{\sqrt{n}}+\frac{m}{n^{3/2}}\right).
\end{equation}
\end{theorem}

\begin{lemma}[Ge, Ma {\cite[Lemma 8-9]{Ge-Ma}}]\label{lem:sixtofourbound} If $m\ll n^{3/2}$, then with high probability
\[ \left(\sum\limits_{i=1}^{m}\langle u_i, x\rangle^4\right)^2\leq \sum\limits_{i=1}^{m}\langle u_i, x\rangle^6+\widetilde{O}\left(\frac{1}{n}+\frac{m^2}{n^3}\right).\]
\end{lemma}

Additionally, the following stronger version of Lemma~\ref{lem:basic-norm-bounds} will be useful frequently.
\begin{lemma}\label{lem:basic-cten-bound} With high probability
\[\left\Vert (U\cten V)^T(U\cten V) -I_m\right\Vert = \wt{O}\left(\dfrac{\sqrt{m}}{n}\right).\] 
\end{lemma}

\subsection{Hadamard product and some useful norm bounds}
\begin{definition}
Let $A$ and $B$ be a pair of $n_1\times n_2$ matrices. Their Hadamard product $A\odot B$ is an $n_1\times n_2$ matrix defined as
\[ (A\odot B)_{ij} = A_{ij}\cdot B_{ij}\quad \forall\ i\in [n_1], j\in [n_2]. \]
\end{definition}

\begin{lemma}[see \cite{johnson} p.113]\label{lem:odot-inequality} Let $A, B \in M_m(\mathbb{R})$. Assume that $A$ is positive semidefinite. Then 
\[ \Vert A\odot B \Vert \leq \max\limits_{i\in [m]} A_{ii} \cdot \Vert B \Vert. \]
\end{lemma}
\begin{proof}
We provide the proof for completeness. We prove a slightly more general result. Assume that $A = X^TY$ for $X, Y \in M_{n\times m} (\mathbb{R})$. Then
\[ \left( \begin{matrix} \Vert B \Vert (I_m\odot X^TX) & A\odot B \\
(A\odot B)^T & \Vert B \Vert (I_m\odot Y^TY)
\end{matrix}\right) = \left( \begin{matrix} \Vert B \Vert I_m &  B \\
B & \Vert B \Vert I_m
\end{matrix}\right) \odot \left( \begin{matrix}  X^TX & A \\
A & Y^TY
\end{matrix}\right).\]
The matrices in the RHS are positive semidefinite, so by Shur's product theorem, the matrix in the LHS is positive semidefinite. Hence,
\[\Vert A\odot B\Vert \leq \Vert B\Vert \sqrt{\Vert I_m\odot X^TX\Vert \cdot \Vert I_m\odot Y^TY\Vert} = \Vert B\Vert \sqrt{\max\limits_{i\in [m]} (XX^T)_{ii}\cdot \max\limits_{i\in [m]} (YY^T)_{ii}} \]
In the case, when $A$ is positive semidefinite there exists $X$ such that $A = X^TX$. Therefore, the claim of the lemma follows.
\end{proof}

\begin{lemma}\label{lem:cten-norm-bound} Let $X$ and $Y$ be $n\times m$ matrices. Assume that every column of $Y$ has norm at most 1. Then 
\[ \Vert X\cten Y \Vert_{F} \leq  \Vert X \Vert_F \quad  \text{and} \quad \Vert X \cten Y \Vert \leq \Vert X \Vert.\]
\end{lemma}
\begin{proof}
The first inequality holds since
\[ \Vert X\cten Y \Vert_{F}^2 = \sum\limits_{i = 1}^m \Vert x_i\otimes y_i \Vert^2 \leq \sum\limits_{i = 1}^m \Vert x_i\Vert^2 = \Vert X \Vert_{F}^2. \]

Note that 
\[  (X\cten Y)^T (X\cten Y) = \left( X^TX\right) \odot \left(Y^{T}Y\right). \]
We have $\max\limits_{i\in [m]} \left(Y^TY\right)_{ii} \leq 1$ and $Y^TY$ is positive semidefinite.  Thus, by Lemma~\ref{lem:odot-inequality}, 
\[  \Vert X\cten Y \Vert^2 = \Vert (X\cten Y)^T (X\cten Y) \Vert  \leq \max\limits_{i\in [m]} \left( Y^TY\right)_{ii} \cdot \left\Vert X^TX\right\Vert \leq \Vert X \Vert^2 . \]
\end{proof}

\subsection{Nuclear norm and its dual certificates}

 \begin{definition} Let $\mathcal{T}$ be a tensor in $\mathbb{R}^{n_1\times n_2 \times \ldots \times n_d}$. The \textit{nuclear norm} of $\mathcal{T}$ is defined as
 \[ \Vert \mathcal{T}\Vert_{*} = \min\left\{\sum\limits_{i = 1}^m |\lambda_i| \mid \mathcal{T} = \sum\limits_{i=1}^{m} \lambda_i a_i^{1}\otimes a_i^2\otimes \ldots \otimes a_i^d,\ a_i^t\in S^{n_t-1}\right\},\]
 where $S^{n_i-1}\subseteq \mathbb{R}^{n_i}$ is a unit sphere.
 \end{definition}
 
 It is not hard to see that the minimum in the definition is well defined since we minimize a continuous function over a compact set (e.g. see \cite{Friedland-Lim}).
 
 \begin{definition} A decomposition $\mathcal{T} = \sum\limits_{i=1}^{m} \lambda_i a_i^{1}\otimes a_i^2\otimes \ldots \otimes a_i^d$ which minimizes the tensor nuclear norm is called a \textit{nuclear decomposition}.
 \end{definition}
 
 As discussed in the introduction, it may be rather non-trivial to verify that a given decomposition minimizes the nuclear norm.  A standard approach is to provide a (dual) certificate for the decomposition. 

\begin{definition} Let $\mathcal{T}$ be a tensor in $\mathbb{R}^{n_1\times n_2 \times \ldots \times n_d}$. The \textit{injective (spectral) norm} of $\mathcal{T}$ is defined as
\[ \Vert \mathcal{T} \Vert_{\sigma} = \max\left\{ \langle \mathcal{T}, x^1\otimes x^2\otimes \ldots \otimes x^d \rangle \mid \Vert x_i \Vert \leq 1, x_i \in \mathbb{R}^{n_i}\right\}  \] 
\end{definition}

The spectral norm is dual to the nuclear norm, i.e.,
\[ \Vert X \Vert_{*} = \max\{ \langle X, Y\rangle \mid \Vert Y \Vert_{\sigma}\leq 1\}. \] 
Since this duality is central for our results, we give an explanation for this duality in Appendix~\ref{dualityappendix}.

The $Y$ for which the maximum is attained is called a \textit{nuclear norm dual certificate} for $X$. In fact, it is more natural to view $Y$ as a dual certificate for the components of a nuclear decomposition of $X$. 

\begin{definition} We say that $\oa{A}$ is a \textit{dual certificate} for a collection of $dm$ unit vectors $\mathcal{V} = \{a_i^t\in S^{n_t-1} \mid i\in [m],\ t\in [d]\}$ if $\Vert \oa{A} \Vert_{\sigma}\leq 1$ and 
\[ \langle \oa{A}, a_i^1\otimes a_i^2\otimes \ldots \otimes a_i^d\rangle = 1, \quad \forall i\in [m].\]
We say that $\oa{A}$ is a \textit{strong dual certificate} for $\mathcal{V}$ if $A$ is a dual certificate and $a_i^1\otimes \ldots \otimes a_i^d$ are the only unit length rank-1 tensors for which the equality is achieved. 
\end{definition}

\subsection{Techniques to bound the norm of a random matrix}

\subsubsection{Concentration inequalities}

\begin{theorem}[Bernstein Inequality] Let $X_1, X_2, \ldots X_n$ be independent zero mean random variables. Suppose that $\vert X_i \vert \leq L$ almost surely for all $i$. Then for any $t>0$,
\[ \mathbb{P}\left(\sum\limits_{i=1}^{n} X_i\geq t\right) \leq \exp\left(\frac{-t^2/2}{\sum\limits_{i=1}^n \mathbb{E}\left(X_i^2\right)+Lt/3}\right). \] 
\end{theorem}

\begin{theorem}[Matrix Bernstein inequality, see {\cite[Theorem 6.1.1]{Tropp}}]
Consider a finite sequence $\{S_k\}$ of independent random matrices of common dimension $d_1\times d_2$. Assume that
\[ \mathbb{E} S_k  = 0 \quad \text{and} \quad \Vert S_k \Vert \leq L \quad \text{for every } k.\]
Define 
\[ Z = \sum\limits_{k} S_k\quad \text{ and} \quad \nu(Z) = \max\left(\left\Vert\sum\limits_{k}\mathbb{E}(S_k S_k^T)\right\Vert, \left\Vert\sum\limits_{k}\mathbb{E}(S_k^T S_k)\right\Vert \right).\]
Then, for all $t\geq 0$,
 \[ \mathbb{P}\left( \Vert Z \Vert \geq t \right) \leq (d_1+d_2)\exp\left(\frac{-t^2/2}{\nu(Z)+Lt/3}\right).\]
\end{theorem}

In Theorems~\ref{thm:inj-norm-bound} and~\ref{thm:A0-twist-bound} will apply the Matrix Bernstein inequality in a combination with the following powerful decoupling theorem.

\begin{theorem}[de la Pena, Montgomery-Smith \cite{var-split}]\label{thm:random-variable-split} Let $X_1, \ldots, X_n, Y_1, \ldots, Y_n$ are independent random variables on a measurable space over $S$, where $X_i$ and $Y_i$ has the same distribution for $i\in [n]$. Let $f_{i, j}:S\times S \rightarrow (B, \Vert \cdot \Vert)$ be a family of functions, where $(B, \Vert \cdot \Vert)$ is a Banach space. Then there exists an absolute constant $C$, such that for any $t>0$,
\[ \mathbb{P}\left(\left\Vert\sum\limits_{i, j:\ i\neq j} f_{i, j}(X_i, X_j)\right\Vert\geq t \right) \leq C \mathbb{P}\left(\left\Vert\sum\limits_{i, j:\ i\neq j} f_{i, j}(X_i, Y_j)\right\Vert\geq \frac{t}{C} \right).\]
\end{theorem}

\begin{theorem}[Matrix Chernoff bound]\label{thm:matrix-chernoff} Consider a sequence $X_k$ of independent positive semidefinite matrices of dimension $d$. Assume that the maximal eigenvalue $\lambda_{\max}(X_k)$ satisfies $\lambda_{\max}(X_k) \leq L$ almost surely. Consider the maximum and the minimal eigenvalues
\[ \mu_{\min} = \lambda_{\min}\left(\sum\limits_{k} \mathbb{E}[ X_k]\right)\quad \text{and}\quad \mu_{\max} = \lambda_{\max}\left(\sum\limits_{k} \mathbb{E}[ X_k]\right).\]
Then
\[ \mathbb{P}\left(\lambda_{\min}\left(\sum\limits_{k} X_k\right) \leq (1-\delta) \mu_{\min}\right) \leq d\left(\dfrac{e^{-\delta}}{(1-\delta)^{1-\delta}}\right)^{\mu_{\min}/L}\quad \text{for } \delta\in [0, 1],\]
\[\mathbb{P}\left(\lambda_{\max}\left(\sum\limits_{k} X_k\right) \geq (1+\delta) \mu_{\max}\right) \leq d\left(\dfrac{e^{\delta}}{(1+\delta)^{1+\delta}}\right)^{\mu_{\max}/L}\quad \text{for } \delta\geq 0.\]
\end{theorem}

\subsubsection{The trace power method}\label{sec:power-trace-prelim}

Another very efficient technique which we will use to bound the norms of a matrix is the trace power method. We will use it in the following form. 

\begin{lemma}[see {\cite[Lemma 3.1]{graph-matrix-bounds}}]\label{lem:trace-power-method-norm}
Let $B(2q)$ be a sequence of positive numbers. Assume that $M$ is a random matrix and that for every $q\in \mathbb{N}$ the bound $\mathbb{E}\left[\Tr\left((MM^T)^q\right)\right] \leq B(2q)$ holds. Then, for every $\varepsilon >0$, 
\[ \mathbb{P}\left[\Vert M \Vert > \min\limits_{q\in \mathbb{N}} \left(\dfrac{B(2q)}{\varepsilon}\right)^{1/(2q)}\right]<\varepsilon. \]
\end{lemma}

\section{Power-trace method bounds using matrix diagrams}

In this paper, we will frequently need to bound the norms of random matrices which have a special structure. Each entry of such a matrix is defined as a multi-index sum of inner products of random vectors from some collection.  In this section we describe a class of matrices for which the sum defining the matrix can be represented with a colored graph, called a matrix diagram. One of the key technical contributions of this paper shows how one can get norm bounds for matrices of such a form using only combinatorial properties of the corresponding matrix diagrams. 

Our technique is inspired by the technique of getting norm bounds for graph matrices developed in~\cite{graph-matrix-bounds} (however the discussion below is self-contained and does not require any knowledge of~\cite{graph-matrix-bounds}).  The key advantage of our technique over graph matrices is that we do not require vertices of the diagram (shape) to have different labels. This allows us to have almost no case analysis, which is required for graph matrices. This is crucial, as in the analysis we encounter matrices that are defined by a sum with hundreds of indicies (corresponding to hundreds of vertices), and so they are infeasible for a non-trivial case analysis. At the same time, in some situations the norm bounds which can be shown through graph matrices are not achievable with our technique (at least not without new tricks).

\subsection{Inner product graph matrices and matrix diagrams}
  
Before giving formal definitions we explain the intuition behind them and we give several examples of IP graph matrices and the corresponding diagrams.

Fix orthonormal bases $\{f_i\}_{i=1}^{m}$ and $\{e_j\}_{j=1}^{n}$ of $\mathbb{R}^m$ and $\mathbb{R}^n$. We consider a collection of vectors $\mathcal{V} = \{u_i, v_i, w_i  \mid i\in [m]\}\subseteq S^{n}$, where we think of each vector having one of the colors $\{u, v, w\}$.  

Consider matrices $R_1\in M_{n^2, m}(\mathbb{R})$, $R_2\in M_{m, n^2}(\mathbb{R})$ and $R_3\in M_{m, m}(\mathbb{R})$ defined as
\begin{equation}\label{eq:graph-matrix-examples}
\begin{gathered}
R_1 = \sum\limits_{i = 1}^{m}\sum\limits_{j:\, j\neq i}\ \sum\limits_{k:\, k\neq i,\, k\neq j} \langle u_i, u_j\rangle \langle v_j, v_k\rangle \langle w_i, w_k\rangle (v_i\otimes w_j)f_k^T,\qquad R_2 = \sum\limits_{t=1}^{m} f_t(v_t\otimes w_t)^T, \\
R_3 = R_2R_1 = \sum\limits_{\substack{i, j, k, t = 1\\ i\neq j,\ j\neq k,\ k\neq i}}^{m} \langle u_i, u_j\rangle \langle v_j, v_k\rangle \langle w_i, w_k\rangle \langle v_i, v_t \rangle \langle w_j, w_t \rangle f_tf_k^T.
\end{gathered}
\end{equation}
\begin{figure}
\begin{subfigure}[b]{0.3\textwidth}
\begin{center}
\includegraphics[height = 3cm]{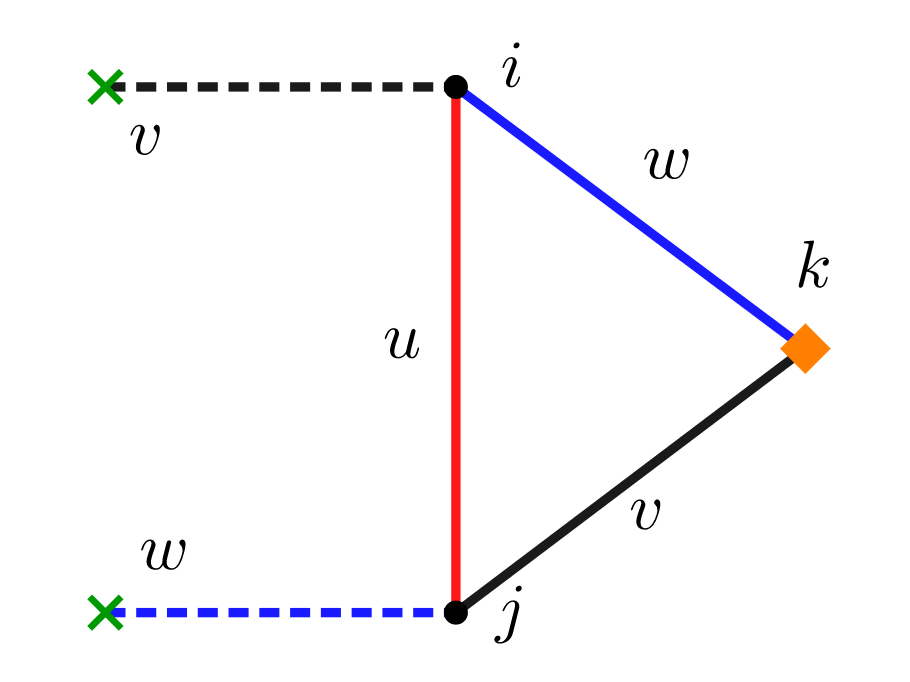}
\end{center}
\end{subfigure}
\begin{subfigure}[b]{0.25\textwidth}
\begin{center}
\includegraphics[height = 3cm]{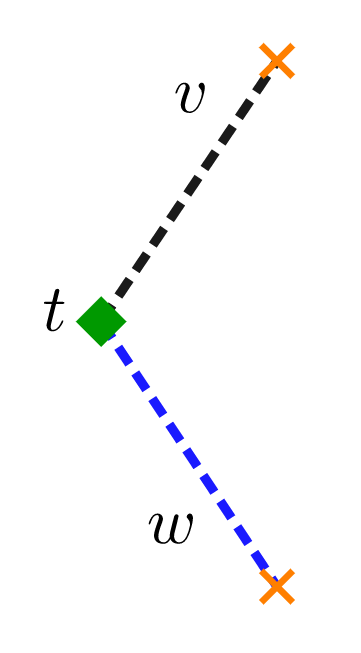}
\end{center}
\end{subfigure}
\begin{subfigure}[b]{0.3\textwidth}
\begin{center}
\includegraphics[height = 3cm]{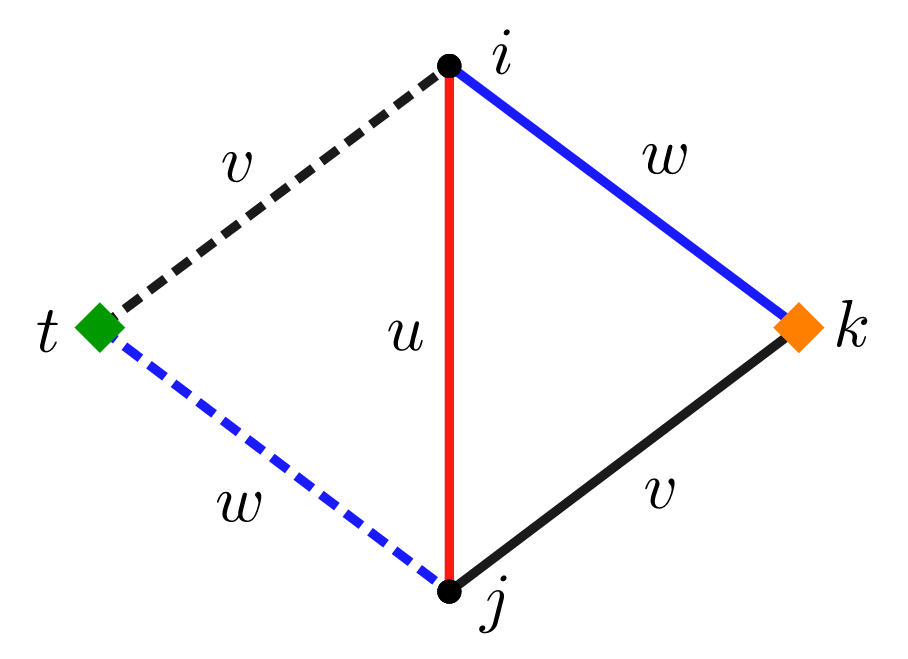}
\end{center}
\end{subfigure}
\caption{Matrix diagrams for $R_1$, $R_2$ and $R_3 = R_2\cdot R_1$, respectively}\label{fig:examples}
\end{figure}
Each of these three matrices can be described by a small colored graph 
 $G = (\myscr{Nod}\cup \myscr{Cr}, E)$ (see Figure~\ref{fig:examples}) and a set of functions $\Phi = \{\phi: \myscr{Nod} \rightarrow [m]\}$. We describe some intuitive rules:

\begin{enumerate}
\item An IP graph matrix $M$ is a sum of rank-1 terms of the same structure, which are described by the corresponding matrix diagram $G$. The sum is indexed by functions $\phi\in \Phi$. More precisely, 
\[ M = \sum\limits_{\phi \in \Phi} val(G, \phi), \]
where $val(G, \phi)$ is a rank-1 term that corresponds to $G$ and a vertex label assignment $\phi$.
For example, the sum defining $R_1$ can be indexed by $\Phi=\{\phi: \{{\tt{i, j, k}}\}\rightarrow [m]\mid \phi \text{ is injective}\}$.
\item The vertices of $G$ have two different types: nodes and crosses. Nodes correspond to vectors $f_i$ from $\mathbb{R}^m$, while crosses correspond to $e_j$ from $\mathbb{R}^n$.   
\item  A matrix diagram $G$ contains the following information that describes a rank-1 term: the tensor product structure of a row and column and a list of inner products of vectors that form the coefficient in front of the term.
\item The colored edges of $G$ correspond to inner products as follows:
\begin{enumerate}
\item an edge of color $u$ between a node labeled $i$ and a cross labeled $j$ corresponds to $\langle u_i, e_j\rangle$
\item an edge of color $u$ between a pair of nodes labeled $i$ and $j$ corresponds to $\langle u_i, u_j\rangle$
\end{enumerate}
We have similar correspondences for edges of color $v$ or $w$.
\item To describe the tensor structure of rows and columns of the matrix we specify the ordered subsets of left (row structure) and right (column structure) vertices. The rest of the vertices are inner vertices. 
\end{enumerate}

For example, the diagram for $R_2$ (Figure~\ref{fig:examples}) has a node on the left and an ordered pair of crosses on the right. Thus, the rows have structure $e_x\otimes e_y$ and the columns have structure $f_t$. Using the above rules, we write
\[  R_2 = \sum\limits_{t=1}^{m}\sum\limits_{x, y=1}^{n} \langle v_t, e_x\rangle\langle w_t, e_y\rangle f_t(e_x\otimes e_y)^T = \sum\limits_{t=1}^{m} f_t(v_t\otimes w_t)^T.\]
We avoid explicit summation over crosses to reduce the complexity of the expression.


Observe that $R_2$ and $R_1$ have ``consistent'' row and column ``types'', hence when we multiply $R_2$ and $R_1$, the product of row $(v_t\otimes w_t)^T$ with the column $(v_i\otimes w_j)$ results in a product of a pair of inner products. Thus, as we see, their product $R_3$ also has a diagram representation. Also, note that we can think of vectors $v_i$ and $w_j$ in the column $v_i\otimes w_j$ of $R_1$ as of being a ``half-edges", which become edges after $R_1$ is multiplied by a matrix of ``consistent" type.

Next, we give formal definitions to describe this (matrix $\leftrightarrow$ diagram) correspondence.

\begin{definition}  Let $\myscr{Ver} = (\myscr{Ver}_{L}\cup \myscr{Ver}_{R})\sqcup \myscr{Ver}_{I} $, where $\myscr{Ver}_{L} = \myscr{Cr}_{L}\sqcup \myscr{Nod}_{L}$ and $\myscr{Ver}_{R} = \myscr{Cr}_{R}\sqcup\myscr{Nod}_{R}$, and $\myscr{Nod}_I = \myscr{Ver}_I$
\begin{itemize}
\item $\myscr{Ver}_{L}$ and $\myscr{Ver}_{R}$ are ordered sets, called \textit{left} and \textit{right outer vertices};
\item $\myscr{Ver}_{I}$ is a set, elements of which are called \textit{inner vertices}.
\item elements of $\myscr{Cr}_{L}$ and $\myscr{Cr}_{R}$ are called \textit{crosses} or \textit{ends of half-edges};
\item elements of $\myscr{Nod} = \myscr{Nod}_{L} \cup \myscr{Nod}_{R} \cup \myscr{Nod}_I$  are called \textit{nodes}.
\end{itemize}
Let $r$ be a positive integer. Consider a colored graph $G = (\myscr{Ver}, E, \mathfrak{c})$, where vertices in $\myscr{Cr}_{L}\cap \myscr{Cr}_{R}$ have degree 0, all other vertices in $\myscr{Cr}_{L}\cup \myscr{Cr}_{R}$ have degree precisely 1 and are adjacent to a node in $\myscr{Nod}$; and  $\mathfrak{c}: E \rightarrow [r]$ assigns a color for each edge of $G$. We call $G$ an \emph{inner product matrix diagram} or just a  \textit{matrix diagram}.
\end{definition}

We call \textit{half edges} the edges in $E$ incident to crosses. Denote by $E_{1/2}$ the set of half edges of $G$ and let $E^v = E\setminus E_{1/2}$.

\begin{definition} Let $G = (\myscr{Ver}, E, \mathfrak{c})$ be a matrix diagram and $\mathcal{U} = \{a_{i}^{t}\in S^{n_t-1}\mid i\in[m], t\in[r]\}$ be a collection of $rm$ unit vectors.  Let $\Phi\subseteq \{\phi:\myscr{Nod} \rightarrow [m]\}$ be a \emph{set of permitted labelings}. 

Fix an orthonormal basis $\{f_i\}$ of $\mathbb{R}^m$.
Denote by $\vect_{\phi}(x)$ the vector $a^{\mathfrak{c}(\{x, y\})}_{\phi(y)}$, if $x$ is in $\myscr{Cr}$, where $y$ is the unique vertex adjacent to $x$, and denote by $\vect_{\phi}(x)$ the vector $f_{\phi(x)}$, if $x$ is in $\myscr{Nod}$. 
 
 The  \emph{inner product (IP) graph matrix} corresponding to $(G, \mathcal{U}, \Phi)$ is defined as
\[ IGM(G, \mathcal{U}, \Phi) = \sum\limits_{\phi\in \Phi}\left(\prod\limits_{e = \{x, y\} \in E^v}\left\langle a^{\mathfrak{c}(e)}_{\phi(x)}, a^{\mathfrak{c}(e)}_{\phi(y)} \right\rangle \left(\bigotimes\limits_{\ell \in \myscr{Ver}_L} \vect_{\phi}(\ell)\right)\left(\bigotimes\limits_{p \in \myscr{Ver}_R} \vect_{\phi}(p)\right)^T\right), \]
where tensor products respect the order of $\Omega_L$ and $\Omega_R$.

We say that that a matrix $M$ is an \emph{inner product graph matrix} if there exists some $(G, \mathcal{U}, \Phi)$ such that $M = IGM(G, \mathcal{U}, \Phi)$. We say that $G$ is a matrix diagram for $M$ and we denote it by $\mathcal{MD}(M)$.

\end{definition}

\begin{definition}\label{def:non-equality-edges}Let $G = (\myscr{Ver}, E, \mathfrak{c})$ be a matrix diagram. We say that a set of permitted labelings $\Phi$ is defined by a set $E_{\neq}\subseteq E^v$ of \textit{non-equality edges}, if 
\[ \Phi = \{\phi: \myscr{Nod}\rightarrow [m]\mid \phi(x)\neq \phi(y)\ \text{ for all } \{x, y\}\in E_{\neq}\}.\]
\end{definition}

\begin{notation}[Diagrams] We use the following convention in drawing matrix diagrams:
\begin{itemize}
\item The nodes in $\myscr{Nod}_I$ are denoted by \textit{bullets} and crosses in $\myscr{Cr}$ are denoted by \textit{crosses}. The nodes in $\myscr{Nod}_L$ and $\myscr{Nod}_R$ are denoted by \textit{diamonds}.
\item The vertices in $\myscr{Ver}_I$ are colored \textit{black}, $\myscr{Ver}_L$ are colored \textit{green}, and the vertices in $\myscr{Ver}_R$ are colored \textit{orange}. If the vertex belongs to both $\myscr{Ver}_L$ and $\myscr{Ver}_R$ we use \textit{purple} color.
\item If $r=3$ and the collection of vectors is $\mathcal{U}=\mathcal{V}=\{u_i, v_i, w_i\mid i\in[m]\}$, we use \textit{red} color for $u$-edges, \textit{black} color for $v$-edges and \textit{blue} color for $w$-edges.
\item If the set of permitted labelings $\Phi$ is defined by non-equality edges, we draw the edges in $E_{\neq}$ with a \textit{solid line} and the edges in $E\setminus E_{\neq}$ with a \textit{dashed line}. We may deviate from this convention for ``schematic'' diagrams or in the situations, when $\Phi$ is not important for the discussion.
\end{itemize}
\end{notation}

\subsection{Product and trace of inner product graph matrices}\label{sec:trace-diagram}
Note that the condition when the product of two IP graph matrices is an IP graph matrix can be described as a simple combinatorial condition on matrix diagrams.

\begin{definition}\label{def:type-omegaL} Let $\myscr{Ver}_L$ be the ordered set of left outer vertices of a matrix diagram with edges of $r$ colors. We call the \textit{type} of $\myscr{Ver}_L$ a sequence $type(\myscr{Ver}_L)$ of $|\myscr{Ver}_L|$ elements from $[r]\cup \{*\}$  defined in the following way. If the $i$-th element of $\myscr{Ver}_L$ is a node, then the $i$-th element of $type(\myscr{Ver}_L)$ is $*$; else if $i$-th element of $\myscr{Ver}_L$ is an end of half edge of color $t$, then $i$-th element of $type(\myscr{Ver}_L)$ is $t$. The type of $\myscr{Ver}_R$ is defined similarly.
\end{definition}


\begin{definition}
We say that a pair of matrix diagrams $(G, H)$ with $G = (\myscr{Ver}, E, \mathfrak{c})$ and $H = (\myscr{Ver}', E', \mathfrak{c}')$ is \textit{compatible} if $type(\myscr{Ver}_R) = type(\myscr{Ver}'_L)$.
\end{definition}

\begin{observation}
If a pair of IP graph matrices $(Q, S)$ has a compatible pair of matrix diagrams $(\mathcal{MD}(Q), \mathcal{MD}(S))$, then their product $Q\cdot S$ is an IP graph matrix and its matrix diagram $G = \mathcal{MD}(QS)$ can be described as $\mathcal{MD}(Q)$ being \textit{``glued''} to $\mathcal{MD}(S)$ in the following way (see also Figure~\ref{fig:examples}). 
\begin{itemize}
\item The ordered set $\myscr{Ver}_L$ for $G$ coincides with the ordered set $\myscr{Ver}_L$ for $\mathcal{MD}(Q)$ and $\myscr{Ver}_R$ for $G$ coincides with $\myscr{Ver}_R$ for $\mathcal{MD}(S)$. 
\item The vertices in $\myscr{Nod}_R$ for $\mathcal{MD}(Q)$ are identified with vertices in $\myscr{Nod}_L$ for $\mathcal{MD}(S)$ and they become inner nodes for $G$.
\item  All inner vertices of $\mathcal{MD}(Q)$ and $\mathcal{MD}(S)$ become inner vertices of $G$ and all edges (which are not half edges) remain untouched. Finally, all ends of half edges are deleted from the graph  and every pair of the corresponding half edges in $\mathcal{MD}(Q)$ and $\mathcal{MD}(S)$ form an edge in $G$. That is, for half edges $\{(\myscr{Ver}_R)_i, x\}$ and $\{(\myscr{Ver}'_L)_i, y\}$ in $\mathcal{MD}(Q)$ and $\mathcal{MD}(S)$, respectively, we add an edge $\{x, y\}$ in $G$ of the same color.
\end{itemize}
\end{observation}

The set of permitted labelings for the product $Q\cdot S$ is defined in a natural way from the sets of permitted labelings for $Q$ and $S$. 

\begin{observation}\label{obs:RRT-compatible} Let $R$ be an IP graph matrix, then $R^T$ is also an IP graph matrix and pairs $(\mathcal{MD}(R), \mathcal{MD}(R^T))$ and   $(\mathcal{MD}(R^T), \mathcal{MD}(R))$ are compatible.
\end{observation}
\begin{observation}\label{obs:traceIP} Assume that for an IP graph matrix $R$ the pair $(\mathcal{MD}(R), \mathcal{MD}(R))$ is compatible. Then $\Tr(R)$ is an $1\times 1$ IP graph matrix with $\myscr{Ver}_L = \myscr{Ver}_R = \emptyset $.
\end{observation}
\begin{proof} Observe that $\Tr(R)$ has a matrix diagram which is obtained from the matrix diagram $\mathcal{MD}(R)$ by ``gluing" $\myscr{Ver}_L$ of $\mathcal{MD}(R)$ with $\myscr{Ver}_R$ of the same copy of $\mathcal{MD}(R)$ as it is described for the product of matrices.
\end{proof}
\begin{corollary} Let $R$ be an IP graph matrix. Then for any $q\geq 1$ the trace $\Tr\left((R^TR)^q\right)=\Tr\left((RR^T)^q\right)$ is an IP graph matrix (and both expressions have the same matrix diagram).
\end{corollary}
\begin{proof}
By Observation~\ref{obs:RRT-compatible}, the matrix $\left(MM^T\right)^q$ is an IP matrix and the pair of matrix diagrams $\left(\mathcal{MD}\left(\left(RR^T\right)^q\right), \mathcal{MD}\left(\left(RR^T\right)^q\right)\right)$ is compatible. Hence, the claim follows from Observation~\ref{obs:traceIP}.
\end{proof}

We will be primarily interested in the regime when $q = O(\log(n)^2)$, as this is usually sufficient for power trace method (see Lemma~\ref{lem:trace-power-method-norm}).

\begin{definition}
For an IP graph matrix $R$ we call $\mathcal{MD}\left(\Tr\left(RR^T\right)^q\right) = \mathcal{MD}\left(\Tr\left(R^TR\right)^q\right)$  a $q$-th power \textit{trace diagram} $\mathcal{TD}_q(R)$ of $R$. 
\end{definition}

\begin{definition} Let $G = (\myscr{Ver}, E, \mathfrak{c})$ be a colored graph. Let $\phi':\myscr{Ver}\rightarrow \myscr{Ver}'$ be a map. Let $H$ be a graph on the set of vertices $\phi(\myscr{Ver})$, in which there are $k$ repeated edges of color $c$ between vertices $x, y\in \phi(\myscr{Ver})$ if there are exactly $k$ edges of color $c$ between $\phi^{-1}(x)$ and $\phi^{-1}(y)$ in $G$, for all $x, y, c$. We say that $H$ is the \emph{graph induced by $G$ and $\phi$}. We use the notation $G(\phi) = H$. 
\end{definition}

\begin{definition} For a permitted labeling $\phi\in \Phi$ of $\mathcal{TD}_q(R)$, define $\mathcal{TD}_q(R, \phi)$ to be the graph induced by $\mathcal{TD}_q(R)$ and $\phi$.
\end{definition}

\begin{definition}
For a permitted labeling $\phi\in \Phi$ of  $\mathcal{TD}_q(R) = (\myscr{Ver}, E, \mathfrak{c})$ with a collection of vectors $\mathcal{U} = \{a_{i}^t| i\in[m], t\in[r]\}$, define
\[ val(\mathcal{TD}_q(R, \phi)) = \prod_{e = \{x, y\}\in E} \term(e),\qquad \text{where}\qquad \term(e) = \langle a^{\mathfrak{c}(e)}_{\phi(x)}, a^{\mathfrak{c}(e)}_{\phi(y)} \rangle.\] 
\end{definition}

Clearly,
\[ \Tr\left((RR^T)^q\right) = \sum\limits_{\phi\in \Phi} val(\mathcal{TD}_q(R, \phi)).\]

\subsection{Trace bounds from combinatorial properties of a matrix diagram}\label{sec:trace-bounds-graph-matr}

To get bounds on the trace of $(RR^T)^q$, we consider a collection $\mathcal{C}$ of subsets of colors in $[r]$ and we study the connectivity of subgraphs  of $\mathcal{TD}_q(R)$ induced by the edges of colors in these subsets. In this section we show how connectivity of these subgraphs implies bounds on the trace.    

\begin{definition} We say that a matrix diagram $G = (\myscr{Ver}, E, \mathfrak{c})$ has at most $d$\ \textit{$\mathcal{C}$-connected components}, if for every $C\in \mathcal{C}$ the graph induced on $\myscr{Nod}$ by edges with color in $C$  has at most $d$ connected components. 

We say that $G$ is \textit{$\mathcal{C}$-connected} if it has at most one $\mathcal{C}$-connected component. 
\end{definition}

\begin{definition} For a set of colors $C$, we say that a path $\ell$ in a matrix diagram $G = (\myscr{Ver}, E, \mathfrak{c})$ is a \textit{$C$-path}, if every edge of $\ell$ is of color in $C$.
\end{definition}

\begin{definition} We say that a matrix diagram $G = (\myscr{Ver}, E, \mathfrak{c})$ is $\mathcal{C}$-boundary-connected, if for every $C\in \mathcal{C}$ and $x\in \myscr{Nod}$ there exists a $C$-path from $x$ to some vertex in $\myscr{Ver}_L$ and there exists a $C$-path from $x$ to some vertex in $\myscr{Ver}_R$. 
\end{definition}

\begin{definition}\label{def:pofk-connected}
Define the class $\mathfrak{CM}^t(\mathcal{C}; e, d)$ to be the collection of inner product graph matrices $R$ which satisfy the following conditions.
\begin{enumerate}
\item $R$ corresponds to $(G, \mathcal{U}, \Phi)$, where $\Phi$ is defined by a set of at least $e$ non-equality edges.
\item $G$ has at most $t$ nodes, i.e., $\vert \myscr{Nod}\vert  \leq t$.
\item $G$ has at most $d$\ $\mathcal{C}$-connected components.
\end{enumerate}
In the special case when $d = 1$ we use the notation $\mathfrak{CM}^t(\mathcal{C}; e)$. 
\end{definition}

\begin{definition}
Define the class $\mathfrak{BCM}^t(\mathcal{C}; e, d)$ to be a collection of inner product graph matrices $R\in \mathfrak{CM}^t(\mathcal{C}; e, d)$, for which the matrix diagram of $R$ is $\mathcal{C}$-boundary-connected.

\end{definition}

An important corollary of these definitions in the following statement about $\mathcal{TD}_q\left(R\right)$.

\begin{lemma} Assume that the matrix diagram of $R$ has at most $d$\ $\mathcal{C}$-connected components and is $\mathcal{C}$-boundary-connected. Then for any $q\geq 1$, the trace diagram $\mathcal{TD}_q\left(R\right)$ has at most $d$\ $\mathcal{C}$-connected components.
\end{lemma}

\begin{definition} Define $\vspan\left(\mathfrak{G},  K\right)$ to be the set of matrices that can be written as a linear combination of matrices from the class $\mathfrak{G}$, where the sum of the absolute values of the coefficients in the linear combination is at most $K$. 
\end{definition}

\begin{theorem}\label{thm:diagram-bound-tool} Let $V \subseteq [m]$ and $G = (V, E)$ be a connected graph without loops, but possibly with repeated edges. Assume that every edge of $G$ is labeled with a number in $[r]$. Assume $|E| =  O(\log(n)^4)$. Let $\{a^{t}_i\mid i\in [m],\, t\in[r]\}$ be a set of independent random vectors sampled uniformly at random from $S^{n-1}\in \mathbb{R}^n$. For $e = (i, j)\in E$ labeled by $t$ define $\term(e) = \langle a^{t}_i, a^{t}_j \rangle$. Then 
\[ \left\vert \mathbb{E} \prod\limits_{e\in E} \term(e) \right\vert  = \left(\wt{O}\left( \dfrac{1}{n}\right)\right)^{\max(|V|-1, |E|/2)}. \] 
\end{theorem}
\begin{proof}
First, assume that there is a vertex $i\in V$ of odd degree. Then for some $t$ vector $a^{t}_i$ appears odd number of times in $\displaystyle{\prod\limits_{e\in E} \term(e)}$. Note that the map $a^t_i \mapsto - a^t_i$ is a probability preserving map. Hence, in this case 
\[\mathbb{E} \prod\limits_{e\in E} \term(e) = \mathbb{E} \left(-\prod\limits_{e\in E} \term(e) \right) = 0. \] 
Thus, we may assume that every vertex in $G$ has even degree. Let $d_2$ be the number of vertices of degree 2. Counting the edges yields 
\[ 2d_2+4(|V|-d_2)\leq 2|E|\quad \Rightarrow \quad 2|V|\leq |E|+d_2.\]
We prove the statement of the theorem by induction on $|V|$. If $|V| = 1$ the statement is obvious.

By Fact~\ref{fact:inner-produc-hp-bound}, w.h.p. for each $e\in E$ we have $\vert \term(e)\vert =\wt{O}(1/\sqrt{n})$. Hence, by union bound, for $|E| =  O(\log(n)^4)$, 
\begin{equation}
 \left\vert \mathbb{E} \prod\limits_{e\in E} \term(e) \right\vert  = \left(\wt{O}\left( \dfrac{1}{\sqrt{n}}\right)\right)^{|E|}. 
\end{equation}
If $d_2 = 0$, then $|E|/2>|V|-1$, so the statement of the theorem follows from the bound above. Otherwise, there exists a vertex $i$ of degree 2. Let $\{i, j\}$ and $\{i, k\}$ be the edges incident with vertex $i$. If these edges have different labels $t_1$ and $t_2$, then $a^{t_1}_i$ appears only once, so expectation is zero. Hence, we may assume that edges $\{i, j\}$ and $\{i, k\}$ have the same label $t$. Define $G' = (V', E')$ to be the graph obtained from $G$ by deleting vertex $i$, and adding edge between $j$ and $k$ labeled $t$, if $j\neq k$. Then, by Fact~\ref{fact:inner-products}, 
\[ \left\vert \mathbb{E} \prod\limits_{e\in E} \term(e) \right\vert  = \frac{1}{n}\left\vert \mathbb{E} \prod\limits_{e\in E'} \term(e) \right\vert. \]
Moreover, $|E'|\geq |E|-2$, $|V'| = |V|-1$ and $G'$ is connected. Hence, the statement of the theorem follows by induction.
\end{proof}

\begin{theorem}\label{thm:main-diagram-tool} 
Let $\mathcal{C}$ be a collection of subsets of $[r]$, such that every $t\in[r]$ belongs to precisely $p$ sets from $\mathcal{C}$. Let $k$ be the number of sets in $\mathcal{C}$. Assume that $\mathcal{U} = \{a_i^t\in S^{n-1} | i\in [m], t\in [r]\}$ is a set of independent uniformly distributed random vectors and  $R = (G, \mathcal{U}, \Phi)$ is an IP graph matrix from the class $\mathfrak{BCM}^v(\mathcal{C}; \alpha, D)$. If $m\ll n^{k/p}$, then for all $q = O(\log(n)^2)$, w.h.p. 
\[ \left\vert \Tr(RR^T)^q \right\vert = m^D(vq)^{vq+1}\wt{O}\left(\dfrac{m^{p/k}}{n}\right)^{\alpha q}.\]
(The variables $v$, $\alpha$, $r$, $k$, $p$ and $D$ are assumed to be constants).  
\end{theorem}
\begin{proof}  Let $\Phi\subseteq \{\phi: \myscr{Ver} \rightarrow [m]\}$ be a set of permitted labelings of $\mathcal{TD}_q(R) = (\myscr{Ver}, E, \mathfrak{c})$. Then, by definition, 
\[ \Tr(RR^T)^q = \sum\limits_{\phi\in \Phi} val(\mathcal{TD}_q(R, \phi)),\quad \text{where} \quad val(\mathcal{TD}_q(R, \phi)) = \prod\limits_{e = \{x, y\} \in E}\langle a^{\mathfrak{c}(e)}_{\phi(x)}, a^{\mathfrak{c}(e)}_{\phi(y)} \rangle. \]
Let $n_\phi = \left\vert \phi(\myscr{Ver}) \right\vert $ be the size of the image of $\phi$. Let $G_{i, \phi}$ be the graph on the set of vertices $\phi(\myscr{Ver})$, in which there are $t$ edges between a pair of distinct vertices $j$ and $j'$, if there are precisely $t$ edges $e = \{x, y\}$ in $\mathcal{TD}_q(R)$ such that $\{\phi(x), \phi(y)\} = \{j, j'\}$ and $\mathfrak{c}(e) \in C_i$. 

Denote by $h_i$ the number of non-equality edges of $\mathcal{TD}_q(R)$ with colors in $C_i$. Then 
\[\sum\limits_{i = 1}^{k} h_i = p(2\alpha q).\]
Clearly, for every $i\in[k]$ and $\phi\in \Phi$ the number of edges in $G_{i, \phi}$ is at least $h_i$. 

For $e = \{x, y\}$ in $\mathcal{TD}_q(R, \phi)$ define 
\[ \term_{\phi}(e) = \langle a^{\mathfrak{c}(e)}_{\phi(x)}, a^{\mathfrak{c}(e)}_{\phi(y)} \rangle. \]
By Theorem~\ref{thm:diagram-bound-tool} applied to connected components of $G_{i, \phi}$,
\[ \left\vert \mathbb{E}\left[ \prod\limits_{e\in E,\, \mathfrak{c}(e) \in C_i} \term_{\phi}(e) \right] \right\vert  = \left(\wt{O}\left( \dfrac{1}{n}\right)\right)^{\max(n_{\phi}-D, h_i/2)}. \] 
Observe that the multiset equality $pE = C_1 \sqcup C_2 \sqcup \ldots  \sqcup C_k$ implies that
\[ \prod\limits_{i = 1}^k  \mathbb{E}\left[ \prod\limits_{e\in E,\, \mathfrak{c}(e) \in C_i} \term_{\phi}(e) \right] = \left(\mathbb{E}\left[ \prod\limits_{e\in E} \term_{\phi}(e) \right] \right)^p = \left(\mathbb{E} \left[val(\mathcal{TD}_q(R, \phi))\right]\right)^p,  \]
since edges with distinct colors are independent. Thus,
\[ \left\vert \mathbb{E} \left[val(\mathcal{TD}_q(R, \phi))\right]\right\vert^p \leq \left(\wt{O}\left( \dfrac{1}{n}\right)\right)^{\displaystyle{\max(k (n_{\phi}-D), p\alpha q)}}\]
At the same time, there are at most $m^jj^{\vert\myscr{Ver}\vert}\leq m^j|\myscr{Ver}|^{|\myscr{Ver}|}$ maps $\phi$ with $n_{\phi} = j$. Therefore, 
\begin{equation}
\begin{gathered}
\left\vert \mathbb{E} \Tr(RR^T)^q \right\vert \leq \sum\limits_{j = 1}^{|\myscr{Ver}|}\ \sum\limits_{\phi\in \Phi,\, n_{\phi} = j} \left\vert \mathbb{E} \left[val(\mathcal{TD}_q(R, \phi))\right]\right\vert \leq \\
 \leq \sum\limits_{j = 1}^{|\myscr{Ver}|} m^{j}|\myscr{Ver}|^{|\myscr{Ver}|}\left(\wt{O}\left( \dfrac{1}{n}\right)\right)^{ \displaystyle{\max\left(\dfrac{k}{p}( j-D), \alpha q\right)}}
\end{gathered}
\end{equation}
The expression under the sum sign is maximized when $j = (p/k)\cdot \alpha q+D$ if $m \ll n^{k/p}$. Hence, the statement of the theorem follows.    
\end{proof}

\begin{remark} Note that in the proof above we use that $\sum\limits_{i=1}^{k} \max(n_{\phi}-D, h_i/2)$ is at least $\max\left(k(n_{\phi}- D),\sum\limits_{i=1}^{k}h_i/2\right)$. In many cases this bound is not the most efficient. In particular, in the special case, when $p=1$ and all non-equality edges belong to some $C_i$, the valid lower bound $\max(k(n_{\phi}- D), (k-1)(n_{\phi}-D)+\alpha q)$ is more efficient than $\max(k(n_{\phi}- D), \alpha q)$ and will yield trace upper bound of order $\left(\dfrac{m}{n^k}\right)^{\alpha q}$ instead of $\left(\dfrac{m^{1/k}}{n}\right)^{\alpha q}$. The latter bound is less efficient if $m\ll n^k$.
\end{remark}

In what follows we will use $\{u, v, w\}$ as a set of possible colors instead of $\{1, 2, 3\}$. We will typically apply the theorem above for collections $\mathcal{C}_{2/3} = \{\{u, v\}, \{u, w\}, \{v, w\}\}$, $\mathcal{C}_{u} = \{\{u\}, \{v, w\}\}$, $\mathcal{C}_{v} = \{\{v\}, \{u, w\}\}$ or $\mathcal{C}_{w} = \{\{w\}, \{u, v\}\}$.

\subsection{Expanded matrix diagrams}

In the definition of a matrix diagram we have two types of edges: edges between nodes and half-edges between a node and a cross. Note that an edge $e$ between nodes $i$ and $j$ corresponds to some inner product $\langle a_i, a_j \rangle$, which can be equivalently written as $\sum\limits_{t=1}^{n} \langle a_i, e_t\rangle \langle e_t, a_j\rangle$. Therefore, we may add a cross $t$ in the middle of the edge $e$ and now treat this edge as a pair of half-edges $\{i, t\}$ and $\{t, j\}$, where cross $t$ is allowed to take all values from $[n]$. 

In Sections~\ref{sec:Omega-dual-certificate} and~\ref{sec:SOS-Omega-certificate} we will be interested in expressions defined by matrix diagrams, where the values for some crosses will be restricted to a given subset $\Omega$. To accommodate such expressions we give a more general definition of an IP graph matrix.

\begin{definition}
Let $\myscr{Ver} = (\myscr{Ver}_{L}\cup \myscr{Ver}_{R})\sqcup \myscr{Ver}_{I} $, where $\myscr{Ver}_{L} = \myscr{Cr}_{L}\sqcup \myscr{Nod}_{L}$ and $\myscr{Ver}_{R} = \myscr{Cr}_{R}\sqcup\myscr{Nod}_{R}$, and $\myscr{Nod}_I = \myscr{Ver}_I\sqcup \myscr{Cr}_I$. Furthermore, the sets $\myscr{Ver}_{L}$ and $\myscr{Ver}_{R}$ are ordered.

Let $r$ be a positive integer. Consider a colored bipartite graph $G = (\myscr{Ver}, E, \mathfrak{c})$ such that 
\begin{itemize}
\item each vertex in $\myscr{Cr}_{L}\cup \myscr{Cr}_{R}$ has degree 1, and each vertex in   $\myscr{Cr}_{I}$ has degree 2;
\item $G$ has parts $\myscr{Nod}$ and $\myscr{Cr}$;
\item the map $\mathfrak{c}: \myscr{Cr} \rightarrow [r]$ assigns a color for each cross of $G$.  
\end{itemize}
We call $G$ an \emph{expanded IP matrix diagram} or just an  \textit{expanded matrix diagram}.
\end{definition}

One can get an expanded matrix diagram from a matrix diagram by introducing a cross in the middle of every edge (which is not a half-edge) and by assigning this cross the color of the corresponding edge.

Next we generalize the definition of permitted labeling and IP graph matrix.

\begin{definition} Let $G = (\myscr{Ver}, E, \mathfrak{c})$ be a matrix diagram and $\mathcal{U} = \{a_{i}^{t}\in S^{n_t-1}\mid i\in[m], t\in[r]\}$ be a collection of $rm$ unit vectors.  Let 
\[\Phi\subseteq \{\phi:\myscr{Ver} \rightarrow [m]\cup [n]\mid \phi(\myscr{Nod})\subseteq [m],\ \phi(\mathfrak{c}^{-1}(t))\subseteq [n_t] \text{ for } t\in [r]\}\]
 be a \emph{set of permitted labelings}. 

Fix an orthonormal basis $\{f_i\}$ of $\mathbb{R}^m$ and $\{e_i^t\}$ of $\mathbb{R}^{n_t}$. For convenience, define $\mathfrak{c}(nod) = *$ for $nod \in \myscr{Nod}$ and denote $e^*_i = f_i$. 
 
 The  \emph{inner product (IP) graph matrix} corresponding to $(G, \mathcal{U}, \Phi)$ is defined as
\[ IGM(G, \mathcal{U}, \Phi) = \sum\limits_{\phi\in \Phi}\left(\prod\limits_{e = \{x, y\} \in E, x\in \myscr{Nod}, y\in \myscr{Cr}}\left\langle a^{\mathfrak{c}(y)}_{\phi(x)}, e^{\mathfrak{c}(y)}_{\phi(y)} \right\rangle \left(\bigotimes\limits_{\ell \in \myscr{Ver}_L} e^{\mathfrak{c}(\ell)}_{\phi(\ell)}\right)\left(\bigotimes\limits_{p \in \myscr{Ver}_R} e^{\mathfrak{c}(p)}_{\phi(p)}\right)^T\right), \]
where tensor products respect the order of $\myscr{Ver}_L$ and $\myscr{Ver}_R$.

We say that that a matrix $M$ is an \emph{inner product graph matrix} if there exists some $(G, \mathcal{U}, \Phi)$ such that $M = IGM(G, \mathcal{U}, \Phi)$. We say that $G$ is an expanded matrix diagram for $M$ and denote it $\mathcal{MD}^{*}(M)$.

\end{definition}

It is not hard to see that if labels of crosses are not restricted, then the definition of an IP graph matrix for an expanded diagram is consistent with the definition of an IP graph matrix for a non-expanded diagram.
\begin{observation} Let $G$ be a matrix diagram with the set of permitted labelings $\Phi$. Let $G^*$ be the expanded matrix diagram obtained from $G$ by adding a cross in the middle of every edge. Denote the set of vertices of $G^*$ by $\myscr{Ver}$. For simplicity, assume that every vector in collection $\mathcal{U}$ has dimension $n$. Consider 
\[\Phi^* = \{\phi:\myscr{Ver} \rightarrow [n]\cup[m]\mid \phi|_{\myscr{Nod}} \in \Phi\} = \Phi\times [n]^{|\myscr{Cr}|}.\]
Then $IGM(G, \mathcal{U}, \Phi) = IGM(G^*, \mathcal{U}, \Phi^*)$.
\end{observation}

One can also check that definitions of compatible matrices, trace diagram, $\mathcal{C}$-connectivity and $\mathcal{C}$-boundary-connectivity transfer with essentially no changes to expanded matrix diagrams. Hence, we will not repeat these definitions explicitly.

Moreover, the techniques and theorems described in Section~\ref{sec:trace-bounds-graph-matr} can be also easily transferred to the case of expanded matrix diagrams. In particular, the following analog of Theorem~\ref{thm:diagram-bound-tool} will be useful in Sections~\ref{sec:Omega-dual-certificate} and~\ref{sec:SOS-Omega-certificate}.

\begin{lemma}\label{lem:matix-diagram-tool-expanded} Let $G = (\myscr{Cr}\cup \myscr{Nod}, E)$ be a bipartite graph with parts $\myscr{Cr}\subseteq [n]\times [r]$ and $\myscr{Nod}\subseteq [m]$, and let $\{a_{i}^{t}\mid i\in [m],\ t\in [r]\}$ be a collection of independent random vectors, uniformly distributed on $S^{n-1}$. Let $\{e_k\mid k\in [n]$ be a fixed orthonormal basis in $\mathbb{R}^n$.  Assume that $G$ has $v(G)$ vertices, $e(G)$ edges and $d(G)$ connected components. Suppose $|E| = O(\log(n)^4)$. Then
\[ val(G) = \left\vert \mathbb{E} \prod\limits_{\{(x, t), i\}\in E:\, (x, t)\in \myscr{Cr},\, i\in \myscr{Nod}} \left\langle a^t_{i}, e_x \right\rangle \right\vert \leq \left(\wt{O}\left(\dfrac{1}{n}\right)\right)^{\displaystyle \max(v(G)-d(G), e(G)/2)}.\]
\end{lemma}
\begin{proof} The proof is very similar to the proof of Theorem~\ref{thm:diagram-bound-tool}. First observe that flipping of the $k$-th coordinate of $a_i^t$: $(a_i^t)_k \mapsto -(a_i^t)_k$ preserves the uniform distribution on a sphere. Therefore, if any edge in $G$ appears in odd degree, then $val(G) = 0$. 

Therefore, we may assume that every edge in $G$ has even multiplicity. If there exists a vertex of degree 2 in $G$, then it is incident with a repeated twice edge, and all other edges in $G$ are independent of this one. Hence, if $G'$ is a graph obtained by deleting this vertex, then by Fact~\ref{fact:inner-produc-hp-bound} 
\[ val(G)\leq val(G')\wt{O}\left(\dfrac{1}{n}\right).\]
Moreover, $G'$ has 1 less vertex and 2 less edges. Therefore, the statement follows by induction.
\end{proof}
\begin{remark}\label{graphmatrixremark}
If we consider expanded IP graph matrices where all of the nodes and crosses must be distinct, this is a special case of the graph matrices studied in \cite{graph-matrix-bounds}. This gives us an alternative way to analyze expanded IP graph matrices. We first split the analysis into cases based on which nodes and crosses are equal to each other and then use the graph matrix norm bounds in \cite{graph-matrix-bounds}. This method allows us to analyze any given expanded IP graph matrix, but it has the disadvantage that there may be a lot of cases. For more details and for examples of this method, see Appendix \ref{graphmatrixappendix}.
\end{remark}

\section{Construction of the certificate candidate}\label{sec:corr-terms}
The theorem below was claimed by Li, Prater, Shen and Tang {\cite[Lemmas 1-3]{Li-et-al}}. However, we believe that their proof has some flaws. We use a completely different technique. 
\begin{theorem}\label{thm:candidate-exists} Let $\mathcal{V} = \{(u_i, v_i, w_i)\}_{i=1}^{m}$ be a collection of $3m$ independent uniformly distributed on $S^n$ random vectors. For $m  \ll n^2$, w.h.p. over the randomness of $\mathcal{V}$, there exist $3m$ vectors $\{(\alpha_i, \beta_i, \gamma_i)\}_{i=1}^{m}$ such that the following conditions hold. The tensor
\begin{equation}\label{eq:certificate-3-form}
 \oa{A} = \sum\limits_{i=1}^{m} u_i\otimes v_i \otimes w_i+ \sum\limits_{i=1}^{m} \alpha_i\otimes v_i \otimes w_i+ u_i\otimes \beta_i \otimes w_i+ u_i\otimes v_i \otimes \gamma_i
 \end{equation}
satisfies
\begin{equation}\label{eq:orthog-u}
\sum\limits_{j, k=1}^{m} \oa{A}_{(i, j, k)} (v_{t})_j (w_{t})_k = (u_{t})_i, \quad\quad\quad \text{for all } i\in [n], \ t\in [m],   
\end{equation}
\begin{equation} 
\sum\limits_{i, k=1}^{m} \oa{A}_{(i, j, k)} (u_{t})_i (w_{t})_k = (v_{t})_j, \quad\quad\quad \text{for all } j\in [n], \ t\in [m],   
\end{equation}
\begin{equation}\label{eq:orthog-w} 
\sum\limits_{i, j=1}^{m} \oa{A}_{(i, j, k)} (u_{t})_i(v_{t})_j = (w_{t})_k, \quad\quad\quad \text{for all } k\in [n], \ t\in [m].   
\end{equation}
Moreover, the matrices $U' = [\alpha_1, \ldots, \alpha_m]$, $V' = [\beta_1, \ldots, \beta_m]$ and $W' = [\gamma_1, \ldots, \gamma_m]$ satisfy
\begin{equation}\label{eq:error-term-bound}
 \max\left(\left\Vert U'\right\Vert, \left\Vert V'\right\Vert, \left\Vert W'\right\Vert\right)\leq \widetilde{O}\left(\frac{\sqrt{m}}{n}+\frac{m}{n^{3/2}}+\dfrac{m^2}{n^3}\right).
\end{equation} 
\end{theorem}

\subsection{System of linear equations for candidate certificates in a matrix form}\label{sec:M-system} 

Let $e_1, e_2, \ldots e_n$ be an orthonormal basis of $\mathbb{R}^n$ and $\{f^t_i \mid t\in[3], i\in [m]\}$ be an orthonormal basis of $\mathbb{R}^m\oplus\mathbb{R}^m\oplus \mathbb{R}^m$.
We define three $nm\times n^3$ matrices
\[ M_1 = \sum\limits_{j=1}^{m}\sum\limits_{t=1}^{n} (e_t\otimes f^1_j)(e_t\otimes v_j\otimes w_j)^T,\]
\[ M_2 = \sum\limits_{j=1}^{m}\sum\limits_{t=1}^{n} (e_t\otimes f^2_j)(u_j\otimes e_t\otimes w_j)^T,\]
\[ M_3 = \sum\limits_{j=1}^{m}\sum\limits_{t=1}^{n} (e_t\otimes f^3_j)(u_j\otimes v_j\otimes e_t)^T.\]
Let $\mvec{U} = \sum\limits_{j}  u_j\otimes f^1_j$, $\mvec{V} = \sum\limits_{j} v_j\otimes f^2_j$ and $\mvec{W} = \sum\limits_{j} w_j \otimes f^3_j$.  We also introduce 
\begin{equation}\label{eq:M-definition}
 M = \left(\begin{matrix} M_1 \\
M_2 \\
M_3
\end{matrix}\right) \quad D = \left(\begin{matrix} \mvec{U} \\
\mvec{V} \\
\mvec{W}
\end{matrix}\right)
\end{equation}
Then conditions, Eq.~\eqref{eq:orthog-u}-\eqref{eq:orthog-w} on $\oa{A}$ are equivalent to
\begin{equation}\label{eq:A-system}
 M\oa{A} = D.
 \end{equation}
We will search for $\oa{A}$ of the form 
\[ \oa{A} = M^T\left(\frac{1}{3}D+Y\right), \quad \text{where } Y\in \mathbb{R}^{3mn}.\]
Then Eq.~\eqref{eq:A-system} takes the form
\begin{equation}\label{eq:Y-system}
\begin{gathered}
(MM^T)Y = D-\frac{1}{3}MM^TD = D - M\left(\sum\limits_{j}^{m} u_j\otimes v_j\otimes w_j \right) =\\
 = -\left(\begin{matrix}
  \sum\limits_{i=1}^{m}\sum\limits_{j:\ j\neq i} (u_j\otimes f_i^1)\langle v_j, v_i \rangle \langle w_j, w_i \rangle\\
 \sum\limits_{i=1}^{m}\sum\limits_{j:\ j\neq i} (v_j\otimes f_i^2)\langle u_j, u_i \rangle \langle w_j, w_i \rangle\\
 \sum\limits_{i=1}^{m}\sum\limits_{j:\ j\neq i} (w_j\otimes f_i^3)\langle v_j, v_i \rangle \langle u_j, u_i \rangle
 \end{matrix} \right) =: E = \left(\begin{matrix}
 E_1 \\
 E_2 \\
 E_3
 \end{matrix}\right).
 \end{gathered}
\end{equation}

To show that this equation has a solution we need to show that $E\in \Ran(MM^T)$. Since $MM^T$ is symmetric, this is equivalent to $E\perp \Ker(MM^T)$.

\subsection{Approximation of $MM^T$ with a simpler matrix}\label{sec:M-approx}
To analyze equation Eq.~\eqref{eq:Y-system}, we show that $MM^T$ is a small perturbation of the matrix
\begin{equation}
R = I+\left(\begin{matrix} 0 & F_{12} & F_{13}\\
F_{21} & 0 & F_{23} \\
F_{31} & F_{32} & 0
\end{matrix}\right), \quad \text{ where } F_{sq} = F_{qs}^T, \text{  and} 
\end{equation}
\[F_{12} = \sum\limits_{j=1}^{m} (u_j\otimes f^1_j)(v_j\otimes f^2_j)^T, \quad F_{13} = \sum\limits_{j=1}^{m} (u_j\otimes f^1_j)(w_j\otimes f^3_j)^T, \quad F_{23} = \sum\limits_{j=1}^{m} (v_j\otimes f^2_j)(w_j\otimes f^3_j)^T\]

First, we show that the spectrum of $R$ has very simple structure. 
\begin{lemma}\label{lem:R-eigenspaces}  The eigenvalues of $R$ are $0$, $1$ and $3$. Moreover, the subspaces $\mathcal{K}$, $\mathcal{I}$, $\mathcal{D}$ defined as
\begin{equation}
 \mathcal{K} = \vspan \{u_j\otimes f^{1}_j - v_j\otimes f^{2}_j,\  u_j\otimes f^{1}_j - w_j\otimes f^{3}_j,\  v_j\otimes f^{2}_j - w_j\otimes f^{3}_j\mid j\in[m]\}. 
 \end{equation}
 \begin{equation}
 \mathcal{D} = \vspan \{u_j\otimes f^{1}_j + v_j\otimes f^{2}_j + w_j\otimes f^{3}_j\mid j\in[m]\}
 \end{equation}
and $\mathcal{I} = \left(\mathcal{K}\cup \mathcal{D}\right)^{\perp}$ are  the 0-, 3-, and 1-eigenspace of $R$, respectively.
\end{lemma}
\begin{proof}
Observe that $F_{ij}F_{jk} = F_{ik}$, for $i, j, k\in [3]$, where $F_{11} = \sum\limits_{j=1}^{m} (u_j\otimes f^1_j)(u_j\otimes f^1_j)^T$ and $F_{22}$, $F_{33}$ are defined similarly.
It is easy to see that
\[R^2-3R+2I = (R-I)^2 - (R-I) = 2\left(\begin{matrix}
F_{11} & 0 & 0 \\
0 & F_{22} & 0 \\
0 & 0 & F_{33}
\end{matrix}\right) \qquad \text{ and}\]
\[ (R-I)\left(\begin{matrix}
F_{11} & 0 & 0 \\
0 & F_{22} & 0 \\
0 & 0 & F_{33}
\end{matrix}\right) = (R-I). \]
Hence,
\[ (R-I)(R^2-3R+2I)-2(R-I) = 0\qquad \Rightarrow \qquad  R(R-I)(R-3I) = 0.\]
Thus, the eigenvalues of $R$ are $0$, $1$ and $3$.

Observe that every vector in $\mathcal{K}$ is a $0$-eigenvector of $R$ and every vector in $\mathcal{D}$ is a $3$-eigenvector of $R$. Moreover, $\dim(\mathcal{K}) = 2m$ and $\dim(\mathcal{D}) = m$. Denote by $d_0$, $d_1$ and $d_3$ the multiplicities of 0, 1, and 3 eigenvalues of $R$. Then we have the following constraints,
\[ d_0+d_1+d_3 = 3mn, \qquad d_1+3d_3 = \Tr(R) = 3mn,\quad \text{and} \]
\[ \quad 2d_0+2d_3 = \Tr(R^2-3R+2I) = 2\cdot 3m = 6m\]
Hence, $d_0 = 2m$ and $d_3 = m$.

Therefore, $\mathcal{K}$ coincides with the $0$-eigenspace and $\mathcal{D}$ coincides with the $3$-eigenspace.
\end{proof}

\begin{proposition}\label{prop:M-cross-approx}
Let $m\ll n^2$. For $s, q\in [3]$ with $s\neq q$,  with high probability 
\[\left\Vert M_sM_q^T - F_{sq} \right\Vert = \wt{O}\left(\dfrac{\sqrt{m}}{n}\right).\]
\end{proposition}
\begin{proof} Without loss of generality we may assume that $s = 1$ and $q=2$. Observe that
\begin{equation}\label{eq:s12-expansion} S_{12}:= M_1M_2^T - F_{12} = \sum\limits_{j\neq j'} \langle w_j, w_{j'} \rangle (u_{j'}\otimes f_j)(v_j\otimes f_{j'})^T
\end{equation}
and $S_{12}$ has the matrix and trace diagram given by Figure~\ref{fig:S12-diagram}.
\begin{figure}
\begin{center}
\begin{subfigure}[b]{0.2\textwidth}
\begin{center}
\includegraphics[scale=0.9]{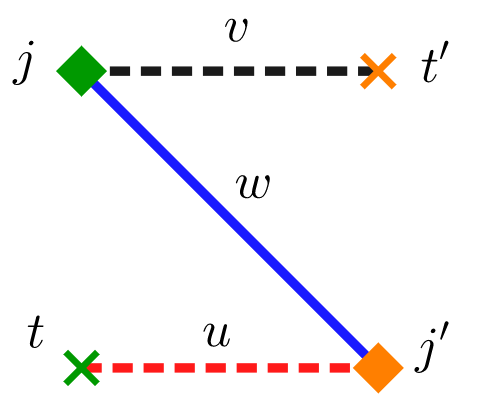}
\end{center}
\end{subfigure}
\begin{subfigure}[b]{0.75\textwidth}
\begin{center}
\includegraphics[scale=0.9]{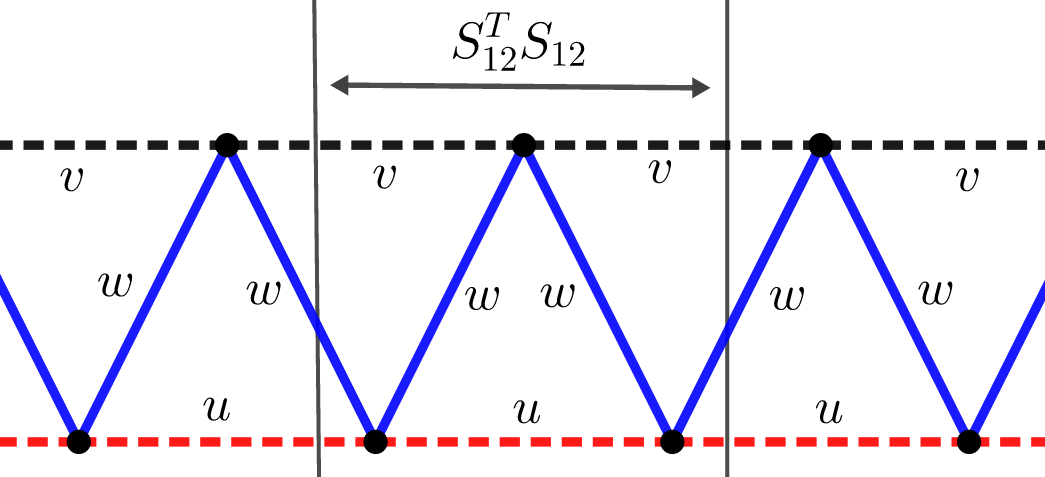}
\end{center}
\end{subfigure}
\caption{The matrix diagram for $S_{12}$ and the trace diagram for $S_{12}^TS_{12}$.}\label{fig:S12-diagram}
\end{center}
\end{figure}
We proceed similarly as in the proof of Theorem~\ref{thm:main-diagram-tool} except that to get a better bound, we take into account that all non-equality edges are of color $w$. Using independence, we can write
\[ \mathbb{E}\Tr\left((S_{12}^TS_{12})^q\right) = \mathbb{E}\sum\limits_{\phi\in \Phi} val(\mathcal{TD}_q(S_{12}, \phi)) = \sum\limits_{\phi\in \Phi} \mathbb{E}\left(\prod\limits_{e\in E_w} \term_{\phi}(e)\right)\mathbb{E}\left(\prod\limits_{e\in E_{uv}} \term_{\phi}(e)\right).\]
We apply Theorem~\ref{thm:diagram-bound-tool} to graphs $G_{\{w\}, \phi}$ and $G_{\{u, v\}, \phi}$ induced on labels of $\phi$ by the edges $E_w$ and $E_{uv}$ of colors $w$ and $\{u,v\}$, respectively (with loops being deleted). We obtain
\[ \left\vert\mathbb{E}\left(\prod\limits_{e\in E_w} \term_{\phi}(e)\right) \right\vert\leq \wt{O}\left(\dfrac{1}{n}\right)^{\max(n_{\phi}-1, q)}\quad  \text{and} \quad \left\vert\mathbb{E}\left(\prod\limits_{e\in E_{uv}} \term_{\phi}(e)\right)\right\vert \leq \wt{O}\left(\dfrac{1}{n}\right)^{n_{\phi}-2},\]
where $n_{\phi}$ is the size of the image of $\phi$. Note that there are at most $m^{n_{\phi}}n_{\phi}^{2q}$ labelings of $\mathcal{TD}_q(S_{12})$ using only $n_{\phi}$ labels from $[m]$. Therefore,
\[ \left\vert\mathbb{E}\Tr\left((S_{12}^TS_{12})^q\right)\right\vert \leq \sum\limits_{j=1}^{2q} m^{j}(2q)^{2q} \wt{O}\left(\dfrac{1}{n}\right)^{\max(2j-3, q+j-2)}\]
Since, $m\ll n^2$, the expression under the sum sign is maximized for $j = q+1$, so
\begin{equation}
 \left\vert\mathbb{E}\Tr\left((S_{12}^TS_{12})^q\right)\right\vert \leq nm(2q)^{2q+1}\wt{O}\left(\dfrac{\sqrt{m}}{n}\right)^{2q}.
\end{equation}
Hence, taking $q = O(\log(n)^2)$, the statement of the theorem follows from the power-trace method  (see Lemma~\ref{lem:trace-power-method-norm}).
\end{proof}

\begin{proposition}\label{prop:S-bound} Let $m\ll n^2$ and $\mathcal{E}_M = MM^T -R$. With high probability
\begin{equation}
\Vert \mathcal{E}_M \Vert = \wt{O}\left(\frac{\sqrt{m}}{n}\right).
\end{equation}
\end{proposition}
\begin{proof}
It is sufficient to show for all $s, q\in [3]$ that 
\[\Vert M_sM_q^T - F_{sq} \Vert = \wt{O}\left(\frac{\sqrt{m}}{n}\right)\quad \text{ and } \quad \Vert M_sM_s^T-I \Vert = \wt{O}\left(\frac{\sqrt{m}}{n}\right).\]
Note that the the first inequality follows from Proposition~\ref{prop:M-cross-approx}. To show the second inequality, note that $\Vert A\otimes B\Vert  = \Vert A\Vert \cdot \Vert B \Vert$. Thus, for $s=1$, by Lemma~\ref{lem:basic-cten-bound}, 
\[\Vert M_1M_1^T-I_{mn}\Vert = \Vert I_n\otimes \left((V\cten W)^{T}(V\cten W)-I_m \right)\Vert = \wt{O}\left(\dfrac{\sqrt{m}}{n}\right).\]
 For $s=2$ and $s=3$ the proof is similar.
\end{proof}


\begin{lemma}\label{lem:MM^T-kernel} If $m \ll n^2$, then the kernels of $MM^T$ and $R$ coincide and are equal to $\mathcal{K}$.
\end{lemma}
\begin{proof}
Observe that $MM^T(u_j\otimes f^{1}_j - v_j\otimes f^{2}_j) = 0$. The eigenvalues of $R$ are 0, 1, and 3 and for $\mathcal{E}_M = MM^T-R$ we have $\Vert \mathcal{E}_M \Vert <1$, so $\Ker(MM^T) = \Ker(R)$. 
\end{proof}

%

\begin{observation}\label{obs:E-not-in-kernel}
The vector 
$E$ defined in Eq.~\eqref{eq:Y-system} belongs to $\mathcal{K}^{\perp}$.
\end{observation}

Therefore, we deduce the following theorem.
\begin{theorem}
If $m \ll n^2$, then w.h.p. a solution to Eq.~\eqref{eq:Y-system} exists and one can take
\begin{equation}\label{eq:Y-formula}
 Y = (MM^T)_{\mathcal{K}^{\perp}}^{-1}E.
 \end{equation}
\end{theorem}
 
 \subsection{Norm bounds for correction terms $U'$, $V'$, $W'$}
 
 Now we would like to think of $Y = \left(\begin{matrix} Y_1\\ Y_2\\ Y_3 \end{matrix}\right)$ as of vector for which $Y_1$, $Y_2$ and $Y_3$ are the vectorizations (reshaping to $mn\times 1$) of three $n\times m$ matrices $U'$, $V'$ and $W'$, respectively. Our goal is to show that $U'$, $V'$ and $W'$ have small norm for $Y$ given by Eq.~\eqref{eq:Y-formula}.
 
 \begin{observation} $\Vert S \Vert = \max\limits_{\Vert x\Vert = \Vert a\Vert =1} x^TSa = \max\limits_{\Vert x\Vert = \Vert a\Vert =1} \langle S_{\mathfrak{v}}, x\otimes a  \rangle$, where $S$ is an $m\times n$ matrix, $S_{\mathfrak{v}}$ is its vectorization, and $x\in\mathbb{R}^n$, $a\in \mathbb{R}^m$.  
 \end{observation}

 \begin{lemma}\label{lem:E-norm-bound} With high probability we have $\Vert E \Vert = \wt{O}\left(\dfrac{m}{n}\right)$.
 \end{lemma}
 \begin{proof} It is sufficient to show that 
 \[ \left\Vert \sum\limits_{i=1}^{m}\sum\limits_{j:\ j\neq i} (u_j\otimes f^1_i)\langle v_j, v_i \rangle \langle w_j, w_i \rangle\right\Vert = \wt{O}\left(\dfrac{m}{n}\right).\]
 Since $\{f^1_i\}_{i=1}^{m}$ is orthonormal, it is sufficient to show that 
 \[ \left\Vert \sum\limits_{j:\ j\neq i} u_j\langle v_j, v_i \rangle \langle w_j, w_i \rangle\right\Vert = \wt{O}\left(\dfrac{\sqrt{m}}{n}\right). \]
 The latter bound follows from Matrix Bernshtein inequality as w.h.p
 $\vert \langle v_j, v_i \rangle \langle w_j, w_i \rangle \vert = \wt{O}\left(1/n\right)$ for $i\neq j$ and 
 \[ \sum\limits_{j:\, j\neq i} \langle v_j, v_i \rangle^2 \langle w_j, w_i \rangle^2 = \wt{O}\left(\dfrac{m}{n^2}\right)\quad \text{and} \quad \left\Vert\sum\limits_{j:\, j\neq i} u_ju_j^T\langle v_j, v_i \rangle^2 \langle w_j, w_i \rangle^2 \right\Vert= \wt{O}\left(\dfrac{m}{n^2}\right). \]
 \end{proof} 
 
 \begin{lemma}\label{lem:Y-simplification}
 Assume that $m \ll n^2$, then for $Y$ given by Eq.~\eqref{eq:Y-formula} we have  
 \[\left\Vert Y - \left(2R_{\mathcal{K}^{\perp}}^{-1}-R_{\mathcal{K}^{\perp}}^{-1}(MM^T)R_{\mathcal{K}^{\perp}}^{-1}\right)E \right\Vert  = \wt{O}\left( \dfrac{m^2}{n^3} \right)\]
 \end{lemma}
 \begin{proof}
 We can write
\begin{equation}
\begin{split}
 (MM^T)_{\mathcal{K}^{\perp}}^{-1} =& (R_{\mathcal{K}^{\perp}}+(\mathcal{E}_M)_{\mathcal{K}^{\perp}})^{-1} = \left(R_{\mathcal{K}^{\perp}}(I+R_{\mathcal{K}^{\perp}}^{-1}(\mathcal{E}_M)_{\mathcal{K}^{\perp}})\right)^{-1} =\\
 =& R_{\mathcal{K}^{\perp}}^{-1}-R_{\mathcal{K}^{\perp}}^{-1}(\mathcal{E}_M)_{\mathcal{K}^{\perp}}R_{\mathcal{K}^{\perp}}^{-1}+ \left(R_{\mathcal{K}^{\perp}}^{-1}(\mathcal{E}_M)_{\mathcal{K}^{\perp}}\right)^2\left(R_{\mathcal{K}^{\perp}}+(\mathcal{E}_M)_{\mathcal{K}^{\perp}}\right)^{-1}.
 \end{split}
 \end{equation}
Observe that the following inequality is implied by Proposition~\ref{prop:S-bound},
\begin{equation}\label{eq:inverse-sec-order-bound}
 \left\Vert \left(R_{\mathcal{K}^{\perp}}^{-1}(\mathcal{E}_M)_{\mathcal{K}^{\perp}}\right)^2\left(R_{\mathcal{K}^{\perp}}+(\mathcal{E}_M)_{\mathcal{K}^{\perp}}\right)^{-1} \right\Vert \leq 2\left\Vert (\mathcal{E}_M)_{\mathcal{K}^{\perp}}\right \Vert^2 = \wt{O}\left(\dfrac{m}{n^2}\right).
\end{equation}
Combining bounds from Lemma~\ref{lem:E-norm-bound} and Eq.~\eqref{eq:inverse-sec-order-bound} we obtain
 \[ \left\Vert R_{\mathcal{K}^{\perp}}^{-1}(\mathcal{E}_M)_{\mathcal{K}^{\perp}}R_{\mathcal{K}^{\perp}}^{-1}(\mathcal{E}_M)_{\mathcal{K}^{\perp}}\left(R_{\mathcal{K}^{\perp}}+(\mathcal{E}_M)_{\mathcal{K}^{\perp}}\right)^{-1} E \right\Vert  = \wt{O}\left(\dfrac{m}{n^2}\right) \wt{O}\left(\dfrac{m}{n}\right) = \wt{O}\left(\dfrac{m^2}{n^3}\right). \]
Finally, since $MM^T(\mathcal{K}^{\perp})\subseteq \mathcal{K}^{\perp}$,
 \begin{equation*}
 \begin{gathered}
 R_{\mathcal{K}^{\perp}}^{-1}-R_{\mathcal{K}^{\perp}}^{-1}(\mathcal{E}_M)_{\mathcal{K}^{\perp}}R_{\mathcal{K}^{\perp}}^{-1} = R_{\mathcal{K}^{\perp}}^{-1}\left(2R_{\mathcal{K}^{\perp}}-(R_{\mathcal{K}^{\perp}}+(\mathcal{E}_M)_{\mathcal{K}^{\perp}})\right)R_{\mathcal{K}^{\perp}}^{-1} = \\
 = 2R_{\mathcal{K}^{\perp}}^{-1}-R_{\mathcal{K}^{\perp}}^{-1}\left(MM^T\right)R_{\mathcal{K}^{\perp}}^{-1}. 
 \end{gathered}
 \end{equation*}
\end{proof}

\begin{observation}\label{obs:Rk-inverse}
Let $P_{\mathcal{X}}$ denote the projector on subspace $\mathcal{X}$ of $\mathbb{R}^{mn}$. Then for $\mathcal{I}$ and $\mathcal{D}$ as in Lemma~\ref{lem:R-eigenspaces}, we have   $R_{\mathcal{K}^{\perp}}^{-1} = P_{\mathcal{I}}+\dfrac{1}{3}P_{\mathcal{D}} = P_{\mathcal{K}^{\perp}}-\dfrac{2}{3}P_{\mathcal{D}}$.
\end{observation}
\begin{proof}
Immediately follows from Lemma~\ref{lem:R-eigenspaces} and the definition of $R_{\mathcal{K}^{\perp}}$.
\end{proof}
Hence, we get the following corollary to Lemma~\ref{lem:Y-simplification}.

\begin{corollary}\label{cor:Y-approx}
Assume that $m \ll n^2$, then for $Y$ given by Eq.~\eqref{eq:Y-formula} we have  
 \begin{equation}
 \left\Vert Y - 2\left(P_{\mathcal{K}^{\perp}}-\dfrac{2}{3}P_{\mathcal{D}}\right)E + \left(P_{\mathcal{K}^{\perp}}-\dfrac{2}{3}P_{\mathcal{D}}\right)MM^T\left(P_{\mathcal{K}^{\perp}}-\dfrac{2}{3}P_{\mathcal{D}}\right)E \right\Vert  = \wt{O}\left( \dfrac{m^2}{n^3} \right).
 \end{equation}
\end{corollary}
\begin{proof} The statement follows from Observation~\ref{obs:Rk-inverse} and Lemma~\ref{lem:Y-simplification}.
 \end{proof}

 Thus, for $a, b, c \in \mathbb{R}^m$, and  $x, y, z\in \mathbb{R}^n$ we need to bound 
 \[ \left\langle  2\left(P_{\mathcal{K}^{\perp}}-\dfrac{2}{3}P_{\mathcal{D}}\right)E - \left(P_{\mathcal{K}^{\perp}}-\dfrac{2}{3}P_{\mathcal{D}}\right)MM^T\left(P_{\mathcal{K}^{\perp}}-\dfrac{2}{3}P_{\mathcal{D}}\right)E 
 ,\  \left(\begin{matrix} x\otimes a \\ y\otimes b \\ z\otimes c \end{matrix}\right) \right\rangle. \]
 We are going to bound each term of this expression separately.

 \begin{lemma}\label{lem:PDE-norm-bound} With high probability over the randomness of $\mathcal{V}$
 \begin{equation}\label{eq:E1-proj-bound}
 \left\Vert \sum\limits_i (u_i\otimes f^{1}_i) \left( \sum\limits_{j:\, j\neq i} \langle u_i, u_j \rangle \langle v_i, v_j \rangle \langle w_i, w_j \rangle \right) \right\Vert = \wt{O}\left(\dfrac{m}{n^{3/2}}\right)
 \end{equation}
 Thus, $\Vert P_{\mathcal{D}}E \Vert  = \wt{O}\left(\dfrac{m}{n^{3/2}}\right)$. 
 \end{lemma}
 \begin{proof} The vectors $(u_i\otimes f_i)$ are orthonormal, and by Bernstein's inequality, w.h.p.
 \[ \left\vert \sum\limits_{j:\, j\neq i} \langle u_i, u_j \rangle \langle v_i, v_j \rangle \langle w_i, w_j \rangle \right\vert = \wt{O}\left(\dfrac{\sqrt{m}}{n^{3/2}}\right).\]
 Hence, the bound in Eq.~\eqref{eq:E1-proj-bound} follows. Now, observe that the LHS of Eq.~\eqref{eq:E1-proj-bound} is the norm of the projection $E^P_1$ of $E_1$ (see Eq.~\eqref{eq:Y-system}) onto $\vspan\{u_i  \otimes f^1_i \mid i\in [m] \}$. Similarly we obtain bounds on the projections $E^P_1$ and $E^P_3$ of $E_2$ and $E_3$ onto $\vspan\{ v_i \otimes f^2_i \mid i\in [m] \}$ and $\vspan\{ w_i \otimes f^3_i \mid i\in [m] \}$, respectively. Finally, note that by the definition of $\mathcal{D}$ we have 
 \[ \Vert P_{\mathcal{D}}E \Vert  = \left( \Vert E^P_1 \Vert^2 + \Vert E^P_2 \Vert^2 + \Vert E^P_3 \Vert^2 \right)^{1/2}. \]
 \end{proof}
 
 \begin{lemma}\label{lem:E-inner-prod-bound}  With high probability over the randomness of $\mathcal{V}$,
 \[\max\limits_{\Vert x\Vert = \Vert a\Vert =1} \left\langle E_1, x\otimes a \right\rangle =  \max\limits_{\Vert x\Vert = \Vert a\Vert =1} \left\langle \sum\limits_{i=1}^{m}\sum\limits_{j:\ j\neq i} (u_j\otimes f^1_i)\langle v_j, v_i \rangle \langle w_j, w_i \rangle , x\otimes a \right\rangle = \wt{O}\left(\dfrac{m}{n^{3/2}}+\dfrac{\sqrt{m}}{n}\right).\]
 Similarly, $\max\limits_{\Vert x\Vert = \Vert a\Vert =1} \left\langle E_t, x\otimes a \right\rangle = \wt{O}\left(\dfrac{m}{n^{3/2}}+\dfrac{\sqrt{m}}{n}\right)$ for $t=2$ and $t=3$.
 \end{lemma}
 \begin{proof} For any $x\in \mathbb{R}^n$ and any $a\in \mathbb{R}^m$, with $\Vert x \Vert = \Vert a \Vert = 1$,
 \begin{equation}
 \begin{gathered}
 \left\langle \sum\limits_{i=1}^{m}\sum\limits_{j:\ j\neq i} (u_j\otimes f^1_i)\langle v_j, v_i \rangle \langle w_j, w_i \rangle , x\otimes a \right\rangle = \\
 = \sum\limits_{i=1}^{m}\langle a, f^1_i\rangle \sum\limits_{j:\ j\neq i} \langle x,  u_j\rangle \langle v_j, v_i \rangle \langle w_j, w_i \rangle = a^T\left((V\cten W)^T(V\cten W)-I\right)U^Tx\leq \\
 \leq \left\Vert (V\cten W)^T(V\cten W)-I \right\Vert\cdot \Vert U^T  x\Vert  = \wt{O}\left(\dfrac{\sqrt{m}}{n} \left(\dfrac{\sqrt{m}}{\sqrt{n}}+1\right) \right) = \wt{O}\left(\dfrac{m}{n^{3/2}}+\dfrac{\sqrt{m}}{n}\right).
 \end{gathered}
 \end{equation}
 \end{proof}
 
 \begin{lemma}\label{lem:MMTE-inner-prod-bound} For $m\ll n^2$, with high probability over the randomness of $\mathcal{V}$
 \begin{equation}
  \max\limits_{a, b, c \in S^{m-1},\ x, y, z \in S^{n-1}}\left\langle MM^TE,\ \left(\begin{matrix}
  x\otimes a \\
  y\otimes b \\
  z \otimes c
\end{matrix}\right)    \right\rangle = \wt{O}\left( \dfrac{m}{n^{3/2}}+\dfrac{\sqrt{m}}{n}\right).
 \end{equation}
 \end{lemma}
 \begin{proof} 
Without loss of generality it is sufficient to bound $\langle M_1M^TE, x\otimes a\rangle$. We compute
 \begin{equation}
 \begin{split}
 M_1M^T(-E) & = M_1M^T
 \left(\begin{matrix}
  \sum\limits_{i=1}^{m}\sum\limits_{j:\ j\neq i} (u_j\otimes f_i^1)\langle v_j, v_i \rangle \langle w_j, w_i \rangle\\
 \sum\limits_{i=1}^{m}\sum\limits_{j:\ j\neq i} (v_j\otimes f_i^2)\langle u_j, u_i \rangle \langle w_j, w_i \rangle\\
 \sum\limits_{i=1}^{m}\sum\limits_{j:\ j\neq i} (w_j\otimes f_i^3)\langle v_j, v_i \rangle \langle u_j, u_i \rangle
 \end{matrix} \right) = \\
 & = M_1  \sum\limits_{i=1}^{m}\sum\limits_{j:\ j\neq i} (u_j\otimes v_i\otimes w_i)\langle v_j, v_i \rangle \langle w_j, w_i \rangle + \\
  &+ M_1  \sum\limits_{i=1}^{m}\sum\limits_{j:\ j\neq i}(u_i\otimes v_j\otimes w_i)\langle u_j, u_i \rangle \langle w_j, w_i \rangle +\\
  &+M_1  \sum\limits_{i=1}^{m}\sum\limits_{j:\ j\neq i}(u_i\otimes v_i\otimes w_j)\langle u_j, u_i \rangle \langle v_j, v_i \rangle =\\
  &= \sum\limits_{k=1}^m \sum\limits_{i=1}^m \sum\limits_{j:\ j\neq i} (u_j\otimes f_k^1)\langle v_k, v_i \rangle \langle w_k, w_i \rangle \langle v_j, v_i \rangle \langle w_j, w_i \rangle +\\
  &+ \sum\limits_{k=1}^m \sum\limits_{i=1}^m \sum\limits_{j:\ j\neq i} (u_i\otimes f_k^1)\langle v_k, v_j \rangle \langle w_k, w_i \rangle \langle u_j, u_i \rangle \langle w_j, w_i \rangle +\\
  &+ \sum\limits_{k=1}^m \sum\limits_{i=1}^m \sum\limits_{j:\ j\neq i} (u_i\otimes f_k^1)\langle v_k, v_i \rangle \langle w_k, w_j \rangle \langle u_j, u_i \rangle \langle v_j, v_i \rangle.  
 \end{split}
 \end{equation}
We bound the first sum in the following way
\begin{equation}\label{eq:UME-term1-norm-bound}
\begin{split}
\left\langle \sum\limits_{k=1}^m \sum\limits_{i=1}^m \sum\limits_{j:\ j\neq i} (u_j\otimes f_k^1)\langle v_k, v_i \rangle \langle w_k, w_i \rangle \langle v_j, v_i \rangle \langle w_j, w_i\rangle, x\otimes a \right\rangle = \\
= a^T(V\cten W)^T(V\cten W)\left((V\cten W)^T(V\cten W)-I\right)U^T x = \\
 = \wt{O}\left(1+\dfrac{\sqrt{m}}{n}\right)\wt{O}\left(\dfrac{\sqrt{m}}{n}\right) \wt{O}\left(\dfrac{\sqrt{m}}{\sqrt{n}}+1\right)  = \wt{O}\left(\dfrac{m}{n^{3/2}}+\dfrac{\sqrt{m}}{n}\right).
\end{split}
\end{equation} 
The second and third sums differs only by renaming $v$ and $w$. So we present an argument only for the second sum. We need to bound

\begin{equation}
\begin{split}
\left\langle \sum\limits_{k=1}^m \sum\limits_{i=1}^m \sum\limits_{j:\ j\neq i} (u_i\otimes f_k^1)\langle v_k, v_j \rangle \langle w_k, w_i \rangle \langle u_j, u_i \rangle \langle w_j, w_i \rangle, x\otimes a \right\rangle = \\
 = \sum\limits_{k, i} a_k \langle x, u_i\rangle  \sum\limits_{j:\ j\neq i} \langle v_k, v_j \rangle \langle w_k, w_i \rangle \langle u_j, u_i \rangle \langle w_j, w_i \rangle.
 \end{split}
 \end{equation}
To bound this sum we consider 3 cases. 

Case 1: $k=i$. Using Bernstein's inequality

\begin{equation}
\begin{gathered}
\left\vert \sum\limits_{i} a_i \langle x, u_i\rangle  \sum\limits_{j:\ j\neq i} \langle u_i, u_j \rangle \langle v_j, v_i \rangle \langle w_j, w_i \rangle \right\vert \leq \\
\leq \Vert U \Vert \cdot \max_{i}\left\vert \sum\limits_{j:\ j\neq i} \langle u_i, u_j \rangle \langle v_j, v_i \rangle \langle w_j, w_i \rangle \right\vert
\leq  \wt{O}\left(\left(1+\dfrac{\sqrt{m}}{\sqrt{n}}\right)\dfrac{\sqrt{m}}{n^{3/2}}\right)= \wt{O}\left(\dfrac{m}{n^2}+\dfrac{\sqrt{m}}{n^{3/2}}\right)
\end{gathered}
\end{equation}


Case 2: $k\neq i$. We use the trace power method (Lemma~\ref{lem:trace-power-method-norm}) in this case.

Consider an $m\times m$ matrix $S$ with the $(k, i)$-th entry given by
\begin{equation}\label{eq:matrix-diagram-expansion}
S_{ki} = \sum\limits_{j:\ j\neq i} \langle v_k, v_j \rangle \langle w_k, w_i \rangle \langle u_j, u_i \rangle \langle w_j, w_i \rangle,
\end{equation}
for $k, i\in [m]$. Then $S$ has the matrix diagram and the trace diagram presented on Figure~\ref{fig:graph-mat-2}.
\begin{figure}
\begin{center}
\begin{subfigure}[b]{0.25\textwidth}
\begin{center}
\includegraphics[scale=0.9]{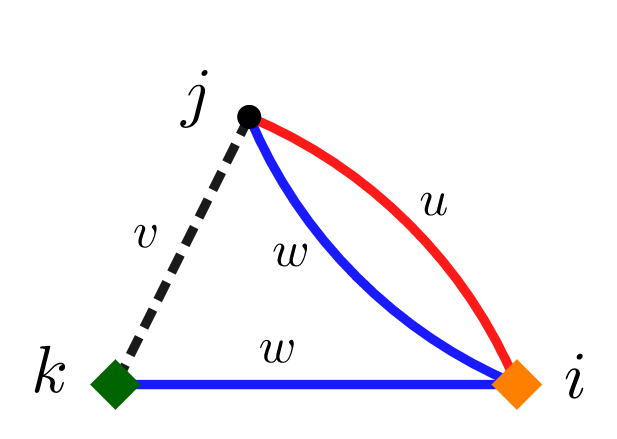}
\end{center}
\end{subfigure}\qquad
\begin{subfigure}[b]{0.7\textwidth}
\begin{center}
\includegraphics[scale=0.85]{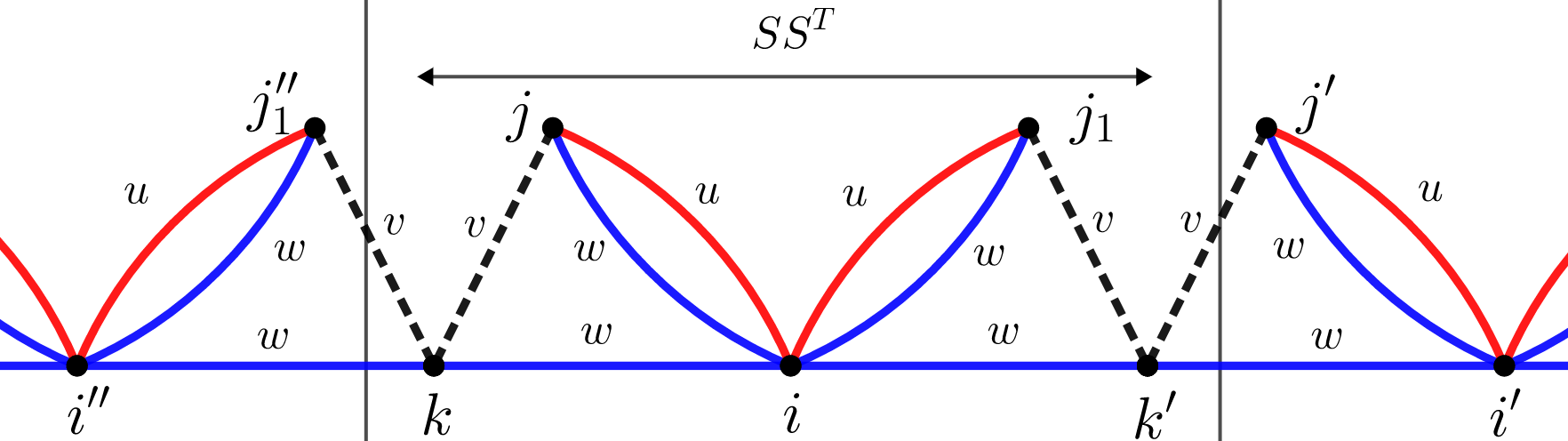}
\end{center}
\end{subfigure}
\caption{The matrix diagram for $S$ and the trace diagram for $SS^T$.}\label{fig:graph-mat-2}
\end{center}
\end{figure} 

We proceed similarly as in the proof of Theorem~\ref{thm:main-diagram-tool}, however, to get a better bound, we take into account that all edges of color $w$ are non-equality. Using independence, we can write
\[ \mathbb{E}\Tr\left((SS^T)^q\right) = \mathbb{E}\sum\limits_{\phi\in \Phi} val(\mathcal{TD}_q(S, \phi)) = \sum\limits_{\phi\in \Phi} \mathbb{E}\left(\prod\limits_{e\in E_w} \term_{\phi}(e)\right)\mathbb{E}\left(\prod\limits_{e\in E_{uv}} \term_{\phi}(e)\right).\]
We apply Theorem~\ref{thm:diagram-bound-tool} to graphs $G_{\{w\}, \phi}$ and $G_{\{u, v\}, \phi}$ induced on labels of $\phi$ by the edges $E_{w}$ and $E_{uv}$ of colors $w$ and $\{u,v\}$, respectively. We obtain
\[ \left\vert\mathbb{E}\left(\prod\limits_{e\in E_w} \term_{\phi}(e)\right)\right\vert \leq \wt{O}\left(\dfrac{1}{n}\right)^{\max(n_{\phi}-1, 2q)}\quad  \text{and} \quad \left\vert\mathbb{E}\left(\prod\limits_{e\in E_{uv}} \term_{\phi}(e)\right)\right\vert \leq \wt{O}\left(\dfrac{1}{n}\right)^{n_{\phi}-1},\]
where $n_{\phi}$ is the size of the image of $\phi$. Note that there are at most $m^{n_{\phi}}n_{\phi}^{4q}$ labelings of $\mathcal{TD}_q(S)$ using only $n_{\phi}$ labels from $[m]$. Therefore,
\[ \left\vert\mathbb{E}\Tr\left((SS^T)^q\right)\right\vert \leq \sum\limits_{j=1}^{4q} m^{j}(4q)^{4q} \wt{O}\left(\dfrac{1}{n}\right)^{\max(2j-2, 2q+j-1)}\]
Since $m\ll n^2$, the expression under the sum sign is maximized for $j = 2q+1$, so
\begin{equation}
 \left\vert\mathbb{E}\Tr\left((SS^T)^q\right)\right\vert \leq m(4q)^{4q+1}\wt{O}\left(\dfrac{m}{n^2}\right)^{2q}.
\end{equation}
Hence, the trace power method (see Lemma~\ref{lem:trace-power-method-norm}) implies that
\[ \Vert S\Vert = \wt{O}\left(\dfrac{m}{n^2}\right).\]

Therefore,
\begin{equation}\label{eq:UME-term2-norm-bound}
\left\vert \sum\limits_{k\neq i} a_k \langle x, u_i\rangle  \sum\limits_{j:\ j\neq i} \langle v_k, v_j \rangle \langle w_k, w_i \rangle \langle u_j, u_i \rangle \langle w_j, w_i \rangle \right\vert \leq \Vert U\Vert \cdot \Vert S \Vert \leq \wt{O}\left(\dfrac{m^{3/2}}{n^{5/2}}+\dfrac{m}{n^{2}} \right).
\end{equation}
\end{proof}
 
\begin{lemma}\label{lem:PDMMTE-inner-prod-bound}
With high probability over the randomness of $\mathcal{V}$
 \begin{equation}
  \max\limits_{a, b, c \in S^{m-1},\ x, y, z \in S^{n-1}}\left\langle P_{\mathcal{D}}MM^TE,\ \left(\begin{matrix}
  x\otimes a \\
  y\otimes b \\
  z \otimes c
\end{matrix}\right)    \right\rangle = \wt{O}\left( \dfrac{m}{n^{3/2}}\right)
 \end{equation}
\end{lemma}
\begin{proof}  By symmetry, it is sufficient to bound the inner product of the first block of $P_{\mathcal{D}}MM^TE$ with $x\otimes a$. The first block of $P_{\mathcal{D}}MM^TE$ is a sum of three IP graph matrices (see Figure~\ref{fig:UPME-matrix-diagram}). We show how to bound the inner product with the first one. The other two are similar.
\begin{equation}\label{eq:PME-bound-expression}
\begin{gathered}
\left\vert  \sum\limits_{k=1}^m \langle u_k\otimes f_k^1, x\otimes a \rangle \sum\limits_{i=1}^m \sum\limits_{j:\ j\neq i}  \langle u_k, u_j \rangle \langle v_k, v_i \rangle \langle w_k, w_i \rangle \langle v_j, v_i \rangle \langle w_j, w_i \rangle,\  \right\vert \leq \\
\leq \Vert a \Vert \left(\sum\limits_{k=1}^m\langle u_k, x\rangle^2 \right)^{1/2}\left( \max\limits_{k} \left(\sum\limits_{i=1}^m \sum\limits_{j:\ j\neq i} \langle u_k, u_j \rangle \langle v_k, v_i \rangle \langle w_k, w_i \rangle \langle v_j, v_i \rangle \langle w_j, w_i \rangle \right) \right)
\end{gathered}
\end{equation}
First, we bound the term which is maximized over $k$.

If $i=k$, then by Bernstein's inequality,
\begin{equation}\label{eq:UPME-len-bound-ik}   \left\vert\sum\limits_{j:\ j\neq k} \langle u_k, u_j \rangle \langle v_j, v_k \rangle \langle w_j, w_k \rangle \right\vert \leq \wt{O}\left(\dfrac{\sqrt{m}}{n^{3/2}}\right).
\end{equation}

If $j=k$, then by Bernstein's inequality,
\begin{equation}\label{eq:UPME-len-bound-ij}
\left\vert\sum\limits_{i:\ i\neq k} \langle v_k, v_i \rangle^2 \langle w_k, w_i \rangle^2 \right\vert \leq \wt{O}\left(\dfrac{m}{n^{2}}\right).
\end{equation}

If $i\neq k$ and $j\neq k$, then treating the expression below as a $1\times 1$ matrix (indexed by a fixed $k$) and applying Lemma~\ref{lem:trace-power-method-norm} and Theorem~\ref{thm:main-diagram-tool} to it with $\mathcal{C} = \{\{u, v\}, \{w\}\}$  (see also Figure~\ref{fig:UPME-matrix-diagram}) we get the following bound
\begin{equation}\label{eq:UPME-len-bound-dist}
\begin{gathered}
\left\vert  \sum\limits_{i:\ i\neq k} \sum\limits_{j:\ k\neq j\neq i} \langle u_k, u_j \rangle \langle v_k, v_i \rangle \langle w_k, w_i \rangle \langle v_j, v_i \rangle \langle w_j, w_i \rangle \right\vert \leq \wt{O}\left(\dfrac{m}{n^2}\right).
\end{gathered}
\end{equation}
Substituting these bounds into~\eqref{eq:PME-bound-expression} and using Lemma~\ref{cor:basic-deg2-sum},

\begin{equation}
\begin{split}
\left\vert  \sum\limits_{k=1}^m \langle u_k\otimes f_k^1, x\otimes a \rangle \sum\limits_{i=1}^m \sum\limits_{j:\ j\neq i}  \langle u_k, u_j \rangle \langle v_k, v_i \rangle \langle w_k, w_i \rangle \langle v_j, v_i \rangle \langle w_j, w_i \rangle,\  \right\vert \leq \\
\leq \wt{O}\left(1+\dfrac{\sqrt{m}}{\sqrt{n}}\right)\wt{O}\left(\dfrac{m}{n^{2}}+\dfrac{m^{1/2}}{n^{3/2}}\right) = \wt{O}\left(\dfrac{m^{3/2}}{n^{5/2}}+ \dfrac{m^{1/2}}{n^{3/2}}\right).
\end{split}
\end{equation}

Similarly, we bound
\begin{equation}
\begin{split}
\left\vert  \sum\limits_{k=1}^m \langle u_k\otimes f_k^1, x\otimes a \rangle \sum\limits_{i=1}^m \sum\limits_{j:\ j\neq i}  \langle u_k, u_i \rangle \langle v_k, v_j \rangle \langle w_k, w_i \rangle \langle u_j, u_i \rangle \langle w_j, w_i \rangle,\  \right\vert \leq \\
\leq \wt{O}\left(1+\dfrac{\sqrt{m}}{\sqrt{n}}\right)\wt{O}\left(\dfrac{m}{n^{2}}+\dfrac{m^{1/2}}{n^{3/2}}\right) = \wt{O}\left(\dfrac{m^{3/2}}{n^{5/2}}+ \dfrac{m^{1/2}}{n^{3/2}}\right).
\end{split}
\end{equation}
 
 \end{proof}

\begin{proof}[Proof of Theorem~\ref{thm:candidate-exists}]
By Theorem~\ref{eq:Y-formula}, for $m \ll n^2$ w.h.p there exists a solution $Y$ to Eq.~\eqref{eq:Y-system}. Hence, $A = M^T(D/3+Y)$ satisfies Eq.~\eqref{eq:orthog-u}-\eqref{eq:orthog-w}. Observe that 
\[M^TD/3 = \sum\limits_{i=1}^{m} u_i\otimes v_i\otimes w_i \quad \text{and} \]
\[ M^T Y = \alpha_i\otimes v_i \otimes w_i+ u_i\otimes \beta_i \otimes w_i+ u_i\otimes v_i \otimes \gamma_i,\]
where $U' = [\alpha_1, \ldots, \alpha_m]$, $V' = [\beta_1, \ldots, \beta_m]$ and $W' = [\gamma_1, \ldots, \gamma_m]$ are the reshaped vectors $Y_1$, $Y_2$ and $Y_3$ into $n\times m$ matrices (see the discussion at the beginning of this subsection). Hence, as discussed above, the desired norm bounds for $U'$, $V'$ and $W'$ are equivalent to 
\begin{equation}
  \max\limits_{a, b, c \in S^{m-1},\ x, y, z \in S^{n-1}}\left\langle Y,\ \left(\begin{matrix}
  x\otimes a \\
  y\otimes b \\
  z \otimes c
\end{matrix}\right)    \right\rangle = \widetilde{O}\left(\frac{\sqrt{m}}{n}+\frac{m}{n^{3/2}}+\dfrac{m^2}{n^3}\right).
 \end{equation}
 By Corollary~\ref{cor:Y-approx},
 \begin{equation*}
 \left\Vert Y - \left(P_{\mathcal{K}^{\perp}}-\dfrac{2}{3}P_{\mathcal{D}}\right)E + \left(P_{\mathcal{K}^{\perp}}-\dfrac{2}{3}P_{\mathcal{D}}\right)MM^T\left(P_{\mathcal{K}^{\perp}}-\dfrac{2}{3}P_{\mathcal{D}}\right)E \right\Vert  = \wt{O}\left( \dfrac{m^2}{n^3} \right).
 \end{equation*}
 By Lemma~\ref{lem:PDE-norm-bound}, $\Vert P_D E \Vert = \wt{O}\left(\dfrac{m}{n^{3/2}}\right)$, and Lemma~\ref{lem:R-eigenspaces} and Proposition~\ref{prop:S-bound} imply 
 \[\left\Vert MM^T \right\Vert = 3+\wt{O}\left(\dfrac{\sqrt{m}}{n}\right).\] 
 Moreover, by Observation~\ref{obs:E-not-in-kernel}, we have $P_{\mathcal{K}^{\perp}}E = E$, and by Lemma~\ref{lem:MM^T-kernel}, $P_{\mathcal{K}^{\perp}} MM^T = MM^T$. Therefore it is sufficient to bound
 \[ \left\langle E,\ \left(\begin{matrix}
  x\otimes a \\
  y\otimes b \\
  z \otimes c
\end{matrix}\right)    \right\rangle, \quad \left\langle MM^TE,\ \left(\begin{matrix}
  x\otimes a \\
  y\otimes b \\
  z \otimes c
\end{matrix}\right)    \right\rangle \quad \text{and} \quad \left\langle P_{\mathcal{D}}MM^TE,\ \left(\begin{matrix}
  x\otimes a \\
  y\otimes b \\
  z \otimes c
\end{matrix}\right)    \right\rangle. \]
The desired bounds are established in Lemmas~\ref{lem:E-inner-prod-bound}, \ref{lem:MMTE-inner-prod-bound} and \ref{lem:PDMMTE-inner-prod-bound}.
\end{proof}

\subsection{Explicit approximations to correction terms $\alpha_i$, $\beta_i$ and $\gamma_i$}\label{sec:explicit-approx}

In this section we extract more information about the correction terms $\alpha_i$, $\beta_i$ and $\gamma_i$ from the proofs above for the purposes of future sections. 

Recall that we get correction terms as a solution to the equation
\[Y = (R_{\mathcal{K}^{\perp}}+(\mathcal{E}_M)_{\mathcal{K}^{\perp}})^{-1}E.\]
As was shown in the previous section, we can write $(R_{\mathcal{K}^{\perp}}+(\mathcal{E}_M)_{\mathcal{K}^{\perp}})^{-1}$ as a series
\begin{equation}\label{eq:rsinv-exp-frob}
\begin{gathered}
(R_{\mathcal{K}^{\perp}}+(\mathcal{E}_M)_{\mathcal{K}^{\perp}})^{-1} = \left(R_{\mathcal{K}^{\perp}}(I+R_{\mathcal{K}^{\perp}}^{-1}(\mathcal{E}_M)_{\mathcal{K}^{\perp}})\right)^{-1} = \sum\limits_{i=0}^{\infty} \left(R_{\mathcal{K}^{\perp}}^{-1}(\mathcal{E}_M)_{\mathcal{K}^{\perp}}\right)^iR_{\mathcal{K}^{\perp}}^{-1} = \\
=  R_{\mathcal{K}^{\perp}}^{-1}-R_{\mathcal{K}^{\perp}}^{-1}(\mathcal{E}_M)_{\mathcal{K}^{\perp}}R_{\mathcal{K}^{\perp}}^{-1}+ \left(R_{\mathcal{K}^{\perp}}^{-1}(\mathcal{E}_M)_{\mathcal{K}^{\perp}}\right)^2 R_{\mathcal{K}^{\perp}}^{-1}+\left(R_{\mathcal{K}^{\perp}}^{-1}(\mathcal{E}_M)_{\mathcal{K}^{\perp}}\right)^3\left(R_{\mathcal{K}^{\perp}}+(\mathcal{E}_M)_{\mathcal{K}^{\perp}}\right)^{-1}
\end{gathered}
\end{equation}
By Proposition~\ref{prop:S-bound} and Lemma~\ref{lem:R-eigenspaces},  
\[ \left\Vert \left(R_{\mathcal{K}^{\perp}}^{-1}(\mathcal{E}_M)_{\mathcal{K}^{\perp}}\right)^3\left(R_{\mathcal{K}^{\perp}}+(\mathcal{E}_M)_{\mathcal{K}^{\perp}}\right)^{-1} \right\Vert \leq \wt{O}\left(\dfrac{m^{3/2}}{n^{3}}\right). \]

By Lemma~\ref{lem:E-norm-bound}, $\Vert E \Vert = \wt{O}\left(m/n\right)$, so we get the following result.
\begin{observation}
For $m\ll n^2$, with high probability
\begin{equation}\label{eq:Y-uptofrob-norm-approx}
 \left\Vert Y - \left(R_{\mathcal{K}^{\perp}}^{-1}-R_{\mathcal{K}^{\perp}}^{-1}(\mathcal{E}_M)_{\mathcal{K}^{\perp}}R_{\mathcal{K}^{\perp}}^{-1}+ R_{\mathcal{K}^{\perp}}^{-1}(\mathcal{E}_M)_{\mathcal{K}^{\perp}}R_{\mathcal{K}^{\perp}}^{-1}(\mathcal{E}_M)_{\mathcal{K}^{\perp}}R_{\mathcal{K}^{\perp}}^{-1}\right)E \right\Vert = \wt{O}\left(\dfrac{m^{5/2}}{n^{4}}\right)
 \end{equation}
\end{observation} 
This bound gives the Frobenius norm bound for the corresponding error for $U'$, $V'$ and $W'$. This is the precision we will be able to tolerate in all future computations. In the rest of this section we analyze the approximation to $U'$, $V'$ and $W'$ given by Eq.~\eqref{eq:Y-uptofrob-norm-approx}.
 
Now, recall that, by Observation~\ref{obs:Rk-inverse},
\begin{equation}
  R_{\mathcal{K}^{\perp}}^{-1} = P_{\mathcal{K}^{\perp}}-\dfrac{2}{3}P_{\mathcal{D}} \quad \text{and} \quad \mathcal{E}_M = MM^T-R. 
 \end{equation}
 Hence, by substituting these into Eq.~\eqref{eq:Y-uptofrob-norm-approx}, we get that the approximation to $Y$ given by that equation can be written as a linear combination of the following terms  
 \begin{equation}
 E,\ P_{\mathcal{D}}E,\ \left(MM^T-I\right)E,\ \left(MM^T-I\right)P_{\mathcal{D}}E,\ P_{\mathcal{D}}\left(MM^T-I\right)E\, \ldots 
 \end{equation}
where each term in the list is some product of $(MM^T-I)$'s and $P_{\mathcal{D}}$'s multiplied by $E$.

We group terms into levels based on how many times $MM^T-I$ appears. We ignore all constant factors. We use notation $(\alpha_k)_{\bullet}$ to denote the contribution to $\alpha_k$ of the term denoted by $\bullet$ . We let $U_{\bullet}$ denote the matrix with $k$-th column $(\alpha_k)_{\bullet}$. These notations generalize to $\beta_i$ and $\gamma_i$.

\paragraph{Level 0}
\begin{enumerate}
\item Contribution of $E$ to $\alpha_i$:
\begin{equation}
(\alpha_i)_E = \sum\limits_{j:\ j\neq i} u_j\langle v_j, v_i \rangle \langle w_j, w_i \rangle.
\end{equation}
The matrix diagram for $U_E$ is presented on Figure~\ref{fig:UE-UPE-matrix-diagram}. The following bounds are implied by the proof of Lemma~\ref{lem:E-norm-bound}.
\begin{equation}
 \Vert U_{E} \Vert = \wt{O}\left(\dfrac{m}{n^{3/2}}+\dfrac{\sqrt{m}}{n}\right),\qquad \Vert U_{E} \Vert_F = \wt{O}\left(\dfrac{m}{n}\right), \qquad  \Vert (\alpha_k)_E \Vert = \wt{O}\left(\dfrac{\sqrt{m}}{n}\right).
 \end{equation}
 
 \begin{figure}
\begin{subfigure}[b]{0.5\textwidth}
\begin{center}
\includegraphics[width = 0.6\textwidth]{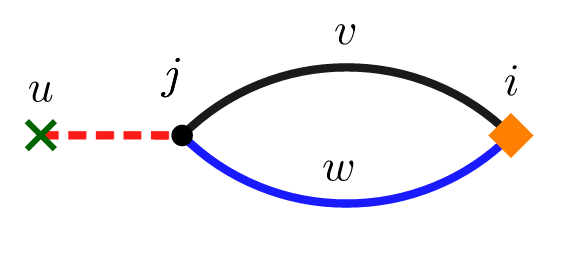}
\end{center}
\end{subfigure}
\begin{subfigure}[b]{0.5\textwidth}
\begin{center}
\includegraphics[width = 0.6\textwidth]{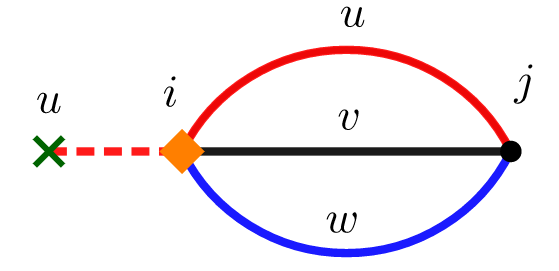}
\end{center}
\end{subfigure}
\caption{Matrix diagrams of $U_E$ (left) and $U_{PE}$ (right).}\label{fig:UE-UPE-matrix-diagram}
\end{figure} 

\item Contribution of $P_D E$ to $\alpha_i$:
\begin{equation}
(\alpha_i)_{PE} = \sum\limits_{j:\ j\neq i} u_i \langle u_j, u_i \rangle \langle v_j, v_i \rangle \langle w_j, w_i \rangle, \quad \text{that is,} 
\end{equation}
\begin{equation}
U_{PE} = U\Diag_{i\in [m]}\left(\sum\limits_{j:\ j\neq i} \langle u_j, u_i \rangle \langle v_j, v_i \rangle \langle w_j, w_i \rangle\right).
\end{equation}
Bernstein's inequality implies that $\displaystyle{\left\vert \sum\limits_{j:\ j\neq i} \langle u_j, u_i \rangle \langle v_j, v_i \rangle \langle w_j, w_i \rangle \right\vert = \wt{O}\left(\dfrac{\sqrt{m}}{n^{3/2}}\right)}$, so
\begin{equation}
 \Vert U_{PE} \Vert = \wt{O}\left(\dfrac{m}{n^{2}}+ \dfrac{\sqrt{m}}{n^{3/2}}\right),\qquad \Vert U_{PE} \Vert_F = \wt{O}\left(\dfrac{m}{n^{3/2}}\right), \qquad  \Vert (\alpha_k)_{PE} \Vert = \wt{O}\left(\dfrac{\sqrt{m}}{n^{3/2}}\right).
 \end{equation}
 The matrix diagram for $U_{PE}$ is presented on Figure~\ref{fig:UE-UPE-matrix-diagram}.
\end{enumerate}

\paragraph{Level 1}
\begin{enumerate}
\item Contribution of $(MM^T-I)E$ to $\alpha_k$:
\begin{equation}\label{eq:UME-def}
\begin{split}
(\alpha_k)_{ME} &=  \sum\limits_{i:\ i\neq k }^m \sum\limits_{j:\ j\neq i}  u_j\langle v_k, v_i \rangle \langle w_k, w_i \rangle \langle v_j, v_i \rangle \langle w_j, w_i \rangle +\\
  &+  \sum\limits_{i=1}^m \sum\limits_{j:\ j\neq i} u_i\langle v_k, v_j \rangle \langle w_k, w_i \rangle \langle u_j, u_i \rangle \langle w_j, w_i \rangle +\\
  &+  \sum\limits_{i=1}^m \sum\limits_{j:\ j\neq i} u_i\langle v_k, v_i \rangle \langle w_k, w_j \rangle \langle u_j, u_i \rangle \langle v_j, v_i \rangle
  \end{split}
  \end{equation}
  As we can see from Eq.~\eqref{eq:UME-def}, $U_{ME}$ is a sum of inner product graph matrices with diagrams presented on Figure~\ref{fig:UME-matrix-diagram}. The bound on $\Vert U_{ME} \Vert$ is implied by Eq.~\eqref{eq:UME-term2-norm-bound} and the argument similar to Eq.~\eqref{eq:UME-term1-norm-bound}. Since rank of $U_{ME}$ is at most $n$, the bound on $\Vert U_{ME} \Vert_F$ follows.
  
  \begin{equation}
 \Vert U_{ME} \Vert = \wt{O}\left(\dfrac{m^{3/2}}{n^{5/2}}+\dfrac{\sqrt{m}}{n^{3/2}}\right),\qquad \Vert U_{ME} \Vert_F = \wt{O}\left(\dfrac{m^{3/2}}{n^2}+\dfrac{m}{n^{3/2}}\right)   
 \end{equation}

 \begin{figure}
 \begin{subfigure}[b]{0.35\textwidth}
\begin{center}
\includegraphics[width = 0.9\textwidth]{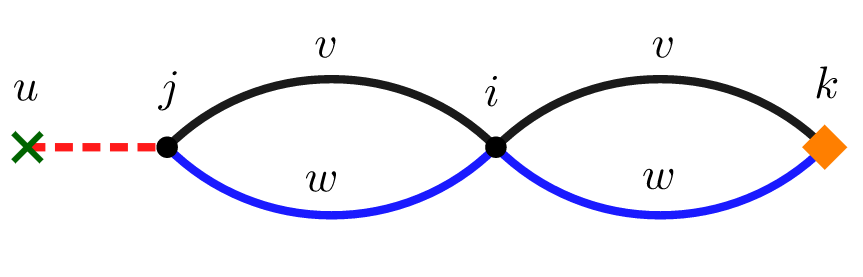}
\end{center}
\end{subfigure}
\begin{subfigure}[b]{0.3\textwidth}
\begin{center}
\includegraphics[width = 0.9\textwidth]{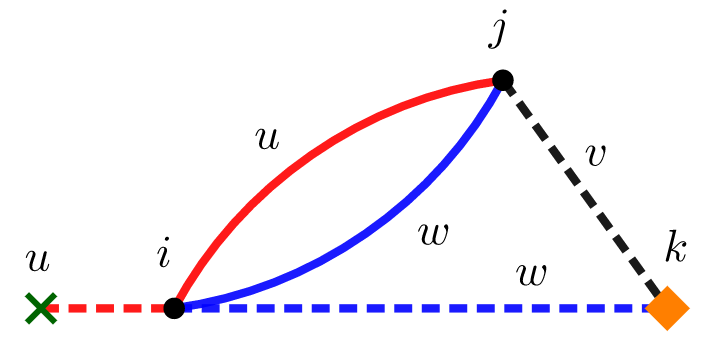}
\end{center}
\end{subfigure}
\begin{subfigure}[b]{0.3\textwidth}
\begin{center}
\includegraphics[width = 0.9\textwidth]{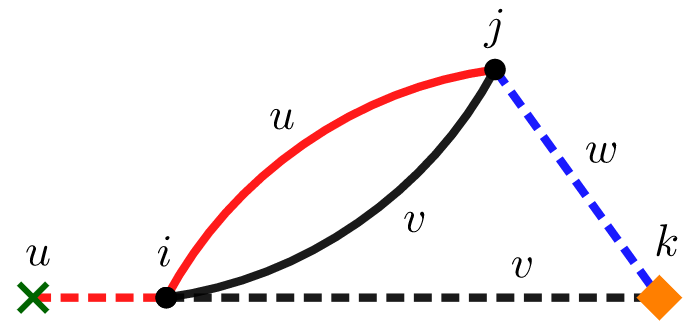}
\end{center}
\end{subfigure}
\caption{Matrix diagrams for $U_{ME}$.}\label{fig:UME-matrix-diagram}
\end{figure} 
  
  \item Contribution of $P_D(MM^T-I)E$ to $\alpha_k$:  
  \begin{equation}
\begin{split}
(\alpha_k)_{PME} &=  \sum\limits_{i:\ i\neq k }^m \sum\limits_{j:\ j\neq i}  u_k \langle u_j, u_k \rangle\langle v_k, v_i \rangle \langle w_k, w_i \rangle \langle v_j, v_i \rangle \langle w_j, w_i \rangle +\\
  &+  \sum\limits_{i=1}^m \sum\limits_{j:\ j\neq i} u_k \langle u_i, u_k \rangle\langle v_k, v_j \rangle \langle w_k, w_i \rangle \langle u_j, u_i \rangle \langle w_j, w_i \rangle +\\
  &+  \sum\limits_{i=1}^m \sum\limits_{j:\ j\neq i} u_k \langle u_i, u_k \rangle\langle v_k, v_i \rangle \langle w_k, w_j \rangle \langle u_j, u_i \rangle \langle v_j, v_i \rangle  
  \end{split}
  \end{equation}
  \begin{figure}[b]
 \begin{subfigure}[b]{0.3\textwidth}
\begin{center}
\includegraphics[width = 0.9\textwidth]{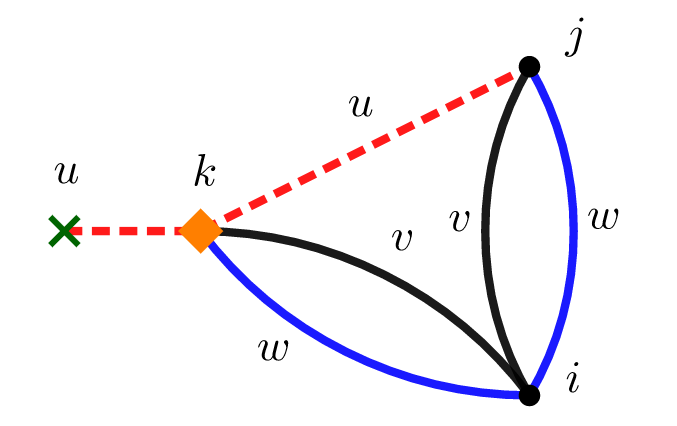}
\end{center}
\end{subfigure}
\begin{subfigure}[b]{0.3\textwidth}
\begin{center}
\includegraphics[width = 0.9\textwidth]{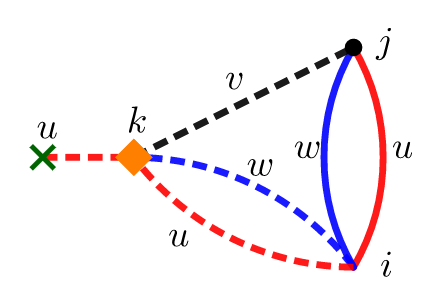}
\end{center}
\end{subfigure}
\begin{subfigure}[b]{0.3\textwidth}
\begin{center}
\includegraphics[width = 0.9\textwidth]{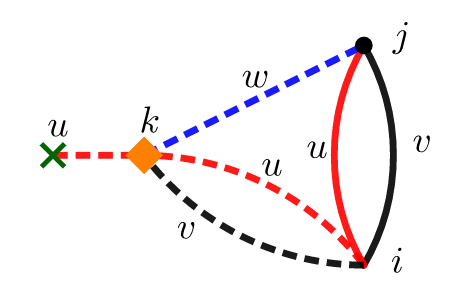}
\end{center}
\end{subfigure}
\caption{Matrix diagrams for $U_{PME}$.}\label{fig:UPME-matrix-diagram}
\end{figure} 
 Then $U_{PME}$ is a sum of the inner product graph matrices with matrix diagrams on Figure~\ref{fig:UPME-matrix-diagram}. The following bound on $\Vert(\alpha_k)_{PME}\Vert$ is shown in the proof of Lemma~\ref{lem:PDMMTE-inner-prod-bound} (see Eq.~\eqref{eq:UPME-len-bound-ik}-\eqref{eq:UPME-len-bound-dist}). The bounds on $\Vert U_{PME} \Vert$  and $\Vert U_{PME} \Vert_F$ follow immediately.
 \begin{equation}\label{eq:UPME-norm-bounds}
 \Vert U_{PME} \Vert = \wt{O}\left(\dfrac{m^{3/2}}{n^{5/2}}+\dfrac{m}{n^{2}}\right),\quad \Vert U_{PME} \Vert_F = \wt{O}\left(\dfrac{m^{3/2}}{n^{2}}\right), \quad  \Vert (\alpha_k)_{PME} \Vert = \wt{O}\left(\dfrac{m}{n^{2}}\right)
 \end{equation}

  \item Contribution of $(MM^T-I)P_DE$ to $\alpha_k$:
  \begin{equation}
  (\alpha_k)_{MPE} = \sum\limits_{i:\ i\neq k} u_i\langle v_i, v_k \rangle \langle w_i, w_k \rangle \sum\limits_{j:\ j\neq i} \langle u_j, u_i \rangle \langle v_j, v_i \rangle \langle w_j, w_i \rangle, \quad \text{so}
  \end{equation}
 \[\displaystyle{U_{MPE} = U\cdot \Diag_{i\in[m]}\left(\sum\limits_{j:\ j\neq i} \langle u_j, u_i \rangle \langle v_j, v_i \rangle \langle w_j, w_i \rangle\right) \cdot \left((V\cten W)^T(V\cten W)-I\right)}.\] 
 Hence
  \begin{equation}\label{eq:UMPE-norm-bounds}
 \Vert U_{MPE} \Vert = \wt{O}\left(\dfrac{m^{3/2}}{n^{3}}+\dfrac{m}{n^{5/2}}\right),\quad \Vert U_{MPE} \Vert_F = \wt{O}\left(\dfrac{m^{3/2}}{n^{5/2}}\right)  
 \end{equation}
  
  \item Contribution of $P_D(MM^T-I)P_DE$ to $\alpha_k$:

  By Lemma~\ref{lem:R-eigenspaces} and Proposition~\ref{prop:S-bound}, using bounds for $U_{MPE}$ from Eq~\eqref{eq:UMPE-norm-bounds}, we get  
\begin{equation}\label{eq:UPMPE-norm-bounds}
\Vert U_{PMPE} \Vert_F = \wt{O}\left(\dfrac{m^{3/2}}{n^{5/2}}\right)
\end{equation}
\end{enumerate}

\paragraph{Level 2:}

\begin{enumerate}
\item Contribution of $(MM^T-I)^2E$ and $P_D(MM^T-I)^2E$ to $\alpha_t$:
\begin{equation}
\begin{split}
(\alpha_t)_{MME} &= \sum\limits_{k:\ k\neq t }\sum\limits_{i:\ i\neq k } \sum\limits_{j:\ j\neq i}  u_j\langle v_t, v_k \rangle\langle w_t, w_k \rangle \langle v_k, v_i \rangle \langle w_k, w_i \rangle \langle v_j, v_i \rangle \langle w_j, w_i \rangle +\\
  &+  \sum\limits_{k:\ k\neq t }\sum\limits_{i=1}^m \sum\limits_{j:\ j\neq i} u_i\langle v_t, v_k \rangle\langle w_t, w_k \rangle\langle v_k, v_j \rangle \langle w_k, w_i \rangle \langle u_j, u_i \rangle \langle w_j, w_i \rangle +\\
  &+ \sum\limits_{k=1}^m\sum\limits_{i:\ i\neq k } \sum\limits_{j:\ j\neq i}  u_k\langle v_t, v_j \rangle\langle w_t, w_k \rangle \langle u_k, u_i \rangle \langle w_k, w_i \rangle \langle u_j, u_i \rangle \langle w_j, w_i \rangle +\\
  &+  \sum\limits_{k=1}^m\sum\limits_{i=1}^m \sum\limits_{j:\ j\neq i} u_k\langle v_i, v_t \rangle\langle w_t, w_k \rangle\langle u_k, u_j \rangle \langle w_k, w_i \rangle \langle v_j, v_i \rangle \langle w_j, w_i \rangle +\\
  &+\text{ 5 more qualitatively similar terms}
\end{split}
\end{equation}
  \begin{figure}
 \begin{subfigure}[b]{0.23\textwidth}
\begin{center}
\includegraphics[width = 0.9\textwidth]{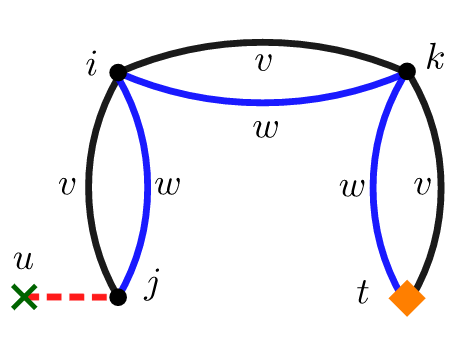}
\end{center}
\end{subfigure}
\begin{subfigure}[b]{0.23\textwidth}
\begin{center}
\includegraphics[width = 0.9\textwidth]{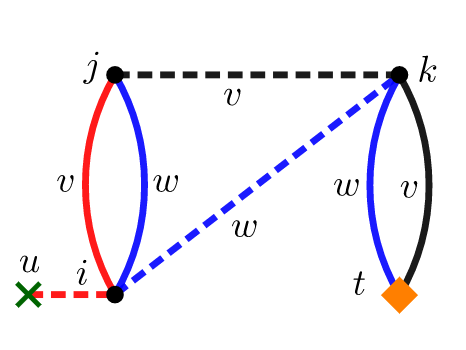}
\end{center}
\end{subfigure}
\begin{subfigure}[b]{0.23\textwidth}
\begin{center}
\includegraphics[width = 0.9\textwidth]{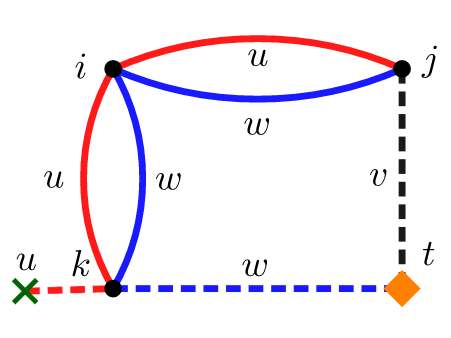}
\end{center}
\end{subfigure}
\begin{subfigure}[b]{0.23\textwidth}
\begin{center}
\includegraphics[width = 0.9\textwidth]{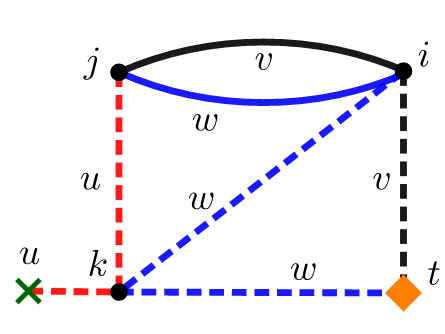}
\end{center}
\end{subfigure}
\caption{Matrix diagrams for $U_{MME}$.}\label{fig:UMME-matrix-diagram}
\end{figure} 
Observe that each of the IP graph matrices with diagrams on Figure~\ref{fig:UMME-matrix-diagram}, can be written in the form $U\cdot S$, where $S$ is $\{\{u, v\}, \{w\}\}$-connected. Hence, in the case when all indicies $i, j, k$ and $t$ are distinct, Theorem~\ref{thm:main-diagram-tool} and Lemma~\ref{lem:trace-power-method-norm} imply that $\Vert S \Vert = \wt{O}\left(m^{3/2}/n^{3}\right)$. Cases when some of the indicies are equal are considered as above. Thus, using Corollary~\ref{cor:basic-deg2-sum}, we get the following bounds 
\begin{equation}
\left\Vert U_{MME} \right\Vert = \wt{O}\left(\dfrac{m^2}{n^{7/2}}\right) \quad \text{and} \quad \left\Vert U_{MME} \right\Vert_F = \wt{O}\left(\dfrac{m^2}{n^{3}}\right).
\end{equation}
Since, $P_{\mathcal{D}}$ is a projector, we also have

\begin{equation}
\left\Vert U_{PMME} \right\Vert_F = \wt{O}\left(\dfrac{m^2}{n^{3}}\right)
\end{equation}

\item Contribution of $(MM^T-I)P_D(MM^T-I)E$ and $\left(P_D(MM^T-I)\right)^2E$ to $U'$:
\begin{equation}
U_{MPME} = U\Diag_{k\in [m]}\left( \Vert (\alpha_k)_{PME}\Vert\right) \cdot \left((V\cten W)^T(V\cten W)-I\right)
\end{equation}
  \begin{figure}
 \begin{subfigure}[b]{0.3\textwidth}
\begin{center}
\includegraphics[width = 0.9\textwidth]{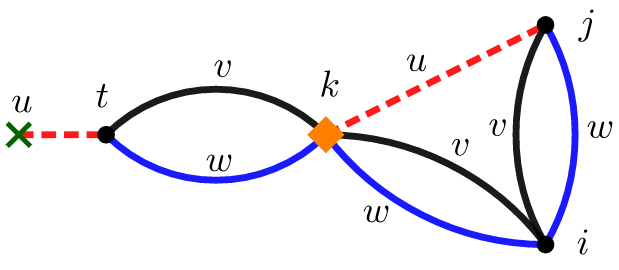}
\end{center}
\end{subfigure}
\begin{subfigure}[b]{0.3\textwidth}
\begin{center}
\includegraphics[width = 0.9\textwidth]{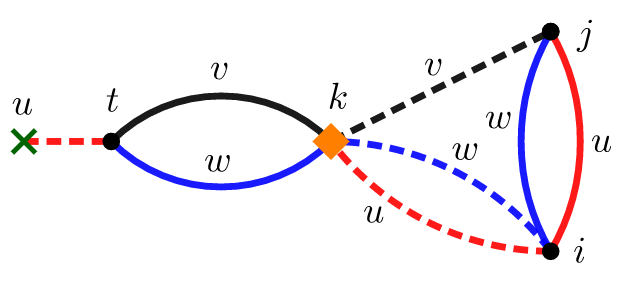}
\end{center}
\end{subfigure}
\begin{subfigure}[b]{0.3\textwidth}
\begin{center}
\includegraphics[width = 0.9\textwidth]{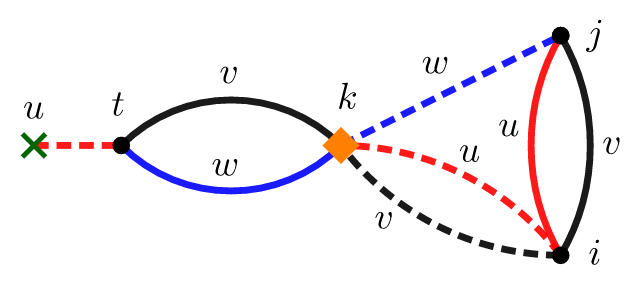}
\end{center}
\end{subfigure}
\caption{Diagram for $U_{MPME}$.}\label{fig:UMPME-matrix-diagram}
\end{figure} 
Hence, the bound for $\Vert (\alpha_k)_{PME}\Vert$ in Eq~\eqref{eq:UPME-norm-bounds} implies 
\begin{equation}
\Vert U_{MPME}\Vert =  \wt{O}\left(\dfrac{m^{2}}{n^{7/2}}+\dfrac{m^{3/2}}{n^{3}}\right) \quad \text{and} \quad \left\Vert U_{MPME} \right\Vert_F = \wt{O}\left(\dfrac{m^2}{n^{3}}\right)
\end{equation}
Since $P_{\mathcal{D}}$ is a projector, we also have
\begin{equation}
\left\Vert U_{PMPME} \right\Vert_F = \wt{O}\left(\dfrac{m^2}{n^{3}}\right)
\end{equation}

\item Contribution of $(MM^T-I)^2P_{\mathcal{D}}E$ and $P_{\mathcal{D}}(MM^T-I)^2P_{\mathcal{D}}E$ to $U'$:

By Lemma~\ref{lem:R-eigenspaces} and Proposition~\ref{prop:S-bound}, using bounds for $U_{MPE}$ from Eq~\eqref{eq:UMPE-norm-bounds}, we get
\begin{equation}
\left\Vert U_{MMPE} \right\Vert_F = \wt{O}\left(\dfrac{m^{3/2}}{n^{5/2}}\right) \quad \text{and} \quad \left\Vert U_{PMMPE} \right\Vert_F = \wt{O}\left(\dfrac{m^{3/2}}{n^{5/2}}\right)
\end{equation}

\item Contribution of $\left((MM^T-I)P_{\mathcal{D}}\right)^2E$ and $P_{\mathcal{D}}\left((MM^T-I)P_{\mathcal{D}}\right)^2E$ to $U'$:

By Lemma~\ref{lem:R-eigenspaces} and Proposition~\ref{prop:S-bound}, using bounds for $U_{PMPE}$ from Eq~\eqref{eq:UPMPE-norm-bounds}, we get 
\begin{equation}
\left\Vert U_{MPMPE} \right\Vert_F = \wt{O}\left(\dfrac{m^{3/2}}{n^{5/2}}\right) \quad \text{and} \quad \left\Vert U_{PMPMPE} \right\Vert_F = \wt{O}\left(\dfrac{m^{3/2}}{n^{5/2}}\right)
\end{equation}

\end{enumerate}

An important corollary of the discussion above is the following lemma.
\begin{theorem}\label{thm:corr-terms-summary} Let $U'$ be the matrix constructed in Theorem~\ref{thm:candidate-exists}. Then $U$ can be written as 
\[ U' = U_{GM}+U_{sm},\quad \text{where}\]
\[\Vert U_{sm} \Vert_F = \wt{O}\left(\dfrac{m^{3/2}}{n^{5/2}}\right)\quad \text{and}\quad U_{GM}\in \vspan\left(\mathfrak{CM}^4(\mathcal{C}; 2), 100\right),\]
for $\mathcal{C} = \{\{u, v\}, \{v, w\}, \{u, w\}\}$.

Moreover, $U_{GM} = U_E+U'_{IGM}$, where $U'_{IGM}\in  \vspan\left(\mathfrak{CM}^4(\mathcal{C}; 3), 100\right)$.
\end{theorem}
\begin{proof} As discussed above, we have $U' = U'_1+U'_2$, where
\[ \Vert U'_{2} \Vert_F = \wt{O}\left(\dfrac{m^{5/2}}{n^4}+\dfrac{m^{3/2}}{n^{5/2}}\right)\quad \text{and}\quad U'_1 \in \vspan(\mathfrak{U}_2, 100),\quad \text{where}\]
\[ \mathfrak{U}_2 = \{U_E, U_{PE}, U_{ME}, U_{PME}, U_{MME}, U_{PMME}, U_{MPME}, U_{PMPME}\}.\]
 If $X\in \mathfrak{U}_2$  then each IP graph matrix involved in $X$ has $\Omega_L = \{\omega_L\}$ and $\Omega_R = \{\omega_R\}$, where $\omega_L$ is an end of a half edge of color $u$ and $\omega_R$ is a node. Note that then the corresponding IP matrix involved in $P_{\mathcal{D}}X$ is obtained by replacing a half-edge $(x, \omega_L)$ by a half edge $(\omega_L, \omega_R)$ of color $u$ and introducing an edge $(\omega_R, x)$ of color $u$. To get all diagrams involved in $P_{\mathcal{D}}X$ one need also to consider all the diagrams obtained by cyclic renaming of colors $(u, v, w)$.
 
  Thus if $X\in \mathfrak{U}_2$ is in $\vspan\left(\mathfrak{CM}^t(\mathcal{C}; e), N\right)$, then $\mathcal{P}_DX$ is in $\vspan\left(\mathfrak{CM}^t(\mathcal{C}; e), N\right)$. Hence, it is sufficient to check the claim of the theorem for the matrix diagrams of $U_E$, $U_{ME}$, $U_{MME}$ and $U_{MPME}$.  It can be seen directly from Figures~\ref{fig:UE-UPE-matrix-diagram},~\ref{fig:UME-matrix-diagram},~\ref{fig:UMME-matrix-diagram}, and~\ref{fig:UMPME-matrix-diagram} that each of these diagrams is $\mathcal{C}$-connected, and each has at most $4$ vertices. 

Note that on the second and the third diagrams for $U_{ME}$ either $k\neq i$ or $k\neq j$ since $i\neq j$. Similarly, on the last diagram for $U_{MME}$ either $k\neq i$ or $k\neq j$.  Thus each diagram except $\mathcal{MD}(U_E)$ has at least $3$ non-equality edges, or can be written as a sum of diagrams with at least $3$ non-equality edges. Finally, note that $U_{PME}$ has $3$ non-equality edges.
 
\end{proof}

The bounds above hold for $V'$ and $W'$, and their components, by a symmetric argument.

\section{Certificate candidate is a dual certificate}\label{sec:dual-certificate}

Theorem~\ref{thm:candidate-exists} states that with high probability over the randomness of $\mathcal{V}$ there exists a certificate candidate for $\mathcal{V}$. It is easy to see that if $\oa{A}$ is a certificate candidate, then Eq.~\eqref{eq:orthog-u}-\eqref{eq:error-term-bound} imply
\[ \langle \oa{A}, u_i\otimes v_i \otimes w_i\rangle = 1.\]

Moreover, one can deduce the following property.

\begin{observation} Let $\oa{A}$ be a certificate candidate for $\mathcal{V}$. Assume that $u_j^{\perp}$ is orthogonal to $u_j$. Then 
\[ \langle \oa{A}, u_j^{\perp}\otimes v_j\otimes w_j\rangle = 0.\]
\end{observation}
\begin{proof}
The statement of the corollary immediately follows from Eq.~\eqref{eq:orthog-u}.
\end{proof}

\begin{corollary}\label{cor:certificate-expansion}
Let $\oa{A}$ be a certificate candidate for $\mathcal{V}$. Assume that $u_j^{\perp}$, $v_j^{\perp}$ and $w_j^{\perp}$ are orthogonal to $u_j$, $v_j$ and $w_j$, respectively. Then
\[ \langle \oa{A}, (u_j+ u_j^{\perp})\otimes (v_j+ v_j^{\perp})\otimes (w_j+ w_j^{\perp})\rangle  = 1+ \langle \oa{A}, u_j^{\perp}\otimes v_j^{\perp}\otimes w_j \rangle+ \langle \oa{A}, u_j^{\perp}\otimes v_j\otimes w_j^{\perp} \rangle +\]
\[+ \langle \oa{A}, u_j\otimes v_j^{\perp}\otimes w_j^{\perp} \rangle+ \langle \oa{A}, u_j^{\perp}\otimes v_j^{\perp}\otimes w_j^{\perp} \rangle.\]
\end{corollary}

Thus we are interested in bounding inner product of $\oa{A}$ with an arbitrary tensor $x\otimes y\otimes z$ and in the special case, when it has the form $u_j^{\perp}\otimes v_j^{\perp}\otimes w_j$. We bound the inner product with each of the four sums involved in the definition of $\oa{A}$ separately.

\begin{lemma}\label{lem:error-term-effect} Let $m\ll n^2$. With high probability, for any vectors $x, y, z\in \mathbb{R}^n$
\[ \left\vert \left\langle\left(\sum\limits_{i=1}^{m} \alpha_i\otimes v_i \otimes w_i\right), x\otimes y\otimes z\right\rangle \right\vert \leq \Vert x\Vert \Vert y\Vert \Vert z \Vert \widetilde{O}\left(\frac{\sqrt{m}}{n}+\frac{m}{n^{3/2}}\right).\]
\end{lemma}
\begin{proof} Applying the Cauchy-Schwarz inequality twice, by Theorem~\ref{thm:deg4weakbound} and Theorem~\ref{thm:candidate-exists}, we obtain 
\[ \left\vert \left\langle\left(\sum\limits_{i=1}^{m} \alpha_i\otimes v_i \otimes w_i\right), x\otimes y\otimes z\right\rangle  \right\vert=  \left\vert \left(\sum\limits_{i=1}^{m} \left\langle\alpha_i, x\right\rangle\left\langle v_i, y\right\rangle \left\langle w_i, z\right\rangle\right) \right\vert \leq  \]
\[ \leq \sqrt{\sum\limits_{i=1}^{m} \left\langle\alpha_i, x\right\rangle^2} \sqrt{\sum\limits_{i=1}^{m}\left\langle v_i, y\right\rangle^2 \left\langle w_i, z\right\rangle^2}  \leq \left\Vert U'\right\Vert\Vert x \Vert 
 \left(\sum\limits_{i=1}^{m}\left\langle v_i, y\right\rangle^4\right)^{1/4}\left( \sum\limits_{i=1}^{m}\left\langle w_i, z\right\rangle^4\right)^{1/4}\leq\]
 \[ \leq \Vert x\Vert \Vert y\Vert \Vert z\Vert \widetilde{O}\left(\frac{\sqrt{m}}{n}+\frac{m}{n^{3/2}}\right)\widetilde{O}(1)\widetilde{O}(1).\]
\end{proof}
The proof of the next theorem is inspired by \cite{Ge-Ma}, where such an inequality was proved in the symmetric case, i.e., for the tensor with components $u_i\otimes u_i\otimes u_i$. 
\begin{theorem}\label{thm:inj-norm-bound} Let $m\ll n^{3/2}$. With high probability, 
for any unit vectors $x, y, z\in \mathbb{S}^{n-1}$
\[\left\vert \left\langle\left(\sum\limits_{i=1}^{m} u_i\otimes v_i \otimes w_i\right), x\otimes y\otimes z\right\rangle\right\vert  \leq 1+\widetilde{O}\left(\frac{1}{\sqrt{n}}+\frac{m}{n^{3/2}}\right).\]
\end{theorem}
\begin{proof} We can bound the desired inner product in the following way.
\begin{equation}\label{eq:in-prod-expansion-certificate}
\begin{gathered}
 \left\langle\left(\sum\limits_{i=1}^{m} u_i\otimes v_i \otimes w_i\right), x\otimes y\otimes z\right\rangle^2  = \left\langle\left(\sum\limits_{i=1}^{m}\langle u_i, x\rangle \langle v_i, y \rangle  w_i\right), z \right\rangle^2\leq \\
 \leq \left\Vert \sum\limits_{i=1}^{m}\langle u_i, x\rangle \langle v_i, y \rangle  w_i \right\Vert^2  = \sum\limits_{i=1}^{m}\langle u_i, x\rangle^2 \langle v_i, y \rangle^2+ \sum\limits_{i, j:\  i\neq j}\langle w_i, w_j\rangle \langle u_i, x\rangle \langle v_i, y \rangle\langle u_j, x\rangle \langle v_j, y \rangle.
\end{gathered}
\end{equation}
Note that by the Cauchy-Schwarz inequality and by Theorem~\ref{thm:deg4bound},
\begin{equation}\label{eq:sqaure-term-bound}
\sum\limits_{i=1}^{m}\langle u_i, x\rangle^2 \langle v_i, y \rangle^2 \leq \sqrt{\sum\limits_{i=1}^{m}\langle u_i, x \rangle^4}\sqrt{\sum\limits_{i=1}^{m}\langle v_i, y \rangle^4}\leq 1+\widetilde{O}\left(\frac{1}{\sqrt{n}}+\frac{m}{n^{3/2}}\right).
\end{equation}
To bound the second term we write
\begin{equation}\label{eq:certificate-core-term-cross-prod}
\begin{gathered}
 \sum\limits_{i, j:\ i\neq j}\langle w_i, w_j\rangle \langle u_i, x\rangle \langle v_i, y \rangle\langle u_j, x\rangle \langle v_j, y \rangle =  \\
 = (x\otimes y)^T\left(\sum\limits_{i, j:\ i\neq j}\langle w_i, w_j\rangle (u_iv_i^T)\otimes (u_jv_j^T)\right)(y\otimes x)\leq \\
 \leq \Vert x\otimes y \Vert \left\Vert \sum\limits_{i, j:\ i\neq j}\langle w_i, w_j\rangle (u_iv_i^T)\otimes (u_jv_j^T) \right\Vert \Vert y\otimes x \Vert.
\end{gathered}
\end{equation}

As in \cite{Ge-Ma}, to bound the norm of the matrix $S = \sum\limits_{i, j:\ i\neq j}\langle w_i, w_j\rangle (u_iv_i^T)\otimes (u_jv_j^T)$ we replace vectors $w_i$'s with $\sigma_i w_i$, where $\sigma_i$ is a random variable with uniform distribution on $\{-1, 1\}$. Since, $w_i$ and $\sigma_i w_i$ have the same distribution this does not affect the distribution of $S$. Moreover, by Theorem~\ref{thm:random-variable-split}, it is sufficient to bound the norm of 
\[S' = \sum\limits_{i, j:\ i\neq j}\sigma_i \tau_j \langle w_i, w_j\rangle (u_iv_i^T)\otimes (u_jv_j^T),\] 
where $\tau_j$ is an independent copy of $\sigma_j$. To bound the norm of $S'$ we use the matrix Bernstein inequality with respect to the randomness of $\sigma_i$ and $\tau_i$. Define
\[ X_{i, j} = \langle w_i, w_j\rangle (u_jv_j^T),\qquad \text{and} \qquad Y_i = \sum\limits_{j:\ j \neq i} \tau_j X_{i, j}, \]
\[  R_i =(u_iv_i^T)\otimes Y_i =  (u_iv_i^T)\otimes\sum\limits_{j:\ j \neq i} \tau_j X_{i, j}, \quad \quad  \text{so that} \quad S' = \sum\limits_{i} \sigma_i R_i.\]

First we bound the norm of $Y_i$ by applying the Bernstein inequality to $\tau_j X_{i, j}$. We check that  with high probability over the randomness of  $\mathcal{V}$, for $i\neq j$, by Fact~\ref{fact:inner-produc-hp-bound} and Corollary~\ref{cor:basic-deg2-sum},
 \[ \Vert X_{i, j}\Vert = \Vert \langle w_i, w_j\rangle u_j v_j^T \Vert \leq \widetilde{O}\left(\frac{1}{\sqrt{n}}\right), \quad \text{and}\]
 \[ \left\Vert\sum\limits_{j:\ j\neq i} X_{i, j}X_{i, j}^T\right\Vert = \left\Vert \sum\limits_{j:\ j\neq i}  \langle w_i, w_j\rangle^2 (u_j u_j^T) \right\Vert \leq \]
 \[  \leq \left(\max_{i\neq j} \langle w_i, w_j \rangle^2\right) \max_{\Vert a \Vert = 1}  \sum\limits_{j} \langle u_j, a\rangle^2 = \widetilde{O}\left(\frac{1}{n}+\frac{m}{n^2}\right).\]
 Similarly,
 \[ \left\Vert\sum\limits_{j:\ j\neq i} X_{i, j}^T X_{i, j}\right\Vert  = \widetilde{O}\left(\frac{1}{n}+\frac{m}{n^2}\right).\]
Therefore, by the Bernstein inequality, using the randomness of $\tau_i$'s, w.h.p., for each $i\in [m]$,
\begin{equation}
\Vert Y_i \Vert = \widetilde{O}\left(\frac{1}{\sqrt{n}}+\frac{\sqrt{m}}{n}\right).
\end{equation}
Using that $X\otimes Y\succeq 0$ if $X\succeq 0$ and $Y\succeq 0$, we bound  
\[ \left\Vert\sum\limits_{i}^{n}  R_iR_i^T \right\Vert = \left\Vert\sum\limits_{i}^{n} \left(u_i u_i^T\right)\otimes Y_iY_i^T \right\Vert \leq \left\Vert\sum\limits_{i}^{n} \left(u_i u_i^T\right)\otimes \widetilde{O}\left(\frac{1}{n}+\frac{m}{n^2}\right) I \right\Vert = \]
\[  = \left\Vert \left( \sum\limits_{i}^{n} u_i u_i^T \right)\otimes \widetilde{O}\left(\frac{1}{n}+\frac{m}{n^2}\right) I \right\Vert \leq ||U U^T||\widetilde{O}\left(\frac{1}{n}+\frac{m}{n^2}\right) = \widetilde{O}\left(\frac{m}{n^2}+\frac{m^2}{n^3}\right)\]

In absolutely the same way we show
\begin{equation}\label{eq:ber-sigma-bound2}
  \left\Vert\sum\limits_{i}^{n}  R_i^TR_i \right\Vert\leq ||V V^T||\widetilde{O}\left(\frac{1}{n}+\frac{m}{n^2}\right) = \widetilde{O}\left(\frac{m}{n^2}+\frac{m^2}{n^3}\right).
\end{equation}  
 
 Hence, applying the Bernstein inequality to $R_i$ using the randomness of $\sigma_i$'s, we get that with high probability 
\begin{equation}\label{eq:certificate-core-term-cross-prod-matrix}
 \Vert S'\Vert = \widetilde{O} \left(\frac{1}{\sqrt{n}}+\frac{m}{n^{3/2}}\right). 
 \end{equation}
Combining this bound with the bound in Eq.~\eqref{eq:sqaure-term-bound} we deduce the statement of the theorem.
\end{proof}

Now we establish a stronger bound when $x \otimes y \otimes z$ is of the special form $u_j^{\perp}\otimes v_j^{\perp}\otimes w_j$.

\begin{lemma}\label{lem:pairwise-terms} Let $m\ll n^2$ and $\oa{A}$ be a certificate candidate for $\mathcal{V}$. With high probability 
\begin{equation}
|\langle A, u_j^{\perp}\otimes v_j^{\perp}\otimes w_j\rangle| = \Vert u_j^{\perp}\Vert \Vert v_j^{\perp} \Vert \widetilde{O}\left(\frac{1}{\sqrt{n}}+\frac{m}{n^{3/2}}\right),
\end{equation}
where $\oa{A}$ is given by Eq.~\eqref{eq:certificate-3-form}, and $u_j^{\perp}, v_j^{\perp}$ are some vectors orthogonal to $u_j$ and $v_j$, respectively.
\end{lemma}
\begin{proof} By Theorem~\ref{thm:candidate-exists}, we can write
\[  \oa{A} = \sum\limits_{i=1}^{m}  
u_i\otimes v_i\otimes w_i+ \sum\limits_{i=1}^{m} \alpha_i\otimes v_i \otimes w_i+ \sum\limits_{i=1}^{m} u_i\otimes \beta_i \otimes w_i +\sum\limits_{i=1}^{m} u_i\otimes v_i \otimes \gamma_i 
.\]
 Using Lemma~\ref{lem:error-term-effect} we can bound the inner product of $u_j^{\perp}\otimes v_j^{\perp}\otimes w_j $ with each of the last three sums by $\displaystyle{\Vert u_j^{\perp}\Vert \Vert v_j^{\perp} \Vert\widetilde{O}\left(\frac{\sqrt{m}}{n}+\frac{m}{n^{3/2}}\right)}$. To bound the inner product with the first sum we first use orthogonality of $u_i^{\perp}$ and $u_i$ and then we apply Cauchy-Schwarz inequality
 \[ \left\langle \sum\limits_{i=1}^{m} u_i\otimes v_i\otimes w_i, u_j^{\perp}\otimes v_j^{\perp}\otimes w_j \right\rangle = \sum\limits_{i=1}^{m}\langle u_i, u_j^{\perp}\rangle \langle v_i, v_j^{\perp}\rangle \langle w_i, w_j\rangle  = \]
 \[ = \sum\limits_{i:\ i\neq j}\langle u_i, u_j^{\perp}\rangle \langle v_i, v_j^{\perp}\rangle \langle w_i, w_j\rangle  \leq \max_{i\neq j} |\langle w_i, w_j\rangle| \sqrt{\sum\limits_{i=1}^{m}\langle u_i, u_j^{\perp}\rangle^2}\sqrt{\sum\limits_{i=1}^{m}\langle v_i, v_j^{\perp}\rangle^2}.\]
Using Fact~\ref{fact:inner-produc-hp-bound} and Corollary~\ref{cor:basic-deg2-sum}, with high probability 
\[ \max_{i\neq j} |\langle w_i, w_j\rangle| \sqrt{\sum\limits_{i=1}^{m}\langle u_i, u_j^{\perp}\rangle^2}\sqrt{\sum\limits_{i=1}^{m}\langle v_i, v_j^{\perp}\rangle^2} = \widetilde{O}\left(\frac{1}{\sqrt{n}}\right)\widetilde{O}\left(1+\frac{m}{n}\right) = \widetilde{O}\left(\frac{1}{\sqrt{n}}+\frac{m}{n^{3/2}}\right).\]
\end{proof}

Recall that we want to prove that $\langle \oa{A}, x\otimes y\otimes z\rangle \leq 1$ for any $x, y, z\in S^{n-1}$. We will consider two cases: $x\otimes y\otimes z$ is ``close'' to some $u_j\otimes v_j\otimes w_j$; and $x\otimes y\otimes z$ is ``far'' from all $u_j\otimes v_j\otimes w_j$. The next theorem deals with the first case.

\begin{theorem}\label{thm:near-vertex-bound} Let $m\ll n^2$ and $\oa{A}$ be a certificate candidate for $\mathcal{V}$. Assume that $u_j^{\perp}$, $v_j^{\perp}$ and $w_j^{\perp}$ are some vectors orthogonal to $u_j$, $v_j$ and $w_j$, respectively. Then with high probability
\[ \left\vert \left\langle \oa{A}, (u_j+ u_j^{\perp})\otimes (v_j+ v_j^{\perp})\otimes (w_j+ w_j^{\perp})\right\rangle \right\vert \leq \]
\[ \leq 1+\varepsilon\left(\Vert u_j^{\perp}\Vert \Vert v_j^{\perp} \Vert + \Vert u_j^{\perp}\Vert \Vert w_j^{\perp} \Vert +\Vert v_j^{\perp}\Vert \Vert w_j^{\perp} \Vert \right) + (1+\varepsilon)\Vert u_j^{\perp}\Vert \Vert v_j^{\perp} \Vert \Vert w_j^{\perp} \Vert,\]
where $0<\displaystyle{\varepsilon = \widetilde{O}\left(\frac{1}{\sqrt{n}}+\frac{m}{n^{3/2}}\right)}$.

Moreover, if the $\varepsilon$ above satisfies $\varepsilon<1/4$ and $\min\left(\Vert u_i^{\perp} \Vert, \Vert v_i^{\perp} \Vert, \Vert w_i^{\perp} \Vert\right)< 1/3$, then
\[ \langle \oa{A}, (u_j+ u_j^{\perp})\otimes (v_j+ v_j^{\perp})\otimes (w_j+ w_j^{\perp})\rangle^2 \leq (1+\Vert u_j^{\perp} \Vert^2)(1+\Vert v_j^{\perp} \Vert^2)(1+\Vert w_j^{\perp} \Vert^2),\]
and equality holds if and only if $u_j^{\perp} = v_j^{\perp} = w_j^{\perp} = 0$.
\end{theorem}
\begin{proof}
The first claim directly follows from Corollary~\ref{cor:certificate-expansion}, Lemma~\ref{lem:pairwise-terms}, Lemma~\ref{lem:error-term-effect} and Theorem~\ref{thm:inj-norm-bound} applied to $x\otimes y\otimes z = u_j^{\perp}\otimes v_j^{\perp}\otimes w_j^{\perp}$. 

To prove the second claim, take $a = \Vert u_j^{\perp}\Vert$, $b = \Vert u_j^{\perp}\Vert$, and $c = \Vert u_j^{\perp}\Vert$. Since $ab+bc+ca\leq a^2+b^2+c^2$, the first claim implies
\begin{equation*}
\begin{gathered}
  \langle \oa{A}, (u_j+ u_j^{\perp})\otimes (v_j+ v_j^{\perp})\otimes (w_j+ w_j^{\perp})\rangle^2 \leq \\
\leq 1+2\varepsilon(a^2+b^2+c^2)+2(1+\varepsilon)abc+[\varepsilon(ab+bc+ac)+(1+\varepsilon)abc]^2\leq \\
\leq 1+2\varepsilon(a^2+b^2+c^2)+2(1+\varepsilon)abc+6\varepsilon^2(a^2b^2+b^2c^2+c^2a^2)+2(1+\varepsilon)^2(abc)^2
\end{gathered}
\end{equation*}

Note that if the inequality $\min\left(a, b, c \right)< 1/3$ holds, and $\varepsilon<1/4$, then
\[ 2\varepsilon(a^2+b^2+c^2)+2(1+\varepsilon)abc\leq \left(2\varepsilon+\dfrac{1+\varepsilon}{3}\right)(a^2+b^2+c^2)\leq  a^2+b^2+c^2, \quad \text{and}\]
\[ 6\varepsilon^2(a^2b^2+b^2c^2+c^2a^2)+2(1+\varepsilon)^2(abc)^2\leq a^2b^2+b^2c^2+c^2a^2.\]
Moreover, the inequalities are strict, unless $u_i^{\perp} = v_i^{\perp} = w_i^{\perp} = 0$. Combining them together, we obtain
\[  \left\langle \oa{A}, (u_j+ u_j^{\perp})\otimes (v_j+ v_j^{\perp})\otimes (w_j+ w_j^{\perp})\right\rangle^2 \leq (1+a^2)(1+b^2)(1+c^2).\]
\end{proof}

In the case when $x\otimes y\otimes z$ is ``far'' from all $u_j\otimes v_j\otimes w_j$ we use the following bound.
 
\begin{theorem}\label{thm:far-bound} Let $m\ll n^{3/2}$ and $\oa{A}$ be a certificate candidate for $\mathcal{V}$. Let $x, y, z\in S^{n-1}$. Assume that $\vert \langle x, u_i \rangle \vert\leq \delta$ for every $i\in [m]$. Then with high probability 
\[ \vert\langle \oa{A}, x\otimes y\otimes z\rangle \vert \leq \delta+\widetilde{O}\left(\frac{1}{n^{1/4}}+\frac{\sqrt{m}}{n^{3/4}}+\frac{\sqrt{m}}{n}+\frac{m}{n^{3/2}}\right).\] 
\end{theorem}
\begin{proof} By Lemma~\ref{lem:error-term-effect}, 
\[\vert \langle \oa{A}, x\otimes y\otimes z\rangle \vert \leq \left\vert\sum\limits_{i=1}^n\langle u_i, x\rangle\langle v_i, y\rangle\langle w_i, z\rangle\right\vert + \widetilde{O}\left(\frac{\sqrt{m}}{n}+\frac{m}{n^{3/2}}\right).\]
As in Eq.~\eqref{eq:in-prod-expansion-certificate}, we can write
\[\left(\sum\limits_{i=1}^n\langle u_i, x\rangle\langle v_i, y\rangle\langle w_i, z\rangle\right)^2 = \sum\limits_{i=1}^{m}\langle u_i, x\rangle^2 \langle v_i, y \rangle^2+ \sum\limits_{i, j:\  i\neq j}\langle w_i, w_j\rangle \langle u_i, x\rangle \langle v_i, y \rangle\langle u_j, x\rangle \langle v_j, y \rangle.\]
Note that in Eq.~\eqref{eq:certificate-core-term-cross-prod} and Eq.~\eqref{eq:certificate-core-term-cross-prod-matrix} we proved that
\[ \sum\limits_{i, j:\  i\neq j}\langle w_i, w_j\rangle \langle u_i, x\rangle \langle v_i, y \rangle\langle u_j, x\rangle \langle v_j, y \rangle \leq \widetilde{O} \left(\frac{1}{\sqrt{n}}+\frac{m}{n^{3/2}}\right).\]
For the other term, Theorem~\ref{thm:deg4bound} implies, 
\[ \left(\sum\limits_{i=1}^{m}\langle u_i, x\rangle^2 \langle v_i, y \rangle^2 \right)^2 \leq \sum\limits_{i=1}^{m}\langle u_i, x\rangle^4 \sum\limits_{i=1}^{m}\langle v_i, y \rangle^4\leq \sum\limits_{i=1}^{m}\langle u_i, x\rangle^4 \left( 1+\widetilde{O}\left(\frac{1}{\sqrt{n}}+\frac{m}{n^{3/2}}\right)\right). \]
By Lemma~\ref{lem:sixtofourbound},  
\[ \left(\sum\limits_{i=1}^{m}\langle u_i, x\rangle^4\right)^2\leq \sum\limits_{i=1}^{m}\langle u_i, x\rangle^6+\widetilde{O}\left(\frac{1}{n}+\frac{m^2}{n^3}\right).\]
Moreover,  
\[ \left(\sum\limits_{i=1}^{m}\langle u_i, x\rangle^6\right) \leq \max_{i}\vert\langle u_i, x\rangle\vert^2  \left(\sum\limits_{i=1}^{m}\langle u_i, x\rangle^4\right)\leq \delta^2 \sum\limits_{i=1}^{m}\langle u_i, x\rangle^4.\]
Therefore, combining the last three equations we get
\[ \left(\sum\limits_{i=1}^{m}\langle u_i, x\rangle^2 \langle v_i, y \rangle^2 \right)^2  = \delta^2+\widetilde{O}\left(\frac{1}{\sqrt{n}}+\frac{m}{n^{3/2}}\right).\]

\end{proof}

Finally, we combine the results of this section into the proof of Theorem~\ref{thm:certificate-existence}.
\begin{proof}[Proof of Theorem~\ref{thm:certificate-existence}] Consider $x\otimes y\otimes z$ for $x, y, z\in S^{n-1}$. We consider two cases.

If for all $j\in [m]$, the inequality $\vert\langle x, u_j\rangle\vert < 0.99$ holds, then the statement of the theorem follows from Proposition~\ref{thm:far-bound}.

Otherwise, there exists an index $j$ with $\langle x, u_j\rangle \geq 0.99$. Assume that $\langle \oa{A}, x\otimes y\otimes z\rangle \rangle >1$, then by continuity we can choose $y$ and $z$ which are not orthogonal to $v_j$ and $w_j$ such that this inequality still holds. Then  we can define $u_{j}^{\perp}$, $v_{j}^{\perp}$ and $w_{j}^{\perp}$, where 
$ u_j^{\perp} = (x - \langle u_j, x\rangle u_j)/\langle x, u_j\rangle$ and $v_j^{\perp}$, $w_{j}^{\perp}$ are defined similarly. Clearly, $u_j^{\perp}\perp u_j$ and $x = \langle x, u_j \rangle (u_j+u_j^{\perp})$. Since $1+ \Vert u_{j}^{\perp}\Vert^2 = \Vert u_j+ u_{j}^{\perp}\Vert^2 = 1/\langle x, u_j \rangle^2$, we get a contradiction with Theorem~\ref{thm:near-vertex-bound}.

Therefore, for any $x, y, z\in S^{n-1}$ the inequality $\langle \oa{A} , x\otimes y\otimes z\rangle \rangle \leq 1$ holds. Moreover, it is easy to see from Theorem~\ref{thm:near-vertex-bound} that equality can be achieved only for $x\otimes y\otimes z = \pm u_i\otimes v_i\otimes w_i$.

\end{proof}

\section{Construction of $A^TA$ polynomial with small matrix norm}\label{sec:B0}
Our goal now is to construct a matrix $B_0\equiv_{poly} A^TA$ with small matrix norm. Recall that $A = A_0+C$, where 
\begin{equation}\label{eq:A-matrix-form}
A_0 = \sum\limits_{i=1}^{ m} u_i(v_i\otimes w_i)^T\quad \text{and} \quad C = \sum\limits_{i=1}^{m} \alpha_i(v_i \otimes w_i)^T+ u_i(\beta_i \otimes w_i)^T+ u_i(v_i \otimes \gamma_i)^T.
\end{equation}
Then 
\[ A^TA = \sum\limits_{i, j = 1}^m \langle u_i, u_j \rangle (v_i\otimes w_i)(v_j\otimes w_j)^T+ C^TA_0+A_0^TC+C^TC.\]
In this section we prove the following result.
\begin{theorem}\label{thm:B0-construction}
Let $m\ll n^{3/2}$. Define $B_0 = \tw_2(A^TA)$. Then $B_0\equiv_{poly} A^TA$ and w.h.p.
\[ \Vert B_0 - P_{\mathcal{L}} \Vert = \wt{O}\left(\dfrac{m}{n^{3/2}}\right), \]
where $\mathcal{L} = \vspan\{v_i\otimes w_i \mid i\in [m]\}$.
\end{theorem}
The proof consists of three parts: analysis for $A_0^TA_0$, analysis for terms of small Frobenius norm, and analysis for the rest of the terms (which have IP graph matrix structure).

\subsection{Analysis for $A_0^TA_0$}
 We split the $A_0^TA_0$ term into two parts
\[ A_0^TA_0 = \sum\limits_{i = 1}^{m} (v_i\otimes w_i)(v_i\otimes w_i)^T+\sum\limits_{i\neq j} \langle u_i, u_j \rangle (v_i\otimes w_i)(v_j\otimes w_j)^T.  \]
The next two lemmas show that the first sum approximates $P_{\mathcal{L}}$ well.
\begin{lemma}\label{lem:basis-in-L} Assume that $x = \sum\limits_{i=1}^{m} \mu_i (v_i\otimes w_i)$ for $\mu \in \mathbb{R}^m$. Then w.h.p. 
$ \Vert x\Vert = \Vert \mu\Vert \left(1+\wt{O}\left(\dfrac{\sqrt{m}}{n}\right)\right)$.
\end{lemma}
\begin{proof} Definition of $x$ can be alternatively written as
$ x = (V\cten W)\mu $, thus By Lemma~\ref{lem:basic-norm-bounds}, w.h.p 
$\Vert x\Vert \leq \Vert \mu\Vert \left(1+\wt{O}\left(\dfrac{\sqrt{m}}{n}\right)\right)$. At the same time, we can write 
\[ (V\cten W)^Tx = (V\cten W)^T(V\cten W)\mu.\]
Hence, by Lemma~\ref{lem:basic-cten-bound} we get the opposite inequality  $\Vert \mu\Vert \leq \Vert x\Vert \left(1+\wt{O}\left(\dfrac{\sqrt{m}}{n}\right)\right)$.
\end{proof}

\begin{lemma}\label{lem:Bvw-bound}
Let $m\ll n^2$ and $B_{vw} = \sum\limits_{i = 1}^{m} (v_i\otimes w_i)(v_i\otimes w_i)^T$, then $\tw_2(B_{vw}) = B_{vw}$ and w.h.p. 
\[ \left\Vert B_{vw} - P_{\mathcal{L}} \right\Vert = \wt{O}\left(\dfrac{\sqrt{m}}{n}\right).\] 
\end{lemma}
\begin{proof} Assume that $x\perp \mathcal{L}$, then $B_{vw}x = P_{\mathcal{L}}x = 0$. Hence, to prove the lemma, it is sufficient to show that for $x\in \mathcal{L}$ with $\Vert x\Vert = 1$ the inequality $\left(B_{vw}-P_{\mathcal{L}}\right)x = \wt{O}\left(\dfrac{\sqrt{m}}{n}\right)$ holds. Every vector $x\in \mathcal{L}$ can be written as $\sum\limits_{j=1}^m \mu_j (v_j\otimes w_j)$ for some $\mu \in \mathbb{R}^{m}$.  Observe that
\begin{equation*}
\begin{gathered}
 \left(B_{vw}-P_{\mathcal{L}}\right)\sum\limits_{j=1}^{m} \mu_{j}(v_j\otimes w_j) = \sum\limits_{j=1}^{m} \mu_{j}\sum\limits_{i=1}^{m}\langle v_i, v_j\rangle\langle w_i, w_j\rangle (v_i\otimes w_i) - \sum\limits_{j=1}^{m} \mu_{j}(v_j\otimes w_j) = \\
  = (V\cten W)\left((V\cten W)^T(V\cten W) - I_m\right)\mu.  
  \end{gathered}
  \end{equation*} 
  Hence, the statement follows from Lemma~\ref{lem:basic-norm-bounds}, Lemma~\ref{lem:basic-cten-bound} and Lemma~\ref{lem:basis-in-L}. 
\end{proof}

Next, we show that the norm of $\tw_2(A^T_0A_0-B_{vw})$ is small.

\begin{theorem}\label{thm:A0-twist-bound} Let $m\ll n^{3/2}$. With high probability 
\begin{equation}
\Vert \tw_2(A^T_0A_0-B_{vw}) \Vert = \left\Vert \sum\limits_{i\neq j} \langle u_i, u_j \rangle (v_iv_j^T)\otimes (w_jw_i^T) \right\Vert = \widetilde{O} \left(\frac{1}{\sqrt{n}}+\frac{m}{n^{3/2}}\right)
\end{equation}
\end{theorem}
\begin{proof}  For independent uniformly distributed on $\{-1, 1\}$ random variables $\sigma_i$ vectors $w_i$ and $\sigma_i w_i$ have the same distribution. Hence, as in the proof of Theorem~\ref{thm:inj-norm-bound}, using the decoupling inequality from Theorem~\ref{thm:random-variable-split}, it is sufficient to show that 
\[ \Vert S \Vert = \widetilde{O} \left(\dfrac{1}{\sqrt{n}}+\dfrac{m}{n^{3/2}}\right)\quad   \text{ for}\quad  
 S = \sum\limits_{i\neq j} \sigma_i \tau_j \langle u_i, u_j \rangle (v_iv_j^T)\otimes (w_jw_i^T),  \]
where $\tau_i$ is an independent copy of a random variable $\sigma_i$. Our key observation is that for 
\[ \tw_R(S) = \sum\limits_{i\neq j} \sigma_i \tau_j \langle u_i, u_j \rangle (v_iw_i^T)\otimes (w_jv_j^T),\]
its norm $\Vert \tw_R(S) \Vert$ equals the norm of $\Vert S\Vert$. To bound the norm of $\tw_R(S)$ we apply the Matrix Bernstein inequality twice: using randomness of $\{\sigma_i\}$ and using randomness of $\{\tau_i\}$.
 
 We introduce 
\[R_i = \sum\limits_{j:\ j\neq i} \tau_j \langle u_i, u_j \rangle (v_iw_i^T)\otimes (w_jv_j^T), \quad \text{so that} \quad \tw_R(S) = \sum\limits_{i=1}^m \sigma_i R_i.\]
To apply the Matrix Bernstein inequality using the randomness of $\sigma_i$, we need to bound $\Vert R_i\Vert $, $\left\Vert \sum\limits_{i=1}^{m} R_iR_i^T\right\Vert$, and $\left\Vert \sum\limits_{i=1}^{m} R_i^TR_i\right\Vert$.
 Consider
\[ X_{i, j} = \langle u_i, u_j \rangle (v_iw_i^T)\otimes (w_jv_j^T), \quad \text{so that} \quad R_i = \sum \limits_{j:\ j\neq i} \tau_j X_{i, j}. \]
Note that $\left\Vert X_{i ,j}\right\Vert = | \langle u_i, u_j \rangle | = \widetilde{O}(1/\sqrt{n})$ w.h.p over the randomness of $\mathcal{V}$. Moreover, since $(v_iv_i^T)\otimes (w_kw_k)^T$ is positive semidefinite, w.h.p.
\[ \left\Vert \sum\limits_{j:\ j\neq i} X_{i, j}X_{i, j}^T\right\Vert = \left\Vert \sum\limits_{j: j\neq i} \langle u_i, u_j \rangle^2 (v_i v_i^T)\otimes (w_j w_j^T) \right\Vert \leq \left\Vert \max_{j:\ j\neq i}\langle u_i, u_j \rangle^2\sum\limits_{j:\ j\neq i}  (v_i v_i^T)\otimes (w_j w_j^T) \right\Vert \leq\]
\[ \leq \left\Vert \widetilde{O}\left( \frac{1}{n}\right)  (v_i v_i^T)\otimes \sum\limits_{j: j\neq i}(w_j w_j^T) \right\Vert =  \widetilde{O}\left( \frac{1}{n}\right) \left\Vert \sum\limits_{j: j\neq i}(w_j w_j^T) \right\Vert = \widetilde{O}\left( \frac{1}{n}+\frac{m}{n^2}\right). \]
Similarly,
\[ \left\Vert \sum\limits_{j:\ j\neq i} X_{i, j}^T X_{i, j}\right\Vert = \widetilde{O}\left( \frac{1}{n}+\frac{m}{n^2}\right). \]
Therefore, using the randomness of $\tau_j$, by the Matrix Bernstein inequality,
\[ \left\Vert R_i \right\Vert = \widetilde{O}\left( \frac{1}{\sqrt{n}}+\frac{\sqrt{m}}{n}\right).\]
The crucial effect of working with $\tw_R(S)$ instead of $S$ is that the summations over $i$ and $j$ in the expression $\sum\limits_{i} R_iR_i^T$ happen in different components of the tensor product. 

 Consider
\[ Y_i = \sum\limits_{j:\ j\neq i} \tau_j \langle u_i, u_j \rangle  (w_jv_j^T), \qquad \text{then}\qquad R_i = (v_iw_i^T)\otimes Y_{i}.\]
As for $R_i$, using randomness of $\tau_j$, by the Matrix Bernstein inequality $\left\Vert Y_{i}\right\Vert = \widetilde{O}\left( \dfrac{1}{\sqrt{n}}+\dfrac{\sqrt{m}}{n}\right)$. 
Now, since $v_i v_i^T$ is positive semidefinite, we can write
\[ \left\Vert \sum\limits_{i} R_iR_i^T \right\Vert = \left\Vert \sum\limits_{i} (v_i v_i^T)\otimes (Y_iY_i^T) \right\Vert \leq \left\Vert \sum\limits_{i} (v_i v_i^T)\otimes \left(\widetilde{O}\left( \frac{1}{n}+\frac{m}{n^2}\right)I\right) \right\Vert =  \]
\[ = \widetilde{O}\left( \frac{1}{n}+\frac{m}{n^2}\right)\left\Vert \sum\limits_{i} (v_i v_i^T) \right\Vert = \widetilde{O}\left( \frac{m}{n^2}+\frac{m^2}{n^3}\right).\]
Similarly, we obtain 
\[ \left\Vert \sum\limits_{i} R_iR_i^T \right\Vert = \widetilde{O}\left( \frac{m}{n^2}+\frac{m^2}{n^3}\right).\]
Therefore, by the Matrix Bernstein inequality, using the randomness of $\sigma_i$, we get
\[ \left\Vert S \right\Vert = \left\Vert \tw_R(S) \right\Vert = \widetilde{O} \left(\frac{1}{\sqrt{n}}+\frac{m}{n^{3/2}}\right).\]

\end{proof}

\subsection{Norm bound for $A^TA - A^T_0A_0$}
Recall that, as discussed in Section~\ref{sec:corr-terms} (see Theorem~\ref{thm:corr-terms-summary}), we can write $ U' = U_{GM}+U_{sm}$, where 
 $U_{GM} \in \vspan(\mathfrak{CM}^4(\mathcal{C}; 2), 100)$ 
and $U_{sm}$ satisfies $\Vert U_{sm} \Vert_F = \wt{O}\left(\dfrac{m^{3/2}}{n^{5/2}}+\dfrac{m^{5/2}}{n^4}\right)$.

Consider  
\begin{equation}\label{eq:Csm-definition}
 C_{sm} = \sum\limits_{i=1}^{m} (\alpha_i)_{sm}(v_i \otimes w_i)^T+ u_i((\beta_i)_{sm} \otimes w_i)^T+ u_i(v_i \otimes (\gamma_i)_{sm})^T
 \end{equation}
and let $C' = C-C_{sm}$. 

\begin{lemma}\label{lem:A-frob-norm-place-indep}
Let $X$ be an $n\times m$ matrix, then w.h.p.
\[\max\left(\Vert X(V\cten W)^T\Vert_F, \Vert U(X\cten W)^T\Vert_F, \Vert U(V\cten X)^T\Vert_F \right)  \leq \left(1+\wt{O}\left(\dfrac{\sqrt{m}}{n}\right)\right)\Vert X \Vert_F.\]
\end{lemma}
\begin{proof}
Observe that
\begin{equation*}
\begin{gathered}
 \Vert U(X\cten W)^T \Vert_F = \left\Vert \sum\limits_{i=1}^n u_i\otimes x_i \otimes w_i \right\Vert = \left\Vert \sum\limits_{i=1}^n x_i\otimes u_i \otimes w_i \right\Vert = \\
 = \Vert X(U\cten W)^T \Vert_F \leq \Vert X \Vert_F\cdot \Vert U\cten W \Vert. 
 \end{gathered}
 \end{equation*}
\end{proof}

\begin{proposition}\label{prop:A-norm-bound} Let $m\ll n^2$. With high probability
$\Vert A\Vert = \wt{O}\left(1+\dfrac{\sqrt{m}}{\sqrt{n}}+\dfrac{m^{5/2}}{n^{7/2}}+\dfrac{m^{2}}{n^3}\right)$.
\end{proposition}
\begin{proof} Recall that $A = U(V\cten W)+U'(V\cten W)+U(V'\cten W)+U(V\cten W')$. By Theorem~\ref{thm:candidate-exists} and Lemma~\ref{lem:cten-norm-bound}, 
\[ \max\left(\Vert U'\Vert, \Vert V'\cten W\Vert, \Vert V\cten W' \Vert\right) =  \widetilde{O}\left(\frac{m}{n^{3/2}}+\dfrac{m^2}{n^3}+\frac{\sqrt{m}}{n}\right).\]
Since, by Lemma~\ref{lem:basic-norm-bounds}, $\Vert U\Vert = \wt{O}\left(1+ \dfrac{\sqrt{m}}{\sqrt{n}}\right)$ and $\Vert V\cten W\Vert = \wt{O}\left(1+ \dfrac{\sqrt{m}}{n}\right)$, the claim of the proposition follows.
\end{proof}

Combining these bounds together we get the following result.
\begin{proposition}\label{prop:B0-small-frob-part} Let $m\ll n^2$ and $C_{sm}$ be as defined in Eq.~\eqref{eq:Csm-definition}. Then w.h.p.
 \[\Vert A^T C_{sm}\Vert_F = \Vert C_{sm}^TA\Vert_F = \wt{O}\left(\dfrac{m^2}{n^3}+\dfrac{m^4}{n^6}+\dfrac{n^{5/2}}{n^4}\right)\quad \text{and} \quad  \Vert C_{sm}^TC_{sm} \Vert_F = \wt{O}\left(\dfrac{m^3}{n^5}+\dfrac{m^5}{n^8}\right) .\]
\end{proposition}
\begin{proof} From Theorem~\ref{thm:corr-terms-summary} we know that $\max\left(\Vert U'_{sm}\Vert_F, \Vert V'_{sm}\Vert_F, \Vert W'_{sm}\Vert_F\right)\leq \wt{O}\left(\dfrac{m^{3/2}}{n^{5/2}}+\dfrac{m^{5/2}}{n^4}\right)$. Hence, the result is implied by Lemma~\ref{lem:A-frob-norm-place-indep}, and Proposition~\ref{prop:A-norm-bound}.
\end{proof}
Next, we analyze the terms with IP graph matrix structure involved in $A^TA$.
\begin{lemma}\label{lem:B0graph-terms} Let $\mathcal{C} = \{\{u, v\}, \{u, w\}, \{v, w\}\}$. The matrices $\tw_2((C')^TC')$, $\tw_2(A_0^TC')$ and $\tw_2((C')^TA_0)$ belong to the class $\vspan\left(\mathfrak{BCM}^{10}(\mathcal{C}; 2, 2), 300\right)$.
Moreover, for $X = A_0U_E(V\cten W)^T$
\[ \tw_2\left((C')^TC'+A_0^TC'+(C')^TA_0 - X - X^T\right)\in \vspan\left(\mathfrak{BCM}^{10}(\mathcal{C}; 3, 2), 10^3\right).\]
\end{lemma}
 
 \begin{figure}
\begin{subfigure}[b]{0.3\textwidth}
\begin{center}
\includegraphics[height = 3cm]{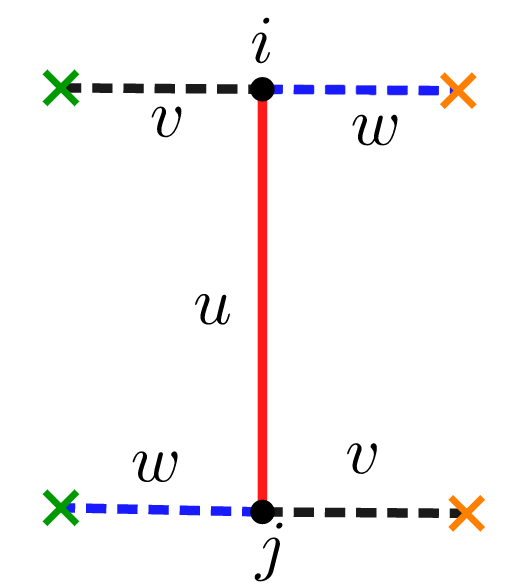}
\caption{Diagram for $\tw_2(A_0^TA_0)$}
\end{center}
\end{subfigure}
\begin{subfigure}[b]{0.3\textwidth}
\begin{center}
\includegraphics[height = 3cm]{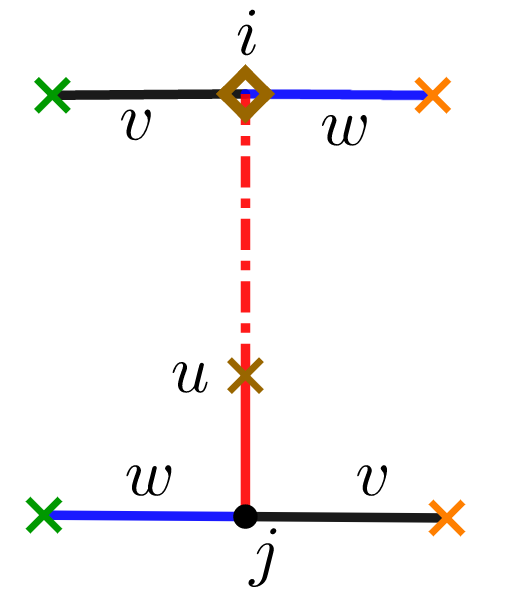}
\caption{Diagram for $X_{ul}$}
\end{center}
\end{subfigure}
\begin{subfigure}[b]{0.3\textwidth}
\begin{center}
\includegraphics[height = 3cm]{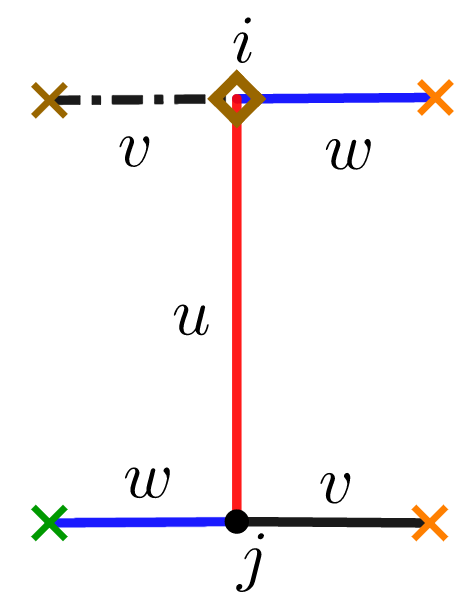}
\caption{Diagram for $X_{vl}$}
\end{center}
\end{subfigure}
\caption{Templates for $(C')^TA$}\label{fig:templates}
\end{figure}

\begin{proof}
Recall that $A_0 = U(V\cten W)^T$ and for $U_{GM}$, $V_{GM}$, $W_{GM}\in \vspan(\mathfrak{CM}^4(\mathcal{C}, 2), 100)$ we have 
\[C' = U_{GM}(V\cten W)^T+U(V_{GM}\cten W)^T+U(V\cten W_{GM})^T.\]
Consider the matrix diagram $G$ of $\tw_2(A_0^TA_0)$, presented on Figure~\ref{fig:templates} (a).  Clearly, $\tw_2(A_0^TA_0)$ is in the class $\mathfrak{BCM}^2(\mathcal{C}; 1, 2)$. Next, observe that $\tw_2((C')^TA_0)$ is a linear combination of matrices of the form 
\[ \tw_2\left((V\cten W)X^TU(V\cten W)^T\right), \quad \tw_2\left((V\cten W)X^TU(V\cten W)^T\right), \quad  \tw_2\left((V\cten W)X^TU(V\cten W)^T\right),\]
where $X$ is some graph matrix involved in $U_{GM}$, $V_{GM}$ and $W_{GM}$, respectively. Denote these matrices $X_{ul}$, $X_{vl}$ and $X_{wl}$. We claim that the matrix diagrams for these matrices are obtained from a matrix diagram $G$ for $\tw_2(A_0^TA_0)$ by replacing a half edge of $G$ with a matrix diagram of $X$. We illustrate such half edge being replaced in $X_{ul}$ and $X_{vl}$ with dash-dot line on Figure~\ref{fig:templates}.

Hence, we may apply the following claim.   

\begin{claim} Assume that $G = (\Omega, E, \mathfrak{c})$ is $\mathcal{C}$-boundary-connected and has at most $d$\ $\mathcal{C}$-connected components. Let $X = (\Omega', E', \mathfrak{c}')$ be a $\mathcal{C}$-connected diagram such that $type(\Omega'_L) = (a)$ and $type(\Omega'_R) = (*)$. Let $\wt{G}$ be a matrix diagram obtained by replacing a half edge of color $a$ in $G$ with $X$. (Here we think of edge between two nodes as a pair of half edges).

Then, $\wt{G}$ is $\mathcal{C}$-boundary-connected and has at most $d$\ $\mathcal{C}$-connected components.
\end{claim}
\begin{proof} Follows immediately from the definitions.
\end{proof}

Finally, note that for $X$ involved in $U_{GM}$, $V_{GM}$ or $W_{GM}$ its matrix diagram $\mathcal{MD}(X)$ has at most $4$ vertices and at least $3$ non-equality edges, unless $X = U_{E}$, $X = V_{E}$ or $X = W_{E}$. Moreover, if $X = V_{E}$ or $X = W_{E}$ the matrix diagram will also contain a non-equality $u$-edge coming from the diagram for $\tw(A_0^TA_0)$.

The proof for $\tw_2(A_0^TC')$ and $\tw_2((C')^TC')$ is similar, with the only difference that in the latter case two half edges are replaced with the appropriate diagrams.
\end{proof}

\begin{lemma}\label{lem:AUE-graph-matrix} Let $\mathcal{C}_{2/3} = \{\{u, v\}, \{u, w\}, \{v, w\}\}$ and $\mathcal{C}_{1/2} = \{\{v\}, \{w\}\}$, then
\[ \tw(A_0U_E(V\cten W)^T) \in \vspan\left(\mathfrak{BCM}^{3}(\mathcal{C}_{2/3}; 3, 2)\cup\mathfrak{BCM}^{2}(\mathcal{C}_{1/2}; 2, 1),\ 2\right). \]
\end{lemma}
\begin{proof} Note that 
\[ \tw_2(A_0U_E(V\cten W)^T) = \sum\limits_{i, j = 1}^{m}\sum\limits_{k\neq j}  \langle u_i, u_k\rangle \langle v_k, v_j \rangle \langle w_k, w_j  \rangle (v_i\otimes w_j)(v_j\otimes w_i).\] 
Hence, $\tw_2(A_0U_E(V\cten W)^T)$ has the left matrix diagram in Figure~\ref{fig:AUE-diagram}. Now, considering two cases in the sum above: $k=i$ and $k\neq i$, we see that $\tw_2(A_0U_E(V\cten W)^T)$ is a sum of the two IP graph matrices with the middle and right matrix diagrams in Figure~\ref{fig:AUE-diagram}. Note that the first one is $\mathcal{C}_{2/3}$-boundary-connected and has three non-equality edges and the second one is $\mathcal{C}_{1/2}$-boundary-connected and has two non-equality edges. Hence, the statement of the lemma follows. 
 \begin{figure}
\begin{subfigure}[b]{0.3\textwidth}
\begin{center}
\includegraphics[height = 3.3cm]{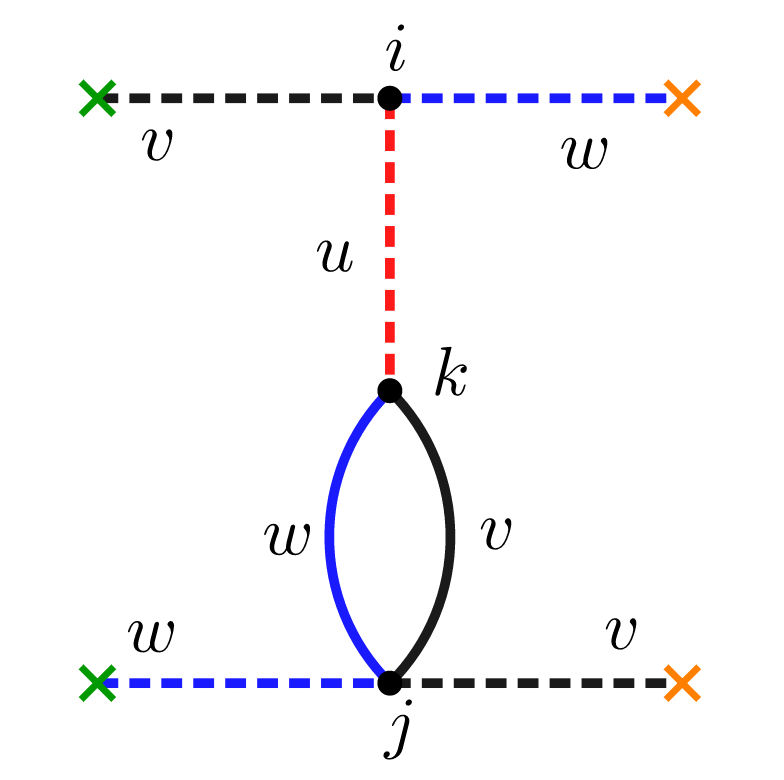}
\end{center}
\end{subfigure}
\begin{subfigure}[b]{0.3\textwidth}
\begin{center}
\includegraphics[height = 3.3cm]{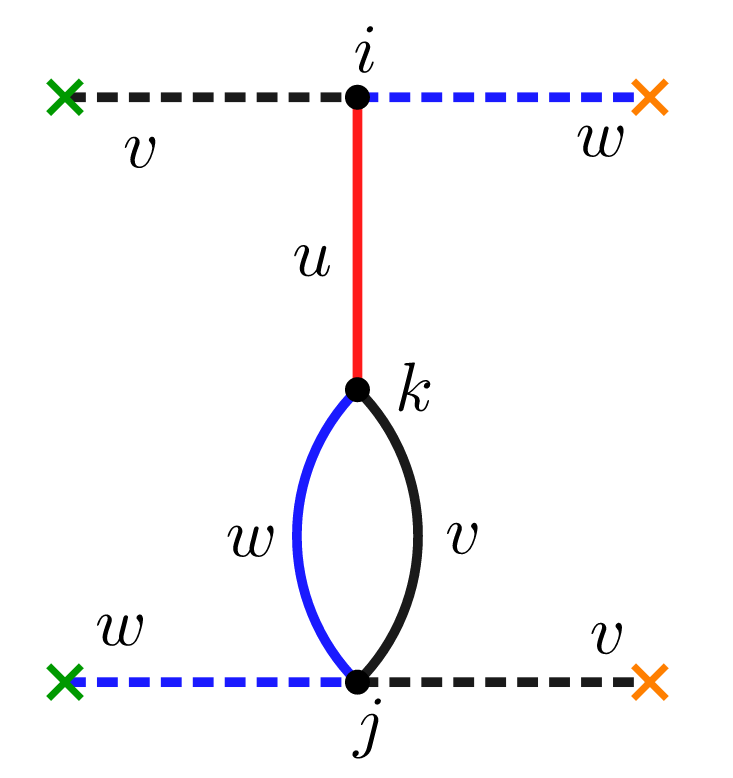}
\end{center}
\end{subfigure}
\begin{subfigure}[b]{0.3\textwidth}
\begin{center}
\includegraphics[height = 3.3cm]{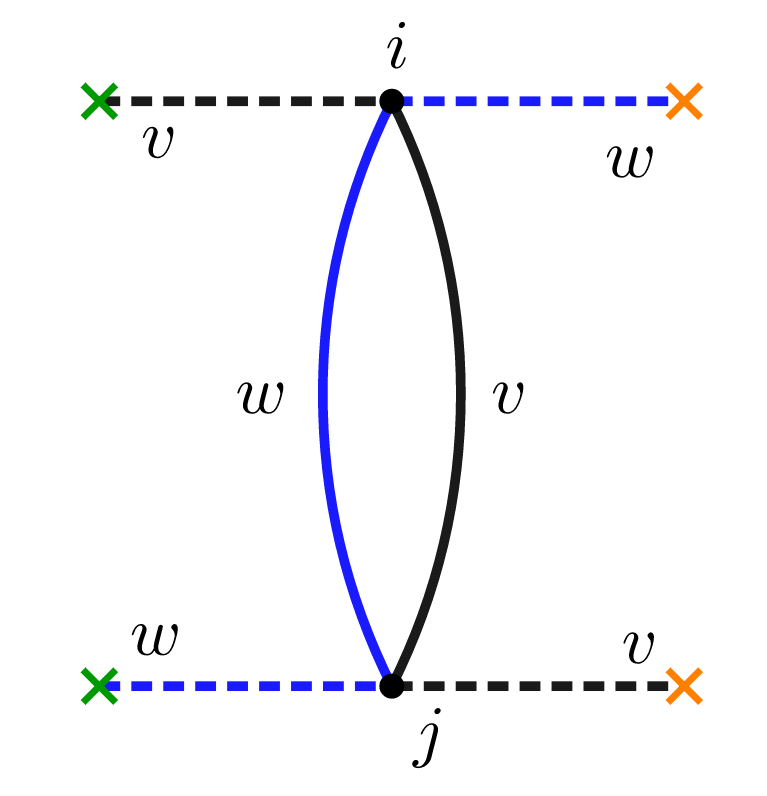}
\end{center}
\end{subfigure}
\caption{Matrix diagram for $(A_0U_E(V\copyright W)^T)$ and diagrams when  $i\neq k$ and $i=k$} \label{fig:AUE-diagram}
\end{figure}
\end{proof}

\begin{remark}
Alternatively, for a term of the form $A_0U_E(V\cten W)^T$ the norm bound can be shown using the Matrix Bernstein inequality, similarly as in the proof of Theorem~\ref{thm:A0-twist-bound}.
\end{remark}

Finally, we combine the above analysis to prove Theorem~\ref{thm:B0-construction}.

\begin{proof}[Proof of Theorem~\ref{thm:B0-construction}] 
We can rewrite
\begin{equation*}
\begin{split}
 A^TA &= \left(A_0+C'+C_{sm}\right)^T\left(A_0+C'+C_{sm}\right) = \\ &= (A_0+C')^T(A_0+C')+A^TC_{sm}+C_{sm}^TA - C_{sm}^TC_{sm} 
 \end{split}
 \end{equation*}
Using the trace power method (see Lemma~\ref{lem:trace-power-method-norm}), we deduce from Theorem~\ref{thm:main-diagram-tool} and Lemmas~\ref{lem:B0graph-terms},~\ref{lem:AUE-graph-matrix} that $\left\Vert (A_0+C')^T(A_0+C') -A_0^TA_0\right\Vert = \wt{O}\left(\dfrac{m}{n^{3/2}}\right)$. Hence, the claim of the theorem follows from Lemma~\ref{lem:Bvw-bound}, Theorem~\ref{thm:A0-twist-bound} and Proposition~\ref{prop:B0-small-frob-part}.
\end{proof}

\section{Construction of a zero polynomial matrix correction}\label{sec:zero-poly-corr}
For the last step, for $B_0 = \tw_2(A^TA)$,  we are looking for a symmetric zero polynomial matrix $Z_0$ of small norm that satisfies 
\begin{equation}\label{eq:Z-conditions}
 Z_0(v_i \otimes w_i) = (B_0-P_{\mathcal{L}})(v_i \otimes w_i) \quad \forall i \in [m].
 \end{equation}
Our analysis consists of two phases. In the first phase, we analyze the existence of a solution to Eq.~\eqref{eq:Z-conditions}. In the second phase, we prove norm bounds for $Z_0$.

\subsection{Representing given linear operator as a zero polynomial matrix}\label{sec:zero-poly-corr-exists}

In this section we show that with high probability there exists a symmetric zero polynomial matrix $Z$ which satisfies the constraints from Eq.~\eqref{eq:Z-conditions}. More precisely, we prove the following.
\begin{theorem}\label{thm:zero-poly-correction}
Assume $m\ll n^2$. Let $D\in M_{n^2}(\mathbb{R})$ be a symmetric matrix such that
\begin{equation}\label{eq:local-zero-poly-cond}
(x\otimes w_i)^T D (v_i\otimes w_i) = 0 \quad \text{and} \quad (v_i\otimes x)^T D (v_i\otimes w_i) = 0\quad \text{for all } x\in \mathbb{R}^n,\  i\in [m].
\end{equation}
Then w.h.p. over the randomness of $\mathcal{V}$, there exists a symmetric matrix $Z\in M_{n^2}(\mathbb{R})$ such that 
\[ Z \equiv_{poly} 0 \qquad \text{and} \qquad Z(v_i\otimes w_i) = D(v_i\otimes w_i) \quad \forall i\in [m]. \]
\end{theorem}
\begin{proof} The statement follows from its constructive version Theorem~\ref{thm:zero-poly-correction-constr}.
\end{proof}

Before proving the constructive version of this theorem let us verify that $B_0-P_{\mathcal{L}}$ satisfies the assumptions of the theorem.

\begin{lemma}\label{lem:ZB0-welldefined} Let $B_0 = \tw_2(A^TA)$. Then $B_0-P_{\mathcal{L}}$ is symmetric and
\[(x\otimes w_i)^T(B_0-P_{\mathcal{L}})(v_i\otimes w_i) = (v_i\otimes x)^T(B_0-P_{\mathcal{L}})(v_i\otimes w_i) = 0 \quad \text{for all } x\in \mathbb{R}^n,\  i\in [m].\]
\end{lemma}
\begin{proof} Clearly $P_{\mathcal{L}}$ is symmetric and it is not hard to see from Eq.~\eqref{eq:A-matrix-form} that $B_0$ is also symmetric. Note that Eq.~\eqref{eq:certificate-3-form}-\eqref{eq:orthog-w} imply
\[ (x\otimes w_i)^T\tw_2(A^TA)(v_i\otimes w_i) = (x\otimes w_i)^T A^TA (v_i\otimes w_i) = (x\otimes w_i)^T A^T u_i = \langle x, v_i \rangle.\] 
Similarly, $(v_i\otimes x)^T\tw_2(A^TA)(v_i\otimes w_i) = \langle x, w_i \rangle$. Hence, since $P_{\mathcal{L}}(v_i\otimes w_i) = (v_i\otimes w_i)$, the result follows.
\end{proof}
\subsubsection{System of linear equations for the desired $Z$}\label{sec:linear-system-for-Z}
As the first step of the proof of Theorem~\ref{thm:zero-poly-correction} we reformulate the conditions on $Z$ as a system of linear equations.

Let $\{e_{(t, s)}\}$ be an orthonormal basis of $\mathbb{R}^{n^2}$ for $t, s\in [n]$ and $\{f_i\}$ be an orthonormal basis of $\mathbb{R}^{m}$ for $i\in [m]$. 
Since $Z$ is symmetric and $Z \equiv_{poly} 0$, we can write
\begin{equation}\label{eq:zero-poly-basis}
Z = \sum\limits_{t<t'}\sum\limits_{s<s'}\left(e_{(t, s)}e_{(t', s')}^T - e_{(t', s)}e_{(t, s')}^T - e_{(t, s')}e_{(t', s)}^T + e_{(t', s')}e_{(t, s)}^T \right) z_{(t, s, t', s')},
\end{equation}
for some vector $z\in \mathbb{R}^{n^2(n-1)^2/4}$.

We introduce the matrix
\begin{equation}\label{eq:Q-definition}
\begin{split}
Q = \sum\limits_{i\in [m]}\sum\limits_{t<t'}\sum\limits_{s<s'} \left( (e_{(t, s)}\otimes f_i)\langle v_i\otimes w_i, e_{(t', s')}\rangle - (e_{(t', s)}\otimes f_i)\langle v_i\otimes w_i, e_{(t, s')}\rangle - \right. \\
\left. - (e_{(t, s')}\otimes f_i)\langle v_i\otimes w_i, e_{(t', s)}\rangle + (e_{(t', s')}\otimes f_i)\langle v_i\otimes w_i, e_{(t, s)}\rangle \right) e^T_{(t, s, t', s')},
\end{split}
\end{equation}
where $\{e_{(t, s, t', s')} \mid t<t',\ s<s',\ t,s, t', s'\in [n]\}$ is an orthonormal basis of $\mathbb{R}^{n^2(n-1)^2/4}$. 

Let $\wt{D} = \sum\limits_{i\in [m]} \wt{D}_i\otimes f_i$, where $\wt{D}_i = D(v_i\otimes w_i)$. Then $Z(v_i\otimes w_i) = D(v_i\otimes w_i)$ holds for all $i\in [m]$ if and only if $z$ given by Eq.~\eqref{eq:zero-poly-basis} satisfies 
\[ Qz = \wt{D}.\]
We look for $z$ of the form $z = Q^TY$ for $Y \in \mathbb{R}^{mn^2}$. Then $(QQ^T)Y = \wt{D}$.


\begin{lemma} Let $Q$ be defined as in Eq.~\eqref{eq:Q-definition}. Then

\[ QQ^T = \sum\limits_{i, j\in [m]} (QQ^T)_{ij}\otimes (f_if_j^T), \quad \text{where} \] 
\begin{equation}\label{eq:MMT-def}
 (QQ^T)_{ij} = \left(\langle v_i, v_j \rangle I_n - v_jv_i^T\right)\otimes \left( \langle w_i, w_j \rangle I_n - w_jw_i^T \right) 
\end{equation}
\end{lemma}
\begin{proof} Let $\sign(x) = \mathbf{1}[x = 0](-1)^{\mathbf{1}[x<0]}$ be a standard sign function. We observe that
\begin{equation*} 
\begin{split}
\left((QQ^T)_{ij}\right)_{(t_1, s_1)(t_2, s_2)} = & \langle v_i, e_{t_2} \rangle \langle w_i, e_{s_2} \rangle \langle v_j, e_{t_1} \rangle \langle w_j, e_{s_1} \rangle \left(\sign(t_2-t_1)\sign(s_2-s_1)\right)^2 = \\
 = & \left(\left(\langle v_i, v_j \rangle I_n - v_jv_i^T\right)\otimes \left( \langle w_i, w_j \rangle I_n - w_jw_i^T \right)\right)_{(t_1, s_1)(t_2, s_2)}.
 \end{split}
\end{equation*}
\end{proof}

Denote by $\mathcal{N}_0$ the null-space of $QQ^T$. Our goal is to show that $\wt{D}\in\mathcal{N}_0^{\perp}$. Then we will be able to find  $Y$ and $z$ as
\begin{equation}\label{eq:linear-system-for-Z}
 Y = (QQ^T)^{-1}_{N_0^{\perp}}\wt{D} \qquad \text{and} \qquad z = Q^T(QQ^T)^{-1}_{N_0^{\perp}}\wt{D}, 
\end{equation}

We will construct a space $\mathcal{N}\subseteq \mathcal{N}_0$ with $\wt{D}\perp \mathcal{N}$ and in Theorem~\ref{thm:QQT-approxim} we will show that $\left\Vert \left(QQ^T - I_{mn^2} - P_{\mathcal{L}_E}\right)_{\mathcal{N}^{\perp}} \right\Vert = \wt{O}\left(\dfrac{\sqrt{m}}{n}\right)$. This will be sufficient to show that Eq.~\eqref{eq:linear-system-for-Z} has a solution.

\subsubsection{Nullspace candidate}

In this subsection we show that there is a good candidate for a nullspace of $QQ^T$ and in the next subsection we will show that with high probability it is indeed  a nullspace.

We start with defining several linear spaces that will be important for the further analysis.

\begin{equation}
\begin{gathered} 
\mathcal{N}_{loc} = \vspan \{v_i\otimes x \otimes f_i,\  x\otimes w_i \otimes f_i \mid i\in [m], \, x\in \mathbb{R}^{n}\}, \\
 \mathcal{N}_S = \vspan \{ (v_i \otimes w_i \otimes f_j - v_j\otimes w_j \otimes f_i) \mid i, j\in [m]\},
\\
 \mathcal{L}_{E} = \vspan\{v_i \otimes w_i \otimes f_j \mid i, j \in [m]\},\qquad \mathcal{L}_{D} = \vspan\{v_i \otimes w_i \otimes f_i \mid i \in [m]\}.
\end{gathered}
\end{equation}

\begin{lemma}\label{lem:N0-inclusion}
We have the following inclusion into the nullspace $\mathcal{N}_0$ of $QQ^T$
\begin{equation}\label{eq:N0-def}
\mathcal{N} = \vspan\{\mathcal{N}_{loc},\  \mathcal{N}_S\} \subseteq \mathcal{N}_0.
\end{equation}
\end{lemma}
\begin{proof}
The claim follows as
\[ (QQ^T)_{ij}(v_j \otimes x) = (\langle v_i, v_j \rangle v_j - v_j \langle v_i, v_j \rangle)\otimes (\langle w_i, w_j \rangle x - w_j \langle w_i, x \rangle) = 0,\]
\[ \text{similarly} \qquad (QQ^T)_{ij}(x \otimes w_j) = 0, \quad \text{and}\]
\[ (QQ^T)_{kj}(v_i\otimes w_i) = (\langle v_k, v_j \rangle v_i - v_j \langle v_k, v_i \rangle) \otimes (\langle w_k, w_j \rangle w_i - w_j \langle w_k, w_i \rangle) = (QQ^T)_{ki}(v_j\otimes w_j).\]
\end{proof}

Additionally, note that $D$ is orthogonal to $\mathcal{N}$.

\begin{observation}\label{obs:NDorthogonal}
 $\wt{D}\perp \mathcal{N}$.
\end{observation}
\begin{proof}
By Eq.~\eqref{eq:local-zero-poly-cond}, $\wt{D}\perp \mathcal{N}_{loc}$ and since $D$ is symmetric, $\wt{D}\perp \mathcal{N}_S$.
\end{proof}

To show that with high probability $\mathcal{N} = \mathcal{N}_0$ we compute the dimension of $\mathcal{N}$ and later we will show that with high probability the nullspace of $QQ^T$ has the same dimension. Since $\mathcal{N}\subseteq \mathcal{N}_0$ this is sufficient.

\begin{lemma}\label{lem:N0dim}
Let $m\ll n^2$ and let $\mathcal{N}$ be given by Eq.~\eqref{eq:N0-def}. Then with high probability 
 \[ \dim(\mathcal{N}_{loc}) = m(2n-1), \quad \dim(\mathcal{N}_S) = m(m-1)/2, \]
 \[ \mathcal{N}_{loc} \cap \mathcal{L}_{E} = \mathcal{L}_{D}, \quad \mathcal{N}_S\subseteq \mathcal{L}_{E},\quad \mathcal{N}_S \cap \mathcal{L}_{D} = \{0\}.\] 
Hence, with high probability
\[\dim(\mathcal{N}) = m(2n-1)+m(m-1)/2.\]
\end{lemma}
\begin{proof} 
Clearly, $\dim(\mathcal{N}_{loc}) = m(2n-1)$, $\mathcal{N}_S\subseteq \mathcal{L}_{E}$ and $\mathcal{L}_{D}\subseteq \mathcal{L}_{E}$. With high probability the vectors $\{v_i\otimes w_i \mid i\in [m]\}$ are linearly independent, so $\dim(\mathcal{L}_E) = m^2$. Note that $\mathcal{L}_E$ is spanned by $m(m-1)/2$ generators of $\mathcal{N}_S$ and $m(m+1)/2$ vectors $(v_i \otimes w_i \otimes f_j + v_j\otimes w_j \otimes f_i)$. This implies that $\dim(\mathcal{N}_S) = m(m-1)/2$.  

Next, we show that $\mathcal{N}_S \cap \mathcal{L}_{D} = \{0\}$. Since $v_i \otimes w_i \otimes f_j$ for $i, j\in [m]$ form a basis of $\mathcal{L}_E$ we can define $\psi_{ii} : \mathcal{L}_{E} \rightarrow \mathbb{R}$ to be a coefficient near $v_i\otimes w_i\otimes f_i$, when the input vector is written in that basis. Clearly, $\psi_{ii}$ restricted on $\mathcal{N}_{S}$ is zero for any $i$, while all $\psi_{ii}$ have value $0$ only for the zero vector in $\mathcal{L}_D$. Thus, $\mathcal{N}_S \cap \mathcal{L}_{D} = \{0\}$.

Finally, we argue that $\mathcal{N}_{loc} \cap \mathcal{L}_{E} = \mathcal{L}_{D}$. For this, it is sufficient to show that 
\[ v_i \otimes x + y \otimes w_i = \sum\limits_{j = 1}^{m} \mu_j v_j \otimes w_j \]
can hold for $\mu_j \in \mathbb{R}$, $x\perp w_i$ and $y\perp v_i$ only if $x = y = 0$. Multiply both sides of equation by $v_i \otimes x$, then by H\"older's inequality with weights $(1/2, 1/4, 1/4)$
\begin{equation}\label{eq:xynormbound}
 \Vert x \Vert^2 = \sum\limits_{j:\, j\neq i} \mu_j\langle v_i, v_j \rangle\langle x, w_j \rangle \leq \left(\sum\limits_{j:\, j\neq i} \mu_j^2 \right)^{1/2} \left (\sum\limits_{j:\, j\neq i} \langle v_i, v_j \rangle^4\right)^{1/4}\left(\sum\limits_{j:\, j\neq i}\langle x, w_j \rangle^4 \right)^{1/4}  
 \end{equation}
Since w.h.p. $\langle v_i, v_j \rangle = \wt{O}(1/\sqrt{n})$, using Theorem~\ref{thm:deg4weakbound}, for $m \ll n^2$ we get that
\[ \left(\sum\limits_{j:\, j\neq i} \langle v_i, v_j \rangle^4\right)^{1/4} = \wt{O}\left(\dfrac{m^{1/4}}{n^{1/2}}\right)\qquad \text{and} \qquad \left(\sum\limits_{j:\, j\neq i}\langle x, w_j \rangle^4 \right)^{1/4} = O\left(1\right)\Vert x \Vert.\]
Additionally, since $(v_i\otimes x) \perp (y\otimes w_i)$ by Lemma~\ref{lem:basis-in-L}, $\Vert \mu \Vert = \left(1+\wt{O}\left(\dfrac{\sqrt{m}}{n}\right)\right)\left(\Vert x \Vert +\Vert y \Vert\right)$.  Thus, Eq.~\eqref{eq:xynormbound} implies,
\[ \Vert x \Vert = \wt{O}\left(\dfrac{m^{1/4}}{n^{1/2}}\right)\left(\Vert x \Vert +\Vert y \Vert\right). \]
Since, an analogous argument shows the same bound for $y$, by adding them together we obtain $\Vert x \Vert +\Vert y \Vert = 0$ for $m \ll n^2$. Therefore, $\Vert x \Vert = \Vert y \Vert  = 0$.

Hence, we deduce that $\mathcal{N}_S \cap \mathcal{N}_{loc} = \{0\}$, and so the claim is proved.
\end{proof}

\subsubsection{Approximating $QQ^T$ with a simpler matrix}

Our goal now is to approximate $QQ^T$ with a simpler matrix, for which it is easier to analyze its eigenvalues. Denote,
\begin{equation}
F_v = \sum\limits_{i, j\in [m]} (F_v)_{ij}\otimes f_i f_j^T, \quad F_w = \sum\limits_{i, j\in [m]} (F_w)_{ij} \otimes f_i f_j^T, \quad \text{and} \quad R = \sum\limits_{i, j\in [m]} R_{ij} \otimes f_i f_j^T, 
\end{equation}
where,  
\begin{equation}
(F_v)_{ij} = \langle w_i, w_j \rangle v_j v_i^T, \quad (F_w)_{ij} = \langle v_i, v_j \rangle w_j w_i^T, \quad \text{and} \quad R_{ij} = (v_j\otimes w_j)(v_i \otimes w_i)^T. 
\end{equation}
Then, Eq.~\eqref{eq:MMT-def} is equivalent to
\begin{equation}\label{eq:MMT-def-equiv}
(QQ^T)_{ij} = \langle v_i, v_j \rangle\langle w_i, w_j \rangle I_{n^2} - (F_v)_{ij} \otimes I_n-I_n\otimes (F_w)_{ij}+R_{ij}.
\end{equation}
We are going to approximate the matrix generated by each of these four terms separately.
\begin{lemma}\label{lem:MMT-core-term-approx}
If $m\ll n^2$, then w.h.p. $\displaystyle \left\Vert \sum\limits_{i, j\in[m]}\langle v_i, v_j \rangle\langle w_i, w_j \rangle I_{n^2}\otimes f_i f_j^T - I_{mn^2} \right\Vert = \wt{O}\left(\dfrac{\sqrt{m}}{n}\right)$.
\end{lemma}
\begin{proof}
Observe that
 $$\sum\limits_{i, j\in[m]}\langle v_i, v_j \rangle\langle w_i, w_j \rangle I_{n^2}\otimes f_i f_j^T = I_{n^2}\otimes \left((V\cten W)^T(V\cten W)\right).$$
  By Lemma~\ref{lem:basic-cten-bound}, $\left\Vert (V\cten W)^T(V\cten W) - I_m \right\Vert = \wt{O}\left(\dfrac{\sqrt{m}}{n}\right)$, so the claim of the lemma follows.
\end{proof}
Consider 
\begin{equation}
 F_v^D = \sum\limits_{i\in[m]} (v_iv_i^T)\otimes (f_if_i^T), \quad \text{and} \quad  F_w^D = \sum\limits_{i\in[m]} (w_iw_i^T)\otimes (f_if_i^T). 
 \end{equation}

\begin{lemma}\label{lem:MMT-F-approx}
 Assume $m\ll n^2$. Then $\left\Vert F_v - F_v^D \right\Vert = \wt{O}\left(\dfrac{\sqrt{m}}{n}\right)$ and $\left\Vert F_w - F_w^D \right\Vert = \wt{O}\left(\dfrac{\sqrt{m}}{n}\right)$.
\end{lemma}
\begin{proof}
Define $F'_v = F_v - F_v^D$. Then
\[ F'_{v} = \sum\limits_{i=1}^{m}\sum\limits_{j:\, j\neq i} \langle w_i, w_j\rangle v_jv_i^T,\]
and it has a matrix diagram and a trace diagram given by Figure~\ref{fig:Fv-diagram}.

\begin{figure}
\begin{center}
\begin{subfigure}[b]{0.22\textwidth}
\begin{center}
\includegraphics[height = 3.2cm]{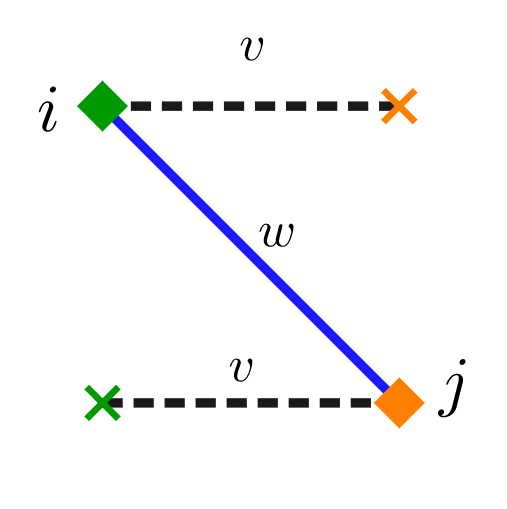}
\end{center}
\end{subfigure}
\begin{subfigure}[b]{0.75\textwidth}
\begin{center}
\includegraphics[height = 3.2cm]{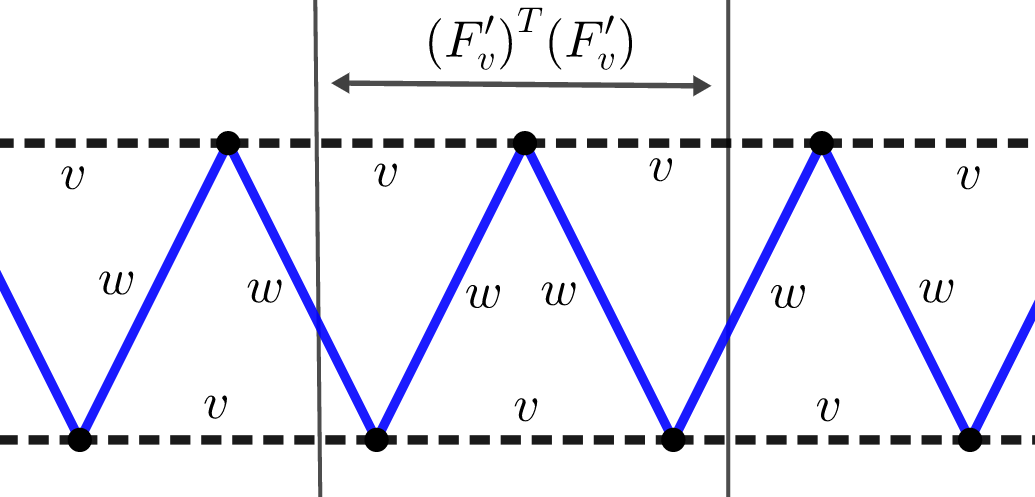}
\end{center}
\end{subfigure}
\caption{The matrix diagram for $F_v'$ and the trace diagram for $((F_v')^TF_v')^q$.}\label{fig:Fv-diagram}
\end{center}
\end{figure}

Using independence, we can write
\[ \mathbb{E}\Tr\left(((F'_v)^TF'_v)^q\right) = \mathbb{E}\sum\limits_{\phi\in \Phi} val(\mathcal{TD}_q(F'_v, \phi)) = \sum\limits_{\phi\in \Phi} \mathbb{E}\left(\prod\limits_{e\in E_w} \term_{\phi}(e)\right)\mathbb{E}\left(\prod\limits_{e\in E_{v}} \term_{\phi}(e)\right).\]
We apply Theorem~\ref{thm:diagram-bound-tool} to the graphs $G_{\{w\}, \phi}$ and $G_{\{v\}, \phi}$ induced on labels of $\phi$ by the edges $E_w$ and $E_{v}$ of colors $w$ and $\{v\}$, respectively (with loops being deleted). We obtain
\[ \left\vert\mathbb{E}\left(\prod\limits_{e\in E_w} \term_{\phi}(e)\right) \right\vert\leq \wt{O}\left(\dfrac{1}{n}\right)^{\max(n_{\phi}-1, q)}\quad  \text{and} \quad \left\vert\mathbb{E}\left(\prod\limits_{e\in E_{v}} \term_{\phi}(e)\right)\right\vert \leq \wt{O}\left(\dfrac{1}{n}\right)^{n_{\phi}-2},\]
where $n_{\phi}$ is the size of the image of $\phi$. Note that there are at most $m^{n_{\phi}}n_{\phi}^{2q}$ labelings of $\mathcal{TD}_q(F'_v)$ using only $n_{\phi}$ labels from $[m]$. Therefore,
\[ \left\vert\mathbb{E}\Tr\left(((F'_v)^TF'_v)^q\right)\right\vert \leq \sum\limits_{j=1}^{2q} m^{j}(2q)^{2q} \wt{O}\left(\dfrac{1}{n}\right)^{\max(2j-3, q+j-2)}\]
Since, $m\ll n^2$, the expression under the sum sign is maximized for $j = q+1$, so
\[ \left\vert \mathbb{E}\Tr ((F'_v)^TF'_v)^q \right\vert \leq m (2q)^{2q+1}\left(\wt{O}\left(\dfrac{m}{n^2}\right)\right)^{q}\]
Taking $q = O(\log(n)^2 )$, the power-trace method (Lemma~\ref{lem:trace-power-method-norm}) implies
\[ \Vert F'_v \Vert  = \wt{O}\left(\dfrac{\sqrt{m}}{n}\right) \]
A similar argument works for $F_w$.
\end{proof}


\begin{theorem}\label{thm:R-approx-QQ} Let ${P}_{\mathcal{N}_{S}}$ and $P_{\mathcal{N}_{S}^{\perp\mathcal{L}}}$ be projectors on $\mathcal{N}_S$ and on its orthogonal complement inside the space $\mathcal{L}_E$. Then, for $m\ll n^2$,
\[\left\Vert R - \left({P}_{\mathcal{N}_{S}^{\perp\mathcal{L}}} - {P}_{\mathcal{N}_{S}}\right) \right\Vert = \wt{O}\left(\dfrac{\sqrt{m}}{n}\right).\]
\end{theorem}
\begin{proof}
Any vector $x$ in $\mathcal{L}_E$ can be uniquely written as 
\begin{equation}\label{eq:x-in-basis-L}
x = \sum\limits_{i=1}^{m}\sum\limits_{j=1}^m \psi_{ij} (v_i\otimes w_i \otimes f_j).
\end{equation}
Denote by $\Psi(x)$ an $m\times m$ matrix $(\psi_{ij})_{i, j\in[m]}$ and by $\psi(x)$ the vector of length $m^2$ composed of $\psi_{ij}$. 
Observe that Eq.~\eqref{eq:x-in-basis-L} can be also written as
\begin{equation}
 x = ((V\cten W)\otimes I_m)\psi(x). 
\end{equation}
Recall that by Lemma~\ref{lem:basic-cten-bound}, $\left\Vert(V\cten W)^T(V\cten W) - I_m\right\Vert = \wt{O}\left(\dfrac{\sqrt{m}}{n}\right)$. 
Using that
\[ \left((V\cten W)^T\otimes I_m\right) x = \left(\left((V\cten W)^T(V\cten W)\right)\otimes I_m \right)\psi(x).\]
 we get
\begin{equation}\label{eq:xphi-relation}
\left(1-\wt{O}\left(\dfrac{\sqrt{m}}{n}\right)\right)\Vert x \Vert \leq \Vert \psi(x) \Vert \leq  \left(1+\wt{O}\left(\dfrac{\sqrt{m}}{n}\right)\right)\Vert x\Vert.
\end{equation}
As a ``complement" to $\mathcal{N}_S$, define
\begin{equation}
\mathcal{N}_{A} = \vspan\{(v_i\otimes w_i\otimes f_j + v_j\otimes w_j \otimes f_i)\mid i, j\in [m]\}.
\end{equation}
Now, note that for $x\in \mathcal{N}_{S}$ we have $\psi_{kj} = -\psi_{jk}$ and for $\psi\in \mathcal{N}_{A}$ we have $\psi_{kj} = \psi_{jk}$. 
Define the following two maps 
\[ \wt{P}_{\mathcal{N}_S} (x) = \sum\limits_{i=1}^{m}\sum\limits_{j=1}^m \dfrac{\psi_{ij}-\psi_{ji}}{2}(v_i\otimes w_i \otimes f_j)  \quad \text{and} \]
\[ \wt{P}_{\mathcal{N}_{A}} (x) = \sum\limits_{i=1}^{m}\sum\limits_{j=1}^m \dfrac{\psi_{ij}+\psi_{ji}}{2}(v_i\otimes w_i \otimes f_j)\]
where $x$ has an expansion given by Eq.~\eqref{eq:x-in-basis-L}. Then 
\[ \wt{P}_{\mathcal{N}_S} : \mathcal{L}_E \rightarrow \mathcal{N}_S,\qquad \wt{P}_{\mathcal{N}_{A}} : \mathcal{L}_E \rightarrow \mathcal{N}_{A} \quad \text{and} \quad  \wt{P}_{\mathcal{N}_S}+\wt{P}_{\mathcal{N}_{A}} = I_{\mathcal{L}_E}. \]
For $x$ given by Eq.~\eqref{eq:x-in-basis-L}, we compute
\begin{equation}
 Rx =  \sum\limits_{i=1}^{m}\sum\limits_{j=1}^{m}\sum\limits_{k=1}^{m} \psi_{ki} \langle v_j, v_k \rangle \langle w_j, w_k \rangle (v_i \otimes w_i \otimes f_j).
\end{equation}
Hence, for $x\in \mathcal{N}_{S}$,
\[ Rx +x =  \sum\limits_{i=1}^{m}\sum\limits_{j=1}^{m}\sum\limits_{k: k\neq j}^{m} (-\psi_{ik}) \langle v_j, v_k \rangle \langle w_j, w_k \rangle (v_i \otimes w_i \otimes f_j) =: \sum\limits_{i=1}^{m}\sum\limits_{j=1}^{m}\psi'_{ij}  (v_i \otimes w_i \otimes f_j).\]
Let $\psi'$ and $\Psi'$ be the vector and the matrix consisting of $\psi'_{ij}$. Observe that 
\[ \Psi' = -\Psi(x)\left((V\cten W)^T(V\cten W) -I_m\right), \quad \text{so}\]
\[ \Vert \psi'\Vert = \Vert \Psi' \Vert_F \leq \Vert \Psi(x) \Vert_F \cdot \Vert (V\cten W)^T(V\cten W) -I_m \Vert  = \Vert \psi(x)\Vert\cdot \wt{O}\left(\dfrac{\sqrt{m}}{n}\right).\]
Thus, using Eq.~\eqref{eq:xphi-relation}, for $x\in \mathcal{N}_{S}$,
\[\Vert Rx+x \Vert = \Vert ((V\cten W)\otimes I_m)\psi' \Vert  = \Vert \psi(x)\Vert\cdot \wt{O}\left(\dfrac{\sqrt{m}}{n}\right) = \Vert x\Vert\cdot \wt{O}\left(\dfrac{\sqrt{m}}{n}\right) . \]
A similar argument shows that for $x\in \mathcal{N}_{A}$,
\[\Vert Rx-x \Vert = \Vert x\Vert\cdot \wt{O}\left(\dfrac{\sqrt{m}}{n}\right) . \]

Using Eq.~\eqref{eq:xphi-relation}, for any $x\in \mathcal{L}_E$,
\begin{equation}\label{eq:PNS-norm-bounds}
 \max\left(\Vert \psi(\wt{P}_{\mathcal{N}_{S}} x) \Vert,\  \Vert \psi(\wt{P}_{\mathcal{N}_{A}} x) \Vert\right) \leq \left(1+\wt{O}\left(\dfrac{\sqrt{m}}{n}\right) \right)\Vert x\Vert.
 \end{equation} 
Thus, extending the definition domain of $\wt{P}_{\mathcal{N}_{S}}$ and $\wt{P}_{\mathcal{N}_{S}}$ to $\mathbb{R}^{n^2}$ (by setting them to be 0 on the orthogonal complement to $\mathcal{L}^E$, we can rewrite the bounds above as
\[ \Vert R\wt{P}_{\mathcal{N}_S} + \wt{P}_{\mathcal{N}_S} \Vert = \wt{O}\left(\dfrac{\sqrt{m}}{n}\right)\quad \text{and} \quad  \Vert R\wt{P}_{\mathcal{N}_{A}} - \wt{P}_{\mathcal{N}_{A}} \Vert = \wt{O}\left(\dfrac{\sqrt{m}}{n}\right).\]
Therefore, since $R\left(\wt{P}_{\mathcal{N}_S}+ \wt{P}_{\mathcal{N}_E}\right) = RP_{\mathcal{L}_E} = R$, we get
\begin{equation}\label{eq:R-approx-first}
 \left\Vert R - \left(\wt{P}_{\mathcal{N}_{A}} - \wt{P}_{\mathcal{N}_{S}}\right) \right\Vert = \wt{O}\left(\dfrac{\sqrt{m}}{n}\right).
 \end{equation}

Let $\mathcal{N}_{S}^{\perp\mathcal{L}}\subseteq \mathcal{L}_{A}$ be an orthogonal complement to $\mathcal{N}_{S}$ in $\mathcal{L}_{E}$. Next, we will show that $\wt{P}_{\mathcal{N}_{A}}$ is close to $P_{\mathcal{N}_{S}^{\perp\mathcal{L}}}$ and $\wt{P}_{\mathcal{N}_{S}}$ is close to $P_{\mathcal{N}_{S}}$. 

Observe that for $x\in \mathcal{N}_{S}$ and $y\in \mathcal{N}_{A}$ we have $\psi(x) \perp \psi(y)$, so by Lemma~\ref{lem:basic-cten-bound},
\begin{equation}
\begin{gathered}
 \langle x, y \rangle  = \psi(x)^T((V\cten W)^T(V\cten W)\otimes I_m)) \psi(y) = \\
 = \psi(x)^T\left(\left((V\cten W)^T(V\cten W)-I_m\right)\otimes I_m\right) \psi(y) = \Vert \psi(x)\Vert \cdot \Vert \psi(y) \Vert \wt{O}\left(\dfrac{\sqrt{m}}{n}\right).
 \end{gathered}
\end{equation}
Hence, using Eq.~\eqref{eq:PNS-norm-bounds} and the identity $x = \wt{P}_{\mathcal{N}_{A}} x+\wt{P}_{\mathcal{N}_S}x = {P}_{\mathcal{N}_S}x+ \left({P}_{\mathcal{N}_S^{\perp\mathcal{L}}}\right)x$, for $x\in \mathcal{L}_E$,
\begin{equation*}
\begin{gathered}
 \left\Vert \wt{P}_{\mathcal{N}_{A}} x - \left({P}_{\mathcal{N}_S^{\perp\mathcal{L}}}\right)x \right\Vert^2 = \left\langle (x-\wt{P}_{\mathcal{N}_{S}} x) - (x-{P}_{\mathcal{N}_S}x), \wt{P}_{\mathcal{N}_{A}} x - \left({P}_{\mathcal{N}_S^{\perp\mathcal{L}}}\right)x \right\rangle = \\
 = \left\langle {P}_{\mathcal{N}_S}x - \wt{P}_{\mathcal{N}_{S}} x , \wt{P}_{\mathcal{N}_{A}} x - \left({P}_{\mathcal{N}_S^{\perp\mathcal{L}}}\right)x \right\rangle \leq \Vert x\Vert \cdot  \left\Vert \wt{P}_{\mathcal{N}_{A}} x - \left({P}_{\mathcal{N}_S^{\perp\mathcal{L}}}\right)x \right\Vert \wt{O}\left(\dfrac{\sqrt{m}}{n}\right).
 \end{gathered} 
 \end{equation*}
This immediately implies that 
\[ \left\Vert \wt{P}_{\mathcal{N}_{S}}  - {P}_{\mathcal{N}_S} \right\Vert  = \left\Vert \wt{P}_{\mathcal{N}_{A}}  - {P}_{\mathcal{N}_S^{\perp\mathcal{L}}} \right\Vert =  \wt{O}\left(\dfrac{\sqrt{m}}{n}\right).\]
 Therefore, using Eq.~\eqref{eq:R-approx-first}, we obtain 
 \begin{equation*}
 \left\Vert R - \left({P}_{\mathcal{N}_{S}^{\perp\mathcal{L}}} - {P}_{\mathcal{N}_{S}}\right) \right\Vert = \wt{O}\left(\dfrac{\sqrt{m}}{n}\right).
 \end{equation*}
\end{proof}

Combining these approximations together we can finally deduce the following.

\begin{theorem}\label{thm:QQT-approxim} For $m\ll n^2$, w.h.p. $\mathcal{N} = \mathcal{N}_0$ and we have the following approximation to $QQ^T$ on $\mathcal{N}^{\perp}$
\[ \left\Vert \left(QQ^T - I_{mn^2} - P_{\mathcal{L}_E}\right)_{\mathcal{N}^{\perp}} \right\Vert = \wt{O}\left(\dfrac{\sqrt{m}}{n}\right).\]
\end{theorem}
\begin{proof}
Lemmas~\ref{lem:MMT-core-term-approx} and \ref{lem:MMT-F-approx} combined with Eq~\eqref{eq:MMT-def-equiv} imply 
\begin{equation}\label{eq:MMT-approx-H}
\left\Vert(QQ^T-R) - H\right\Vert = \wt{O}\left(\dfrac{\sqrt{m}}{n}\right), \quad \text{where}
\end{equation}
\begin{equation}
H = I_{mn^2} - \sum\limits_{i \in [m]}\left(v_iv_i^T \otimes I_n+I_n\otimes w_iw_i^T\right)\otimes f_if_i^T
\end{equation}
Note that $H$ has $(-1)$-eigenspace $\mathcal{H}_{-1} = \mathcal{L}_D$, has $0$-eigenspace $\mathcal{H}_0$, the complement to $\mathcal{L}_D$ inside $\mathcal{N}_{loc}$,  and has 1-eigenspace $\mathcal{H}_{1}$ being the complement to $\mathcal{N}_{loc}$. The matrix $R$ is zero on the complement to $\mathcal{L}_E$, and as was shown in Lemma~\ref{lem:N0dim}, $\mathcal{L}_E\cap \mathcal{N}_{loc} = \mathcal{L}_D$. Therefore,
$\mathcal{H}_0$ is contained in the 0-eigenspace of $H+R$ and the orthogonal complement to $\vspan(\mathcal{L}_E, \mathcal{N}_{loc})$ is contained in $1$-eigenspace of $H+R$. Note that $QQ^T$ is $0$ on $\mathcal{L}_D$. 

Thus, we need to understand $(H+R)$ on the orthogonal complement to $\mathcal{L}_{D}$ in $\mathcal{L}_{E}$. Denote this subspace $\mathcal{L}_{ND}$. Since, $H$ is the identity on $\mathcal{L}_{ND}$, by Theorem~\ref{thm:R-approx-QQ}, 
\[\left\Vert (H+R)_{\mathcal{L}_{ND}} - I_{\mathcal{L}_{ND}}-\left({P}_{\mathcal{N}_{S}^{\perp\mathcal{L}}} - {P}_{\mathcal{N}_{S}}\right)_{\mathcal{L_{ND}}}\right\Vert = \wt{O}\left(\dfrac{\sqrt{m}}{n}\right).\]
Thus $(H+R)_{\mathcal{L_{ND}}}$ has $m(m-1)/2$ eigenvalues in the interval $(-\varepsilon, \varepsilon)$ and $m(m-1)/2$ eigenvalues in the interval $(2-\varepsilon, 2+\varepsilon)$ for some $0<\varepsilon = \wt{O}\left(\dfrac{\sqrt{m}}{n}\right)$. 

Hence, by Eq.~\eqref{eq:MMT-approx-H}, $QQ^T$ has at most $\dim(\mathcal{N}_{loc})+m(m-1)/2 = m(2n-1)+m(m-1)/2$ eigenvalues in the interval $(-\varepsilon, \varepsilon)$ for some $0<\varepsilon = \wt{O}\left({\sqrt{m}}/{n}\right)$. Hence, 
\[\dim(\mathcal{N}_0)\leq m(2n-1)+m(m-1)/2.\]
Thus, by Lemma~\ref{lem:N0-inclusion} and Lemma~\ref{lem:N0dim}, $\mathcal{N}_0 = \mathcal{N}$.

Finally, using that $H_{\mathcal{N}^{\perp}} = I_{\mathcal{N}^{\perp}}$ and using Theorem~\ref{thm:R-approx-QQ}, we deduce from Eq.~\eqref{eq:MMT-approx-H} that
\[ \left\Vert (QQ^T)_{\mathcal{N}^{\perp}} - I_{\mathcal{N}^{\perp}} - \left(P_{\mathcal{N}_S^{\perp\mathcal{L}}}\right)_{\mathcal{N}^{\perp}} \right\Vert = \wt{O}\left(\dfrac{\sqrt{m}}{n}\right).\] 
\end{proof}

\subsubsection{Collecting pieces together}

\begin{theorem}\label{thm:zero-poly-correction-constr} Suppose that assumptions of Theorem~\ref{thm:zero-poly-correction} hold. Let $Q$ be given by Eq.~\eqref{eq:Q-definition}, and
\[ \wt{Y} =  (QQ^T)^{-1}_{\mathcal{N}^{\perp}}\wt{D},\quad  \text{where}\quad \wt{D} = \sum\limits_{i=1}^{m} (D(v_i\otimes w_i))\otimes f_i,\] 
and $\mathcal{N}$ is a nullspace of $QQ^T$. Consider the decomposition $\wt{Y} = \sum\limits_{i=1}^{m} Y_i\otimes f_i$.  Define 
\[\mathcal{Z}(D) = X+X^T-\tw_2(X)-\tw_2(X)^T,\quad  \text{where}\quad  X = \sum\limits_{i=1}^{m} Y_i(v_i\otimes w_i)^T, \]
Then w.h.p. $Z = \mathcal{Z}(D)$ satisfies the conditions of Theorem~\ref{thm:zero-poly-correction}.
\end{theorem}
\begin{proof} As explained in Section~\ref{sec:linear-system-for-Z}, The desired matrix $Z$ can be written as
\[ Z = \sum\limits_{t<t'}\sum\limits_{s<s'}\left(e_{(t, s)}e_{(t', s')}^T - e_{(t', s)}e_{(t, s')}^T - e_{(t, s')}e_{(t', s)}^T + e_{(t', s')}e_{(t, s)}^T \right) z_{(t, s, t', s')},
\]
for some vector $z\in \mathbb{R}^{n^2(n-1)^2/4}$. Moreover, $Z$ satisfies the conditions of Theorem~\ref{thm:zero-poly-correction} if and only if $z$ is a solution to the equation $Qz = \wt{D}$. By Observation~\ref{obs:NDorthogonal}, $\wt{D}\perp \mathcal{N}$, and $\mathcal{N}$ is a nullspace of $QQ^T$ by Theorem~\ref{thm:QQT-approxim}. Hence, w.h.p. $\wt{Y} =  (QQ^T)^{-1}_{\mathcal{N}^{\perp}}\wt{D}$ is well-defined and taking $z = Q^T\wt{Y}$ we get a solution to the equation $Qz = \wt{D}$. Finally, note that for so defined $z$,
\begin{equation}
\begin{split}
 z_{(t, s, t', s')} = \left(Q^T\wt{Y}\right)_{(t, s, t', s')} =  \sum\limits_{i\in [m]} &\left[ (Y_i)_{(t, s)}(v_i\otimes w_i)_{(t', s')} - (Y_i)_{(t', s)}(v_i\otimes w_i)_{(t, s')} -\right. \\
  &  - \left.(Y_i)_{(t, s')}(v_i\otimes w_i)_{(t', s)} + (Y_i)_{(t', s')}(v_i\otimes w_i)_{(t, s)}\right].
\end{split}
 \end{equation}
 Hence, substituting this into the equation for $Z$, we deduce the statement of the theorem. 
\end{proof}

\subsection{Norm bound for $Z_0$}\label{sec:zero-poly-corr-norm}

The last step is to show a norm bound for the $Z$ constructed in Theorem~\ref{thm:zero-poly-correction-constr}. Recall that

\[ Z_0 = X+X^T-\tw_2(X)-\tw_2(X)^T,\quad \text{where} \quad  X = \sum\limits_{i=1}^{m} Y_i(v_i\otimes w_i)^T.\]

\subsubsection{Approximating $(QQ^T)^{-1}_{\mathcal{N}^{\perp}}$ with IP graph matrices}
Note that the definition of $X$ in Theorem~\ref{thm:zero-poly-correction-constr} involves the matrix $(QQ^T)^{-1}_{\mathcal{N}^{\perp}}$. Thus, we start our analysis by approximating $(QQ^T)^{-1}_{\mathcal{N}^{\perp}}$ with a linear combination of IP graph matrices. 

From Theorem~\ref{thm:QQT-approxim} we know that 
\begin{equation}\label{eq:QQT-formula-1}
 \left(QQ^T\right)_{\mathcal{N}^{\perp}}  =  \left(I_{mn^2} + P_{\mathcal{L}_E}\right)_{\mathcal{N}^{\perp}}+\mathcal{E}_Q\quad \text{where} \quad \left\Vert  \mathcal{E}_Q \right\Vert = \wt{O}\left(\dfrac{\sqrt{m}}{n}\right).
 \end{equation}
Note that $\left(I_{mn^2}+P_{\mathcal{L}_E}\right)^{-1} = I_{mn^2}-\dfrac{1}{2}P_{\mathcal{L}_E}$, hence 
\begin{equation}\label{eq:QQTinverse-approx-1}
 \left\Vert \left(QQ^T|_{\mathcal{N}^{\perp}}\right)^{-1} -  \left(I_{mn^2} - \dfrac{P_{\mathcal{L}_E}}{2}\right)_{\mathcal{N}^{\perp}}\sum\limits_{j=0}^{t}   \left(\mathcal{E}_Q \left(I_{mn^2} - \dfrac{P_{\mathcal{L}_E}}{2}\right)\right)^j_{\mathcal{N}^{\perp}} \right\Vert = \wt{O}\left(\dfrac{\sqrt{m}}{n}\right)^{t+1}.
\end{equation}
For our purposes it is sufficient to have an approximation of $\left(QQ^T|_{\mathcal{N}^{\perp}}\right)^{-1}$ up to a term with norm $\wt{O}\left(\dfrac{m^2}{n^4}\right)$. Therefore in the expression above it is sufficient to consider $t=3$.

Using Eq.~\eqref{eq:QQT-formula-1} we express $\mathcal{E}_Q$, as 
\begin{equation}\label{eq:EQ-definition}
\mathcal{E}_Q = \left(QQ^T - I_{mn^2}-P_{\mathcal{L}_E}\right)_{\mathcal{N}^{\perp}}.
\end{equation}  
Additionally, note that $QQ^T$ maps ${\mathcal{N}^{\perp}}$ to ${\mathcal{N}^{\perp}}$, thus if input vector is in ${\mathcal{N}^{\perp}}$, we can replace in the formula all $QQ^T\vert_{\mathcal{N}^{\perp}}$ with $QQ^T$.  We can do the same for $P_{\mathcal{L}_E}$ and $I_{mn^2}$. Hence, after substituting this expression into Eq.~\eqref{eq:QQTinverse-approx-1}, the projector $P_{\mathcal{L}_E}$ is the only ingredient which does not have an explicit formula. However, we know that
\[ P_{\mathcal{L}_E} = P_{\mathcal{L}}\otimes I_{m}, \quad \text{and} \quad \Vert P_{\mathcal{L}} - B_{vw} \Vert = \wt{O}\left(\dfrac{\sqrt{m}}{n}\right), \]
where $B_{vw} = (V\cten W)(V\cten W)^T$. Moreover, $B_{vw}$ maps $\mathcal{L}$ to $\mathcal{L}$ and $\mathcal{L}^{\perp}$ to $0$, hence for $t>0$
\[ \left\Vert P_{\mathcal{L}} - \left(\left(I_{n^2} - B_{vw}\right)^{2t} - I_{n^2}\right) \right\Vert = \wt{O}\left(\dfrac{m^t}{n^{2t}}\right). \]
Thus, first substituting $\mathcal{E}_Q$ with the expression from Eq.~\eqref{eq:EQ-definition} and then using the approximation to $P_{\mathcal{L}_E}$ given in terms of $B_{vw}$ we obtain the following statement.
\begin{lemma}\label{lem:QQT-approx}
For $m\ll n^2$, w.h.p. there exists an IP graph matrix $Q_{inv}^{[2t]}$ such that 
\begin{equation}\label{eq:Happrox-def}
\begin{gathered}
\left\Vert \left(Q_{inv}^{[2t]}\right)_{\mathcal{N}^{\perp}} - \left((QQ^T)_{\mathcal{N}^{\perp}}\right)^{-1} \right\Vert  = \wt{O}\left(\dfrac{m^t}{n^{2t}}\right), \quad \text{where} \\
Q_{inv}^{[2t]} = poly\left( QQ^T,\, B_{vw}\otimes I_{m} \right), \quad \text{with} \quad \deg(poly) \leq 8t^2.
\end{gathered}
\end{equation} 
\end{lemma} 
For the purposes of this section the approximation $Q_{inv}^{[4]}$ will be sufficient. 

%
\subsubsection{Analysis for small terms}\label{sec:Z-small-terms}
As in Section~\ref{sec:B0} (proof of Theorem~\ref{thm:B0-construction}), we can write
\[B_0 = \tw_2\left((A_0+C')^T(A_0+C')+A^TC_{sm}+C_{sm}^TA - C_{sm}^TC_{sm}\right).\]
Denote $B_{GM} = \tw_2\left((A_0+C')^T(A_0+C')\right)-P_{\mathcal{L}}$. By Proposition~\ref{prop:B0-small-frob-part}, we know that for 
\[ B_{sm} = (B_0 -P_{\mathcal{L}}) - B_{GM} \quad \text{we have} \quad  \Vert B_{sm} \Vert_F = \wt{O}\left(\dfrac{m^2}{n^3}+\dfrac{m^4}{n^6}\right).\]
Consider $\wt{D} = \wt{D}_{sm}+\wt{D}_{GM}$, where
\begin{equation}\label{eq:DGM-def}
 \wt{D}_{sm} = \sum\limits_{i}\left(B_{sm}(v_i\otimes w_i)\right)\otimes f_i \quad \text{and} \quad \wt{D}_{GM} = \sum\limits_{i}\left(B_{GM}(v_i\otimes w_i)\right)\otimes f_i.
 \end{equation}
By Lemma~\ref{lem:A-frob-norm-place-indep}, we can bound
\[ \Vert\wt{D}_{sm}\Vert = \left\Vert\sum\limits_{i}\left(B_{sm}(v_i\otimes w_i)\right)\otimes f_i \right\Vert = \Vert B_{sm}(V\cten W) \Vert_F 
 = \wt{O}\left(\dfrac{m^2}{n^3}+\dfrac{m^4}{n^6}\right).\]
Let $H = Q_{inv}^{[4]}$ be a matrix from Eq.~\eqref{eq:Happrox-def} and let $\mathcal{E}_H = \left(QQ^T\right)_{\mathcal{N}^{\perp}}^{-1} - H_{\mathcal{N}^{\perp}}$. By Observation~\ref{obs:NDorthogonal}, $\wt{D}\in \mathcal{N}^{\perp}$, so 
\[ \wt{Y} = (QQ^T)^{-1}_{\mathcal{N}^{\perp}}\wt{D} = \left(H_{\mathcal{N}^{\perp}}+\mathcal{E}_H\right)\wt{D} = H\wt{D} + \mathcal{E}_H\wt{D}  = H\wt{D}_{GM}+ H\wt{D}_{sm}+ \mathcal{E}_H\wt{D}. \]
Therefore,
\[ \left\Vert\wt{Y} - H\wt{D}_{GM}\right\Vert\leq \Vert H\Vert \Vert \wt{D}_{sm}\Vert+\Vert \mathcal{E}_H\Vert \Vert \wt{D}\Vert = \wt{O}\left(\dfrac{m^2}{n^3}+\dfrac{m^4}{n^6}\right).\] 
Reshaping $\wt{Y} - H\wt{D}_{GM}$ and $H\wt{D}_{GM}$ into $n^2\times m$ matrices $Y_{sm}$ and $Y_{GM}$, we get
\begin{equation}\label{eq:Xsmall-definition}
 \Vert X - Y_{GM}(V\cten W) \Vert_F = \Vert Y_{sm}(V\cten W) \Vert_F \leq \Vert Y_{sm}\Vert_F\Vert V\cten W\Vert  = \wt{O}\left(\dfrac{m^2}{n^3}+\dfrac{m^4}{n^6}\right),
 \end{equation}
where we use $\Vert Y_{sm}\Vert_F = \left\Vert\wt{Y} - H\wt{D}_{GM}\right\Vert$ and Lemma~\ref{lem:basic-norm-bounds}.

\begin{proposition}\label{prop:Zsm-bound} Assume $m\ll n^{3/2}$. Let $X_{sm} = X - Y_{GM}(V\cten W)$ (see Eq.~\eqref{eq:Xsmall-definition}). Define
\[  Z_{sm} = X_{sm}+X_{sm}^T-\tw_2(X_{sm})-\tw_2(X_{sm}^T).\]
Then $\Vert Z_{sm} \Vert_F = \wt{O}\left(\dfrac{m^2}{n^{3}}\right)$.
\end{proposition}
\begin{proof} By Eq.~\eqref{eq:Xsmall-definition}, $\Vert X_{sm}\Vert_{F} = \wt{O}\left(\dfrac{m^2}{n^{3}}\right)$, since $\Vert X_{sm} \Vert_F =  \Vert \tw_2(X_{sm}) \Vert_F$, the bound for $Z_{sm}$ holds.
\end{proof}
Therefore, we get that $Z_0$ is approximated well by the IP graph matrix
\[ Z_{GM} = \mathcal{Z}_{GM}(B_0-P_{\mathcal{L}}) = X_{GM}+X_{GM}^T+\tw_2(X_{GM})+\tw_2(X_{GM}^T),\]
where $X_{GM} = Y_{GM}(V\cten W)$.

\subsubsection{Analysis for the essential IP graph matrices involved in $Z_0$}

In this subsection we analyze the contribution of $H \wt{D}_{GM}$ to $Z_0$.

\begin{definition}
We say that an inner product graph matrix $\wt{S}$ with matrix diagram $G = (\Omega, E, \mathfrak{c})$ belongs to the class $(v, w, *)\text{-}\mathfrak{G}$ if \begin{itemize} 
\item $type(\Omega_R) = ()$ and $type(\Omega_L) = (v, w, *)$, and
\item $S = \sum\limits_{i=1}^{m} S_i f_i^T$ is in the class $\mathfrak{G}$, where $\wt{S} = \sum\limits_{i=1}^{m} S_i\otimes f_i$.
\end{itemize}
\end{definition}

\begin{observation}\label{obs:boundary-con-impl-con} Assume that $S\in (v, w, *)\text{-}\mathfrak{BCM}^{s}(\mathcal{C}; e, d)$, then $S\in (v, w, *)\text{-}\mathfrak{CM}^{s}(\mathcal{C}; e)$.
\end{observation}
\begin{proof}
By the assumption, $S$ can be written as $S = \sum\limits_{i, j, k = 1}^{m} s_{kji}(v_k\otimes w_j\otimes f_i)$ and its matrix diagram can be schematically drawn as on Figure~\ref{fig:QQT-B-transform} (a). Moreover, by assumption the matrix diagram for $S^* = \sum\limits_{i, j, k = 1}^{m} s_{kji}(v_k\otimes w_j)f_i^T$ is $\mathcal{C}$-boundary connected. In particular this means that for any node $x$ of $\mathcal{MD}(S)$ and any $C\in \mathcal{C}$ there exists a $C$-path to $i$ (since $\Omega_R = \{i\}$ for $S^*$). Therefore, any two nodes in $\mathcal{MD}(S)$ can be connected by a $C$-path. 
\end{proof}

\begin{lemma}\label{lem:QQT-matrix-transform}
Consider $\mathcal{C}_{2/3} = \{\{u, v\}, \{v, w\}, \{u, w\}\}$, $\mathcal{C}_v =\{\{v\},\{u, w\}\}$ and $\mathcal{C}_w =\{\{w\},\{u, v\}\}$. Let $\mathcal{C}\in \{\mathcal{C}_{2/3}, \mathcal{C}_{v}, \mathcal{C}_{w}\}$. If $S$ is from the class $(v, w, *)\text{-}\mathfrak{CM}^{s}(\mathcal{C}; e)$, then 
\begin{enumerate}
\item $\left(B_{vw}\otimes I_m\right)S$ is from the class $(v, w, *)\text{-}\mathfrak{CM}^{s}(\mathcal{C}; e)$, and 
\item $QQ^TS\in \vspan\left((v, w, *)\text{-}\mathfrak{CM}^{s}(\mathcal{C}; e), 4\right)$. 
\end{enumerate}
\end{lemma}
\begin{proof}
Since $S$ is in $(v, w, *)\text{-}\mathfrak{CM}^{s}(\mathcal{C}; e)$ we can write $S = \sum\limits_{i, j, k = 1}^{m} s_{kji}(v_k\otimes w_j\otimes f_i)$. We schematically draw the matrix diagram for $S$ as on Figure~\ref{fig:QQT-B-transform} (a).
\begin{figure}
\begin{subfigure}[b]{0.3\textwidth}
\begin{center}
\includegraphics[width = 0.9\textwidth]{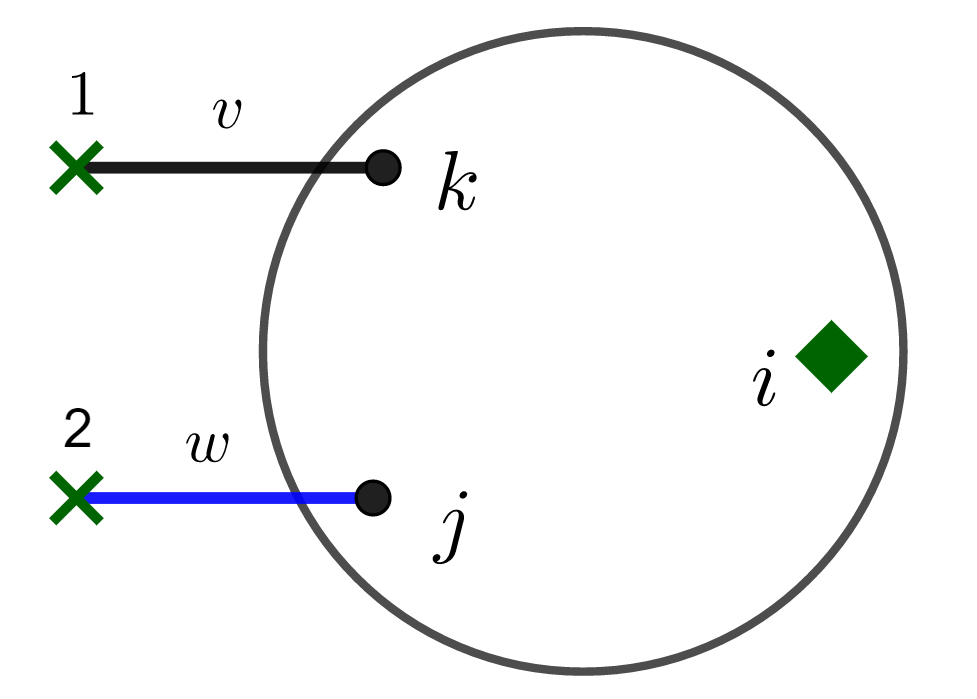}
\caption{Diagram for $S$}
\end{center}
\end{subfigure}
\begin{subfigure}[b]{0.3\textwidth}
\begin{center}
\includegraphics[width = 0.9\textwidth]{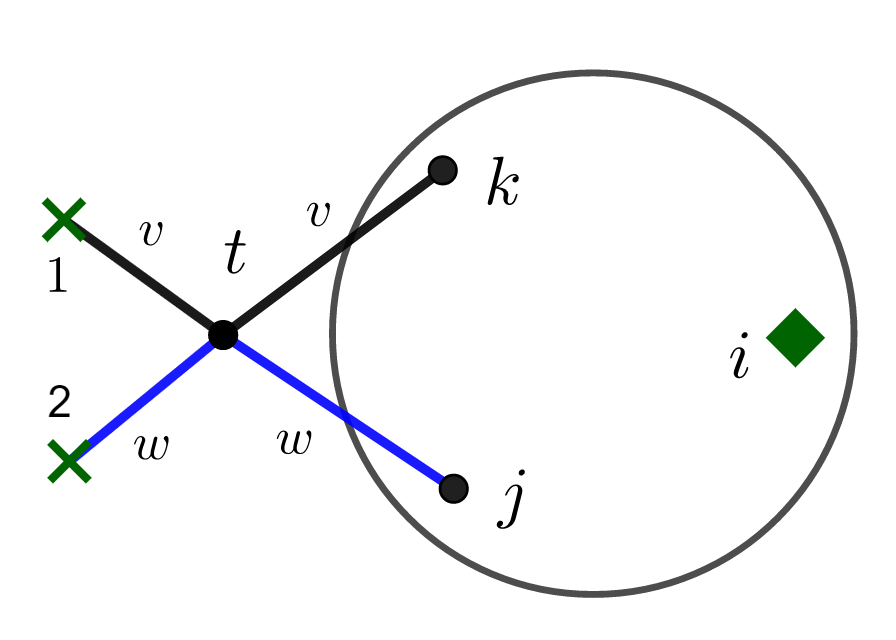}
\caption{Diagram for $\left(B_{vw}\otimes I_m\right)S$}
\end{center}
\end{subfigure}
\begin{subfigure}[b]{0.3\textwidth}
\begin{center}
\includegraphics[width = 0.9\textwidth]{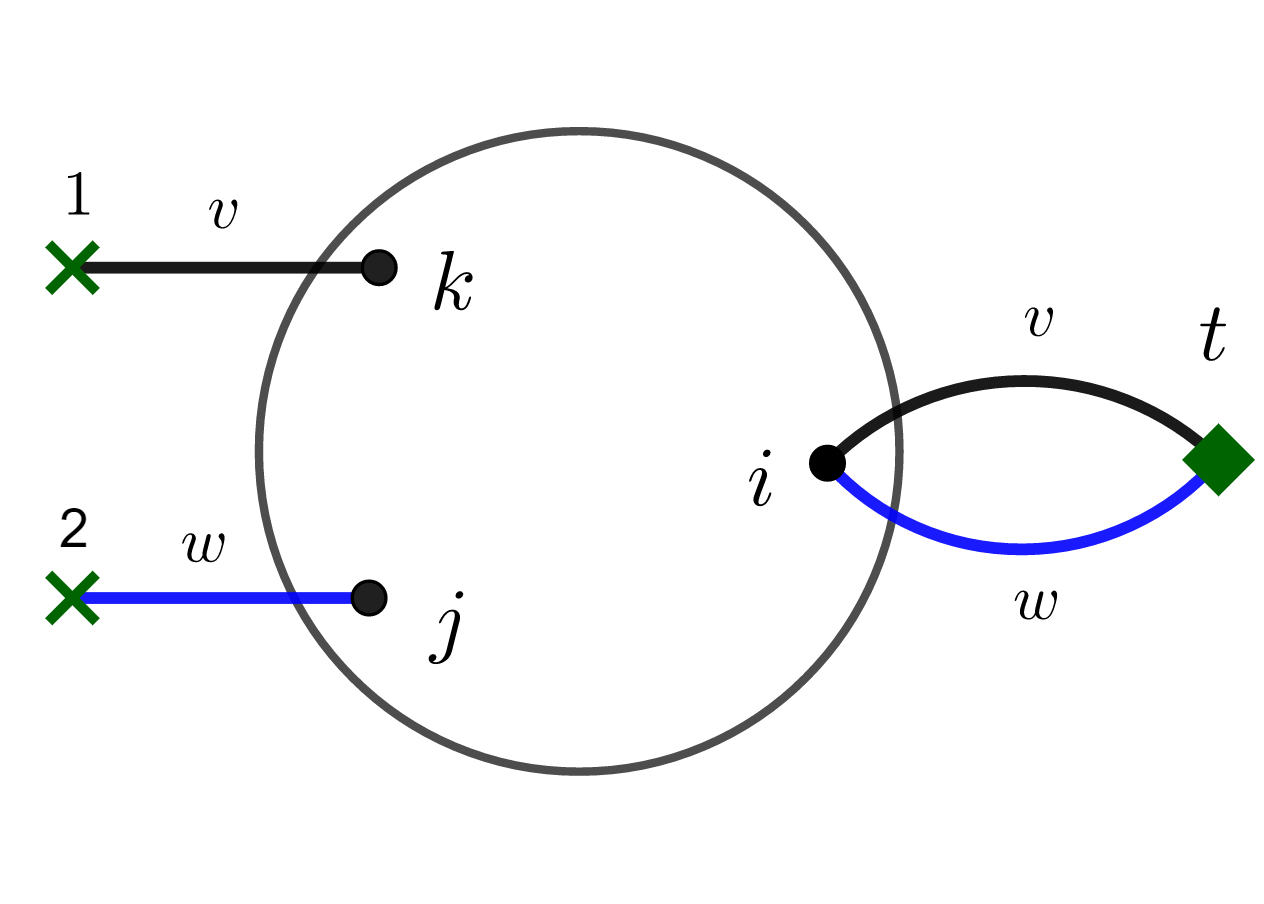}
\caption{Diagram for $(QQ^TS)_{[1]}$}
\end{center}
\end{subfigure}

\begin{subfigure}[b]{0.3\textwidth}
\begin{center}
\includegraphics[width = 0.9\textwidth]{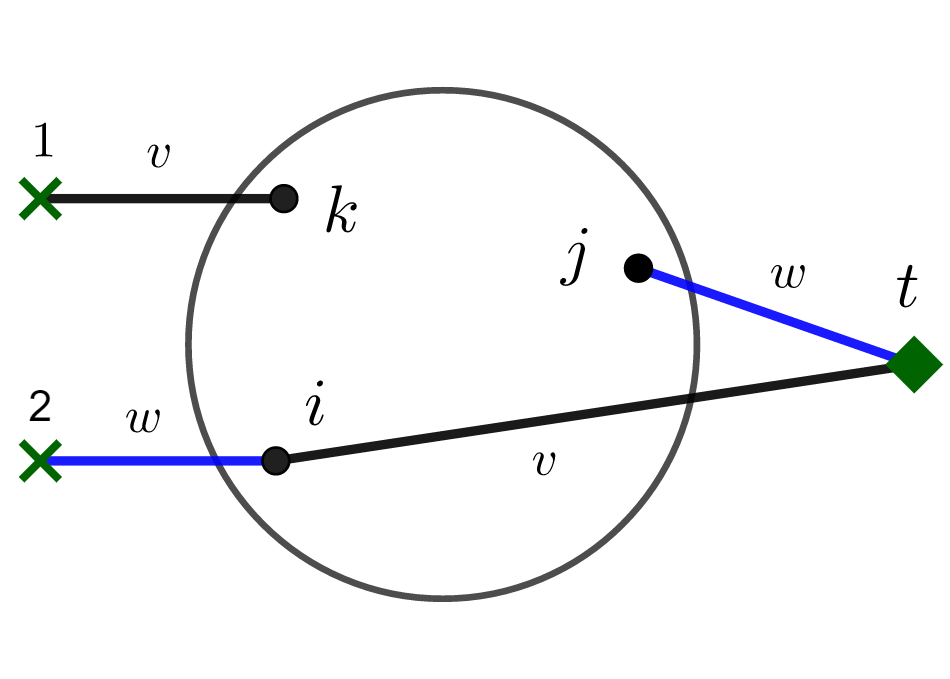}
\caption{Diagram for $(QQ^TS)_{[2]}$}
\end{center}
\end{subfigure}
\begin{subfigure}[b]{0.3\textwidth}
\begin{center}
\includegraphics[width = 0.9\textwidth]{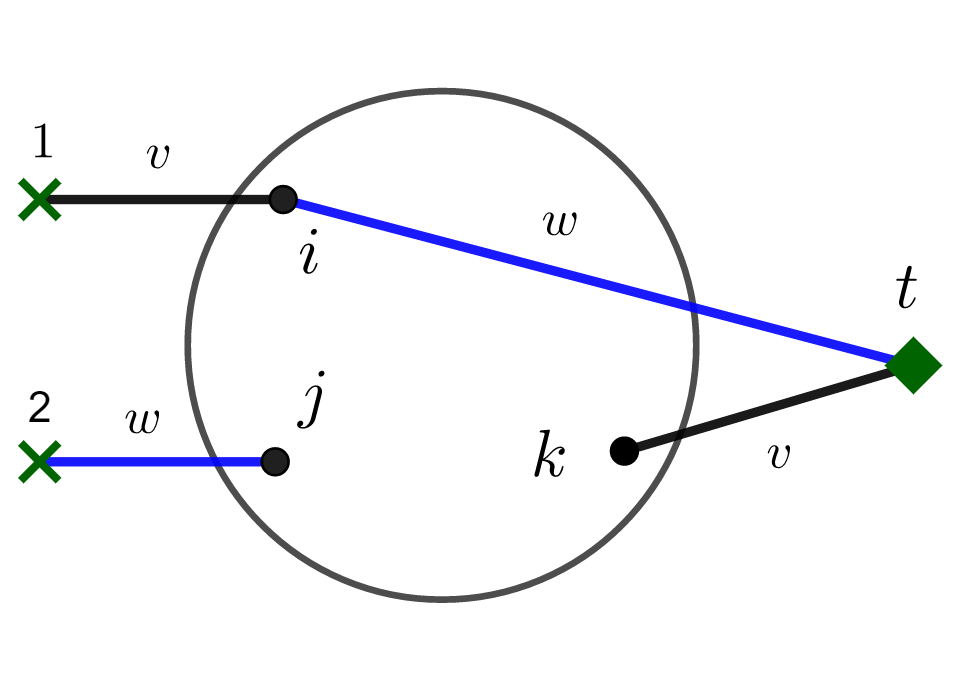}
\caption{Diagram for $(QQ^TS)_{[3]}$}
\end{center}
\end{subfigure}
\begin{subfigure}[b]{0.3\textwidth}
\begin{center}
\includegraphics[width = 0.9\textwidth]{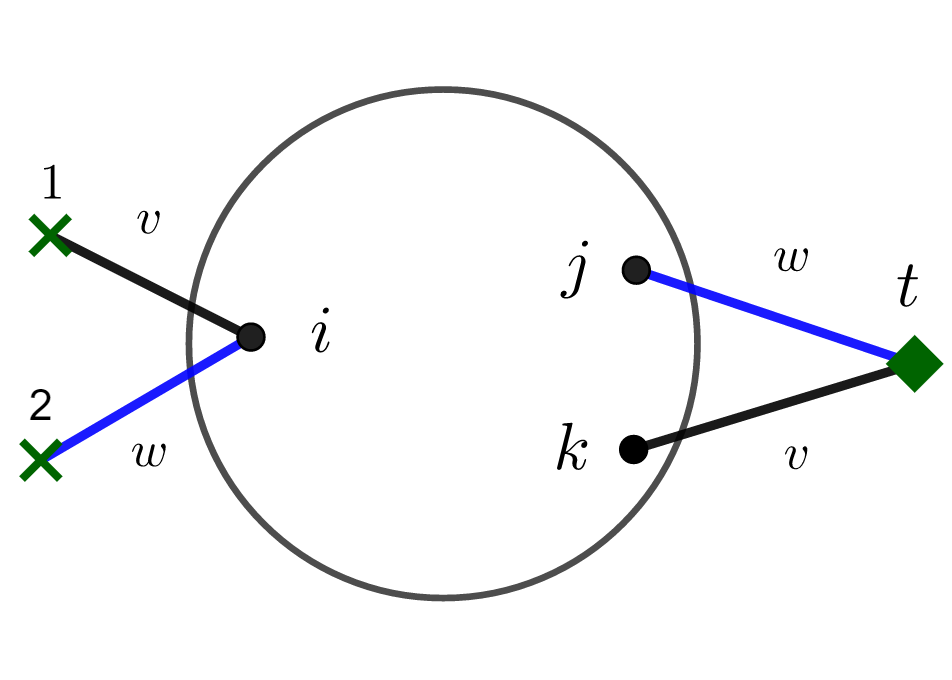}
\caption{Diagram for $(QQ^TS)_{[4]}$}
\end{center}
\end{subfigure}
\caption{Diagrams illustrating multiplication by $QQ^T$ and $B_{vw}\otimes I_m$.}\label{fig:QQT-B-transform}
\end{figure}

Compute
\[ \left(B_{vw}\otimes I_{m} \right)S = \left(\left((V\cten W)(V\cten W)^T\right)\otimes I_{m} \right)S = \sum\limits_{t=1}^{m}\sum\limits_{i, j, k = 1}^{m} s_{kji}\langle v_t, v_k \rangle \langle w_t,  w_j\rangle(v_t\otimes w_t\otimes f_i) \]
and 
\[QQ^TS = \sum\limits_{t, i, j, k} \left(s_{kji}\left(\langle v_t, v_i \rangle I_n - v_iv_t^T\right)\otimes \left( \langle w_t, w_i \rangle I_n - w_iw_t^T \right)(v_k\otimes w_j)\right)\otimes f_{t} = \]
\[ = \sum\limits_{t, i, j, k} s_{kji} \langle v_t, v_i \rangle \langle w_t, w_i \rangle (v_k\otimes w_j \otimes f_t)- \sum\limits_{t, i, j, k} s_{kji} \langle v_t, v_i \rangle \langle w_t, w_j \rangle (v_k\otimes w_i \otimes f_t) - \]
\[ - \sum\limits_{t, i, j, k} s_{kji} \langle v_t, v_k \rangle \langle w_t, w_i \rangle (v_i\otimes w_j \otimes f_t)+ \sum\limits_{t, i, j, k} s_{kji} \langle v_t, v_k \rangle \langle w_t, w_j \rangle (v_i\otimes w_i \otimes f_t)\]
Hence, $\left(B_{vw}\otimes I_m\right)S$ has matrix diagram as on Figure~\ref{fig:QQT-B-transform} (b) and $QQ^TS$ is a signed sum of four  IP graph matrices with diagrams on Figure~\ref{fig:QQT-B-transform} (c)-(f). The ``hidden" part of $X$, depicted with a circle, remains unchanged in all these diagrams. Thus, in particular, each diagram on Figure~\ref{fig:QQT-B-transform} has at least $e$ non-equality edges, and at most $N+1$ vertices. 


Moreover, it is easy to see that if $S$ is $\mathcal{C}$-connected, then each diagram on Figure~\ref{fig:QQT-B-transform} is $\mathcal{C}$-connected. Indeed, the only new added node is $t$ and it is connected by edges of color $v$ and $w$ to some nodes present in $S$. Since every $C\in \mathcal{C}$ contains either $v$ or $w$, $t$ is connected to the rest of the nodes (which form a connected graph by the assumption). 
 
\end{proof}

\begin{theorem}\label{thm:ZGM-bound} Assume $m\ll n^{3/2}$. Let $H = Q_{inv}^{[4]}$ and $\wt{D}_{GM}$ be the matrices given by Eq.~\eqref{eq:Happrox-def} and~\eqref{eq:DGM-def}. Consider the decomposition $H\wt{D}_{GM} = \sum\limits_{i=1}^{m} (Y_{GM})_i\otimes f_i$ for $(Y_{GM})_i\in \mathbb{R}^{n^2}$. Let
\[ X_{GM} = \sum\limits_{i=1}^{m} (Y_{GM})_i(v_i\otimes w_i)^T \quad \text{and} \quad Z_{GM} = X_{GM}+X_{GM}^T-\tw_2(X_{GM})-\tw_2(X_{GM}^T).\]
Then $\Vert Z_{GM} \Vert = \wt{O}\left(\dfrac{m}{n^{3/2}}+\dfrac{\sqrt{m}}{n}\right)$.

\end{theorem}
\begin{proof} Consider $\mathcal{C}_{2/3} = \{\{u, v\}, \{v, w\}, \{w, u\}\}$, $\mathcal{C}_{v} = \{\{v\}, \{u, w\}\}$ and $\mathcal{C}_{w} = \{\{w\}, \{u, v\}\}$. 

Recall, that the matrix $Y_{GM}$ is defined using the matrix $B_{GM} = \tw_2\left((A_0+C')^T(A_0+C')\right)-P_{\mathcal{L}}$ as described in Section~\ref{sec:Z-small-terms}. As before, we decompose 
\[B_{GM} = \tw_2\left((A_0+C')^T(A_0+C') - A_0^TA_0\right)+ \left(\tw_2\left( A_0^TA_0\right) - P_{\mathcal{L}}\right),  \] 
and we analyze each part separately.

First, by Lemma~\ref{lem:B0graph-terms} and Lemma~\ref{lem:AUE-graph-matrix},
\[ \tw_2\left((A_0+C')^T(A_0+C') - A_0^TA_0\right) \in \vspan\left(\mathfrak{BCM}^{10}(\mathcal{C}_{2/3}; 3, 2)\cup\mathfrak{BCM}^{2}(\mathcal{C}_{v}; 2, 1), 10^3+3 \right).\]
Let $R$ be some IP graph matrix that participates in the linear combination for this part of $B_{GM}$. For $\mathcal{MD}(R)$ we have $type(\Omega_L) = type(\Omega_R) = (v, w)$. Hence, 
\[R = \sum\limits_{j_1, j_2, k_1, k_2} r_{j_1 j_2 k_1 k_2} (v_{j_1}\otimes w_{j_2})(v_{k_1}\otimes w_{k_2})^T\]
and we can schematically depict $R$ as on Figure~\ref{fig:BGMdiagram}(a).

 Then $S = \sum\limits_{i=1}^{m} R(v_i\otimes w_i) f_i^T$ is an IP graph matrix with matrix diagram on Figure~\ref{fig:BGMdiagram} (b). Since  $R$ is in $\mathfrak{BCM}^{10}(\mathcal{C}_{2/3}; 3, 2)$ or $\mathfrak{BCM}^{2}(\mathcal{C}_{v}; 2, 1)$, then  
 $S$ is in $\mathfrak{BCM}^{10}(\mathcal{C}_{2/3}; 3, 2)$ or $\mathfrak{BCM}^{2}(\mathcal{C}_{v}; 2, 1)$. 
 \begin{figure}
\begin{subfigure}[b]{0.24\textwidth}
\begin{center}
\includegraphics[width = 0.9\textwidth]{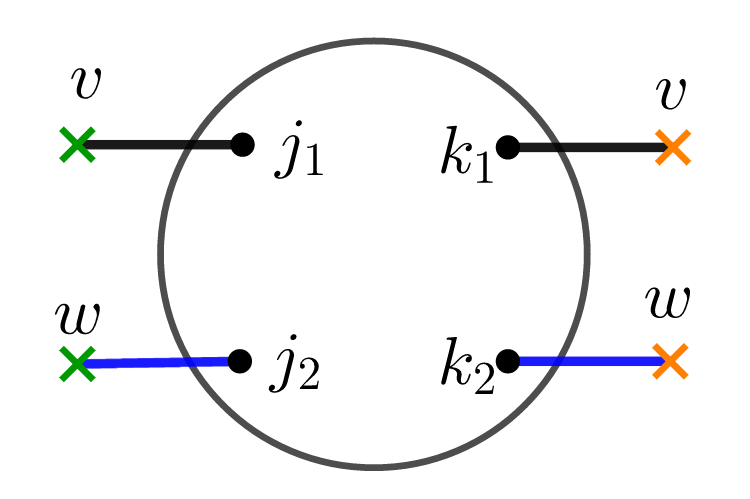}
\caption{Diagram for $R$.}
\end{center}
\end{subfigure}
\begin{subfigure}[b]{0.24\textwidth}
\begin{center}
\includegraphics[width = 0.9\textwidth]{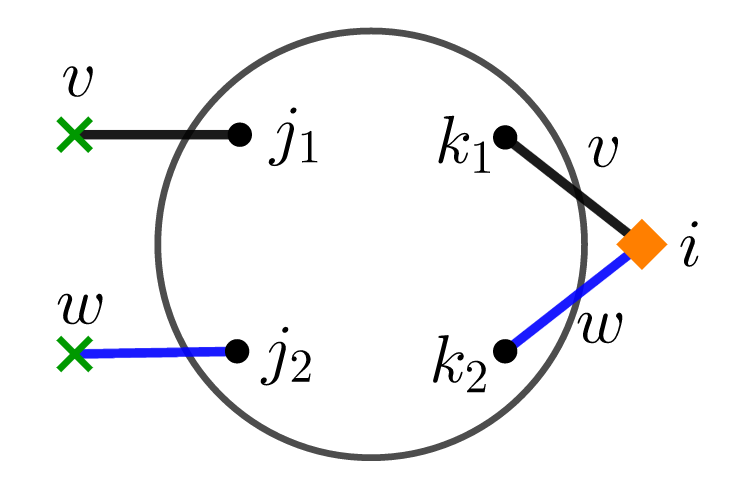}
\caption{Diagram for $S$.}
\end{center}
\end{subfigure}
\begin{subfigure}[b]{0.24\textwidth}
\begin{center}
\includegraphics[width = 0.9\textwidth]{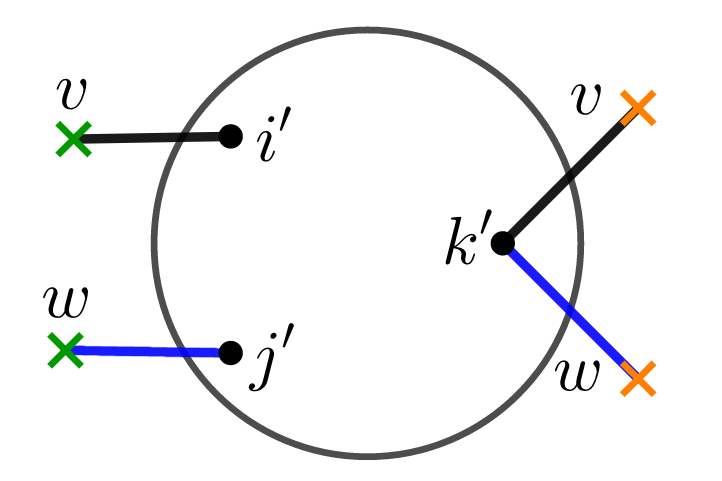}
\caption{Diagram for $X_S$}
\end{center}
\end{subfigure}
\begin{subfigure}[b]{0.24\textwidth}
\begin{center}
\includegraphics[width = 0.9\textwidth]{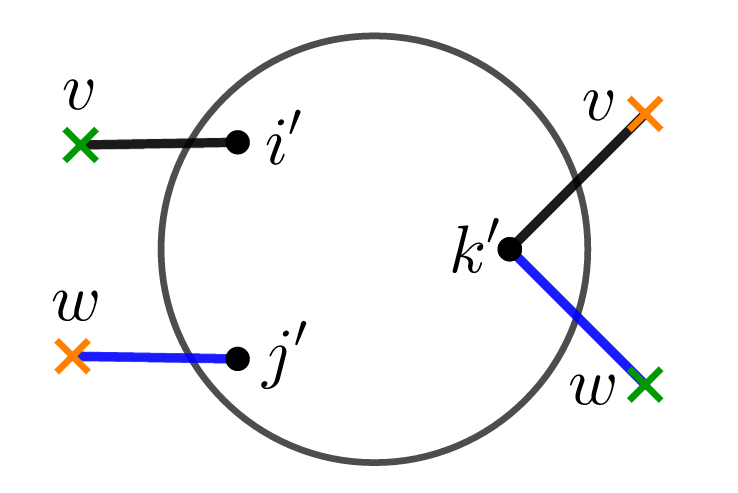}
\caption{Diagram of $\tw_2(X_S)$}
\end{center}
\end{subfigure}
\caption{Diagram illustrating stages of computing $Z_S$}\label{fig:BGMdiagram}
\end{figure} 
Therefore, $\wt{S} = \sum\limits_{i=1}^{m} R(v_i\otimes w_i) \otimes f_i$ is in $(v, w, *)\text{-}\mathfrak{BCM}^{11}(\mathcal{C}_{2/3}; 3, 2)$ or $(v, w, *)\text{-}\mathfrak{BCM}^{3}(\mathcal{C}_{v}; 2, 1)$ and so by Observation~\ref{obs:boundary-con-impl-con}, $\wt{S}\in \mathfrak{CM}^{11}(\mathcal{C}_{2/3}; 3)$ or $\wt{S}\in \mathfrak{CM}^{3}(\mathcal{C}_{v}; 2)$.

Next, we consider $\tw_2(A_0^TA_0)-P_{\mathcal{L}}$. As in Section~\ref{sec:B0}, we write
\[\tw_2(A_0^TA_0)-P_{\mathcal{L}} = \tw_2(A_0^TA_0 - B_{vw})+(B_{vw} - P_{\mathcal{L}}).\]
We compute
\[ \tw_2(A_0^TA_0 - B_{vw})(V\cten W) = \sum\limits_{k=1}^{m}\sum\limits_{i, j:\, i\neq j} \langle u_i, u_j \rangle\langle w_i, w_k\rangle \langle v_k, v_j \rangle (v_i\otimes w_j)f_k^T. \] 
Hence, $ \tw_2(A_0^TA_0 - B_{vw})(V\cten W)$ can be written as a sum of three IP graph matrices $S_1$, $S_2$ and $S_3$, which correspond to the cases $i\neq k\neq j$, $i=k$ and $j=k$, respectively, in the summation above. The matrix diagrams for them are presented on Figure~\ref{fig:A0-twist-diagram}. Note that, $S_1 \in \mathfrak{CM}^{3}(\mathcal{C}_{2/3}; 3)$, $S_2 \in \mathfrak{CM}^{2}(\mathcal{C}_{v}; 2)$ and $S_3 \in \mathfrak{CM}^{2}(\mathcal{C}_{w}; 2)$. Furthermore,
\[ S_4:= (B_{vw} - P_{\mathcal{L}})(V\cten W) = \sum\limits_{i, j:\, i\neq j} \langle v_i, v_j\rangle \langle w_i, w_j\rangle (v_i\otimes w_i)f_j^T,\]
so $S_4$ is in the $\mathfrak{CM}^{2}(\mathcal{C}_{v}; 2)$  (see Figure~\ref{fig:A0-twist-diagram} (d)).

Therefore, we summarize the discussion above into the following inclusion 
\begin{equation}\label{eq:BGM-linear-graph-matrix-expression}
 \sum\limits_{i=1}^{m} B_{GM}(v_i\otimes w_i)f_i^T \in \vspan\left(\mathfrak{CM}^{11}(\mathcal{C}_{2/3}; 3)\cup \mathfrak{CM}^{3}(\mathcal{C}_{v}; 2)\cup \mathfrak{CM}^{3}(\mathcal{C}_{w}; 2), 10^4\right).
 \end{equation}

 Recall that $\wt{D}_{GM} = \sum\limits_{i=1}^{m} (B_{GM}(v_i\otimes w_i))\otimes f_i$. Hence, the same inclusion, as in Eq.~\eqref{eq:BGM-linear-graph-matrix-expression}, is true for $\wt{D}_{GM}$. For any IP graph matrix $\wt{S}$ involved in the linear  expression for $\wt{D}_{GM}$, given by Eq.~\eqref{eq:BGM-linear-graph-matrix-expression} define
\[\wt{Y}_{S} = H\wt{S} \quad \text{and}\quad Y_{S} = \sum\limits_{k=1}^{m}(Y_{S})_kf_k^T, \quad \text{where} \quad \wt{Y}_{S} = \sum\limits_{k=1}^{m}(Y_{S})_k\otimes f_k.\]
 Then, by Lemma~\ref{lem:QQT-matrix-transform}, $Y_S$ can be written as 
$Y_S = \sum\limits_{i', j', k'} y_{i'j'k'}(v_{i'}\otimes w_{j'})f_{k'}^T$ and since the degree of the polynomial defining $H = Q^{[4]}_{inv}$ is at most 32, $Y_S$ is in  $\mathfrak{CM}^{43}(\mathcal{C}_{2/3}; 3)$ or is in $\mathfrak{CM}^{35}(\mathcal{C}_{v}; 2)$ or $\mathfrak{CM}^{35}(\mathcal{C}_{w}; 2)$. 
Therefore, it is easy to see from Figure~\ref{fig:BGMdiagram} (c) and (d), that the matrix diagrams of $X_S = Y_S(V\cten W)$ and $\tw_2(X_S)$ are $\mathcal{C}$-connected and $\mathcal{C}$-boundary-connected, as any node is $\mathcal{C}$-connected to $i'$, $j'$ and $k'$, and every $C\in \mathcal{C}$  contains either $v$ or $w$, for $\mathcal{C}\in \{\mathcal{C}_{2/3}, \mathcal{C}_{v}, \mathcal{C}_w\}$.

Thus, by linearity, for $Z_{GM} = X_{GM}+X_{GM}^T+\tw_2(X_{GM})+\tw_2(X_{GM})^T$, 
\[ Z_{GM} \in \vspan\left(\mathfrak{CM}^{43}(\mathcal{C}_{2/3}; 3)\cup \mathfrak{CM}^{35}(\mathcal{C}_{v}; 2)\cup \mathfrak{CM}^{35}(\mathcal{C}_{w}; 2), L_0\right), \]
for some absolute constant $L_0$, which can be bounded by $10^4$ times the sum of the absolute values of coefficients of $H$ given by Eq.~\eqref{eq:Happrox-def}. Therefore, using Theorem~\ref{thm:main-diagram-tool} and the trace power method (see Lemma~\ref{lem:trace-power-method-norm}), we obtain
\[\left\Vert Z_{GM} \right\Vert = \wt{O}\left(\dfrac{m}{n^{3/2}}+\dfrac{\sqrt{m}}{n}\right).\]  

\end{proof}

\begin{figure}
\begin{subfigure}[b]{0.24\textwidth}
\begin{center}
\includegraphics[width = 0.9\textwidth]{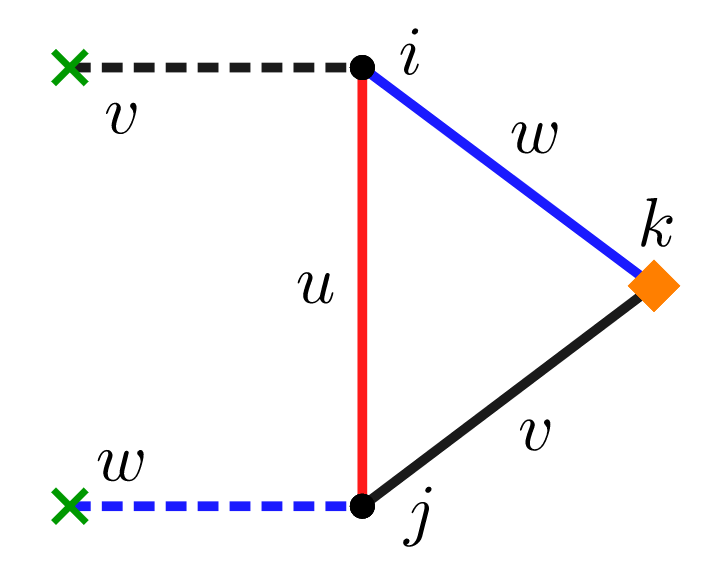}
\caption{Diagram for $S_1$}
\end{center}
\end{subfigure}
\begin{subfigure}[b]{0.24\textwidth}
\begin{center}
\includegraphics[width = 0.8\textwidth]{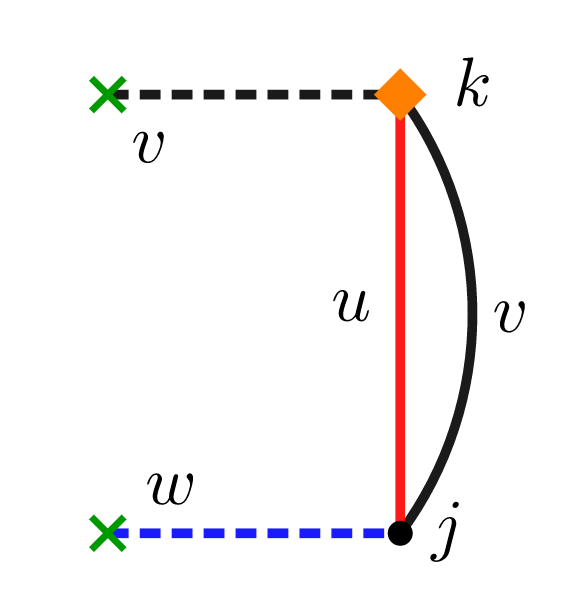}
\caption{Diagram for $S_2$}
\end{center}
\end{subfigure}
\begin{subfigure}[b]{0.24\textwidth}
\begin{center}
\includegraphics[width = 0.9\textwidth]{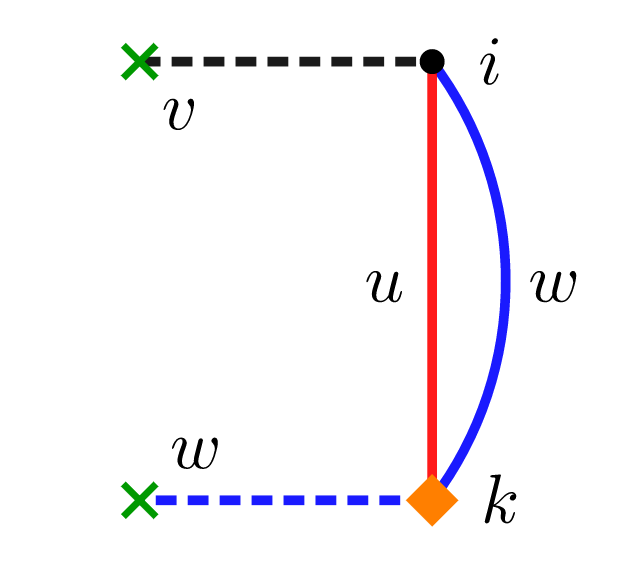}
\caption{Diagram for $S_3$}
\end{center}
\end{subfigure}
\begin{subfigure}[b]{0.24\textwidth}
\begin{center}
\includegraphics[width = 0.8\textwidth]{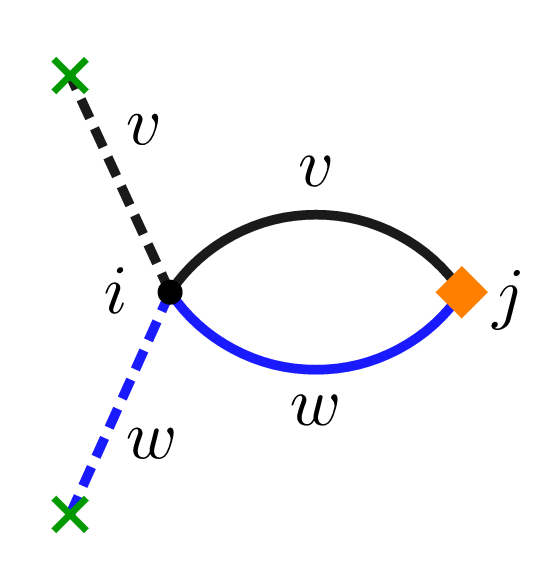}
\caption{Diagram for $S_4$}
\end{center}
\end{subfigure}
\caption{Diagrams for IP graph matrices involved in $\tw_2(A_0^TA_0)-\mathcal{P}_{\mathcal{L}}$}\label{fig:A0-twist-diagram}
\end{figure}

Combining this with the bound for $Z_{sm}$ we obtain the desired bound on the norm of $Z_0$.
 
 \begin{theorem}\label{thm:Z0-norm-bound} Let $Z_0 = \mathcal{Z}(B_0-P_{\mathcal{L}})$ be as defined in Theorem~\ref{thm:zero-poly-correction-constr}. If $m\ll n^{3/2}$, then
\[\Vert \mathcal{Z}(B_0-P_{\mathcal{L}}) \Vert = \wt{O}\left(\dfrac{m}{n^{3/2}}+\dfrac{\sqrt{m}}{n}\right)\]  
\end{theorem}
\begin{proof} Follows from Proposition~\ref{prop:Zsm-bound} and Theorem~\ref{thm:ZGM-bound}, as $Z_0 = Z_{GM}+Z_{sm}$.
\end{proof}

Therefore we may deduce one of our main theorems.

\begin{theorem}\label{thm:ABZ-exists}
Let $m \ll n^{3/2}$. Then w.h.p. over the randomness of $\mathcal{V}$, there exists a solution $(\oa{A}, B, Z)$ to~\eqref{eq:opt-constraints-id-intr} that satisfies~\eqref{eq:conditions-on-B}, where $\oa{A}$ is a strong dual certificate for $\mathcal{V}$. 
\end{theorem}

\section{Construction of {$\Omega$}-restricted dual certificate candidate}\label{sec:Omega-dual-certificate}

\subsection{Construction of $\oa{A}_{\Omega}$}\label{sec:AOmega-small-construction}

The goal of this section is to construct a vector $\oa{A}_{\Omega}\in \mathbb{R}^{n^3}$ which satisfies conditions~\eqref{eq:AOmega-necessary-cond}.

We follow the strategy used in \cite{potechin-steurer-exact}. Let $\oa{A}\in \mathcal{S}$ be a dual certificate constructed in Theorems~\ref{thm:candidate-exists},~\ref{thm:certificate-existence}.
We define $\Omega = \Omega_1\cup \ldots \cup \Omega_k$ and define
\begin{equation}\label{eq:Aomega-constr}
\oa{A}_{\Omega} = \sum\limits_{j=1}^{k} R_{\Omega_j}\left(\prod_{i=1}^{j-1}P_{\mathcal{S}}\ol{R}_{\Omega_{j-i}}\right)\oa{A}+\oa{A}_{\Omega, sm},
\end{equation}
for a sufficiently large constant $k$ (for our analysis would be sufficient to have $k = 30$ if $N\gg mn^{3/2}$) and an unknown vector $\oa{A}_{\Omega, sm}$ which we treat as a correction for an error with a small norm. The conditions for $\oa{A}_{\Omega}$ translate into the following conditions for $\oa{A}_{\Omega, sm}$.
\begin{lemma}\label{lem:AOmega-small-criteria} The vector $\oa{A}_{\Omega}$ defined in Eq.~\eqref{eq:Aomega-constr} satisfies 
 $P_{\mathcal{S}}\oa{A}_{\Omega} = \oa{A}$ and $(\oa{A}_{\Omega})_{\omega} = 0$ for $\omega \notin \Omega$ if and only if $\oa{A}_{\Omega, sm}$ satisfies
 \[ P_{\mathcal{S}}\oa{A}_{\Omega, sm} = \left(\prod\limits_{i=0}^{k-1}P_{\mathcal{S}}\ol{R}_{\Omega_{k-i}}\right)\oa{A}, \qquad \text{and} \qquad (\oa{A}_{\Omega})_{\omega} = 0\quad  \text{for } \omega \notin \Omega.\]
\end{lemma} 
\begin{proof} Note that since $P_{\mathcal{S}}$ is a projector it satisfies $P_{\mathcal{S}}^2 = P_{\mathcal{S}}$. Therefore, we have
\[
\begin{gathered}
 P_{\mathcal{S}} \oa{A}_{\Omega} - P_{\mathcal{S}}\oa{A}_{\Omega, sm}= P_{\mathcal{S}}\sum\limits_{j=1}^{k} \left(I_{n^3} - \ol{R}_{\Omega_j}\right)\left(\prod_{i=1}^{j-1}P_{\mathcal{S}}\ol{R}_{\Omega_{j-i}}\right)\oa{A} = \\
 = P_{\mathcal{S}}\sum\limits_{j=1}^{k} \left(\prod_{i=1}^{j-1}P_{\mathcal{S}}\ol{R}_{\Omega_{j-i}}\right)\oa{A} - \sum\limits_{j=1}^{k}P_{\mathcal{S}}\ol{R}_{\Omega_j}\left(\prod_{i=1}^{j-1}P_{\mathcal{S}}\ol{R}_{\Omega_{j-i}}\right)\oa{A} = \\
 = P_S\oa{A} - \left(\prod\limits_{i=0}^{k-1}P_{\mathcal{S}}\ol{R}_{\Omega_{k-i}}\right)\oa{A}
 \end{gathered}
 \] 
Since $\Omega_j\subseteq \Omega$, for any $j\in[k]$, the equality $(\oa{A}_{\Omega})_{\omega} = (\oa{A}_{\Omega, sm})_{\omega}$ holds for all $\omega \notin \Omega$.
\end{proof}

Thus, to construct the correction $\oa{A}_{\Omega, sm}$ we will argue that $P_{\mathcal{S}}\ol{R}_{\Omega_{i}}$ w.h.p. shrinks an input vector and so it will be sufficient to use the theorem below.

 Consider an $3mn\times n^3$ matrix $M$ with rows $u_i\otimes v_i\otimes e_j$, $u_i\otimes e_j\otimes w_i$ and $e_j\otimes v_i\otimes w_i$, as defined in Section~\ref{sec:M-system}. Then $P_{\mathcal{S}} = M^T(MM^T)^{-1}M$.
 
  Additionally, with high probability over the randomness of $\mathcal{V}$, for all $\omega = (a, b, c)\in [n]^3$, 
\begin{equation}\label{eq:M-entrywise-bound}
(M^TM)_{(\omega, \omega)} = \sum\limits_{i = 1}^m \langle u_i, e_a\rangle^2\langle v_i, e_b\rangle^2+ \langle u_i, e_a\rangle^2\langle w_i, e_c\rangle^2+ \langle v_i, e_b\rangle^2\langle w_i, e_c\rangle^2 = \wt{O}\left(\dfrac{m}{n^2}\right).
\end{equation}

\begin{theorem}\label{thm:small-ROmega-correction} Assume that $m\ll n^2$, $N\gg mn$ and $M$ satisfies the bound in Eq.~\eqref{eq:M-entrywise-bound}. Let $E\in \mathcal{S}$. Then with high probability over the randomness of $\Omega$, there exists $Y \in \mathbb{R}^{n^3}$ such that $P_{\mathcal{S}}Y = E$, $Y_{\omega} = 0$ for all $\omega \notin \Omega$ and $\Vert Y\Vert =\wt{O}\left(\dfrac{n^{3/2}}{N^{1/2}}\right)\Vert E\Vert$.
\end{theorem}
\begin{proof}
Denote by $D_{\Omega}$ the diagonal $n^3\times n^3$ matrix with $(\omega, \omega)$ entry being 1 if $\omega\in \Omega$, and being 0, otherwise. Let $M_{\Omega} = MD_{\Omega}$. Since $Y = D_{\Omega}Y$ we can write
\[ E = P_{\mathcal{S}}Y = M^T(MM^T)^{-1}MY = M^T(MM^T)^{-1}M_{\Omega}Y.\]
Thus $ME = M_{\Omega}Y$. We look for $Y$ of the form $Y = M_{\Omega}^TX$. If $ME$ is in the range of $M_{\Omega}M_{\Omega}^T$, we obtain
\[ ME = M_{\Omega}M_{\Omega}^T X\quad  \Rightarrow \quad X = (M_{\Omega}M_{\Omega}^T)^{-1}ME \quad \Rightarrow \quad Y = M_{\Omega}^T(M_{\Omega}M_{\Omega}^T)^{-1}ME\]
First, note that for any $X$ this form ensures that $Y$ has non-zero entries only in the coordinates inside $\Omega$. Moreover, since $MM_{\Omega}^T = M_{\Omega}M_{\Omega}^T$ and $E\in \mathcal{S}$ we obtain 
\[ P_{\mathcal{S}}Y = M^T(MM^T)^{-1}\left(MM_{\Omega}^T(M_{\Omega}M_{\Omega}^T)^{-1}\right)ME = M^T(MM^T)^{-1}ME = E. \]
Hence, we only need to verify that $(M_{\Omega}M_{\Omega}^T)^{-1}$ is well-defined on $ME$. Let $D_a$ be a random matrix which has the unique non-zero entry $(a, a)$ equal 1 with probability $N/n^3$ and which is a 0 matrix otherwise. We can view $D_{\Omega}$ as
\[ D_{\Omega} = \sum\limits_{\omega \in [n]^3} D_a\quad \text{and so} \quad M_{\Omega}M_{\Omega}^T = \sum\limits_{\omega\in [n]^3} MD_{\omega}M^T. \]
We are going to apply a Chernoff bound to $MD_{\omega}M^T$. Note that for $\omega\in [n^3]$, $MD_{\omega}M^T\succeq 0$, and
\[\Vert MD_{\omega}M^T \Vert = \Vert D_{\omega}M^TMD_{\omega} \Vert = \left|(M^TM)_{(\omega, \omega)} \right|  = \wt{O}\left(\dfrac{m}{n^2}\right).\] 
At the same time, 
\[ \mathbb{E}\left[M_{\Omega}M_{\Omega}^T\right] = M\mathbb{E}\left[D_{\Omega}\right]M^T = \dfrac{N}{n^3}MM^T\]

By Proposition~\ref{prop:M-cross-approx} and by Lemma~\ref{lem:R-eigenspaces}, $M$ restricted to the $\Ran(M) = \mathcal{K}^{\perp}$ has eigenvalues in the interval $(0.9, 3.1)$ if $m\ll n^2$. Since $N/n^3\gg m/n^2$, by Matrix Chernoff's bound (see Theorem~\ref{thm:matrix-chernoff}), w.h.p. over the randomness of $\Omega$, $M_{\Omega}M_{\Omega}^T$ restricted to $\mathcal{K}^{\perp}$ has eigenvalues in the interval $(N/(2n^3),7N/(2n^3))$. Thus, $(M_{\Omega}M_{\Omega}^T)^{-1}ME$ is well-defined and 
\[ \Vert Y\Vert  \leq \Vert M_{\Omega}^T \Vert \cdot \Vert (M_{\Omega}M_{\Omega}^T)^{-1}\Vert \cdot \Vert M\Vert \cdot \Vert E \Vert  = \wt{O}\left(\dfrac{n^{3/2}}{N^{1/2}}\right) \Vert E\Vert. \]
\end{proof}

\subsection{Projection operator $P_{\mathcal{S}}$ is well-balanced}\label{sec:PS-wellbalanced}

In this section we show how to approximate the projector $P_{\mathcal{S}}$ with graph matrices and we prove entrywise bounds for  $P_{\mathcal{S}}$.

\subsubsection{Approximating $P_{\mathcal{S}}$ with an IP graph matrix}\label{sec:PS-approx}

Consider the following four subspaces of $\mathbb{R}^{n^3}$.
\begin{equation}
\begin{gathered}
\mathcal{S}_{uv} = \vspan\{ u_i\otimes v_i\otimes x \mid x\in \mathbb{R}^{n},\  i\in [m]\},\quad  \mathcal{S}_{uw} = \vspan\{ u_i\otimes x\otimes w_i \mid x\in \mathbb{R}^{n},\  i\in [m]\},\\
 \mathcal{S}_{vw} = \vspan\{ x\otimes v_i\otimes w_i \mid x\in \mathbb{R}^{n},\  i\in [m]\},\quad \text{and}\quad \mathcal{S}_{uvw} = \vspan\{ u_i\otimes v_i\otimes w_i \mid  i\in [m]\}.
 \end{gathered}
\end{equation}
Clearly,
 $ \mathcal{S} = \vspan\{ \mathcal{S}_{uv}, \mathcal{S}_{uw}, \mathcal{S}_{vw}\}$.
 
Additionally, define $P_{\bullet}$ to be the orthogonal projector onto subspace $\mathcal{S}_{\bullet}$ for $\bullet\in  \{uv, uw, vw, uvw\}$. Define 
 \begin{equation}\label{eq:Puvtilde-def}
  \wt{P}_{uv} = \sum\limits_{i=1}^{m}\sum\limits_{t=1}^{n} (u_i\otimes v_i\otimes f_t)(u_i\otimes v_i\otimes f_t)^T,\qquad \wt{P}_{uvw} = \sum\limits_{i=1}^{m} (u_i\otimes v_i\otimes w_i)(u_i\otimes v_i\otimes w_i)^T, 
  \end{equation}
 and define $\wt{P}_{vw}$, $\wt{P}_{uw}$ similarly to $\wt{P}_{uv}$ (see also Figure~\ref{fig:S-projectors}).  
 
\begin{lemma}\label{lem:projector-four-approx} For $\bullet \in \{uv, uw, vw, uvw\}$, $\Ker(\wt{P}_{\bullet})^{\perp} = \Ran(\wt{P}_{\bullet}) = \mathcal{S}_{\bullet}$ and w.h.p., for $m\ll n^2$, 
\[\left\Vert \wt{P}_{\bullet} - P_{\bullet}\right\Vert = \wt{O}\left(\dfrac{\sqrt{m}}{n}\right). \]
\end{lemma} 
\begin{proof} Clearly, $\Ker(\wt{P}_{vw})^{\perp} = \Ran(\wt{P}_{vw}) = \mathcal{S}_{vw}$, and since $\wt{P}_{vw} = I_{n}\otimes B_{vw}$, by Lemma~\ref{lem:Bvw-bound},  $\Vert \wt{P}_{vw} - P_{vw} \Vert = \wt{O}\left(\dfrac{\sqrt{m}}{n}\right)$. By the symmetrical argument, the claim follows for $\wt{P}_{uv}$ and $\wt{P}_{uw}$.

For $\wt{P}_{uvw}$ it is clear from the definition that $\Ker(\wt{P}_{uvw})^{\perp} = \Ran(\wt{P}_{uvw}) = \mathcal{S}_{uvw}$. Hence, to establish the desired inequality it is sufficient to check that for 
\[ x = \sum\limits_{i=1}^{m}\mu_{i}(u_i\otimes v_i \otimes w_i)\]
the inequality $\left\Vert\wt{P}_{uvw}x - x\right\Vert =  \wt{O}\left(\dfrac{\sqrt{m}}{n}\right)\Vert x\Vert$ holds. We have
\begin{equation*}
 \left\Vert\wt{P}_{uvw}x - x\right\Vert = \left\Vert(U\cten V\cten W)\left((U\cten V\cten W)^T(U\cten V\cten W) - I_{m}\right)\mu \right\Vert \leq \wt{O}\left(\dfrac{\sqrt{m}}{n^{3/2}}\right)\Vert \mu \Vert,
 \end{equation*}
 where the bound $ \left\Vert(U\cten V\cten W)^T(U\cten V\cten W) - I_{m}\right\Vert = \wt{O}\left(\dfrac{\sqrt{m}}{n^{3/2}}\right)$ follows from Theorem~\ref{thm:main-diagram-tool}.
 Considering $(U\cten V\cten W)^Tx$ we get that $\Vert \mu\Vert \leq \wt{O}(1)\Vert x\Vert$. Thus, the statement of the lemma holds.
\end{proof}

\begin{figure}
\begin{center}
\begin{subfigure}[b]{0.77\textwidth}
\begin{center}
\includegraphics[height = 2.5cm]{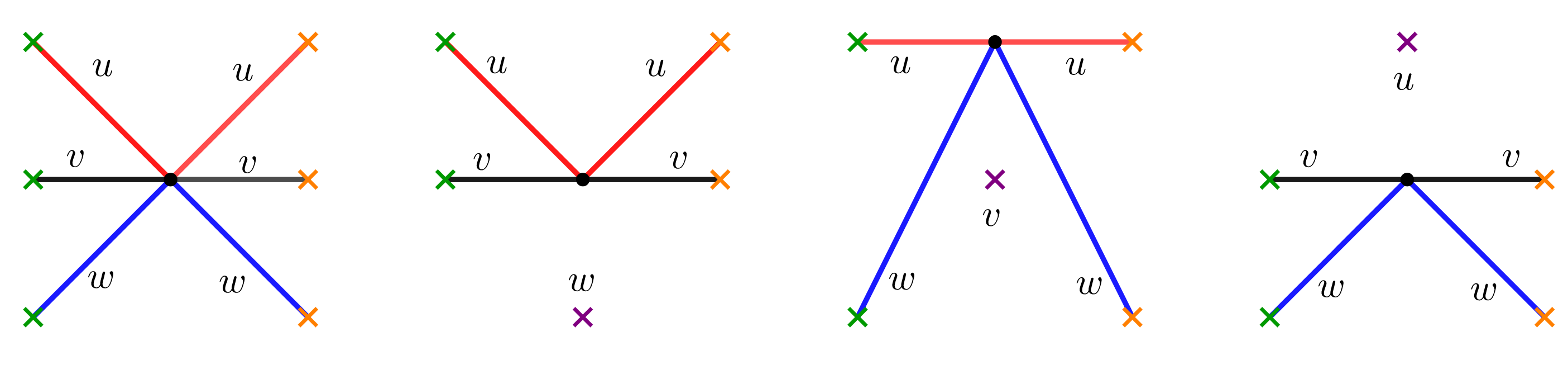}
\end{center}
\end{subfigure}
\begin{subfigure}[b]{0.2\textwidth}
\begin{center}
\includegraphics[height = 2.5cm]{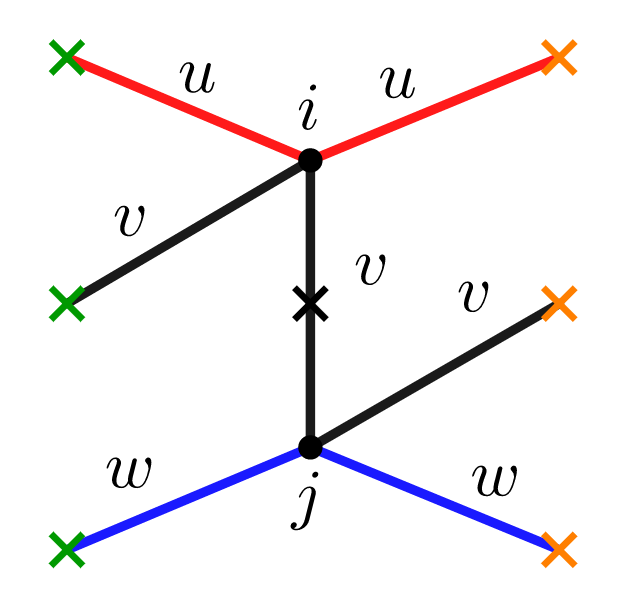}
\end{center}
\end{subfigure}
\caption{Matrix diagrams for $\wt{P}_{uvw}$, $\wt{P}_{uv}$, $\wt{P}_{uw}$, $\wt{P}_{vw}$; and $\wt{P}_{uv}\wt{P}_{vw}-\wt{P}_{uvw}$. } \label{fig:S-projectors}
\end{center}
\end{figure}

\begin{lemma}\label{lem:projector-products-approx} Assume that $m\ll n^2$. Then with high probability,
\[\left\Vert \wt{P}_{uv}\wt{P}_{vw} - \wt{P}_{uvw} \right \Vert = \wt{O}\left(\dfrac{\sqrt{m}}{n}\right). \] 
\end{lemma}
\begin{proof} By definition, 
\[ F := \wt{P}_{uv}\wt{P}_{vw} - \wt{P}_{uvw} = \sum\limits_{i\neq j} \langle v_i, v_j\rangle (u_i\otimes v_i\otimes w_j)(u_i\otimes v_j\otimes w_j)^T.\]
The matrix diagram (see Figure~\ref{fig:S-projectors}) of $\wt{P}_{uv}\wt{P}_{vw} - \wt{P}_{uvw}$ is $\{\{u, w\}, \{v\}\}$-boundary-connected. Additionally, it has one non-equality edge. Therefore, applying Theorem~\ref{thm:main-diagram-tool} with a slight modification as in the proof of Proposition~\ref{prop:M-cross-approx}, we deduce the statement of the lemma from Lemma~\ref{lem:trace-power-method-norm}.
\end{proof}

\begin{corollary}\label{cor:subspace-inter} If $m\ll n^2$, then with high probability
\[ \mathcal{S}_{uvw}  = \mathcal{S}_{uv}\cap \mathcal{S}_{vw} = \mathcal{S}_{uv}\cap \mathcal{S}_{uw} = \mathcal{S}_{uw}\cap \mathcal{S}_{vw}.\]
\end{corollary}
\begin{proof} By symmetry, it is enough to check that $\mathcal{S}_v'  = \mathcal{S}_{uv}\cap \mathcal{S}_{vw}$ coincides with $\mathcal{S}_{uvw}$. It directly follows from the definition that $\mathcal{S}_{uvw}\subseteq \mathcal{S}'_{v}$. At the same time, by Lemma~\ref{lem:projector-products-approx}, 
\[ \Vert P_{uv}P_{vw} - P_{uvw}\Vert \leq \left\Vert P_{uv}P_{vw}-\wt{P}_{uv}\wt{P}_{vw}\right\Vert+\left\Vert \wt{P}_{uv}\wt{P}_{vw} - \wt{P}_{uvw} \right\Vert + \left\Vert \wt{P}_{uvw} - P_{uvw} \right\Vert \leq \wt{O}\left(\dfrac{\sqrt{m}}{n}\right).\]
Note that $P_{uv}P_{vw}$ restricted to $\mathcal{S}'_{v}$ is an identity operator and thus it has at least $\dim(\mathcal{S}_{v}')$ eigenvalues 1. So, the projector $P_{uvw}$ has at least $\dim(\mathcal{S}_{v}')$ eigenvalues in $(0.9, 1.1)$ for $m\ll n^2$.
Hence, $\dim(\mathcal{S}'_{v}) \leq \dim(\mathcal{S}_{uvw})$. Thus, using the inclusion above, we get $\mathcal{S}_{uvw} = \mathcal{S}_{uv}\cap \mathcal{S}_{vw}$.
\end{proof}

%
\begin{corollary}\label{cor:PS-approx-weak} If $m\ll n^2$, then with high probablity
$$\left\Vert P_{\mathcal{S}} - \left(\wt{P}_{uv}+\wt{P}_{uw}+\wt{P}_{vw}-2\wt{P}_{uvw}\right) \right\Vert = \wt{O}\left(\dfrac{\sqrt{m}}{n}\right).$$
\end{corollary}
\begin{proof} Denote $\wt{P}_{\mathcal{S}} = \wt{P}_{uv}+\wt{P}_{uw}+\wt{P}_{vw}-2\wt{P}_{uvw}$. Note that for any $x\perp \mathcal{S}$, we have $\wt{P}_{\mathcal{S}}(x)=0$. 

At the same time, by the definition of $\mathcal{S}$, any $x\in \mathcal{S}$ can be written as $x = x_{uv}+x_{uw}+x_{vw}$, such that $x_{\bullet}\in \mathcal{S}_{\bullet}$ for $\bullet\in \{uv, uw, vw\}$. Using that $x_{\bullet} = P_{\bullet}x_{\bullet}$ for $\bullet\in \{uv, uw, vw\}$, we get
\[  \wt{P}_{uv}x_{vw}= \wt{P}_{uv}P_{vw}x_{vw} = \wt{P}_{uvw}x_{vw}+\wt{P}_{uv}\left( P_{vw} - \wt{P}_{vw}\right)x_{vw} + \left( \wt{P}_{uv}\wt{P}_{vw} - \wt{P}_{uvw}\right)x_{vw}.\]
Using a symmetrical argument and Lemmas~\ref{lem:projector-four-approx} and~\ref{lem:projector-products-approx}, we obtain
\[ \left\Vert \wt{P}_{uv}x - (x_{uv}+\wt{P}_{uvw}x_{uw}+\wt{P}_{uvw}x_{uw}) \right\Vert = \wt{O}\left(\dfrac{\sqrt{m}}{n}\right).\]
Thus, using a symmetrical argument for $\wt{P}_{uw}x$ and $\wt{P}_{vw}x$ as well, for $x\in \mathcal{S}$, with high probability
\[ \left\Vert \wt{P}_{uv}x +\wt{P}_{vw}x+\wt{P}_{uw}x-2 \wt{P}_{uvw}x-x\right\Vert = \wt{O}\left(\dfrac{\sqrt{m}}{n}\right).\]
Therefore, the claim of the corollary holds.
\end{proof}

\begin{theorem}\label{thm:PS-approx-strong} For any $k\in \mathbb{N}$, there exists a matrix $\wt{P}^{(k)}_{\mathcal{S}}$ which can be written as a polynomial of $\wt{P}_{uv}$, $\wt{P}_{uw}$, $\wt{P}_{vw}$, and $\wt{P}_{uvw}$ such that
\[ \left\Vert \wt{P}^{(k)}_{\mathcal{S}} - P_{\mathcal{S}}\right\Vert = \wt{O}\left(\left(\dfrac{\sqrt{m}}{n}\right)^{k}\right).\]
\end{theorem}
\begin{proof} Observe that for $\bullet\in \{uv, uw, vw, uvw\}$ we have
 $\Ker( \wt{P}_{\bullet})^{\perp} = \Ran(\wt{P}_{\bullet}) = \mathcal{S}_{\bullet}$. Therefore, $\wt{P}_{\bullet}P_{\mathcal{S}} = P_{\mathcal{S}}\wt{P}_{\bullet} = \wt{P}_{\bullet}$. Moreover, for any $t\geq 1$, the projector $P_{\mathcal{S}}$ satisfies $(P_{\mathcal{S}})^t = P_{\mathcal{S}}$. Thus, 
 \[ \wt{P}_{\mathcal{S}}^{(k)} = P_{\mathcal{S}} - \left( P_{\mathcal{S}} - \left(\wt{P}_{uv}+\wt{P}_{uw}+\wt{P}_{vw}-2\wt{P}_{uvw}\right) \right)^k   \] 
 is a degree-$k$ polynomial of $\wt{P}_{uv}$, $\wt{P}_{uw}$, $\wt{P}_{vw}$, and $\wt{P}_{uvw}$ and the desired norm bound follows from Corollary~\ref{cor:PS-approx-weak}.
\end{proof}

\subsubsection{Balanced matrices}
\begin{definition} Let $m$ and $n$ be fixed numbers. For $a, a'\in [n]$ denote 
\[ \bnd(a, a') = \begin{cases}
      \wt{O}\left(\dfrac{\sqrt{m}}{n^2}\right) & \text{if $a'\neq a$}, \\
      \wt{O}\left(\dfrac{m}{n^2}\right) & \text{if $a' = a$}.
    \end{cases}
    \]
\end{definition}
Let $M$ be an $n^d\times n^d$ matrix, whose rows and columns are indexed by tuples in $[n]^d$ in a natural way. For a subset $S\subseteq [d]$ we say that $M$ is \textit{$S$-balanced} (with respect to a function $\bnd()$) if 
\[M^2_{(a_1, a_2, \ldots, a_d)(a_1', a_2', \ldots, a_d')} = 
\begin{cases}
      \prod\limits_{t\in S} \bnd(a_t, a'_t) & \text{if $a_i = a'_i$ for any $i\notin S$}, \\
      0 & \text{if $a_i \neq a'_i$ for some $i\notin S$.}
    \end{cases}\]
    In the case when $d=3$, we will use set $\{u,v,w\}$ instead of $\{1,2,3\}$ in the definition above. For example, an $n^3\times n^3$ matrix $M$ is $\{u, v\}$-balanced if 
    \[ M^2_{(a, b, c)(a', b', c')} = 
\begin{cases}
      \bnd(a, a')\cdot\bnd(b, b') & \text{if $c=c'$}, \\
      0 & \text{if $c\neq c'$.}
    \end{cases}
\]
This can be alternatively written as
\[ M^2_{(a, b, c)(a', b', c')} = \mathbf{1}[c=c']\wt{O}\left(\left(\dfrac{\sqrt{m}}{n^2}\right)^2\right)(\sqrt{m})^{\displaystyle \mathbf{1}[a=a']+\mathbf{1}[b=b']}.\]
\begin{definition} We say that an $n^3\times n^3$ matrix $M$ is \emph{well-balanced} if it can be written as
\[ M = M_{\{u, v\}}+M_{\{u, w\}}+M_{\{v, w\}}+M_{\{u, v, w\}},\]
where $M_{S}$ is $S$-balanced for each $S\in \{\{u, v\}, \{u, w\}, \{v, w\}, \{u, v, w\}\}$.
\end{definition}

\subsubsection{Well-balanced IP graph matrices}

\begin{lemma}\label{lem:matrix-diagram-entrywise-tool} Let $G = (\myscr{Nod}\sqcup \myscr{Cr}, E)$ be a bipartite graph with parts $\myscr{Nod}\subseteq [m]$ and $\myscr{Cr}\subseteq [n]$. Let $X$ and $Y$ be disjoint subsets of $\myscr{Cr}$. Let $\{a_i \mid i\in [m]\}$ be a collection of independent uniformly distributed on $S^{n-1}$ vectors. Assume that $G$ is connected, and there are $2dq$ edge disjoint paths from vertices in $X$ to vertices in $Y$. Let $v(G) = |\myscr{Nod}|+|\myscr{Cr}|$, and $e(G) = |E|$. Assume $|E| = O(\log(n)^4)$. Then w.h.p.
\[ \vert val(G)\vert = \left\vert \mathbb{E}\prod_{\{j, t\}\in E, j\in \myscr{Nod}, t\in \myscr{Cr}}\langle a_i, e_t \rangle \right\vert  = \wt{O}\left(\dfrac{1}{n}\right)^{\displaystyle \max(v(G)+qd - |X|-|Y|, e(G)/2)}.\] 
\end{lemma}
\begin{proof} As in the proof of Lemma~\ref{lem:matix-diagram-tool-expanded}, if some edge of $G$ has odd multiplicity, then $val(G) = 0$. Hence, we may assume that every edge has even multiplicity.

Assume that all vertices in $G$, which are not in $X$ and $Y$,  have degree at least $4$. Since, there are $2dq$ edge disjoint paths between $X$ and $Y$, there are at least $4dq$ edges incident with vertices in $X\cup Y$ (here we treat each instance of a repeated edge as distinct). Therefore, $G$ has at least $2(v(G)-|X| - |Y|)+2dq$ edges. Hence, the desired bound is implied by Fact~\ref{fact:inner-produc-hp-bound}. 

Now we run induction by the number of vertices of degree 2, not in $X$ and $Y$. Assume that $G$ has a vertex $z$ of degree $2$, where $z\notin (X\cup Y)$. Note that $z$ is incident to a repeated edge. Let $G'$ be the graph obtained by deleting $z$ from $G$. Then $G'$ is connected, and there are still $2dq$ edge disjoint paths between $X$ and $Y$ in $G'$. Clearly, $v(G')=v(G)-1$ and $e(G') = e(G)-2$. Moreover, by independence, using Fact~\ref{fact:inner-products},
$val(G) = val(G')/n$. 
Therefore, the claim of the lemma follows by induction. 
\end{proof}

The lemma above can be combined with Lemma~\ref{lem:matix-diagram-tool-expanded} to achieve efficient entrywise bounds for IP-graph matrices. We do not attempt to prove some general bounds, and so we just show how this lemma is used to deduce that a matrix is well-balanced from its matrix diagram. We concentrate on the properties which are satisfied by matrices involved in $\wt{P}^{(k)}_{\mathcal{S}}$.

\begin{definition} For $C\subseteq [r]$ we say that a matrix diagram $G$ (with $r$ colors) is \emph{weakly-$C$-connected} if the subgraph induced by edges of color $C$ is connected (a vertex which is not incident to an edge of color from $C$ is deleted from the subgraph).

We say that $G$ is weakly-$\mathcal{C}$-connected, if for every set $C\in \mathcal{C}$\ $G$ is weakly-$C$-connected.  
\end{definition}

\begin{definition}
Let $G = (\myscr{Ver}, E, \mathfrak{c})$ be an $n^d\times n^d$ matrix diagram. We say that $G$ is \textit{$S$-nontrivial}, where $S = \{i\mid (\myscr{Ver}_L)_i \neq (\myscr{Ver}_R)_i\}$.
\end{definition}

Note that the set $S$ in the definition above is equal to the indices of crosses that are not incident to any edge. The corresponding graph matrix restricted to these indices is just an identity matrix. For example, $\wt{P}_{uv}$ is $\{u, v\}$-nontrivial. 

\begin{theorem}\label{thm:balanced-criteria} Assume that an $n^3\times n^3$ IP graph matrix $M$ has a weakly-$\{\{u\}, \{v\}, \{w\}\}$-connected matrix diagram. Suppose also that each node in $\mathcal{MD}(M)$ is incident with edges of at least 2 distinct colors. If $m\ll n^2$, $type(\myscr{Ver}_L) = type(\myscr{Ver}_R) = (u, v, w)$ for $\mathcal{MD}(M)$, and $M$ is $S$-nontrivial, then $M$ is $S$-balanced.
\end{theorem}
\begin{proof} We need to verify that for $a = (a_1, a_2, a_3), b = (b_1, b_2, b_3)\in [n]^3$ we have
\[ M_{a, b} = \left(\prod_{j\notin S} \mathbf{1}[a_j = b_j]\right)\prod_{t\in S} \bnd(a_t, b_t).\]
Fix $a$ and $b$ and consider the expanded matrix diagram $G$ of $(M_{a,b})^{2q}$. 

Let $\phi$ be a labeling of $G$, which respects the indicies assigned by $a$ and $b$. By assumption, the subgraph $G_u(\phi)$ of $G(\phi)$ induced by the edges of color $u$ is connected. Applying Lemma~\ref{lem:matix-diagram-tool-expanded} if $a_1 = b_1$, and Lemma~\ref{lem:matrix-diagram-entrywise-tool} if $a_1\neq b_1$ we get
\[ \vert val(G_u(\phi)) \vert = \wt{O}\left(\dfrac{1}{n}\right)^{\displaystyle \max\left(v(G_u(\phi))+q\mathbf{1}[a_1 \neq b_1]-2, e(G_u(\phi))/2\right)},\]  
where $v(G_u(\phi))$ and $e(G_u(\phi))$ is the number of vertices and edges in $G_u(\phi)$, respectively. 

By the same argument we get symmetrical bounds for $val(G_{v}(\phi))$ and $val(G_{w}(\phi))$. 

Using independence of $\{u_i\}$, $\{v_i\}$ and $\{w_i\}$ we get that 
\[ val(G(\phi)) = val(G_u(\phi))\cdot val(G_v(\phi))\cdot val(G_w(\phi)). \] 

Next, note that every edge and every cross of $G(\phi)$ appears in precisely one of $G_u(\phi)$, $G_v(\phi)$ and $G_w(\phi)$. On the other hand, every node of $G(\phi)$ appears in at least 2 of them. Let $\myscr{nod}(\phi)$, $\myscr{cr}(\phi)$ and $e(\phi)$ be the number of nodes,  crosses and edges in $G(\phi)$. We also denote $d = \mathbf{1}[a_1 \neq b_1]+\mathbf{1}[a_2 \neq b_2]+\mathbf{1}[a_3 \neq b_3]$, then
\[ \vert val(G(\phi)) \vert = \wt{O}\left(\dfrac{1}{n}\right)^{\displaystyle \max\left(2\myscr{nod}(\phi)+\myscr{cr}(\phi)+qd-6, e(\phi)/2\right)}\]

There are at most $(t+p)^{(t+p)}n^tm^p$ distinct labelings $\phi$ with $\myscr{cr}(\phi) = t$ and $\myscr{nod}(\phi) = p$. Hence, for $q = O(\log (n)^2)$,

\[ \left|\mathbb{E}\left[(M_{a, b})^{2q}\right]\right| = q^{O(q)}\max\limits_{\phi}\left(n^{\myscr{cr}(\phi)}m^{\myscr{nod}(\phi)}\wt{O}\left(\dfrac{1}{n}\right)^{\displaystyle \max\left(2\myscr{nod}(\phi)+\myscr{cr}(\phi)+qd-6 , e(\phi)/2\right)}\right)\]

Note that $G$ is bipartite and every cross has degree at least 2. Additionally, vertices in $\{a_i, b_i\}$ are adjacent with at least $4|S|q$ edges. Thus, $e(\phi)\geq 2(cr(\phi)-6)+4|S|q$.  Then

\[ \left|\mathbb{E}\left[(M_{a, b})^{2q}\right]\right| = q^{O(q)}\max\limits_{\phi}\left(m^{\myscr{nod}(\phi)}\wt{O}\left(\dfrac{1}{n}\right)^{\displaystyle \max\left(2\myscr{nod}(\phi)+qd-6 , 2|S|q-6\right)}\right)\]

Since $m\ll n^{2}$ the expression above is maximized for $\myscr{nod}(\phi) = (2|S|-d)q/2$. Thus, by Lemma~\ref{lem:trace-power-method-norm}, for $q = O(\log n)$, w.h.p.
\[  M_{a,b}^2  = \wt{O}\left(\dfrac{m^{|S|}}{n^{2|S|}}\cdot \left(\dfrac{1}{\sqrt{m}}\right)^{\mathbf{1}[a_1 \neq b_1]+\mathbf{1}[a_2 \neq b_2]+\mathbf{1}[a_3 \neq b_3]}\right).\]

\end{proof}

Now we verify that each matrix involved in the definition of $\wt{P}^{(k)}_{\mathcal{S}}$ satisfies the assumptions of the theorem above.

\begin{lemma}\label{lem:PuvPuvw-prop} Let $\wt{P}_{uv}$, $\wt{P}_{uw}$, $\wt{P}_{vw}$ and $\wt{P}_{uvw}$ be as in Eq.~\eqref{eq:Puvtilde-def}. The following properties hold.
\begin{enumerate}
\item $\mathcal{MD}(\wt{P}_{X})$ is $X$-nontrivial for $X\in \{\{u, v\}, \{u, w\}, \{v, w\}, \{u, v, w\}\}$.
\item $\mathcal{MD}(\wt{P}_{X})$ is weakly-$\{\{u\}, \{v\}, \{w\}\}$-connected for each $X\in \{uv, uw, vw, uvw\}$. 
\item Each node of $\mathcal{MD}(\wt{P}_{X})$ is incident with edges of \ $\geq 2$ distinct colors for $X\in \{uv, uw, vw, uvw\}$.
\item If $\mathcal{MD}(X)$ is $S_X$-nontrivial, $\mathcal{MD}(Y)$ is $S_Y$-nontrivial and $(X, Y)$ are compatible, then $\mathcal{MD}(XY)$ is $(S_X\cup S_Y)$-nontrivial.
\item If $\mathcal{MD}(X)$ and $\mathcal{MD}(Y)$ are weakly-$\{\{u\}, \{v\}, \{w\}\}$-connected, and $\myscr{Ver}_L$ and $\myscr{Ver}_R$ for both $X$ and $Y$ have type $(u, v, w)$, then $\mathcal{MD}(XY)$ is weakly-$\{\{u\}, \{v\}, \{w\}\}$-connected.
\item Every product of the form $\prod P_i$, where $P_i\in \{\wt{P}_{uv}, \wt{P}_{uw}, \wt{P}_{vw}, \wt{P}_{uvw}\}$, is weakly-$\{\{u\}, \{v\}, \{w\}\}$-connected.
\end{enumerate} 
\end{lemma}
\begin{proof} The first three statements follow from Figure~\ref{fig:S-projectors}. Statements 4 and 5 easily follow from the definitions. The last statement is the direct corollary of 2 and 5.
\end{proof}

\begin{corollary}\label{cor:PSk-wellbalanced} Assume $m\ll n^2$. For every $k\geq 1$, w.h.p. $\wt{P}^{(k)}_{\mathcal{S}}$ is well-balanced.
\end{corollary}
\begin{proof}
We need to show that there exist matrices $\wt{P}_{\mathcal{S}}^{uv, k}$, $\wt{P}_{\mathcal{S}}^{uw, k}$, $\wt{P}_{\mathcal{S}}^{vw, k}$ and $\wt{P}_{\mathcal{S}}^{uvw, k}$ such that 
\[\wt{P}_{\mathcal{S}}^{(k)} = \wt{P}_{\mathcal{S}}^{uv, k}+\wt{P}_{\mathcal{S}}^{uw, k}+\wt{P}_{\mathcal{S}}^{vw, k}+\wt{P}_{\mathcal{S}}^{uvw, k},\]
and $P_{\mathcal{S}}^{X, k}$ is $S$-balanced for each $S \in \{uv, uw, vw, uvw\}$.

By the definition of $\wt{P}^{(k)}_{\mathcal{S}}$, it can be written as a signed sum of at most $6^k$ products of at most $k$ matrices form $\{\wt{P}_{uv}, \wt{P}_{uw}, \wt{P}_{vw}, \wt{P}_{uvw}\}$. By Lemma~\ref{lem:PuvPuvw-prop}, any such product is weakly-$\{\{u\}, \{v\}, \{w\}\}$-connected and each node of its matrix diagram is incident with edges of at least 2 distinct colors. Therefore, by Theorem~\ref{thm:balanced-criteria}, every such product is $S$-balanced for some $S \in  \{\{u, v\}, \{u, w\}, \{v, w\}, \{u, v, w\}\}$. Thus, $\wt{P}^{(k)}_{\mathcal{S}}$ is well-balanced.
\end{proof}

Finally, we can deduce that $\mathcal{P}_S$ is well-balanced itself if $m<n^{2-\delta}$. 
\begin{theorem}\label{thm:PS-well-balanced} Assume $m\ll n^{2-\delta}$, for $\delta>0$. W.h.p. the projector $\mathcal{P}_S$ is well-balanced.
\end{theorem}
\begin{proof} By Theorem~\ref{thm:PS-approx-strong}, for $k>12/\delta$, $\Vert \mathcal{P}_S -  \wt{P}^{(k)}_{\mathcal{S}}\Vert = \wt{O}(n^{-6})$. Thus, $\mathcal{P}_S -  \wt{P}^{(k)}_{\mathcal{S}}$ is well-balanced and so the statement follows from Corollary~\ref{cor:PSk-wellbalanced}.
\end{proof}

\subsection{Entrywise bound for the dual certificate $\mathcal{A}$}
As another application of Lemma~\ref{lem:matrix-diagram-entrywise-tool} we establish an entrywise bound for the dual certificate constructed in Theorem~\ref{thm:candidate-exists}.
\begin{theorem}\label{thm:A-infty-norm-bound} Assume that $m\ll n^{3/2}$ and let $\oa{A}$ be the dual certificate constructed in Theorems~\ref{thm:candidate-exists} and~\ref{thm:certificate-existence}. Then w.h.p.  $\Vert \oa{A} \Vert_{\infty} = \wt{O}\left(\dfrac{\sqrt{m}}{n^{3/2}}\right)$.
\end{theorem}
\begin{proof} Recall that the dual certificate constructed in Theorem~\ref{thm:candidate-exists} has the form
\[ \oa{A} = \sum\limits_{i=1}^m u_i\otimes v_i \otimes w_i + \alpha_i\otimes v_i \otimes w_i+u_i\otimes \beta_i \otimes w_i+u_i\otimes v_i \otimes \gamma_i.\]
An argument similar to one in Section~\ref{sec:explicit-approx} implies that the matrix $U'$ with columns $\alpha_i$ can be written as $U' = U''_{GM}+U''_{sm}$ with $\Vert U''_{sm} \Vert = \wt{O}\left(\dfrac{1}{\sqrt{n}}\right)$ and $U''_{GM}$ being a signed sum of IP graph matrices with $\mathcal{C}_{2/3}$-connected matrix diagram for $\mathcal{C}_{2/3} = \{\{u, v\}, \{u, w\}, \{v, w\}\}$ (see also Lemma~\ref{lem:ord-d-M-approx} and Corollary~\ref{cor:ord-d-GM-structure-corr}).

Using Fact~\ref{fact:inner-produc-hp-bound}, by the Bernstein inequality, for basis vectors $e_a, e_b, e_c \in \mathbb{R}^n$,
\[\left\vert\sum\limits_{i=1}^{m} \langle u_i, e_a\rangle \langle v_i, e_b\rangle \langle w_i, e_c\rangle\right\vert = \wt{O}\left(\dfrac{\sqrt{m}}{n^{3/2}}\right).\]
Let $\alpha_{i, sm}$ denote the $i$-th column of $U''_{sm}$. Then by the Cauchy-Schwarz inequality,
\[
\begin{gathered}\left\vert\sum\limits_{i=1}^{m} \langle \alpha_{i, sm}, e_a\rangle \langle v_i, e_b\rangle \langle w_i, e_c\rangle\right\vert \leq \left(\sum\limits_{i=1}^{m} \langle \alpha_{i, sm}, e_a\rangle^2\right)^{1/2}\left(\sum\limits_{i=1}^{m} \langle v_i, e_b\rangle^2 \langle w_i, e_c\rangle^2\right)^{1/2}\leq  \\
\leq \Vert U''_{sm}\Vert \cdot \wt{O}\left(\dfrac{\sqrt{m}}{n}\right) = \wt{O}\left(\dfrac{\sqrt{m}}{n^{3/2}}\right).
\end{gathered}\]
Finally let $\alpha_{i, gm}$ be the column of an IP graph matrix involved in $U''_{GM}$. Consider an expanded matrix diagram $G$ of 
\[ \left\vert\sum\limits_{i=1}^{m} \langle \alpha_{i, gm}, e_a\rangle \langle v_i, e_b\rangle \langle w_i, e_c\rangle\right\vert^{2q}.\] 
This diagram is $\mathcal{C}_{2/3}$-connected, moreover it contains $2q$-edge disjoint paths between $u$-cross labeled $a$ and $v$-cross labeled $b$, between $u$-cross labeled $a$ and $w$-cross labeled $c$, and between $v$-cross labeled $b$ and $w$-cross labeled $c$. Let $\phi$ be some labeling of this expanded matrix diagram such that every edge appears at least twice. Then, similarly as in the proof of Theorem~\ref{thm:balanced-criteria}, by applying Lemma~\ref{lem:matrix-diagram-entrywise-tool} to the labeled subgraphs $G_{uv}(\phi)$, $G_{uw}(\phi)$ and $G_{vw}(\phi)$  induced by edges of colors $\{u, v\}$, $\{u, w\}$ and $\{v, w\}$, respectively, we obtain
\[ \left| val(G_{uv}(\phi))\right| = \wt{O}\left(\dfrac{1}{n}\right)^{\displaystyle \max\left(\myscr{nod}(\phi)+\myscr{cr}_{uv}(\phi)+q/2-2, e_{uv}(\phi)/2\right)}\] 
and symmetric statements for $G_{uw}(\phi)$ and $G_{vw}(\phi)$. Here $\myscr{nod}(\phi) = \myscr{nod}_{uv}(\phi)$, $\myscr{cr}_{uv}(\phi)$ and $e_{uv}(\phi)$ denote the number of nodes,  crosses and edges in $G_{uv}(\phi)$. Clearly, the number of nodes crosses and edges of $G(\phi)$ is $\myscr{nod}(\phi)$, 
\[ \myscr{cr}(\phi) = \left(\myscr{cr}_{uv}(\phi)+\myscr{cr}_{uw}(\phi)+\myscr{cr}_{vw}(\phi)\right)/2\quad \text{and} \quad e(\phi) = \left(e(\phi)_{uv}+e(\phi)_{uw}+e(\phi)_{vw}\right)/2\]. Using independence of $u_i$, $v_i$ and $w_i$,
\[ 
\begin{gathered}
\vert val(G(\phi)) \vert^2 = \vert val(G_{uv}(\phi))\cdot val(G_{uw}(\phi))\cdot val(G_{vw}(\phi))\vert \qquad \Rightarrow \\
 \vert val(G(\phi)) \vert = \wt{O}\left(\dfrac{1}{n}\right)^{\displaystyle \max\left(3\myscr{nod}(\phi)/2+\myscr{cr}(\phi)+3q/2-3, e(\phi)/2\right)}
 \end{gathered}
 \] 
  There are at least $6q$ edges adjacent with crosses $a, b, c$, so there are at least $2(\myscr(cr)(\phi) - 3)+6q$ edges in $G(\phi)$. Therefore,
 
\[
\begin{gathered}
\mathbb{E}\left[\left\vert\sum\limits_{i=1}^{m} \langle \alpha_{i, gm}, e_a\rangle \langle v_i, e_b\rangle \langle w_i, e_c\rangle\right\vert^{2q}\right]  = \left|\sum\limits_{\phi\in \Phi_0} val(G(\phi))\right| \leq \\
   \leq q^{O(q)}\max_{\phi\in \Phi_0}\left( n^{\myscr{cr}(\phi)}m^{\myscr{nod}(\phi)}\wt{O}\left(\dfrac{1}{n}\right)^{\displaystyle \max (3\myscr{nod}(\phi)/2+\myscr{cr}(\phi)+3q/2 - 3, \myscr{cr}(\phi)-3 +3q)}\right)  = \\
   = q^{O(q)}\max_{\phi\in \Phi_0}\left( m^{\myscr{nod}(\phi)}\wt{O}\left(\dfrac{1}{n}\right)^{\displaystyle \max (3\myscr{nod}(\phi)/2, 3q/2)}\right)  
\end{gathered}\]
Since $m\ll n^{3/2}$ the expression in the last line is maximized when $\myscr{nod}(\phi) = q$. Therefore, using the trace power method (see Lemma~\ref{lem:trace-power-method-norm}) for $q = O(\log(n)^2)$ we get
\[ \left\vert\sum\limits_{i=1}^{m} \langle \alpha_{i, gm}, e_a\rangle \langle v_i, e_b\rangle \langle w_i, e_c\rangle\right\vert \leq \wt{O}\left(\dfrac{\sqrt{m}}{n^{3/2}}\right).\]
Thus, applying a symmetrical argument to the terms involving $\beta_i$ and $\gamma_i$  we deduce the statement of the theorem.
\end{proof}

\subsection{Entrywise bound for $ X\ol{R}_{\Omega}Y $ when $X$ is well-balanced}\label{sec:infty-to-infty-norm-bound}

Let $\overline{R}_{\Omega}$ be an $n^3\times n^3$ matrix defined in Eq.~\eqref{eq:ROmega-definition}. Denote $\overline{R}_{\Omega}(\omega) = (\overline{R}_{\Omega})_{\omega, \omega}$.

\begin{lemma}[{\cite[Proposition 5.5]{potechin-steurer-exact}}]\label{lem:ROmega-moments} For a given $(a, b, c)\in [n]^3$,
\begin{enumerate}
\item $\mathbb{E}\left[\ol{R}_{\Omega}(a, b, c)\right] = 0$,
\item For all $k>1$, $\mathbb{E}\left[\left(\ol{R}_{\Omega}(a, b, c)\right)^k \right] \leq \left(\dfrac{n^3}{N}\right)^{k-1}$.
\end{enumerate}
\end{lemma}

Let $X$ be an $n_1\times n^3$ matrix and $Y$ be an $n^3\times n_2$ matrix. To bound the entries of $X\ol{R}_{\Omega}Y$ we use power method described in Section~\ref{sec:power-trace-prelim}.  For $j\in [n_1]$ and $k\in [n_2]$ we consider 
\begin{equation}\label{eq:Romega-power-method}
 \mathbb{E}_{\Omega}\left(\left(X\ol{R}_{\Omega}Y\right)_{j,k}\right)^{2q} = \mathbb{E}_{\Omega} \sum\limits_{\substack{\phi: [2q] \rightarrow [n]^3}} \prod\limits_{i=1}^{2q} X_{j,\phi(i)}\cdot \ol{R}_{\Omega}(\phi(i))\cdot Y_{\phi(i), k}.
 \end{equation}
Now for the function $\phi$ in the expression above it is convenient to consider a hypergraph $H = H(\phi)$ defined in the following way. We consider the set $V_0$ consisting of three independent copies of $[n]$ and we think of $\phi$ as a function whose values are 3-element sets of $V_0$. Then we can define $H$ to be a hypergraph on $V_0$ with $2q$  (possibly repeating) hyperedges defined by the image of $\phi$. The set of vertices of $H$ is equal $V(H) = \bigcup\limits_{i=1}^{2q} \phi(i)$. We denote the number of vertices and hyperedges of $H$ by
\[ v(\phi) = \left\vert \bigcup\limits_{i=1}^{2q} \phi(i) \right\vert \quad \text{and} \quad h(\phi) = \left\vert \phi([2q])\right\vert \]
respectively.

If in the hypergraph $H(\phi)$ some hyperedge appears only once, using Lemma~\ref{lem:ROmega-moments} and the independence of entries of $\ol{R}_{\Omega}$, we obtain 
\[ \mathbb{E}_{\Omega}  \prod\limits_{i=1}^{2q} X_{j,\phi(i)}\cdot \ol{R}_{\Omega}(\phi(i))\cdot Y_{\phi(i), k} = 0.\]
Hence, using the linearity of expectation, in Eq.~\eqref{eq:Romega-power-method}, it is sufficient to sum over $\phi\in \Phi_0$, where 
 \[ \Phi_0 = \{\phi:[2q]\rightarrow [n]^3 \mid \text{ for any } (x, y, z)\in [n]^3,\ |\phi^{-1}(x, y, z)|\neq 1\}\]
 is the set of maps such that every image appears at least twice.
 
 Thus, we can bound the expression in Eq~\eqref{eq:Romega-power-method}, as 
 \begin{equation*}
 \begin{gathered}
  \mathbb{E}_{\Omega}\left(\left(X\ol{R}_{\Omega}Y\right)_{j,k}\right)^{2q} \leq \\
  \leq  \sum\limits_{t=0}^{\infty} \left\vert\{\phi \mid \phi\in \Phi_0,\ v(\phi) = t \} \right\vert \cdot \max\limits_{\displaystyle \phi\in \Phi_0,\, v(\phi) = t} \left\vert \mathbb{E}_{\Omega} \prod\limits_{i=1}^{2q} X_{j,\phi(i)}\cdot \ol{R}_{\Omega}(\phi(i))\cdot Y_{\phi(i), k} \right\vert.
  \end{gathered}
  \end{equation*}
  Since every image of $\phi$ appears at least twice, $\phi$ takes at most $q$ distinct values. Hence, for any $\phi \in \Phi_0$, we have $h(\phi)\leq q$ and $v(\phi)\leq 3h(\phi)\leq 3q$. Note that there at most $t^{2q}n^{t}\leq 3q^{2q}n^{t}$ maps $\phi:[2q]\rightarrow [n]^3$ with $v(\phi) = t$. Moreover, recall that by Lemma~\ref{lem:ROmega-moments},
  \[ \left\vert \mathbb{E}_{\Omega} \prod \ol{R}_{\Omega}(\phi(i))\right\vert \leq \wt{O}\left(\dfrac{n^3}{N}\right)^{2q-h(\phi)}.\]
  Therefore, we can summarize this discussion into the lemma below.
  \begin{lemma}\label{lem:XRY-trace-bound}
  Let $\ol{R}_{\Omega}$ be as in Eq.~\eqref{eq:ROmega-definition}, $X\in \mathbb{R}^{n_1\times n^3}$ and $Y\in \mathbb{R}^{n^3\times n_2}$. Then   for all $j\in [n_1]$ and $k\in [n_2]$
  \[ \mathbb{E}_{\Omega}\left(\left(X\ol{R}_{\Omega}Y\right)_{j,k}\right)^{2q} \leq \sum\limits_{t=1}^{3q}(3q)^{2q}n^{t}\max\limits_{\displaystyle \phi\in \Phi_0,\, v(\phi) = t} \left\vert \left(\dfrac{n^3}{N}\right)^{2q-h(\phi)}\prod\limits_{i=1}^{2q} X_{j,\phi(i)} Y_{\phi(i), k} \right\vert. \]
  \end{lemma} 
  
  In the cases when $X$ and $Y$ have some special structure we are able to get good bounds on $\left\vert \prod\limits_{i=1}^{2q} X_{j,\phi(i)} Y_{\phi(i), k} \right\vert$ and $e(\phi)$, which yields efficient entrywise bounds for $X\ol{R}_{\Omega}Y$.
  
\begin{theorem}\label{thm:infinity-norm-balanced} Let $X\in \mathbb{R}^{n^3\times n^3}$ be a well-balanced matrix and $Y\in R^{n^3}$. Suppose $N\gg mn$ and $m\ll n^2$. Then w.h.p.
\[ \Vert  X\ol{R}_{\Omega}Y \Vert_{\infty}  = \wt{O}\left(\dfrac{n^{1/2}m^{1/2}}{N^{1/2}}\right)\Vert Y\Vert_{\infty}.\]
\end{theorem}
\begin{proof} Since $X$ is well-balanced, we can write $X = X^{\{u, v\}}+X^{\{u, w\}}+X^{\{v, w\}}+X^{\{u, v, w\}}$, where $X^S$ is $S$-balanced for $S\in \{\{u, v\}, \{u, w\}, \{v, w\}, \{u, v, w\}\}$.

For $X^{\{u, v\}}$ we bound
\[ \left\vert \prod\limits_{i=1}^{2q} X^{\{u, v\}}_{(a, b, c),\phi(i)} Y_{\phi(i)} \right\vert \leq \Vert Y\Vert^{2q}_{\infty}\left(\wt{O}\left(\dfrac{m^2}{n^4}\right)\right)^{q}\left(\dfrac{1}{\sqrt{m}}\right)^{v(\phi)-3}.\]
Moreover, the product of $X^{\{u, v\}}_{(a, b, c),\phi(i)}$ is non-zero only if $\phi(i) = c$ for any $i$. For such $\phi$ we can bound $v(\phi)\leq 2h(\phi)+1$.
Thus, Lemma~\ref{lem:XRY-trace-bound} implies
\[ \mathbb{E}_{\Omega}\left(\left(X^{\{u, v\}}\ol{R}_{\Omega}Y\right)_{(a, b, c)}\right)^{2q} \leq \sum\limits_{t=1}^{3q}(3q)^{2q}n^{t}\left(\dfrac{n^3}{N}\right)^{2q-(t-1)/2}\cdot \left(\wt{O}\left(\dfrac{m^2}{n^4}\right)\right)^{q}\left(\dfrac{1}{\sqrt{m}}\right)^{t-3}\Vert Y\Vert^{2q}_{\infty}. \]
Note that $t\leq 2q+1$ and $N>nm$ implies
\[ n^t \left(\dfrac{n^3}{N}\right)^{2q-t/2}\left(\dfrac{m^2}{n^4}\right)^{q}\left(\dfrac{1}{\sqrt{m}}\right)^t = \left(\dfrac{nm}{N}\right)^{2q}\left(\dfrac{N}{nm}\right)^{t/2}\leq \left(\dfrac{nm}{N}\right)^{q-1/2}.\]
Hence,
\[ \mathbb{E}_{\Omega}\left(\left(X^{\{u, v\}}\ol{R}_{\Omega}Y\right)_{(a, b, c)}\right)^{2q} \leq (3q)^{2q+1}\left(\wt{O}\left(\dfrac{nm}{N}\right)\right)^{q}\left(nm\right)\Vert Y\Vert^{2q}_{\infty}. \]

Taking $q = O(\log(n)^2)$, by Lemma~\ref{lem:trace-power-method-norm} we get,
\[ \left\vert \left(X^{\{u, v\}}\ol{R}_{\Omega}Y\right)_{(a, b, c)} \right\vert \leq \wt{O}\left(\dfrac{n^{1/2}m^{1/2}}{N^{1/2}}\right)\Vert Y\Vert_{\infty}.\]
By the symmetrical argument we get that $\Vert X^{\{u, w\}}\ol{R}_{\Omega} Y\Vert_{\infty}$ and $\Vert X^{\{v, w\}}\ol{R}_{\Omega} Y\Vert_{\infty}$ are bounded by $\wt{O}\left(\dfrac{n^{1/2}m^{1/2}}{N^{1/2}}\right)\Vert Y\Vert_{\infty}$ as well.

For $X^{\{u, v, w\}}$ we have $v(\phi)\leq 3h(\phi)$ and
 \[ \left\vert \prod\limits_{i=1}^{2q} X^{\{u, v, w\}}_{(a, b, c),\phi(i)} Y_{\phi(i)} \right\vert \leq \Vert Y\Vert^{2q}_{\infty}\left(\wt{O}\left(\dfrac{m^3}{n^6}\right)\right)^{q}\left(\dfrac{1}{\sqrt{m}}\right)^{v(\phi)-3}.\]
 Hence, Lemma~\ref{lem:XRY-trace-bound} implies that
  \[ \mathbb{E}_{\Omega}\left(\left(X^{\{u, v, w\}}\ol{R}_{\Omega}Y\right)_{(a, b, c)}\right)^{2q} \leq \sum\limits_{t=1}^{3q}(3q)^{2q}n^{t}\left(\dfrac{n^3}{N}\right)^{2q-t/3}\cdot \left(\wt{O}\left(\dfrac{m^3}{n^6}\right)\right)^{q}\left(\dfrac{1}{\sqrt{m}}\right)^{t-3}\Vert Y\Vert^{2q}_{\infty}. \]
Note $t\leq 3q$ and $N>nm\geq m^{3/2}$ implies
\[ n^t \left(\dfrac{n^3}{N}\right)^{2q-t/3}\left(\dfrac{m^3}{n^6}\right)^{q}\left(\dfrac{1}{\sqrt{m}}\right)^t = \left(\dfrac{m^{3/2}}{N}\right)^{2q}\left(\dfrac{N}{m^{3/2}}\right)^{t/3}\leq \left(\dfrac{m^{3/2}}{N}\right)^{q}.\]
Hence,
\[ \mathbb{E}_{\Omega}\left(\left(X^{\{u, v, w\}}\ol{R}_{\Omega}Y\right)_{(a, b, c)}\right)^{2q} \leq (3q)^{2q+1}\left(\wt{O}\left(\dfrac{m^{3/2}}{N}\right)\right)^{q}\left(m^{3/2}\right)\Vert Y\Vert^{2q}_{\infty}. \]
Taking $q = O(\log(n)^2)$, by Lemma~\ref{lem:trace-power-method-norm}, we get,
\[ \left\vert \left(X^{\{u, v, w\}}\ol{R}_{\Omega}Y\right)_{(a, b, c)} \right\vert \leq \wt{O}\left(\dfrac{m^{3/4}}{N^{1/2}}\right)\Vert Y\Vert_{\infty} \leq \wt{O}\left(\dfrac{n^{1/2}m^{1/2}}{N^{1/2}}\right)\Vert Y\Vert_{\infty}.\]
Therefore, $\Vert  X\ol{R}_{\Omega}Y \Vert_{\infty}  = \wt{O}\left(\dfrac{n^{1/2}m^{1/2}}{N^{1/2}}\right)\Vert Y\Vert_{\infty}$.
\end{proof}

\subsection{Structure of {$A_{\Omega}$}}\label{sec:AOmega-struct}

We summarize the results we proved in this section into the structural theorems for $\oa{A}_{\Omega}$ below.

In Sections~\ref{sec:PS-wellbalanced}-\ref{sec:infty-to-infty-norm-bound} we proved that the projector $P_{\mathcal{S}}$ has small entries and we proved that this implies that $P_{\mathcal{S}}\ol{R}_{\Omega}$ w.h.p. does not increase the $\Vert\cdot \Vert_{\infty}$ norm of a given vector. Combining this with the results of Section~\ref{sec:AOmega-small-construction} we obtain the following theorem.  

\begin{theorem}\label{thm:AOmega-structure-pseudorandom} Let $m\ll n^{2-\delta}$ and $N\gg nm$ for $\delta>0$. Let $\mathcal{V} = \{(u_i, v_i, w_i)\mid i\in [m]\}\subset S ^{n-1}$ be a collection of vectors. Let $\mathcal{S}$ be the subspace as in Eq.~\eqref{eq:S-def} and $M$ be the matrix from Eq.~\eqref{eq:M-definition}. Suppose that the projector $P_{\mathcal{S}}$ is well-balanced and $M$ satisfies Eq.~\eqref{eq:M-entrywise-bound}. Take $C>0$ and define $K_N = O(\log n)$, if $N\gg nm$; and $K_N = (6+2C)/\varepsilon$, if $N\gg n^{1+\varepsilon}m$. Consider $X\in \mathcal{S}$. Then, w.h.p. over the randomness of $\Omega = \bigcup_{i=1}^{K_{N}} \Omega_i$ there exists a vector $X_{\Omega}$ of the form
\[ X_{\Omega} = \sum\limits_{j=1}^{K_N} R_{\Omega_j}\left(\prod_{i=1}^{j-1}P_{\mathcal{S}}\ol{R}_{\Omega_{j-i}}\right)X+X_{\Omega, sm}, \]
such that 
\[P_{\mathcal{S}}X_{\Omega}=X, \qquad (X_{\Omega})_{\omega} = 0,\ \text{ for } \omega\notin \Omega,\quad \text{and}\quad \Vert X_{\Omega, sm} \Vert = \wt{O}\left(n^{-C}\right)\Vert X\Vert_{\infty}.\]  
\end{theorem}
\begin{proof} By Theorem~\ref{thm:infinity-norm-balanced},  
\[\left\Vert \left(\prod\limits_{i=1}^{K_N} P_{\mathcal{S}}\ol{R}_{\Omega} \right)X\right\Vert \leq n^{3/2}\wt{O}\left(\dfrac{nm}{N}\right)^{K_N/2}\Vert X\Vert_{\infty}. \]
Thus, by Lemma~\ref{lem:AOmega-small-criteria} and Theorem~\ref{thm:small-ROmega-correction}, there exists $X_{sm}$ such that $X_{\Omega}$ satisfies all the desired properties and 
\[\Vert X_{\Omega, sm} \Vert = \wt{O}\left(\dfrac{n^3}{\sqrt{N}}\right)\wt{O}\left(\dfrac{nm}{N}\right)^{K_N/2}\Vert X\Vert_{\infty} = \wt{O}\left(n^{-C}\right)\Vert X \Vert_{\infty}.\]
\end{proof}

Thus, as a corollary of the analysis above we obtain the following norm bound for $\mathcal{A}_{\Omega}$.
\begin{proposition} Suppose that assumptions of Theorem~\ref{thm:AOmega-structure-pseudorandom} hold and $\Vert X \Vert_{\infty} = \wt{O}\left(\dfrac{\sqrt{m}}{n^{3/2}}\right)$. Then 
\[\Vert X_{\Omega}\Vert = \wt{O}\left(n\sqrt{\dfrac{nm}{N}}\right).\]
\end{proposition}
\begin{proof} Observe that w.h.p. $\Vert R_{\Omega_i}Y \Vert  \leq  \wt{O}\left(\dfrac{n^3}{\sqrt{N}}\right)\Vert Y\Vert_{\infty}$ for any $Y\in \mathbb{R}^{n^3}$. Hence, using Theorem~\ref{thm:infinity-norm-balanced}, w.h.p.,  each term in Eq.~\eqref{eq:Aomega-constr} of the form $R_{\Omega_k}\left(\prod_{j=1}^{k-1} P_{\mathcal{S}}\ol{R}_{\Omega_{k-j}}\right)X$ has norm bounded by $\wt{O}\left(n^{3/2}m^{1/2}/N^{1/2}\right)$ and there are at most $(\log n)^{O(1)}$ such terms. Finally, the bound for $\oa{A}_{\Omega, sm}$ follows from Theorem~\ref{thm:small-ROmega-correction}.
\end{proof}

\begin{corollary}\label{cor:AOmega-norm-bound} Assume that $m\ll n^{3/2}$ and $N\gg mn$. Let $\oa{A}\in \mathcal{S}$ be a dual certificate for $\mathcal{V}$ constructed in Sections~\ref{sec:corr-terms} and \ref{sec:dual-certificate}. Then w.h.p over the randomness of $\Omega$, for $\mathcal{A}_{\Omega}$ as in Eq.~\eqref{eq:Aomega-constr}, we have $\Vert \oa{A}_{\Omega}\Vert = \wt{O}\left(n\sqrt{\dfrac{nm}{N}}\right)$.
\end{corollary}
\begin{proof} By Theorem~\ref{thm:PS-well-balanced}, w.h.p. the projector  $P_{\mathcal{S}}$ is well-balanced, and by Theorem~\ref{thm:A-infty-norm-bound}, w.h.p. $\Vert \oa{A}\Vert\leq \wt{O}\left(\dfrac{\sqrt{m}}{n^{3/2}}\right)$.
\end{proof}

For the case of random components we additionally want to keep track of essential IP graph matrix terms present in definition of $\oa{A}_{\Omega}$.

\begin{definition}\label{def:class-GP} Define $\mathfrak{G}_P$ to be the class of matrix diagrams $G = (\myscr{Ver}, E, \mathfrak{c}:E\rightarrow \{u, v, w\})$ with a number of vertices bounded by an absolute constant such that $type(\myscr{Ver}_L) = type(\myscr{Ver}_R) = (u, v, w)$, $G$ is weakly-$\{\{u\}, \{v\}, \{w\}\}$-connected and every node in $G$ is incident with edges of at least two distinct colors. We also require a diagram in $\mathfrak{G}_P$ to be $S$-nontrivial for $S\in \{\{u, v\}, \{u, w\}, \{v, w\}, \{u, v, w\}\}$. 
\end{definition}

\begin{theorem}\label{thm:AOmega-structure} Assume that $m\ll n^{3/2}$ and $N\gg n^{1+\delta}m$ for $\delta>0$. Let $\oa{A}\in \mathcal{S}$ be a dual certificate for $\mathcal{V} = \{(u_i, v_i, w_i)\mid i\in [m]\}$ constructed in Sections~~\ref{sec:corr-terms}-\ref{sec:dual-certificate}. Then, w.h.p. over the randomness of $\mathcal{V}$ and the randomness of $\Omega = \bigcup_{i=1}^{K_{N}} \Omega_i$ for $K_N = 14/\delta$, the vector $\oa{A}_{\Omega}$ given by Eq.~\eqref{eq:Aomega-constr} can be written as
\[ \oa{A}_{\Omega} = \oa{A}+\oa{A}_{\Omega, GM}+\oa{A}'_{\Omega, sm}, \]
where $\Vert \oa{A}'_{\Omega, sm} \Vert = \wt{O}\left(\dfrac{m^2}{n^4}\right)$ and $\oa{A}_{\Omega, GM} \in \vspan\left(\mathfrak{G}_{\Omega, N}, C^{K_N}\right)$ for some absolute constant $C$ and
\[\mathfrak{G}_{\Omega, N} =  \left\lbrace\left(\prod\limits_{i=0}^{k-1} P_{k-i}\ol{R}_{\Omega_{k-i}}\right)\oa{A} \mid k\in[K_N],\ P_{j}\in \mathfrak{G}_P,\text{ for } j<k, \text{ and } P_{k} \in \mathfrak{G}_P\cup\{I_{n^3}\} \right\rbrace. \] 
\end{theorem}
\begin{proof} By Theorem~\ref{thm:AOmega-structure-pseudorandom} and Theorem~\ref{thm:A-infty-norm-bound}, there exists a vector $\oa{A}_{\Omega, sm}$ such that
$\Vert \oa{A}_{\Omega, sm} \Vert = \wt{O}\left(\dfrac{m^2}{n^4}\right)$ and $\oa{A}_{\Omega}$ given by Eq.~\eqref{eq:Aomega-constr} satisfies all the desired properties. 

Using that $R_{\Omega_k} = I_{n^3}-\ol{R}_{\Omega_k}$, we get that $\oa{A}_{\Omega} - \oa{A}_{\Omega, sm} - \oa{A}$ can be written as a linear combination of terms of the form
$\left(\prod\limits_{i=0}^{k-1} P_{\mathcal{S}}\ol{R}_{\Omega_{k-i}}\right)\oa{A}$ or $\ol{R}_{\Omega_{k}}\left(\prod\limits_{i=1}^{k-1} P_{\mathcal{S}}\ol{R}_{\Omega_{k-i}}\right)\oa{A}$, for $k\in [K_N]$.

By Theorem~\ref{thm:PS-approx-strong} and Corollary~\ref{cor:PSk-wellbalanced}, one can write $P_{\mathcal{S}} = \wt{P}^{uv}_{\mathcal{S}}+\wt{P}^{uw}_{\mathcal{S}}+\wt{P}^{vw}_{\mathcal{S}}+\wt{P}^{uvw}_{\mathcal{S}}+P_{\mathcal{S}, sm}$, such that $\wt{P}^X_{\mathcal{S}}\in \mathfrak{G}_P$ is $X$-nontrivial, and $\Vert P_{\mathcal{S}, sm} \Vert = \wt{O}\left(\dfrac{m^{20}}{n^{40}}\right)$. Substitute this expression for $P_{\mathcal{S}}$ instead of every occurrence of $P_{\mathcal{S}}$ in the expression for $\oa{A}$ and open all the parenthesis. Let $\oa{A}_{\Omega, sm}'$ be the term obtained by collecting all summands which involve $P_{\mathcal{S}, sm}$ and $\mathcal{A}_{\Omega, sm}$. Clearly, $\oa{A}_{\Omega} - \oa{A}_{\Omega, sm}' - \oa{A}$ belongs to the span of matrices from $\mathfrak{G}_{\Omega, N}$.

Finally, note that any $n^3\times n^3$ matrix $X$ satisfies 
$n^{-3/2}\Vert X\Vert_{\infty \rightarrow \infty} \leq \Vert X\Vert \leq  n^{3/2}\Vert X\Vert_{\infty \rightarrow \infty}$. Additionally, for $\Omega$ with $N$ elements, we have
$\Vert\mathcal{R}_{\Omega_i}\Vert_{\infty\rightarrow \infty} = \dfrac{n^3K_N}{N}$. Therefore, by Theorem~\ref{thm:infinity-norm-balanced}, for $N\gg nm$, any summand which involves $P_{\mathcal{S}, sm}$ w.h.p. has norm at most
\begin{equation*}
\begin{gathered}
 n^{3/2}\Vert P_{\mathcal{S}, sm}\Vert_{\infty\rightarrow \infty} \left(\max\{\Vert\ol{R}_{\Omega_i}\Vert_{\infty\rightarrow \infty} \mid i\in [K_N]\}\right)^2\Vert \oa{A}_{\Omega}\Vert_{\infty} = \\
 = n^{3/2}\cdot n^{3/2}\cdot \wt{O}\left(\dfrac{m^{20}}{n^{40}}\cdot\dfrac{n^6}{N^2}\cdot \dfrac{\sqrt{m}}{n^{3/2}} \right) = \wt{O}\left(\dfrac{m^2}{n^4}\right)
 \end{gathered}
 \end{equation*}
 There are at most $K_N5^{K_N}$ summands involving $P_{\mathcal{S}, sm}$, thus $\Vert \oa{A}_{\Omega, sm}'\Vert = \wt{O}\left(\dfrac{m^2}{n^4}\right)$.
\end{proof}

\section{Existence of an $\Omega$-restricted SOS dual certificate {$(\oa{A}_{\Omega}, B_{\Omega}, Z_{\Omega})$} }\label{sec:SOS-Omega-certificate}

In this section, we prove that for $m\ll n^{3/2}$ and $N\gg n^{3/2}m$, with high probability over the randomness of $\mathcal{V}$ and randomness of $\Omega$ there exists a solution $(\oa{A}_{\Omega}, B_{\Omega}, Z_{\Omega})$ to program \eqref{eq:opt-objective-compl}-\eqref{eq:opt-constraints-compl}, where $\oa{A}_{\Omega}$ is an $\Omega$-restricted dual certificate for $\mathcal{V}$.

We are going to show that with high probability, the same construction we used in Sections~\ref{sec:B0}-\ref{sec:zero-poly-corr} works for $\oa{A}_{\Omega}$ as well. More precisely, similarly as in Section~\ref{sec:ideas-for-B}, we define
\[ B_{\Omega, 0} = \mathcal{B}(\oa{A}_{\Omega}, \oa{A}_{\Omega}) = \tw_2(A_{\Omega}^TA_{\Omega}), \qquad Z_{\Omega, 0} = \mathcal{Z}(B_{\Omega, 0} - P_{\mathcal{L}})\quad \text{and}\]
\[ B_{\Omega} = B_{\Omega, 0}+Z_{\Omega, 0}\qquad Z_{\Omega} = A_{\Omega}^TA_{\Omega} - B_{\Omega}.\quad \ \ \ \]
\begin{definition}\label{def:B0-construction} For $X, Y\in \mathbb{R}^{n^3}$ define an $n^2\times n^2$ matrix
\[ \mathcal{B}(X, Y)_{(b, c)(b', c')} = \sum\limits_{a = 1}^{n} X_{(a, b, c')}Y_{(a, b', c)}. \]
\end{definition}
\begin{observation}\label{obs:ZBOmega-well-defined} Let $B_{\Omega, 0}$ be defined as above. Then for all  $x\in \mathbb{R}^n$, and  $i\in [m]$
\begin{equation*}
(x\otimes w_i)^T (B_{\Omega, 0} - P_{\mathcal{L}}) (v_i\otimes w_i) = 0 \quad \text{and} \quad (v_i\otimes x)^T (B_{\Omega, 0} - P_{\mathcal{L}}) (v_i\otimes w_i) = 0.
\end{equation*}
Therefore, w.h.p. over the randomness of $\mathcal{V}$ and $\Omega$ the matrix  $\mathcal{Z}(B_{\Omega, 0} - P_{\mathcal{L}})$ is well-defined. 
\end{observation}
\begin{proof} Recall that $\mathcal{B}(\oa{A}_{\Omega}, \oa{A}_{\Omega}) = \tw_2\left(A_{\Omega}^TA_{\Omega}\right)$, so 
\[
\begin{gathered}
  (x\otimes w_i)^T (B_{\Omega, 0} - P_{\mathcal{L}}) (v_i\otimes w_i) = (x\otimes w_i)^T\left(A_{\Omega}^TA_{\Omega}\right)(v_i\otimes w_i) -  (x\otimes w_i)^T(v_i\otimes w_i) =\\
  = (x\otimes w_i)^TA_{\Omega}^T u_i - \langle x, v_i \rangle = \langle u_i\otimes x\otimes w_i, \oa{A}_{\Omega} \rangle - \langle x, v_i \rangle = 0.
  \end{gathered}
  \]
  The second equality follows from the symmetric argument. Finally, by Theorem~\ref{thm:zero-poly-correction-constr} w.h.p. $\mathcal{Z}(B_{\Omega, 0} - P_{\mathcal{L}})$ is well-defined. 
\end{proof}

Hence, to verify that $(\oa{A}_{\Omega}, B_{\Omega}, Z_{\Omega})$ is an SOS dual certificate, we only need to establish norm bounds for $B_{\Omega, 0}$ and $Z_{\Omega, 0}$. As before, we separate terms involved in $B_{\Omega, 0}$ and $Z_{\Omega, 0}$ into two groups: terms with an IP graph matrix structure and terms with a sufficiently small norm.

\subsection{Analysis for small terms}\label{sec:BOmega-small-terms}

Using Theorem~\ref{thm:AOmega-structure}, for $N\gg n^{3/2}m$, we can write $\oa{A}_{\Omega} = \oa{A}+\oa{A}_{\Omega, GM}+\oa{A}'_{\Omega, sm}$. Hence, we can write
\[B_{\Omega, 0}  = \mathcal{B}\left(\oa{A}+\oa{A}_{\Omega, GM}+\oa{A}'_{\Omega, sm}, \oa{A}+\oa{A}_{\Omega, GM}+\oa{A}'_{\Omega, sm}\right) = B_{\Omega, GM}+B_{\Omega, sm},\]
where we define
\begin{equation}\label{eq:BGM-defin}
\begin{split} 
 B_{\Omega, GM} = & \mathcal{B}\left(\oa{A}+\oa{A}_{\Omega, GM}, \oa{A}+\oa{A}_{\Omega, GM}\right) = \\
 = & B_{0}+\mathcal{B}\left(\oa{A}_{\Omega, GM}, \oa{A}\right)+\mathcal{B}\left(\oa{A}, \oa{A}_{\Omega, GM}\right)+\mathcal{B}\left(\oa{A}_{\Omega, GM}, \oa{A}_{\Omega, GM}\right),\qquad \text{and} \end{split}  
\end{equation}
\begin{equation}
 B_{\Omega, sm} =  \mathcal{B}\left(\oa{A}'_{\Omega, sm}, \oa{A}_{\Omega}\right)+\mathcal{B}\left(\oa{A}_{\Omega}, \oa{A}'_{\Omega, sm}\right) - \mathcal{B}\left(\oa{A}'_{\Omega, sm}, \oa{A}'_{\Omega, sm}\right)
 \end{equation}
 
 \begin{lemma}\label{lem:BOmega-sm-norm} Assume $m\ll n^{3/2}$ and $N\gg nm$. W.h.p. 
 $\Vert B_{\Omega, sm}\Vert_F = \wt{O}\left(\dfrac{m^2}{n^3}\right)$.
 \end{lemma}
 \begin{proof} Using Theorem~\ref{thm:AOmega-structure} and Corollary~\ref{cor:AOmega-norm-bound} we can bound
 \[ \Vert B_{\Omega, sm}\Vert_F \leq 2 \Vert \oa{A}'_{\Omega, sm} \Vert \cdot \Vert \oa{A}_{\Omega}\Vert +\Vert\oa{A}'_{\Omega, sm}\Vert^2 = \wt{O}\left(n\sqrt{\dfrac{nm}{N}}\right)\wt{O}\left(\dfrac{m^2}{n^4}\right) = \left(\dfrac{m^2}{n^3}\right).\]
 \end{proof}
 
 \begin{lemma}\label{lem:BOmega-GM-norm} Assume $m\ll n^{3/2}$ and $N\gg nm$. W.h.p. $\Vert B_{\Omega, GM}-B_0\Vert_F = \wt{O}\left(\dfrac{n^3m}{N}\right)$.
 \end{lemma}
\begin{proof} By Theorem~\ref{thm:A-infty-norm-bound}, $\Vert B_0\Vert\leq \Vert \oa{A} \Vert^2= \wt{O}(m) = \wt{O}\left(\dfrac{n^3m}{N}\right)$. Clearly,
\[ \Vert B_{\Omega, GM} \Vert_F \leq \Vert B_{\Omega, 0}\Vert_F - \Vert B_{\Omega, sm}\Vert_F \leq \Vert \oa{A}_{\Omega} \Vert^2 - \Vert B_{\Omega, sm}\Vert_F = \wt{O}\left(\dfrac{n^3m}{N}\right)\]
\end{proof} 
 
 Now, we can separate 
 \[ Z_{\Omega, 0} = \mathcal{Z}(B_{\Omega, 0} - P_{\mathcal{L}}) = \mathcal{Z}(B_{0} - P_{\mathcal{L}})+ \mathcal{Z}(B_{\Omega, GM}+ B_{\Omega, sm} - B_0),\]
where all the terms are well-defined by Observation~\ref{obs:ZBOmega-well-defined}, Lemma~\ref{lem:ZB0-welldefined} and linearity. Finally, we are going to decompose $\mathcal{Z}(B_{\Omega, 0} - B_0)$ as a sum of essential graph matrix terms and a small correction term. As in Section~\ref{sec:Z-small-terms}, we apply such decomposition on every step of the construction for $\mathcal{Z}(\cdot)$. Let
\[ \wt{D}_{GM} = \sum\limits_{i=1}^{m}((B_{\Omega, GM} - B_0)(v_i\otimes w_i))\otimes f_i \quad \text{and} \quad \wt{D}_{sm} = \sum\limits_{i=1}^{m}(B_{\Omega, sm}(v_i\otimes w_i))\otimes f_i. \]
Next, consider an approximation $Q^{[2t]}_{inv}$ to $(QQ^T)^{-1}_{\mathcal{N}^{\perp}}$ with $t = 4$ given by Eq.~\eqref{eq:Happrox-def}. Then

\[ \wt{Y} = (QQ^T)_{\mathcal{N}^{\perp}}^{-1}(\wt{D}_{GM}+\wt{D}_{sm}) = \wt{Y}_{GM}+\wt{Y}_{sm}, \quad \text{where}\]
\begin{equation}\label{eq:YZOmega-small}
\wt{Y}_{GM} = Q^{[2t]}_{inv}\wt{D}_{GM}\quad \text{and}\quad \wt{Y}_{sm} = \left((QQ^T)_{\mathcal{N}^{\perp}}^{-1} - Q^{[2t]}_{inv}\right)\wt{D}_{GM} + (QQ^T)_{\mathcal{N}^{\perp}}^{-1}\wt{D}_{sm}
\end{equation}
Let $Y_{GM}$ and $Y_{sm}$ be $n^2\times m$ matrices obtained by reshaping vectors $\wt{Y}_{GM}$ and $\wt{Y}_{sm}$. Finally, consider  $X_{GM} = Y_{GM}(V\cten W)^T$ and $X_{sm} = Y_{sm}(V\cten W)^T$ and define
\begin{equation}\label{eq:ZGM-defin}
\begin{gathered}
\mathcal{Z}_{GM}(B_{\Omega, GM} - B_0) = X_{GM}+X_{GM}^T-\tw_2\left(X_{GM}\right)-\tw_2\left(X_{GM}\right)^T\quad \text{and}\\  Z_{\Omega, sm} = X_{sm}+X_{sm}^T-\tw_2\left(X_{sm}\right)-\tw_2\left(X_{sm}\right)^T
\end{gathered}
\end{equation} 
Clearly,
 \[ \mathcal{Z}(B_{\Omega, 0} - B_0) = \mathcal{Z}_{GM}(B_{\Omega, GM} - B_0)+\mathcal{Z}_{\Omega, sm}.\]

 \begin{lemma} Assume $m\ll n^{3/2}$ and $N\gg nm$. W.h.p. 
 $\Vert Z_{\Omega, sm}\Vert_F = \wt{O}\left(\dfrac{m^2}{n^3}\right)$.
 \end{lemma}
 \begin{proof} First, note that by Lemma~\ref{lem:basic-norm-bounds}, w.h.p.
 \[  \Vert Z_{\Omega, sm}\Vert_F \leq 4\Vert X_{sm}\Vert_F\leq 4\Vert Y_{sm}\Vert_F\Vert V\cten W\Vert\leq 5 \Vert \wt{Y}_{sm}\Vert. \]
 At the same time, by Eq.~\eqref{eq:YZOmega-small}
 \[\Vert \wt{Y}_{sm}\Vert \leq \left\Vert (QQ^T)_{\mathcal{N}^{\perp}}^{-1} - Q^{[2t]}_{inv}\right\Vert\cdot \Vert \wt{D}_{GM}\Vert  + \left\Vert(QQ^T)_{\mathcal{N}^{\perp}}^{-1}\right\Vert\cdot \Vert \wt{D}_{sm}\Vert\]
 Using Lemma~\ref{lem:basic-norm-bounds}, Lemmas~\ref{lem:BOmega-sm-norm} and~\ref{lem:BOmega-GM-norm} imply 
 \[ \Vert \wt{D}_{GM}\Vert \leq \Vert B_{\Omega, GM} - B_0 \Vert_F\Vert V\cten W\Vert \leq \wt{O}\left(\dfrac{n^3m}{N}\right), \qquad \Vert \wt{D}_{sm}\Vert \leq \Vert B_{\Omega, sm} \Vert_F\Vert V\cten W\Vert \leq \wt{O}\left(\dfrac{m^2}{n^3}\right).\]
 Therefore, using that $N\gg nm$ and $m^2\ll n^3$, by Theorem~\ref{thm:QQT-approxim} and Eq.~\eqref{eq:Happrox-def} for $t=4$,
 \[ \Vert Z_{\Omega, sm}\Vert_F\leq \wt{O}\left(\dfrac{m^4}{n^8}\right)\wt{O}\left(\dfrac{n^3m}{N}\right)+\wt{O}\left(\dfrac{m^2}{n^3}\right) = \wt{O}\left(\dfrac{m^2}{n^3}\right).\]
 
 \end{proof}
 
 We summarize the discussion above into the following lemma.
 
 \begin{lemma}\label{lem:BOmega-small-term-analysis} Let $\oa{A}_{\Omega}$ be as in Theorem~\ref{thm:AOmega-structure}, $m\gg n^{3/2}$ and $N\gg n^{3/2}m$. Then  we can write
 \[ B_{\Omega, 0} = B_{\Omega, GM}+B_{\Omega, sm}\quad \text{and} \quad \mathcal{Z}(B_{\Omega, 0} - B_0) = \mathcal{Z}_{GM}(B_{\Omega, GM} - B_0)+\mathcal{Z}_{\Omega, sm}, \]
 where $\Vert \mathcal{Z}_{\Omega, sm} \Vert_F =  \wt{O}\left(\dfrac{m^2}{n^3}\right)$ and  $\Vert B_{\Omega, sm} \Vert_F =  \wt{O}\left(\dfrac{m^2}{n^3}\right)$ and matrices $B_{\Omega, GM}$ and $\mathcal{Z}_{GM}(B_{\Omega, GM}- B_0)$  are defined in Eq.~\eqref{eq:BGM-defin} and~\eqref{eq:ZGM-defin} (and so have IP graph matrix structure). 
 \end{lemma}

\subsection{Norm bounds for terms with graph matrix structure}\label{sec:BOmega-GM}
\subsubsection{Structure of the terms involved in {$B_{\Omega, GM} - B_0$} and {$\mathcal{Z}_{GM}(B_{\Omega, GM} - B_0)$}}

\begin{definition}\label{def:B-general-form} Let $X, Y\in \mathbb{R}^{n^3}$ and $M\in M_{n^5}(\mathbb{R})$. Define $\mathcal{B}(X, Y, M)$ to be the $n^2\times n^2$ matrix, whose $((b, c), (b', c'))$ entry for $(b, c), (b', c')\in [n]^2$ is defined as
\[ \mathcal{B}(X, Y, M)_{(b, c), (b', c')} = \sum\limits_{a^U, b^U, c^U, a^L, b^L, c^L = 1}^{n} X_{(a^U, b^U, c^U)}M_{(a^L, b, c', b^L, c^L) (a^U, b^U, c^U, b', c)}Y_{(a^L, b^L, c^L)}.\]
In the case when $M = I_{n^5}$, the definition agrees with Def.~\ref{def:B0-construction}, i.e., $\mathcal{B}(X, Y) = \mathcal{B}(X, Y, I_{n^5})$.
\end{definition}

The definition above becomes of great use because of the following observation. 

\begin{observation}\label{obs:structure-BOmega-terms} Let $K_N\in \mathbb{N}$ be as in Theorem~\ref{thm:AOmega-structure}. There exists an absolute constant $c_0$,  such that $(B_{\Omega, GM} - B_0)$ and $\mathcal{Z}_{GM}(B_{\Omega, GM} - B_0)$ can be written as signed sums of at most $c_0^{K_N}$ matrices of the form 
\begin{equation}\label{eq:Btl-definition}
 B_{t, \ell} = \mathcal{B}\left(\ol{R}_{\Omega_{t}}P^U_{t-1}\ol{R}_{\Omega_{t-1}}\ldots P_1^U\ol{R}_{\Omega_{1}}X^U,\ \ol{R}_{\Omega_{\ell}}P^L_{\ell-1}\ol{R}_{\Omega_{\ell-1}}\ldots P_1^L\ol{R}_{\Omega_{1}}X^L,\ P^M\right),
 \end{equation}
and the following properties hold.
\begin{enumerate}
\item Each matrix $P_{i}^U$ and $P_{j}^L$ is some bounded product of $\wt{P}_{uv}$, $\wt{P}_{uw}$, $\wt{P}_{vw}$ and $\wt{P}_{uvw}$. So, in particular, it is in the class $\mathfrak{G}_P$ described in Def.~\ref{def:class-GP}.
\item $t\geq 0$, $\ell\geq 0$ and $t+\ell\geq 1$.
\item $P_M$ is an $n^5\times n^5$ IP graph matrix with $type(\myscr{Ver}_L) = type(\myscr{Ver}_R) = (u, v, w, v, w)$. We assign the names to crosses in $\myscr{Ver}_L$ and $\myscr{Ver}_R$ as in Def.~\ref{def:B-general-form}. Let $\delta_u$, $\delta_v^U$, $\delta_w^U$, $\delta_v^L$ and $\delta_w^L$ be the indicator variables that crosses $a^L$, $b^U$, $c^U$, $b^L$ and $c^L$, respectively, are incident to an edge (i.e. are not in the intersection of $\myscr{Ver}_L$ and $\myscr{Ver}_R$). These indicators satisfy
\[\delta_u+\delta_v^U+\delta_w^U+\delta_v^L+\delta_w^L\neq 1\]
\item $P_M$ has at most 2 $\mathcal{C}_{2/3}$-connected components and is $\mathcal{C}_{2/3}$-boundary-connected for $\mathcal{C}_{2/3} = \{\{u, v\}, \{u, w\}, \{v, w\}\}$.
\item $X^U = X^L=\oa{A}_{GM}$, so in particular, they have at most 2 $\mathcal{C}_{2/3}$-connected components and are $\mathcal{C}_{2/3}$-boundary-connected for $\mathcal{C}_{2/3} = \{\{u, v\}, \{u, w\}, \{v, w\}\}$.
\end{enumerate}
\end{observation}

\begin{remark} The matrix $P^M$ takes into account possible operators $P^U_{t}$ and $P^L_{\ell}$ and the influence of the transformation through which matrix goes in Theorem~\ref{thm:zero-poly-correction-constr} to get a zero polynomial correction.

For the purpose of showing norm bounds for $B_{t, \ell}$, by taking transpositions,  we may assume that $t\geq \ell$.
\end{remark}

\begin{figure}
\begin{center}
\includegraphics[height = 2.5cm]{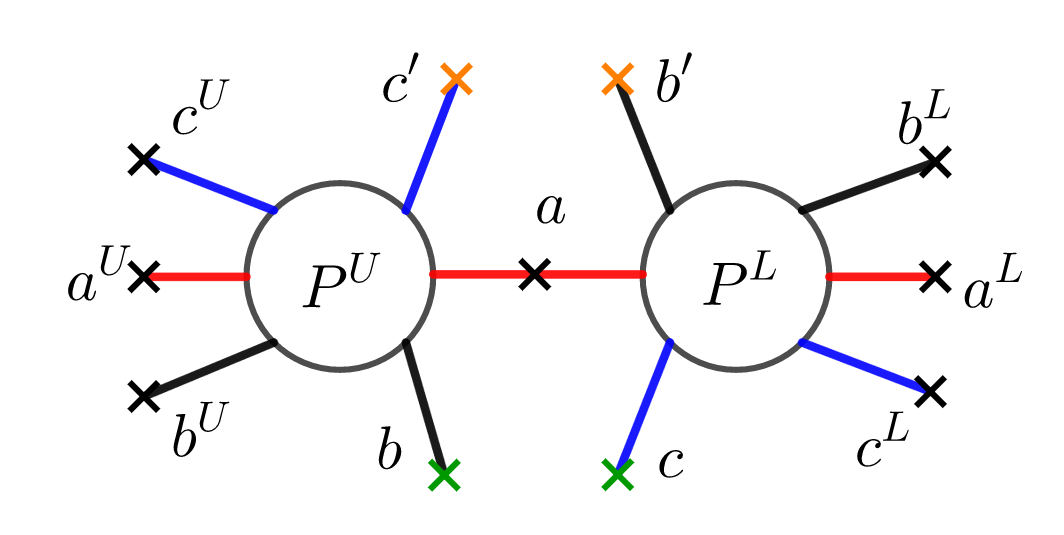}
\caption{Schematic matrix diagram for $P^M$}\label{fig:PM-diagram}
\end{center}
\end{figure}

\begin{proof}[Proof of Obs.~\ref{obs:structure-BOmega-terms}] First, we argue that the statement is true for $(B_{\Omega, GM} - B_0)$. By the definition of $B_{\Omega, GM}$ and Theorem~\ref{thm:AOmega-structure}, the matrix $(B_{\Omega, GM} - B_0)$ can be written as a sum of matrices of the form $B_{t, \ell}$, with 
\[P^M_{(a^L, b, c', b^L, c^L) (a^U, b^U, c^U, b', c)}  = \sum\limits_{a} P^{U}_{(a^U, b^U, c^U)(a, b, c')}P^{L}_{(a^L, b^L, c^L)(a, b', c)}\]
where $P^L$ and $P^U$ are in $\mathfrak{G}_P$ or are identities (see Def.~\ref{def:class-GP}). The first statement follows from Theorem~\ref{thm:AOmega-structure}. The last statement follows from Theorem~\ref{thm:corr-terms-summary}. Clearly, statements 2 and 3 hold as well, as in the case when $P_M$ is not identity, we have $\delta_u+\delta^U_v+\delta^U_w\geq 2$, or $\delta_u+\delta^L_v+\delta^L_w\geq 2$.  Moreover, it is easy to see from Figure~\ref{fig:PM-diagram} that $P^M$ has at most 2 $\mathcal{C}_{2/3}$-connected components and is $\mathcal{C}_{2/3}$-boundary-connected, since $P^U$ and $P^L$ are $\mathcal{C}_{2/3}$-connected and $\mathcal{C}_{2/3}$-boundary-connected.

Next, we consider $\mathcal{Z}_{GM}(B_{\Omega, GM} - B_0)$. We argue that each term involved in $\mathcal{Z}_{GM}(B_{\Omega, GM} - B_0)$ can be obtained from some terms involved in $(B_{\Omega, GM} - B_0)$ just by modifying $P^M$ in a proper way. Indeed, as can be seen from Figure~\ref{fig:BGMdiagram}, the transformations involved in the definition of $\mathcal{Z}( \cdot )$, other than multiplication by  $(QQ^T)_{\mathcal{N}^{\perp}}^{-1}$, only influence $P^M$ and preserve its IP graph matrix structure. Indeed, these transformations only "see" indices $b, c, b', c'$ and they are the crosses of $P^M$. Clearly, these transformations preserve $\mathcal{C}_{2/3}$-boundary-connectivity and the number of $\mathcal{C}_{2/3}$-connected components. Furthermore, in $\mathcal{Z}_{GM}$ we apply $Q_{inv}^{[8]}$ instead of $(QQ^T)_{\mathcal{N}^{\perp}}^{-1}$ and $Q_{inv}^{[8]}$ is a polynomial of $QQ^T$ and $B_{vw}$. In Lemma~\ref{lem:QQT-matrix-transform}, we proved that multiplication by these matrices preserves IP graph matrix structure, and moreover, preserves the desired colored connectivity properties (see Figure~\ref{fig:QQT-B-transform}).

Hence, properties 1, 2, 4, and 5 hold. We only need to verify that property 3 is satisfied. Note that if $P^M$ before the $\mathcal{Z}_{GM}$ transform was not identity, then $\delta_u+\delta_v^U+\delta_w^U+\delta_v^L+\delta_w^L\geq 2$ was satisfied. It is easy to see that any step of $\mathcal{Z}_{GM}$ construction cannot decrease the value of any $\delta\in \{\delta_u, \delta_v^U, \delta_w^U, \delta_v^L, \delta_w^L\}$. In the case when $P^M$ is identity, by Figure~\ref{fig:BGMdiagram}, we have $\delta^U_v=\delta^L_w=1$.   
\end{proof}

\subsubsection{Layered matrix diagram for {$B_{t, \ell}$} and its combinatorial properties}\label{sec:Btl-comb-prop}

Without loss of generality, let $t\geq \ell$.

To get a norm bound for $B_{t, \ell}$ we study its expanded matrix diagram, however, we additionally keep track of the indices at which $\ol{R}_{\Omega_i}$ are evaluated for each $i$. Thus we think of a matrix diagram of $B_{t, \ell}$ as having $t = \max(t, \ell)$ layers, corresponding to each $\Omega_i$, and some intermediate vertices and edges.

We assign labels to vertices of $\mathcal{MD}^*(B_{t, \ell})$ and $\mathcal{TD}^*_q(B_{t, \ell})$ in two stages. On the first stage we treat vertices from different layers which received the same label as being distinct. On the second stage we treat all vertices that received the same label as equal, no matter which layer they are at. The graph obtained on the first stage is convenient to use the randomness of $\Omega_i$, at the same time, to the graph on the second stage we apply an analog of Theorem~\ref{thm:main-diagram-tool} to use the randomness of $\mathcal{V}$.

Prior to describing the matrix diagram of $B_{t, \ell}$ we note that some graph matrices which are present in the approximation of $P_{\mathcal{S}}$, such as $\wt{P}_{uv}$ or $\wt{P}_{uw}$ have non-zero entries only if some indicies are equal to each other. That corresponds to the fact that for them $\myscr{Ver}_L\cap \myscr{Ver}_R \neq \emptyset$.

\begin{definition} Let $S\subseteq [d]$  and define $\pi:n^{[d]}\rightarrow n^{[d]\setminus S}$ be the projector on coordinates not in $S$. We say that a matrix $M$ is \emph{$S$-diagonal IP graph matrix} if there exists an IP graph matrix $M^*\in M_{d-|S|}(\mathbb{R})$ such that for any $x, y\in [n]^{[d]}$
\[ M_{xy} = M^*_{\pi(x)\pi(y)}.\]
In other words, $M$ is a tensor product of $M^*$ with an identity matrix in $S$ coordinates.

For $x\in [d]$, we say that $M$ is $x$-diagonal, if $M$ is $S$-diagonal and $x\in S$.
\end{definition}
  
Let $S_1^U, S_2^U, \ldots S_t^U, S_1^L, \ldots S_{\ell}^L$ be disjoint copies of the set $\{u, v, w\}$ and let $S$ be their disjoint union. We define an equivalence relation on the elements of $S$. For $*\in \{L, U\}$, if $P_i^*$ is $x$-diagonal for $x\in \{u, v, w\}$ we define the element of color $x$ in $S_i^*$ to be equivalent to the element of color $x$ in $S_{i-1}^*$. We also define the $u$-color element in $S^U_t$ to be equivalent to $u$-color element in $S^L_{\ell}$, if $P^M$ is $u$-diagonal.

Let $\myscr{Cr}^R$ be a set of equivalence classes of $S$ under the defined equivalence relation and let $\mu:S\rightarrow \myscr{Cr}^R$ be the map which sends an element to its equivalence class. We think of $\myscr{Cr}^R$ as being obtained from $S$ by deleting all but one element from every equivalence class, moreover, we keep the element which came from $S_i^*$ with the largest index. To capture this, for $i\leq t$, define $\myscr{Cr}^R_i$ to be the set of equivalence classes of $S$, which contain an element from $S^U_i\cup S^L_i$, but does not contain an element from $S^U_j\cup S^L_j$ for any $j>i$. Clearly, $\myscr{Cr}^R = \bigsqcup_{i=1}^{t} \myscr{Cr}^R_i$. 

If we look at the indicies at which $\ol{R}_{\Omega_j}$ are evaluated in the expression for $B_{t, \ell}$ in the order $j = t, t-1, \ldots, 1$, the vertices of $\myscr{Cr}^R_i$ correspond to the indicies at step $j=i$, which were not forced to be equal to indices from the previous steps. (Note that $P_{i}^U$ or $P_{i}^L$ may force some indices to be equal in order for term to be non-zero.) We treat $\myscr{Cr}^R_i$ as \textit{$i$-th layer} of $\mathcal{MD}^*(B_{t, \ell})$.     

\begin{observation} For $i\notin \{t, \ell\}$, we have $|\myscr{Cr}_i|\in \{2, 3, 4, 5, 6\}$.
\end{observation}

\begin{figure}
\begin{center}
\includegraphics[width = 0.9\textwidth]{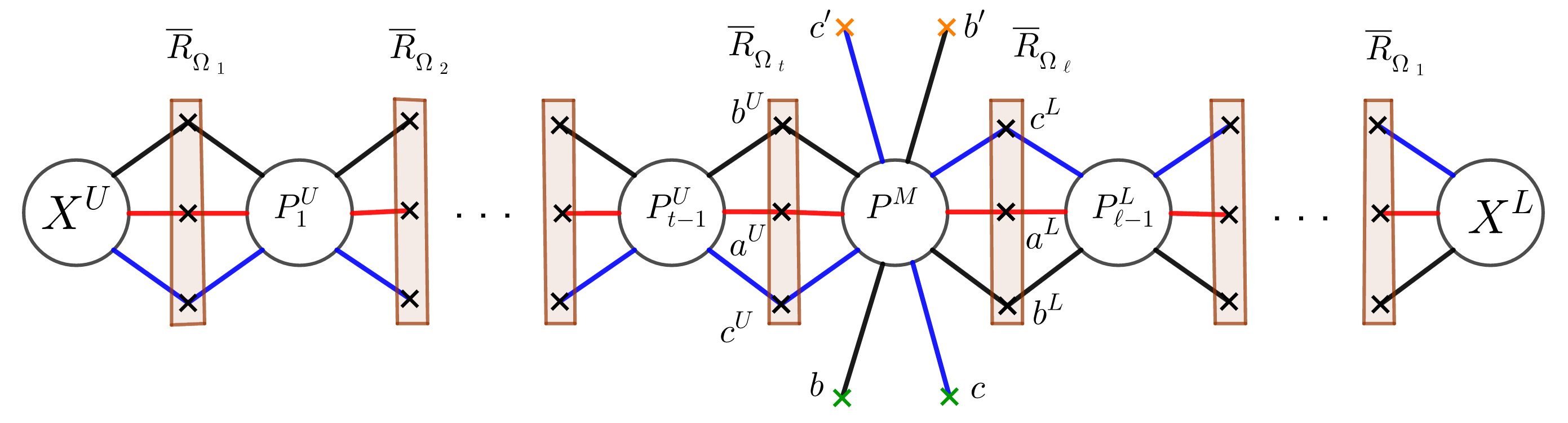}
\caption{Schematic representation of the matrix diagram $\Gamma$ for $B_{t, l}$}\label{fig:Gamma-giagram}
\end{center}
\end{figure}

The expanded matrix diagram $\Gamma$ of $B_{t, \ell}$ is a graph obtained from the expanded matrix diagrams of $X^U$, $X^L$, $\{P_i^L\mid i\in [\ell-1]\}$, $\{P_j^R\mid j\in [t-1]\}$ and $P^M$ by gluing them to $\myscr{Cr}^R$, so that $P_i^L$ is glued to $\mu(S_{i-1}^L)$ and $\mu(S_{i}^L)$, and $P^M$ is glued to $\mu(S_t^U)$ and $\mu(S_{\ell}^L)$, see Figure~\ref{fig:Gamma-giagram}. Denote by $\myscr{Cr}^{P}$ the crosses in $\Gamma$ that are not in $\myscr{Cr}^{R}$ and let $\myscr{Nod}$ denote the nodes of $\Gamma$. 

Clearly, $\Gamma$ is a bipartite graph with parts $\myscr{Cr} = \myscr{Cr}^P\sqcup \myscr{Cr}^R$ and $\myscr{Nod}$. Define $\myscr{Ver} = \myscr{Cr}\sqcup \myscr{Nod}$.

Now, we define a labeling of $\Gamma$, which ``remembers" the layers of $\Gamma$.

\begin{definition} We call $\phi: \myscr{Ver} \rightarrow ([n]\times \{u, v, w\}\times ([t]\cup P))\cup [m]$ \textit{a layered labeling} of $\Gamma$, where 
\begin{enumerate}
\item $\phi(\myscr{Cr}^R_i)\subseteq [n]\times \{u, v, w\}\times \{i\}$ and  $\phi(\myscr{Cr}^P)\subseteq [n]\times \{u, v, w\}\times \{P\}$;
\item for $x\in \myscr{Cr}$, the second coordinate of $\phi(x)$ indicates the color of $x$;
\item $\phi(\myscr{Nod})\subseteq [m]$.
\end{enumerate}  
Define $\Phi^{Lay}(\Gamma)$ to be the set of layered labelings of $\Gamma$.
\end{definition}  

Let $\mathcal{TD}_q(\Gamma):= \mathcal{TD}^*_q(B_{t, \ell})$ be the expanded matrix diagram of $\left(B_{t, \ell}B_{t, \ell}^T\right)^{q}$. It is obtained from $2q$ copies of $\Gamma$ by gluing them in a proper way, as discussed in Section~\ref{sec:trace-diagram}

\begin{definition} We say that $\phi: [2q]\rightarrow \Phi^{Lay}(\Gamma)$ is a layered labeling of $\mathcal{TD}_q(\Gamma)$, if $\phi(i)$ is defined on the $i$-th copy of $\Gamma$ in $\mathcal{TD}_q(\Gamma)$ and all $\phi(j)$ are consistent with gluing of copies of $\Gamma$ in $\mathcal{TD}_q(\Gamma)$.

We use the notation $\phi(i, x)$ to denote the image of $\phi(i)$ applied to $x\in \myscr{Ver}$.
\end{definition} 

The vertices of the graph $\mathcal{TD}_q(\Gamma)$ correspond to indicies   in the expansion of $\Tr\left(B_{t, \ell}B_{t, \ell}^T\right)^{q}$. 

\begin{definition}
For a layered labeling $\phi$ of $\mathcal{TD}_q(\Gamma)$ let $\mathcal{TD}_q(\Gamma, \phi)$ be the graph induced by $\phi$.
\end{definition}

For every layered labeling $\phi$ of $\Gamma$ we define $t = \max(t, \ell)$ hypergraphs $\mathcal{H}_i(\phi)$, which store the triples at which $\ol{R}_{\Omega_i}$ is evaluated in the term that corresponds to $\phi$ in the expansion of $\Tr\left(B_{t, \ell}B_{t, \ell}^T\right)^{q}$.

Define $\mathcal{H}_i(\phi)$ to be a 3-uniform hypergraph, whose hyperedges are 3-element sets $\phi(j, \mu(S_i^U))$ and $\phi(j, \mu(S_i^L))$ (if $i\leq \ell$) for $j\in [2q]$. Note that every vertex in the hypergraph $\mathcal{H}_i(\phi)$ has its last coordinate being $\geq i$ and every hyperedge of $\mathcal{H}_i(\phi)$ has precisely one vertex with second coordinate being $u$, $v$ and $w$, respectively. Thus we may identify each hyperedge of $\mathcal{H}_i(\phi)$ with a triple from $[n]^3$ by ordering first coordinates of its vertices with respect to $(u, v, w)$ order on the second coordinates. Hence, the expression $\ol{R}_{\Omega_i}(y)$ is well-defined for a hyperedge $y$ of $\mathcal{H}_i(\phi)$.

\begin{definition}\label{def:hi} Define $h_i(\phi)$ to be the number of distinct hyperedges of $\mathcal{H}_i(\phi)$ and define $\chi R_i(\phi)$ to be the number of vertices of $\mathcal{H}_i(\phi)$. Let $h(\phi) = \sum\limits_{i=1}^{t} h_i(\phi)$.
\end{definition} 

\begin{definition} Define $eR_i(\phi)$ to be the number of edges of $\mathcal{TD}_q(\Gamma, \phi)$ incident with the vertices from the set $\bigcup\limits_{j\in [2q]} \phi(j, \myscr{Cr}^{R}_i)$.
\end{definition}  

\begin{lemma}\label{lem:eRi-set-size} Let $i<t$ and $i\neq \ell$. Then $eR_i(\phi) = 2|\myscr{Cr}^R_i|\cdot (2q)$.
\end{lemma}
\begin{proof}  Note that for any map $f$ and any graph $G$, the number of edges incident with vertex $x$ in the induced graph $G(f)$ equals to the number of edges incident with vertices in $f^{-1}(x)$ in $G$. The third (layer) coordinate of $\phi$ assures that any vertex not in $\myscr{Cr}^R_i$ cannot have the same image as a vertex in $\myscr{Cr}^R_i$. Since $i\neq t$ and $i\neq \ell$ all $2q$ copies of $\myscr{Cr}^R_i$ in $\mathcal{TD}_q(\Gamma)$ are disjoint. 

Finally, note that any vertex in $\myscr{Cr}^R_i$ is incident with precisely 1 edge which belongs to $\{P^U_i, P^L_i\}$, and precisely one edge which belongs to $\{P^U_{i-1}, P^L_{i-1}, \ldots, P^U_1, P^L_1, X^U, X^L\}$. 
\end{proof} 

\begin{definition}
Let $\Phi_0^{Lay}$ be the family of layered labelings $\phi$ such that every hyperedge in the hypergraph $\mathcal{H}_i(\phi)$ appears at least twice.
\end{definition}

One of the key parts of the analysis is to get efficient bounds on $\chi R_i(\phi)$ in terms of $h_i(\phi)$ for $\phi \in \Phi_0$. We will look into graphs $H^{i}_{uv}$, $H^{i}_{uw}$ and $H^{i}_{vw}$ induced by the hyperedges of $\mathcal{H}_i$ on the vertices  only of colors in $\{u, v\}$ , colors in $\{u, w\}$ , or colors in $\{v, w\}$,  respectively.

\begin{lemma}\label{lem:induced-graphs-obv-bounds} Let $G = (V, E)$ be a graph and $\phi: V\rightarrow V'$ be a map such that for any edge $e$ in the graph $G(\phi)$ induced by $\phi$ the preimage $\phi^{-1}(e)$ has size at least 2. Denote by $ver(\phi)$ the number of vertices of $G(\phi)$ and by $ed(\phi)$ the number of distinct edges in $G(\phi)$.
\begin{enumerate} 
\item $ver(\phi)\leq 2ed(\phi)$.
\item Let $d$ be the number of connected components of $G$. Then $ver(\phi)\leq ed(\phi)+d$.
\item Assume that $G$ is a disjoint union of a cycle of length $2q$ and $2q$ disconnected edges. Then $ver(\phi)\leq q+1+ed(\phi)$. 
\end{enumerate} 
\end{lemma}
\begin{proof} The first inequality is obvious as every edge is incident with at most 2 vertices. 

The second inequality holds since the number of connected components of $G(\phi)$ is at most the number of connected components of $G$.

Finally, for the third part note that since every edge in the image appears at least twice, there could be at most $q+1$ connected components in $G(\phi)$.
\end{proof}

\subsubsection{Bound without taking layers into account}

Let $\pi: ([n]\times \{u, v, w\}\times ([t]\cup P))\cup [m]\rightarrow ([n]\times \{u, v, w\})\cup [m]$ be a map defined as 
\[ \begin{cases}
\pi(x) = x  & \text{ for } x\in [m],\\
\pi((x, c, j)) = (x, c) & \text{ for } (x, c, j) \in [n]\times \{u, v, w\}\times ([t]\cup P).
\end{cases}
\] 

Then $\mathcal{TD}_q(\Gamma, \pi\phi))$ is a graph, which forgets about the layers and treats all vertices which have the same label and color as being equal. 

Define 
\begin{equation}
\begin{split}
 \myscr{Cr}^{\pi\phi, R} = \bigcup\limits_{j\in [2q]} \pi\left(\phi(j, \myscr{Cr}^R)\right)\quad &\text{and} \quad \chi R(\phi) = \left|\myscr{Cr}^{\pi\phi, R}\right|, \\
 \myscr{Cr}^{\pi\phi, P} = \bigcup\limits_{j\in [2q]} \pi\left(\phi(j, \myscr{Cr}^P)\right)\quad &\text{and} \quad \chi(\phi) = \left|\myscr{Cr}^{\pi\phi, P}\setminus \myscr{Cr}^{\pi\phi, R}\right|, \\
 \myscr{Nod}^{\pi\phi} = \bigcup\limits_{j\in [2q]} \pi\left(\phi(j, \myscr{Nod})\right)\quad &\text{and} \quad \myscr{nod}(\phi) = \left|\myscr{Nod}^{\pi\phi}\right|,\\
 \myscr{Cr}^{\pi\phi, R} = \myscr{Cr}^{\pi\phi, R}\cup \myscr{Cr}^{\pi\phi, P} \quad &\text{and} \quad \myscr{cr}(\phi) = \left|\myscr{Cr}^{\pi\phi}\right|.
 \end{split}
 \end{equation}
 
 \begin{lemma}\label{lem:num-of-edges-withR}  Let $e(\phi)$ be the number of edges of $\mathcal{TD}_q(\Gamma, \pi\phi))$, then $e(\phi)\geq 2\chi(\phi)+\sum\limits_{i=1}^{t} eR_i(\phi)$. The number of vertices of $\mathcal{TD}_q(\Gamma, \pi\phi))$ is equal $\chi R(\phi)+\chi(\phi)+\myscr{nod}(\phi)$.
 \end{lemma}

Let $val(\mathcal{TD}_q(\Gamma, \pi\phi))$ be defined as in Lemma~\ref{lem:matix-diagram-tool-expanded}. Consider

\[ val_R(\mathcal{TD}_q(\Gamma, \phi)) = val(\mathcal{TD}_q(\Gamma, \pi\phi))\prod_{i=1}^{t} \prod_{y\in \mathcal{H}_i}\ol{R}_{\Omega_i}(y)\]

\begin{theorem}\label{thm:val-gamma-bound} W.h.p. over the randomness of $\mathcal{V}$, the quantities above satisfy the inequality
\[\left|\mathbb{E}\left[val(\mathcal{TD}_q(\Gamma, \phi)\right]\right| = \wt{O}\left(\dfrac{1}{n}\right)^{\displaystyle \max(\myscr{cr}(\phi)+3\myscr{nod}(\phi)/2-3,\ e(\phi)/2)}\] 
\end{theorem}
\begin{proof} The assumptions on $P_i^L$, $P_i^U$, $P^M$ and $X^L, X^U$  (see Obs.~\ref{obs:structure-BOmega-terms}) imply that $\mathcal{TD}_q(\Gamma)$ has at most 2 $\mathcal{C}$-connected components for $\mathcal{C} = \{\{u, v\}, \{u, w\}, \{v, w\}\}$. Thus, $\mathcal{TD}_q(\Gamma, \phi)$ has at most 2 $\mathcal{C}$-connected components. Recall, that this means that the graph $G_{C}$ induced by edges of colors in $C$ for $C\in \mathcal{C}$ has at most 2 connected components. Denote by $k_u$, $k_v$ and $k_w$ the number of crosses in $\mathcal{TD}_q(\Gamma, \phi)$ of the colors $u$, $v$ and $w$, respectively. Also, denote by $e_u$, $e_v$ and $e_w$ the number of edges of colors $u$, $v$ and $w$, respectively. Then $G_{\{u, v\}}$ has $k_u+k_v+\myscr{nod}(\phi)$ vertices and $e_u+e_v$ edges.

By Lemma~\ref{lem:matix-diagram-tool-expanded}, we obtain
\[\mathbb{E}\left[val(G_{u, v})\right] = \left(\wt{O}\left(\dfrac{1}{n}\right)\right)^{\displaystyle \max(k_u+k_v+\myscr{nod}(\phi)-2,\ (e_u+e_v)/2)}\]
We also get symmetric bounds for $G_{u, w}$ and $G_{v, w}$. Moreover, using independence, as in Theorem~\ref{thm:main-diagram-tool}, we have
\[\left(\mathbb{E}\left[val(\mathcal{TD}_q(\Gamma, \phi)\right]\right)^2 = \mathbb{E}\left[val(G_{u, v})\right]\mathbb{E}\left[val(G_{u, w})\right]\mathbb{E}\left[val(G_{v, w})\right]\]
Therefore,
\[\left|\mathbb{E}\left[val(\mathcal{TD}_q(\Gamma, \phi)\right]\right| = \wt{O}\left(\dfrac{1}{n}\right)^{\displaystyle \max(\myscr{cr}(\phi)+3\myscr{nod}(\phi)/2-3,\ e(\phi)/2)}\]
 
\end{proof}

\begin{corollary}\label{cor:valr-gamma-bound} W.h.p. over the randomness of $\mathcal{V}$ and $\Omega$,
\[ \left\vert \mathbb{E}\left[ val_R(\mathcal{TD}_q(\Gamma, \phi))\right] \right\vert \leq \left(\dfrac{n^3}{N}\right)^{\displaystyle 2(t+\ell)q-h(\phi)}\left(\wt{O}\left(\dfrac{1}{n}\right)\right)^{\displaystyle \max(\myscr{cr}(\phi)+3nod(\phi)/2-3, e(\phi)/2)} \]
\end{corollary}
\begin{proof} Follows from Theorem~\ref{thm:val-gamma-bound} and Lemma~\ref{lem:ROmega-moments}.
\end{proof}

\begin{definition} For $0<cr\leq 2q|\myscr{Cr}|$ and $0<nod\leq 2q|\myscr{Nod}|$, define $\Phi(cr, nod)$ to be the set of layered labelings $\phi$ of $\mathcal{TD}_q(\Gamma)$ with $\myscr{cr}(\phi) = cr$ and $\myscr{nod}(\phi) = nod$. 
\end{definition}

\begin{lemma}  The size of $\Phi(cr, nod)$ is at most $(2q|\myscr{Ver}|)^{2q|\myscr{Ver}|}n^{cr}m^{nod}$.
\end{lemma}
\begin{proof} There are at most $n^{cr}m^{nod}$ possibilities to chose $cr$ distinct values for crosses and $nod$ distinct values nodes. The graph $\mathcal{TD}_q(\Gamma)$ has at most $(2q|\myscr{Ver}|)$ vertices. Finally, there are at most $(cr+nod)^{2q|\myscr{Ver}|}\leq (2q|\myscr{Ver}|)^{2q|\myscr{Ver}|}$ possibilities to assign values to vertices of a graph.
\end{proof}

\begin{theorem}\label{thm:BOmega-bound-without-layers} Let $m\ll n^{3/2}$. Then w.h.p. over the randomness of $\mathcal{V}$ and $\Omega$,
\[ \mathbb{E}\left[\Tr\left(B_{t, \ell}B_{t, \ell}^T\right)^q\right]  =  (2q|\myscr{Ver}|)^{2q|\myscr{Ver}|+2} \max\limits_{\phi\in \Phi_0^{Lay}}\left( \left(\dfrac{n^3}{N}\right)^{ 2(t+\ell)q-h(\phi)}\left(\dfrac{m^{2/3}}{n}\right)^{ eR(\phi)/2}\left(\dfrac{n}{m^{2/3}}\right)^{\chi R(\phi)}  \right)\]
\end{theorem}
\begin{proof}
Recall that $\Phi_0^{Lay}$ is the family of layered labelings $\phi$ such that every hyperedge in the image on layer $i$ appears at least twice. Using Lemma~\ref{lem:ROmega-moments},

\begin{equation}\label{eq:expected-value-expansion-completion}
\begin{gathered}
 \mathbb{E}\left[\Tr\left(B_{t, \ell}B_{t, \ell}^T\right)^q\right]  = \sum\limits_{\phi\in \Phi_0^{Lay}} \mathbb{E}\left[ val_R(\mathcal{TD}_q(\Gamma, \phi)) \right] \leq \\
 \leq \sum\limits_{cr = 1}^{2q|\myscr{Cr}|}\sum\limits_{nod = 1}^{2q|\myscr{Nod}|} |\Phi(cr, nod)\cap \Phi_0^{Lay}|\cdot \max\limits_{\phi\in \Phi_0\cap \Phi(cr, nod)} |\mathbb{E}\left[ val_R(\mathcal{TD}_q(\Gamma, \phi))\right]|\leq \\
 \leq (2q|\myscr{Ver}|)^{2q|\myscr{Ver}|+2}\cdot \max\limits_{\phi\in \Phi_0^{Lay}} \left( n^{\myscr{cr}(\phi)}m^{\myscr{nod}(\phi)} \left| \mathbb{E}\left[ val_R(\mathcal{TD}_q(\Gamma, \phi)) \right] \right|\right)
 \end{gathered}
 \end{equation}
 
Using Corollary~\ref{cor:valr-gamma-bound} we can upper bound the expression inside the maximum  as
\[n^{\displaystyle \myscr{cr}(\phi)}m^{\displaystyle \myscr{nod}(\phi)} \cdot  \left(\dfrac{n^3}{N}\right)^{\displaystyle 2(t+\ell)q-h(\phi)}\left(\dfrac{1}{n}\right)^{\displaystyle \max(\myscr{cr}(\phi)+3\myscr{nod}(\phi)/2-3, e(\phi)/2)}\]
By Lemma~\ref{lem:num-of-edges-withR}, $e(\phi)/2\geq eR(\phi)/2+\chi(\phi)$, and by definition, $\myscr{cr}(\phi) = \chi R(\phi)+\chi(\phi)$, so the expression above is upper bounded by 
\[ n^{\displaystyle \chi R(\phi)}m^{\displaystyle \myscr{nod}(\phi)} \cdot  \left(\dfrac{n^3}{N}\right)^{\displaystyle 2(t+\ell)q-h(\phi)}\left(\dfrac{1}{n}\right)^{\displaystyle \max(\chi R(\phi)+3\myscr{nod}(\phi)/2-3, eR(\phi)/2)}\]
Since $m\ll n^{3/2}$, this expression is maximized when 
\[ \myscr{nod}(\phi) = \dfrac{2}{3}\left(\dfrac{eR(\phi)}{2} - \chi R(\phi)\right)+2 \]
Hence, the expression inside the maximum of Eq.~\eqref{eq:expected-value-expansion-completion}, is upper bounded by
\[ \max\limits_{\phi\in \Phi_0^{Lay}}\left( \left(\dfrac{n^3}{N}\right)^{\displaystyle 2(t+\ell)q-h(\phi)}\left(\dfrac{1}{n}\right)^{\displaystyle eR(\phi)/2}n^{\displaystyle \chi R(\phi)}m^{\displaystyle eR(\phi)/3-2 \chi R(\phi)/3}  \right)\]
Therefore, by recombining the terms we obtain the statement of the theorem.
\end{proof}

\subsubsection{Bounds using analysis for layers}

Define 
\begin{equation}\label{eq:xRi-def}
xR_i(\phi) = \left|\bigcup\limits_{j\in [2q]} \pi\left(\phi(j, \myscr{Cr_i}^R)\right)\setminus \left(\bigcup\limits_{k = i+1}^{t}\bigcup\limits_{j\in [2q]} \pi\left(\phi(j, \myscr{Cr}_k^R)\right)\right)\right|
\end{equation}
Clearly, $xR_i(\phi)\leq \chi R_i(\phi)$ and the following equalities hold.
\[ h(\phi) = \sum\limits_{i=1}^{t} h_i(\phi),\qquad \chi R(\phi) = \sum\limits_{i=1}^t xR_i(\phi),\qquad eR(\phi)=\sum\limits_{i=1}^{t} eR_i(\phi) \]
Denote $s_i = 1+\mathbf{1}[i\leq \ell]$, then the maximum in Theorem~\ref{thm:BOmega-bound-without-layers} can be rewritten as
\[ \max\limits_{\phi\in \Phi_0^{Lay}}\left(\prod\limits_{i=1}^{t} \left(\dfrac{n^3}{N}\right)^{\displaystyle 2s_iq-h_i(\phi)}\left(\dfrac{m^{2/3}}{n}\right)^{\displaystyle eR_i(\phi)/2}\left(\dfrac{n}{m^{2/3}}\right)^{\displaystyle xR_i(\phi)}   \right).\]

\begin{theorem}\label{thm:main-layer-analysis}
Assume $N\gg n^{3/2}m$ and $n^{3/2}\gg m$. For every $\phi \in \Phi_0^{Lay}$, 
\[ \prod\limits_{i=1}^{t}\left(\dfrac{n^3}{N}\right)^{\displaystyle 2s_iq-h_i(\phi)}\left(\dfrac{m^{2/3}}{n}\right)^{\displaystyle eR_i(\phi)/2}\left(\dfrac{n}{m^{2/3}}\right)^{\displaystyle xR_i(\phi)} \leq \left(\dfrac{n}{m^{2/3}}\right)^{2}\left(\dfrac{n^{3/2}m}{N}\right)^{(t+\ell)q}. \]
\end{theorem}
\begin{proof} Since $m<n^{3/2}$, it is sufficient to bound this expression with $xR_i(\phi)$ being replaced with $\chi R_i(\phi)$, and for most $i$ that is what we will do.

Fix the map $\phi\in \Phi_0^{Lay}$. To make our expressions shorter we use the notation defined above without mentioning the dependence on $\phi$.

Denote by $H^i_{u, v}$, $H^i_{u, w}$ and $H^i_{v, w}$ the graphs induced by the hypergraph $\mathcal{H}_i$ on the vertices of colors $\{u, v\}$, $\{u, w\}$ and $\{v, w\}$, respectively. Denote by $x_u^i$, $x_v^i$ and $x_w^i$ the number of vertices of corresponding color in $\mathcal{H}_i$ with last coordinate $i$. By definition, $xR_i = x_u^i+x_v^i+x_w^i$.

We study the desired expression for each $i$ separately. We consider four major cases: 1) $i\notin\{t, \ell\}$, 2) $i = t = \ell$, 3) $i = \ell$ for $t\neq \ell$ and 4) $i=t$ for $t\neq \ell$.

\noindent\textbf{Case 1.} Assume that $i\notin\{t, \ell\}$.

\noindent\textbf{Case 1.a} Assume $|\myscr{Cr}_i| = 5$. In this case, $s_i = 2$ and by Lemma~\ref{lem:eRi-set-size}, $eR_i/2 = 10q$. Moreover, clearly, $\chi R_i \leq 3h_i$ and $h_i\leq 2q$. Thus
\begin{equation}
\begin{gathered}
 \left(\dfrac{n^3}{N}\right)^{4q-h_i}\left(\dfrac{m^{2/3}}{n}\right)^{eR_i/2}\left(\dfrac{n}{m^{2/3}}\right)^{\chi R_i} \leq \left(\dfrac{n^3}{N}\right)^{4q-h_i}\left(\dfrac{m^{2/3}}{n}\right)^{10q}\left(\dfrac{n}{m^{2/3}}\right)^{3h_i} \leq \\
 \leq \left(\dfrac{\sqrt{n}m^{5/3}}{N}\right)^{4q}\left(\dfrac{N}{m^2}\right)^{h_i}\leq \left(\dfrac{nm^{4/3}}{N}\right)^{2q}.
 \end{gathered}
 \end{equation}
 
\noindent\textbf{Case 1.b} Assume $|\myscr{Cr}_i| \neq 5$. Denote $k = |\myscr{Cr}_i|/s_i$. In this case, $\chi R_i\leq k h_i$.  Again, $h_i\leq s_i q$, and by Lemma~\ref{lem:eRi-set-size}, $eR_i/2 = 2|\myscr{Cr}_i|q = 2s_i kq$. Thus
\begin{equation}\label{eq:boundcase1b-interm}
\begin{gathered}
 \left(\dfrac{n^3}{N}\right)^{2s_iq-h_i}\left(\dfrac{m^{2/3}}{n}\right)^{eR_i/2}\left(\dfrac{n}{m^{2/3}}\right)^{\chi R_i} \leq \left(\dfrac{n^3}{N}\right)^{2s_iq-h_i}\left(\dfrac{m^{2/3}}{n}\right)^{2s_ikq}\left(\dfrac{n}{m^{2/3}}\right)^{kh_i} \leq \\
 \leq \left(\dfrac{n^{3-k}m^{2k/3}}{N}\right)^{2s_iq}\left(\dfrac{N}{m^{2k/3}n^{3-k}}\right)^{h_i}.
 \end{gathered}
 \end{equation}
 Since $k\in \{2, 3\}$ in this case, and since $m<n^{3/2}$ and $n^{3/2}m<N$ we have 
 \[ 1<\dfrac{N}{n^{3/2}m}\leq \dfrac{N}{nm^{4/3}} \leq \dfrac{N}{m^{2k/3}n^{3-k}}\leq \dfrac{N}{m^2}.\]
 Thus, the last expression in Eq.~\eqref{eq:boundcase1b-interm} is maximized when $h_i = s_iq$. Hence,
 \begin{equation}\label{eq:boundcase1b}
 \left(\dfrac{n^3}{N}\right)^{2s_iq-h_i}\left(\dfrac{m^{2/3}}{n}\right)^{eR_i/2}\left(\dfrac{n}{m^{2/3}}\right)^{\chi R_i} \leq \left(\dfrac{nm^{4/3}}{N}\right)^{s_iq}.
 \end{equation}
 
\noindent\textbf{Case 2.} Assume that $i= t= \ell$. \\ In this case, $s_t = 2$ and $h_t\leq 2q$.  Using notation from Obs.~\ref{obs:structure-BOmega-terms}, it is easy to see that 
\[ eR_t = (6+2\delta_u+\delta^U_v+\delta^U_w+\delta^L_v+\delta^L_w)\cdot 2q\]
By Obs.~\ref{obs:structure-BOmega-terms}, $2\delta_u+\delta^U_v+\delta^U_w+\delta^L_v+\delta^L_w\neq 1$. We claim that 
\[ \chi R_t \leq \begin{cases}
 3h_t\quad & \text{if}\quad 2\delta_u+\delta^U_v+\delta^U_w+\delta^L_v+\delta^L_w \geq 3 \\
 2h_t+\min(h_t, q+1) \quad & \text{if}\quad 2\delta_u+\delta^U_v+\delta^U_w+\delta^L_v+\delta^L_w = 2\\ 
 3(h_t+1)/2 \quad & \text{if}\quad 2\delta_u+\delta^U_v+\delta^U_w+\delta^L_v+\delta^L_w = 0
\end{cases}\]
\noindent\textbf{Case 2.a} $2\delta_u+\delta^U_v+\delta^U_w+\delta^L_v+\delta^L_w \geq 3$. \\
In this case, $eR_t \geq 18q$. The bound $\chi R_t\leq 3h_t$ is obvious. Thus
\begin{equation}\label{eq:xRi-bound-18}
\begin{gathered}
 \left(\dfrac{n^3}{N}\right)^{4q-h_t}\left(\dfrac{m^{2/3}}{n}\right)^{eR_t/2}\left(\dfrac{n}{m^{2/3}}\right)^{\chi R_t} \leq \left(\dfrac{n^3}{N}\right)^{4q-h_t}\left(\dfrac{m^{18q/3}}{n^{9q}}\right)\left(\dfrac{n}{m^{2/3}}\right)^{3h_t} \leq \\
 \leq \left(\dfrac{n^{3/2}m^{3}}{N^2}\right)^{2q}\left(\dfrac{N}{m^2}\right)^{h_t}\leq \left(\dfrac{n^{3/2}m}{N}\right)^{2q}.
 \end{gathered}
 \end{equation}  

\noindent\textbf{Case 2.b} $2\delta_u+\delta^U_v+\delta^U_w+\delta^L_v+\delta^L_w = 2$. \\
In this case, $eR_t= 16q$. Moreover, by part 3 of Obs.~\ref{obs:structure-BOmega-terms}, in this case $\delta_u = 0$. Note also, among the sums $\delta_v^U+\delta_v^L$ and $\delta_w^U+\delta_w^L$ at least one is $\leq 1$. Without loss of generality, assume that $\delta_v^U+\delta_v^L\leq 1$. Then the graph $H^t_{u, v}$ has at most $q$ connected components. Thus, by Lemma~\ref{lem:induced-graphs-obv-bounds} we have $x_u^t+x_v^t\leq \min(q+1+h_t, 2h_t)$. There are at most $h_t$ vertices of color $w$ in $\mathcal{H}_t$.

 Thus, $\chi R_t\leq 2h_t+\min(h_t, q+1)$ and we can bound

\begin{equation}\label{eq:xRi-bound-16}
\begin{gathered}
 \left(\dfrac{n^3}{N}\right)^{4q-h_t}\left(\dfrac{m^{2/3}}{n}\right)^{eR_t/2}\left(\dfrac{n}{m^{2/3}}\right)^{\chi R_t} \leq \left(\dfrac{n^3}{N}\right)^{4q-h_t}\left(\dfrac{m^{16q/3}}{n^{8q}}\right)\left(\dfrac{n}{m^{2/3}}\right)^{\chi R_t} \leq \\
 \leq \left(\dfrac{nm^{4/3}}{N}\right)^{4q}\left(\dfrac{N}{nm^{4/3}}\right)^{h_t}\left(\dfrac{n}{m^{2/3}}\right)^{\min(h_t, q+1)}
 \leq   \left(\dfrac{nm^{4/3}}{N}\right)^{2q}\left(\dfrac{n}{m^{2/3}}\right)^{q+1} \leq\\
 \leq \left(\dfrac{n^{3/2}m}{N}\right)^{2q}\left(\dfrac{n}{m^{2/3}}\right),
 \end{gathered}
 \end{equation}  
 where to deduce the inequality in the middle line we use that $N>m^2$, so $h_t\geq q+1$ is optimal, and then, since $N>nm^{4/3}$ and $h_t\leq 2q$, $h_t = 2q$ is optimal.

\noindent\textbf{Case 2.c} $2\delta_u+\delta^U_v+\delta^U_w+\delta^L_v+\delta^L_w = 0$. \\
In this case, $eR_t= 12q$ and each graph $H^t_{u, v}$, $H^t_{u, w}$ and $H^t_{v, w}$ is connected. Thus, by Lemma~\ref{lem:induced-graphs-obv-bounds},
\[x^t_u+x^t_v\leq h_t+1, \qquad x^t_u+x^t_w\leq h_t+1, \qquad x^t_v+x^t_w\leq h_t+1,\]
and so the bound $\chi R_t\leq 3(h_t+1)/2$ holds. Hence,
\begin{equation}
\begin{gathered}
 \left(\dfrac{n^3}{N}\right)^{4q-h_t}\left(\dfrac{m^{2/3}}{n}\right)^{eR_t/2}\left(\dfrac{n}{m^{2/3}}\right)^{\chi R_t} \leq \left(\dfrac{n^3}{N}\right)^{4q-h_t}\left(\dfrac{m^{4q}}{n^{6q}}\right)\left(\dfrac{n}{m^{2/3}}\right)^{3(h_t+1)/2} \leq \\
 \leq \left(\dfrac{n^{3/2}m}{N}\right)^{4q}\left(\dfrac{N}{n^{3/2}m}\right)^{h_t}\left(\dfrac{n^{3/2}}{m}\right)\leq \left(\dfrac{n^{3/2}m}{N}\right)^{2q}\left(\dfrac{n^{3/2}}{m}\right)
 \end{gathered}
 \end{equation} 
 
 \noindent\textbf{Case 3.} Assume $i= \ell$ and $t\neq \ell$. In this case, $s_{\ell} = 2$ and $h_{\ell}\leq 2q$.
 
 Let $\delta_P$ be the indicator of the event that $P^U_{\ell}$ has edges of all three colors. Then
 \[ eR_{\ell} = \left((4+2\delta_P)+2+2\delta_u+\delta_w^L+\delta_v^L\right)\cdot 2q\]
We claim that
 \[ xR_{\ell} \leq \begin{cases}
 3h_{\ell}\quad & \text{if}\quad 2\delta_P+2\delta_u+\delta_w^L+\delta_v^L\geq 3 \\
 2h_{\ell}+\min(h_{\ell}, q) \quad & \text{if}\quad 2\delta_P+2\delta_u+\delta_w^L+\delta_v^L = 2\\ 
 h_{\ell}+2\min(h_{\ell}, q) \quad & \text{if}\quad 2\delta_P+2\delta_u+\delta_w^L+\delta_v^L = 1 \\
 h_{\ell}+\min(h_{\ell}, q+1) \quad & \text{if}\quad 2\delta_P+2\delta_u+\delta_w^L+\delta_v^L = 0
\end{cases}\]

\noindent\textbf{Case 3.a} $2\delta_P+2\delta_u+\delta_w^L+\delta_v^L\geq 3$. \\
In this case $eR_{\ell}\geq 18q$, and clearly $xR_{\ell}\leq 3h_{\ell}$. So as in Eq.~\eqref{eq:xRi-bound-18},
\[ \left(\dfrac{n^3}{N}\right)^{4q-h_{\ell}}\left(\dfrac{m^{2/3}}{n}\right)^{eR_{\ell}/2}\left(\dfrac{n}{m^{2/3}}\right)^{xR_{\ell}} \leq \left(\dfrac{n^{3/2}m}{N}\right)^{2q}\]

\noindent\textbf{Case 3.b} $2\delta_P+2\delta_u+\delta_w^L+\delta_v^L= 2$. In this case $eR_{\ell}= 16q$. 

If $\delta_u = 0$ or $P_{\ell}^U$ is $u$-diagonal, then $x_u^{\ell} \leq \min(h_{\ell}, q)$, as every vertex is covered by an hyperedge of $\mathcal{H}_{\ell}$ at least twice. Clearly, $x^{\ell}_v\leq h_{\ell}$ and $x^{\ell}_w\leq h_{\ell}$, so $xR_{\ell}\leq 2h_{\ell}+\min(h_{\ell}, q)$.

Next, assume $\delta_u = 1$, $\delta_P = 0$ and $P_{\ell}^{U}$ is $v$- or $w$-diagonal. Then the graph $H^{\ell}_{v, w}$ restricted to vertices of layer $\ell$ (i.e. with last coordinate being $\ell$) is an image of a cycle of length $2q$ and $2q$ isolated vertices. Therefore, by Lemma~\ref{lem:induced-graphs-obv-bounds}, there are at most $q+1$ distinct vertices in a cycle of $H^{\ell}_{v, w}$. Since every hyperedge in $\mathcal{H}_{\ell}$ appears at least twice, there are at most $q$ isolated vertices in $H^{\ell}_{v, w}$. Thus, $x_v^{\ell}+x_w^{\ell}\leq \min(2h_{\ell}, 2q+1)$. Therefore, in this case, \[xR_{\ell}\leq h_{\ell}+\min(2h_{\ell}, 2q+1)\leq 2h_{\ell}+\min(h_{\ell}, q),\]
where the last inequality can be checked by considering cases $q\geq h_{\ell}$ and $q<h_{\ell}$.
So, similarly as in Eq.~\eqref{eq:xRi-bound-16},
\[
 \left(\dfrac{n^3}{N}\right)^{4q-h_{\ell}}\left(\dfrac{m^{2/3}}{n}\right)^{eR_{\ell}/2}\left(\dfrac{n}{m^{2/3}}\right)^{xR_{\ell}} \leq  \left(\dfrac{n^{3/2}m}{N}\right)^{2q}
\] 

\noindent\textbf{Case 3.c} $2\delta_P+2\delta_u+\delta_w^L+\delta_v^L= 1$. \\
In this case $eR_{\ell}= 14q$ and $\delta_P = \delta_u = 0$.
If $P^U_{\ell}$ is of $u$-diagonal, then $x^{\ell}_u = 0$, so $xR_{\ell}\leq 2h_{\ell}$. If $P^U_{\ell}$ is $v$-diagonal, then $x^{\ell}_u, x^{\ell}_v \leq \min(h_{\ell}, q)$. Similarly, if $P^U_{\ell}$ is of $w$-diagonal, then $x^{\ell}_u, x^{\ell}_w \leq \min(h_{\ell}, q)$. Thus, since $h_{\ell}\leq 2q$, in any of these three cases, $xR_{\ell} \leq h_{\ell}+2\min(h_{\ell}, q)$. Hence,
\begin{equation}
\begin{gathered}
\left(\dfrac{n^3}{N}\right)^{4q-h_{\ell}}\left(\dfrac{m^{2/3}}{n}\right)^{eR_{\ell}/2}\left(\dfrac{n}{m^{2/3}}\right)^{xR_{\ell}} \leq \left(\dfrac{n^3}{N}\right)^{4q-h_{\ell}}\left(\dfrac{m^{14q/3}}{n^{7q}}\right)\left(\dfrac{n}{m^{2/3}}\right)^{h_{\ell}+2\min(h_{\ell}, q)} \leq \\
\leq \left(\dfrac{n^5m^{14/3}}{N^4}\right)^{q}\left(\dfrac{N}{n^2m^{2/3}}\right)^{h_{\ell}}\left(\dfrac{n}{m^{2/3}}\right)^{2\min(h_{\ell}, q)}\leq \max\left(\dfrac{n^{5q}m^{8q/3}}{N^{3q}},\  \dfrac{n^{3q}m^{2q}}{N^{2q}} \right),
\end{gathered}
\end{equation}
where we use that $N>m^2$, so in the second line it is optimal to have $h_{\ell}\geq q$. For $h_{\ell}\geq q$, the dependence on $h_{\ell}$ is monotone, and $h_{\ell}\leq 2q$, so the expression is maximized when $h_{\ell}\in \{q, 2q\}$.

\noindent\textbf{Case 3.d} $2\delta_P+2\delta_u+\delta_w^L+\delta_v^L= 0$. In this case $eR_{\ell}= 12q$ and $\delta_P = \delta_u = \delta_v^L = \delta_w^L = 0$.

If $P^U_{\ell}$ is $u$-diagonal, then $x^{\ell}_u = 0$ and by Lemma~\ref{lem:induced-graphs-obv-bounds}, $xR_{\ell} = x^{\ell}_v+x^{\ell}_w\leq \min(2h_{\ell}, h_{\ell}+q+1)$. 

If $P^U_{\ell}$ is $v$-diagonal, consider hypergraphs $H^U$ and $H^L$ whose hyperedges are sets $\phi(j, \mu(S_{\ell}^U))$ and $\phi(j, \mu(S_{\ell}^L))$ for $j\in [2q]$, respectively.  Then $\mathcal{H}_{\ell} = H^U\cup H^L$, and each hyperedge in $H^U$ has $v$-color vertex with last coordinate $>\ell$ , and each hyperedge in $H^L$ has $u$-color vertex with last coordinate $>\ell$ ($= t$). Let $h^U$ and $h^L$ be the number of distinct hyperedges that appear only in $H^U$ and only in $H^L$, respectively. Denote also by $h'$ the number of distinct hyperedges that appear in both $H^U$ and $H^L$. Clearly, $h_{\ell} = h^L+h^U+h'$.

Since every hyperedge of $\mathcal{H}_{\ell}$  appears at least twice, we have $h^U\leq (2q-h')/2$. Let $H^L_{v, w}$ be the graph induced by the hyperedges of $H^L\setminus H^U$ on the vertices of color $v$ and $w$. Denote by $d$ the number of connected components of $H^L_{v, w}$. This graph has $h_u$ edges, so it has at most $d+h_u$ vertices. Moreover, note that $H^L_{v, w}$ is obtained as an image of a cycle of length $2q$ from which we deleted some $v$-$w$ edges which appear inside hyperedges of $H^U$. Thus, if $d>1$, every connected component of $H^L_{v, w}$ contains at least one vertex that appears in $H^U$. Therefore, there are at most $d+h^L-(d-1) = h^L+1$ vertices of $\mathcal{H}_{\ell}$ with last coordinate being $\ell$, and which don't appear as vertices of $H^U$. 

Every hyperedge which appears in both $H^U$ and $H^L$ has at most one vertex with last coordinate equal to $\ell$. Also, every hyperedge which appears in $H^U$ has at most two vertices with last coordinate $\ell$. Therefore,
\[ xR_{\ell}\leq (h^L+1)+h'+2h^U = h_{\ell}+1+h^U \leq   h_{\ell}+1+ (2q-h')/2 \leq h_{\ell}+q+1,\]
and since $h^U\leq h_{\ell}-1$, we have $xR_{\ell}\leq h_{\ell}+\min(h_{\ell}, q+1)$. 

Applying a symmetrical argument, we get the same bound when $P^U_{\ell}$ is $w$-diagonal. Hence,

\begin{equation*}
\begin{gathered}
\left(\dfrac{n^3}{N}\right)^{4q-h_{\ell}}\left(\dfrac{m^{2/3}}{n}\right)^{eR_{\ell}/2}\left(\dfrac{n}{m^{2/3}}\right)^{xR_{\ell}} \leq \left(\dfrac{n^3}{N}\right)^{4q-h_{\ell}}\left(\dfrac{m^{4q}}{n^{6q}}\right)\left(\dfrac{n}{m^{2/3}}\right)^{h_{\ell}+\min(h_{\ell}, q+1)} \leq \\
\leq \left(\dfrac{n^6m^{4}}{N^4}\right)^{q}\left(\dfrac{N}{n^2m^{2/3}}\right)^{h_{\ell}}\left(\dfrac{n}{m^{2/3}}\right)^{\min(h_{\ell}, q+1)}\leq \max\left(\dfrac{n^{5q-1}m^{(8q-4)/3}}{N^{3q-1}},\  \dfrac{n^{3q+1}m^{2q-2/3}}{N^{2q}} \right),
\end{gathered}
\end{equation*}
where in the last line we first use that it is optimal to have $h_{\ell}\geq q+1$, as $N>nm^{4/3}$. And then, for $h_{\ell}\geq q+1$, the dependence on $h_{\ell}$ is monotone, and $h_{\ell}\leq 2q$, so the expression is maximized when $h_{\ell}\in \{q+1, 2q\}$.

\noindent\textbf{Case 4.} Assume $i= t$ and $t\neq \ell$. In this case, $s_t = 1$ and $h_t\leq q$.
 Then
 \[ eR_t = (4+\delta_v^U+\delta_w^U)\cdot 2q.\]
\noindent\textbf{Case 4.a} Assume that $\delta_v^U+\delta_w^U\geq 1$. Then $eR_t\geq 10q$.

 Clearly $\chi R_t\leq 3h_t$, so
\begin{equation}
\begin{gathered}
 \left(\dfrac{n^3}{N}\right)^{2q-h_t}\left(\dfrac{m^{2/3}}{n}\right)^{eR_t/2}\left(\dfrac{n}{m^{2/3}}\right)^{\chi R_t} \leq \left(\dfrac{n^3}{N}\right)^{2q-h_t}\left(\dfrac{m^{10q/3}}{n^{5q}}\right)\left(\dfrac{n}{m^{2/3}}\right)^{3h_t} \leq \\
 \leq \left(\dfrac{nm^{10/3}}{N^2}\right)^{q}\left(\dfrac{N}{m^2}\right)^{h_t}\leq \left(\dfrac{nm^{4/3}}{N}\right)^{q}
 \end{gathered}
 \end{equation}

\noindent\textbf{Case 4.b} Assume that $\delta_v^U=\delta_w^U=0$. Then $eR_t = 8q$.

Clearly, $x^t_u\leq h_t$ and graph $H^{t}_{v, w}$ is connected. So, by Lemma~\ref{lem:induced-graphs-obv-bounds}, $x^t_v+x^t_w\leq h_t+1$. Thus $xR_t\leq 2h_t+1$ and

\begin{equation}
\begin{gathered}
 \left(\dfrac{n^3}{N}\right)^{2q-h_t}\left(\dfrac{m^{2/3}}{n}\right)^{eR_t/2}\left(\dfrac{n}{m^{2/3}}\right)^{x R_t} \leq \left(\dfrac{n^3}{N}\right)^{2q-h_t}\left(\dfrac{m^{8q/3}}{n^{4q}}\right)\left(\dfrac{n}{m^{2/3}}\right)^{2h_t+1} \leq \\
 \leq \left(\dfrac{n}{m^{2/3}}\right)\left(\dfrac{nm^{4/3}}{N^2}\right)^{2q}\left(\dfrac{N}{nm^{4/3}}\right)^{h_t}\leq \left(\dfrac{n}{m^{2/3}}\right)\left(\dfrac{nm^{4/3}}{N}\right)^{q}
 \end{gathered}
 \end{equation}
 
 Finally, we collect the analysis of these cases together. Denote
\[ X = \prod\limits_{i=1}^{t}\left(\dfrac{n^3}{N}\right)^{\displaystyle 2s_iq-h_i(\phi)}\left(\dfrac{m^{2/3}}{n}\right)^{\displaystyle eR_i(\phi)/2}\left(\dfrac{n}{m^{2/3}}\right)^{\displaystyle xR_i(\phi)}  \]
Recall that $n>m^{2/3}$ and $N>n^{3/2}m>nm^{4/3}>m^2$. Thus
\begin{enumerate}
\item $X\leq \left(\dfrac{n^{3/2}}{m}\right)\left(\dfrac{n^{3/2}m}{N}\right)^{2q}\left(\dfrac{nm^{4/3}}{N}\right)^{2q(t-1)}$\quad  if $t= \ell\geq 1$;
\item $X\leq \left(\dfrac{n}{m^{2/3}}\right)\left(\dfrac{nm^{4/3}}{N}\right)^{tq}$ \quad if $t>\ell = 0$.
\item $X\leq \left(\dfrac{n}{m^{2/3}}\right)^2\left(\dfrac{nm^{4/3}}{N}\right)^{q}\max\left(\dfrac{n^{5q-1}m^{(8q-4)/3}}{N^{3q-1}},\  \dfrac{n^{3q}m^{2q}}{N^{2q}} \right)\left(\dfrac{nm^{4/3}}{N}\right)^{q(t+\ell-3)}$, if $t>\ell\geq 1$.
\end{enumerate}
Hence, in all of these cases, the inequality from the statement of the theorem holds. 
\end{proof}

As an immediate corollary we get the following bound.

\begin{theorem}\label{thm:Btl-norm-bound} Let $m\ll n^{3/2}$ and $N\gg n^{3/2}m$. Then a matrix $B_{t, \ell}$, defined as in Obs.~\ref{obs:structure-BOmega-terms}, w.h.p. over the randomness of $\mathcal{V}$ and $\Omega$, satisfies
\[\Vert B_{t, \ell}\Vert = \wt{O}\left(\left(\dfrac{n^{3/2}m}{N}\right)^{t+\ell}\right).\]
\end{theorem}
\begin{proof} The bound follows from Theorems~\ref{thm:BOmega-bound-without-layers},~\ref{thm:main-layer-analysis} and Lemma~\ref{lem:trace-power-method-norm} by taking $q = O(\log (n)^2)$.
\end{proof}

\subsection{Collecting pieces together}

Finally, we are ready to prove the existence of an $\Omega$-restricted SOS dual certificate $(\oa{A}_{\Omega}, B_{\Omega}, Z_{\Omega})$.

\begin{theorem} Let $m\ll n^{3/2}$ and $N\gg n^{3/2}m$. W.h.p. over the randomness of $\mathcal{V} = \{(u_i, v_i, w_i)\mid i\in [m]\}$ and randomness of $\Omega = \bigcup\limits_{j=1}^{30} \Omega_j$ there exists a triple $(\oa{A}_{\Omega}, B_{\Omega}, Z_{\Omega})$ such that
\begin{enumerate}
\item $\oa{A}_{\Omega} \in \mathbb{R}^{n^3}$ is a dual certificate for $\mathcal{V}$ and $(\oa{A}_{\Omega})_{\omega} = 0$ for $\omega\notin \Omega$;
\item $B_{\Omega}$ is an $n^3\times n^3$ symmetric matrix with $\Vert B_{\Omega}\Vert \leq 1$, and
\item $Z_{\Omega}$ is an $n^3\times n^3$ zero-polynomial matrix such that $B_{\Omega}+Z_{\Omega} = A_{\Omega}^TA_{\Omega}$, where $A_{\Omega}$ is an $n\times n^2$ matrix with $(A_{\Omega})_{a, (b, c)} = (\oa{A}_{\Omega})_{(a, b, c)}$.
\end{enumerate}
\end{theorem}
\begin{proof} By Theorem~\ref{thm:main-sos-dual-certificate} w.h.p. there exists an SOS dual certificate $(\oa{A}, B, Z)$ for $\mathcal{V}$. Define $\oa{A}_{\Omega}$ as in Eq.~\eqref{eq:Aomega-constr} with $k=30$. Then, by Theorem~\ref{thm:AOmega-structure}, w.h.p.  $\oa{A}_{\Omega}$ is a dual certificate and we can write
\[ \oa{A}_{\Omega} = \oa{A}+\oa{A}_{\Omega, GM}+\oa{A}'_{\Omega, sm}, \]
where $\Vert \oa{A}'_{\Omega, sm} \Vert = \wt{O}\left(\dfrac{m^2}{n^4}\right)$ and $\oa{A}_{\Omega, GM} \in \vspan\left(\mathfrak{G}_{\Omega, N}, C^{30}\right)$ for some absolute constant $C$ (see Theorem~\ref{thm:AOmega-structure} for more details).  Define
\[ B_{\Omega, 0} = \mathcal{B}(\oa{A}_{\Omega}, \oa{A}_{\Omega}) = \tw_2\left(A_{\Omega}^TA_{\Omega}\right), \qquad Z_{\Omega, 0} = \mathcal{Z}(B_{\Omega, 0} - P_{\mathcal{L}})\quad \text{and}\]
\[ B_{\Omega} = B_{\Omega, 0}+Z_{\Omega, 0}\qquad Z_{\Omega} = A_{\Omega}^TA_{\Omega} - B_{\Omega}.\quad \ \ \ \]
Obviously, $Z_{\Omega}$ is a zero polynomial. Let $B_0 = A^TA = \mathcal{B}(\oa{A}, \oa{A})$. Then, by Lemma~\ref{lem:BOmega-small-term-analysis},
\[B_{\Omega, 0} = B_{\Omega, GM}+B_{\Omega, sm}\quad \text{and} \quad \mathcal{Z}(B_{\Omega, 0} - B_0) = \mathcal{Z}_{GM}(B_{\Omega, GM} - B_0)+\mathcal{Z}_{\Omega, sm}, \]
 where $\Vert \mathcal{Z}_{\Omega, sm} \Vert_F =  \wt{O}\left(\dfrac{m^2}{n^3}\right)$ and  $\Vert B_{\Omega, sm} \Vert_F =  \wt{O}\left(\dfrac{m^2}{n^3}\right)$ and matrices $B_{\Omega, GM}$, $\mathcal{Z}_{GM}(B_{\Omega, GM}- B_0)$  are defined in Eq.~\eqref{eq:BGM-defin} and~\eqref{eq:ZGM-defin}.
 
 By Observation~\ref{obs:structure-BOmega-terms}, matrices $B_{\Omega, GM}-B_0$ and $\mathcal{Z}_{GM}(B_{\Omega, GM}- B_0)$ can be written as a signed sums of the matrices of form $B_{t, \ell}$ as in Eq.~\eqref{eq:Btl-definition} (and the number of summands is bounded by an absolute constant). Therefore, by Theorem~\ref{thm:Btl-norm-bound},
 \[ \Vert B_{\Omega, GM}\Vert = \wt{O}\left(\dfrac{n^{3/2}m}{N}\right) \quad \text{and} \quad \Vert \mathcal{Z}_{GM}(B_{\Omega, GM}- B_0)\Vert = \wt{O}\left(\dfrac{n^{3/2}m}{N}\right).\]
 Also recall that by Theorem~\ref{thm:B0-construction} and Theorem~\ref{thm:Z0-norm-bound},
\[ \Vert B_0 - P_{\mathcal{L}} \Vert = \wt{O}\left(\dfrac{m}{n^{3/2}}\right),\quad \text{and}\quad \Vert \mathcal{Z}(B_0-P_{\mathcal{L}}) \Vert = \wt{O}\left(\dfrac{m}{n^{3/2}}\right). \]
By construction of $\mathcal{Z}(B_{\Omega, 0} - P_{\mathcal{L}})$ (see Theorems~\ref{thm:zero-poly-correction} and~\ref{thm:zero-poly-correction-constr}),
\[ (B_{\Omega})_{\mathcal{L}} = I_{\mathcal{L}},\]
and restricting to $\mathcal{L}^{\perp}$ we have
\[ \Vert (B_{\Omega, 0})_{\mathcal{L}^{\perp}}\Vert \leq \Vert B_0 - P_{\mathcal{L}} \Vert+ \Vert B_{\Omega, GM}-B_0\Vert+ \Vert B_{\Omega, sm}\Vert \leq \wt{O}\left(\dfrac{n^{3/2}m}{N}\right)+\wt{O}\left(\dfrac{m}{n^{3/2}}\right), \quad \text{and}\]
\[ \Vert (Z_{\Omega, 0})_{\mathcal{L}^{\perp}}\Vert \leq \Vert \mathcal{Z}(B_0 - P_{\mathcal{L}}) \Vert+ \Vert \mathcal{Z}_{GM}(B_{\Omega, GM}- B_0)\Vert+ \Vert \mathcal{Z}_{\Omega, sm}\Vert \leq \wt{O}\left(\dfrac{n^{3/2}m}{N}\right)+\wt{O}\left(\dfrac{m}{n^{3/2}}\right). \]
Therefore, w.h.p. $\Vert B_{\Omega}\Vert \leq 1$.
\end{proof}

\section{Numerical experiments}\label{sec:numerical}
In this section we provide results of numerical experiments in support of our theoretical claims. We implement all three algorithms described in Section~\ref{sec:setup-overview}: SDP for nuclear norm minimization, tensor completion algorithm and tensor decomposition algorithm. 

Our implementation of the nuclear norm minimization and tensor completion algorithms uses the SDPNALv1.0 matlab software package \cite{SDPNAL} to solve large scale semidefinite programs. This is a first order method SDP solver that implements an augmented Lagrangian based method. For tensor decomposition we use MOSEK (python API)~\cite{MOSEK}, an interior point method SDP solver. 

All experiments were run on a laptop with Intel Core i5-6200U 2.3GHz processor and 8GM DDR4 memory.

\paragraph{Input data} We use synthetic input data for the algorithms. That is, for a given dimension $n$ and a number of components $m$ we generate $3m$ independent standard Gaussian random vectors $\{u_i, v_i, w_i\mid i\in [m]\}$ in $\mathbb{R}^n$ and form a tensor as
\begin{equation}\label{eq:ten-numer}
 \mathcal{T} = \sum\limits_{i=1}^m u_i\otimes v_i\otimes w_i.
\end{equation}
Note, that this is consistent with our model, as after normalization, components are uniformly distributed on a unit sphere.

For tensor completion, we generate the set $\Omega$ by uniformly sampling $N$ random entries from $[n]^3$ without replacement.

\paragraph{Experiments for nuclear norm} To find the nuclear norm a given tensor we implement the dual semidefinite program~\eqref{eq:opt-constraints-id-sec2}

\begin{Optbox}
Input: $\mathcal{T}$\\
{Maximize}\quad $\quad 2\langle \oa{A}, \mathcal{T} \rangle$\\
Subject to:\quad  $ \displaystyle
\quad \left(\begin{matrix}
I_n & -A\\
-A^T & Z+I_{n^2}
\end{matrix}\right) \succeq 0, \quad Z\equiv_{poly} 0, \quad A_{i, (j, k)} =\oa{A}_{(i, j, k)}.
$\\
Output:\quad $\oa{A}$, $2\langle \oa{A}, \mathcal{T} \rangle$.
\end{Optbox}

We study the maximal $m$ for which our algorithm is able to compute the nuclear norm of $\mathcal{T}$ and to certify that the decomposition given by Eq.~\eqref{eq:ten-numer} minimizes the nuclear norm. Note that the tensor nuclear norm of $\mathcal{T}$ is at most $\sum\limits_{i=1}^{m} \Vert u_i\otimes v_i\otimes w_i\Vert$. Moreover, as explained in Section~\ref{sec:nuclear-norm-programs}, the value of the SDP we solve can be at most this value. 

Hence, if optimal solution of our SDP has value $\sum\limits_{i=1}^{m} \Vert u_i\otimes v_i\otimes w_i\Vert$, we can declare that this is indeed a nuclear norm of $\mathcal{T}$ and that the SDP found it successfully. In the experiments we allow $10^{-6}$ relative error for the value of nuclear norm.

\begin{figure}[h]
\begin{subfigure}[b]{0.5\textwidth}
\begin{center}
\includegraphics[width = 0.9\textwidth]{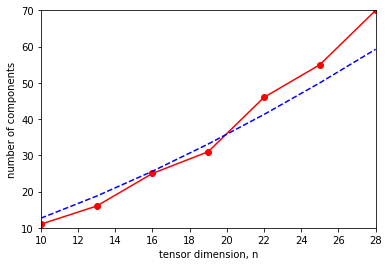}
\caption{Largest $m$ for which SDP find nuclear norm with $\geq80\%$ success rate. Dashed line is $0.4\cdot n^{1.5}$}
\end{center}
\end{subfigure}
\begin{subfigure}[b]{0.5\textwidth}
\begin{center}
\includegraphics[width = 0.9\textwidth]{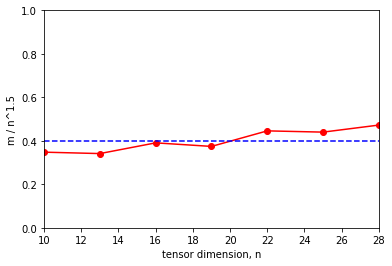}
\caption{Value of $m/n^{3/2}$, for largest $m$ for which SDP succeeds. Dashed line is $0.4$}
\end{center}
\end{subfigure}
\caption{Experimental results for nuclear norm minimization}\label{fig:nuc_norm_min}
\end{figure}

In the experiment we take values of $n$ in the interval $[10, 31]$ with step 3 and $m\in [n, 1.5n^3]$ with step 3. For every pair $(n, m)$ we generate 5 instances of a random tensor and for each of them we solve the corresponding SDP problem. For each $n$ we record the largest $m$ for which our SDP solved nuclear norm minimization successfully for at least 4 out of 5 tensors. The experimental results are presented in Figure~\ref{fig:nuc_norm_min}, where the graph on the right illustrates the value of the fraction $m/n^{3/2}$ for the values of $m$ in the graph on the left.     

\paragraph{Experiments for tensor completion} To solve the tensor completion problem we implement the primal SDP given by Eq.~\eqref{eq:tensor-compl-primal}

\begin{Optbox}
Input: $\mathcal{T}_{(i, j, k)}$ for $(i, j, k) \in \Omega$.
\begin{equation}
\begin{split}
 &\text{Minimize:}\quad  {\Tr\left(\begin{matrix}
M_{U} & X\\
X^T & M_{VW}
\end{matrix}\right).} \\
 &\text{Subject to:}\quad  
 {\left(\begin{matrix}
M_{U} & X\\
X^T & M_{VW}
\end{matrix}\right) \succeq 0}, \quad 
 \forall (i,j,k) \in \Omega, X_{i, (j, k)} = \mathcal{T}_{(i, j, k)} \\
&\ \quad  \forall j,k,j',k' \in [n], (M_{VW})_{jkj'k'} =  (M_{VW})_{j'kjk'} = (M_{VW})_{jk'j'k} =  (M_{VW})_{j'k'jk}
 \end{split}
\end{equation}
Output: $X$.
\end{Optbox}

We study the smallest number of random entries $N$ needed to reconstruct the tensor $\mathcal{T}$ up to a small error. We say that the tensor $\mathcal{T}$ is completed successfully if for the completed tensor $\mathcal{T}_{ex}$ the relative Frobenius error $\Vert \mathcal{T} - \mathcal{T}_{ex}\Vert_F/ \Vert \mathcal{T} \Vert_F$ is at most $10^{-5}$.

We introduce a parameter $p = N/n^3$ which measures the proportion of the entries which are given to us. For $n = 20, 25, 30$, the number of components $m$, and a fixed random tensor $\mathcal{T}$ for each pair $(n, m)$ we search for the smallest $p$ for which our SDP is able to reconstruct the tensor successfully. To find such $p$, we run the binary search on the interval $[0.01, 1]$ until the gap is at most $0.01$. We illustrate the dependence of $p$ on the number of components $m$ in Figure~\ref{fig:compl_on_m}. The graph shows the linear dependence on $m$, which is established in Theorem~\ref{thm:tensor-completion-intro}. 

\begin{figure}[h!]
\begin{center}
\includegraphics[width = 0.45\textwidth]{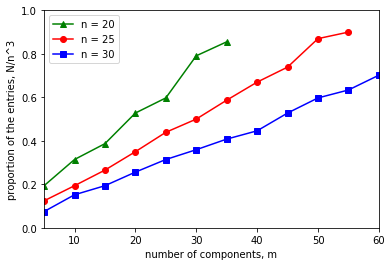}
\end{center}
\caption{Dependence on $m$ of the number of entries needed for our SDP for successful tensor completion for $n = 20, 25, 30$}\label{fig:compl_on_m}
\end{figure}

Additionally, we study the dependence of $p$ on the tensor dimension $n$. For this, we fix the number of components $m=20$, and we generate random tensors with $m=20$ components for $n \in [15, 40]$ with step 3. As before, we run binary search for the smallest $p$ in the region $[0.01, 1]$ until the gap is at most $0.01$. We present our results in Figure~\ref{fig:compl_on_m}.

\begin{figure}[h]
\begin{subfigure}[b]{0.5\textwidth}
\begin{center}
\includegraphics[width = 0.9\textwidth]{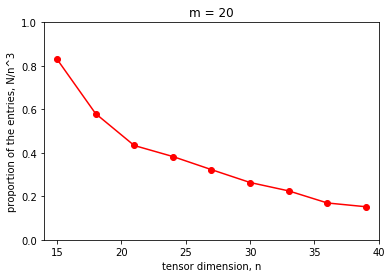}
\caption{Smallest rate $p = N/n^3$ for which tensor completion SDP succeeds}
\end{center}
\end{subfigure}
\begin{subfigure}[b]{0.5\textwidth}
\begin{center}
\includegraphics[width = 0.9\textwidth]{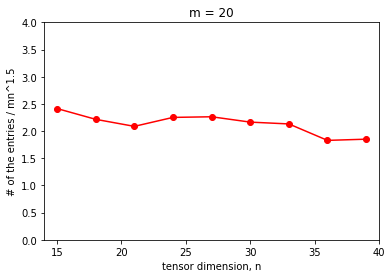}
\caption{Smallest number of entries needed for successful tensor completion divided by $mn^{3/2}$}
\end{center}
\end{subfigure}
\caption{Dependence on $n$ of the number of entries needed for our SDP for successful tensor completion}\label{fig:compl_on_n}
\end{figure}    

\paragraph{Experiments for tensor decomposition}

We used the interior-point method solver MOSEK to implement the tensor decomposition algorithm described in Section~\ref{sec:tensor-decomposition-algorithm}.

While we are still honing our tensor decomposition algorithm, so far we have gotten it to consistently decompose slightly overcomplete tensors with an error of Frobenius norm $10^{-8}$ for small $n$ (e.g. $n = 14$, $m = 16$). We expect that this can be improved considerably.

\pagebreak

\appendix
\section{Duality Arguments}\label{dualityappendix}
In this appendix, we give an explanation for the dualities claimed in Section~\ref{sec:setup-overview}. To do this, we use the following standard framework:
\begin{enumerate}
    \item We construct a two-player zero-sum game.
    \item We then observe that if the first player goes first, this gives the primal program and if the second player goes first, this gives the dual program.
\end{enumerate}
\begin{remark}
As shown by minimax theorems such as Von Neumann's minimax theorem, for such two-player zero sum games, the outcome is almost always the same whether the first player or the second player goes first (while there are some exceptions, these exceptions are pathological). When this holds, the primal and dual programs give the same value (i.e. we have strong duality). It can be shown that strong duality holds for our programs, but since we only need weak duality, we do not show this here.
\end{remark}
\subsection{Duality for Tensor Nuclear Norm}
Here we give an explanation for why the tensor nuclear norm and the injective norm are dual to each other. More precisely, we have the following primal and dual.
\begin{enumerate}
    \item Primal: $\Vert\mathcal{T}\Vert_{*} = \min\left\{\sum\limits_{i = 1}^m |\lambda_i| \mid \mathcal{T} = \sum\limits_{i=1}^{m} \lambda_i a_i^{1}\otimes a_i^2\otimes \ldots \otimes a_i^d,\ a_i^t\in S^{n_t-1}\right\}$
    \item Dual: $\max\{ \langle \oa{A}, \mathcal{T}\rangle \mid \Vert \oa{A} \Vert_{\sigma}\leq 1\}$
\end{enumerate}
To see this duality, consider the following two-player zero sum game. The objective function, which the fist player wants to minimize and the second player wants to maximize, is $c - c\langle   \oa{A}, a^1 \otimes \ldots \otimes a^d \rangle + \langle \oa{A},  \mathcal{T}\rangle$ and the rules for the players are as follows.
\begin{enumerate}
\item Player $1$ chooses a constant $c \geq 0$ and unit vectors $a^1,a^2,\ldots,a^d$ where $a^t \in S^{n_t - 1}$.
\item Player $2$ chooses an arbitrary $\oa{A}$.
\end{enumerate}
Let's consider the game where player $1$ first chooses a probabilistic strategy and then player $2$ chooses $\oa{A}$. If player $1$ chooses $c_i,a^1_i,a^2_i,\ldots,a^d_i$ with probability $p_i$ then the expected value of the objective function is 
\[
E\left[c - c\langle \oa{A}, a^1 \otimes \ldots \otimes a^d \rangle + \langle \oa{A},  \mathcal{T}\rangle\right] = \left(\sum_{i}{p_i\cdot c_i}\right) + \langle \oa{A},(\mathcal{T} - \sum_{i}{p_i\cdot c_i\cdot (a_i^{1} \otimes \ldots \otimes a_i^d)})  
\rangle \]
For player $1$'s optimal strategy, we have that $\sum_{i}{p_i\cdot c_i\cdot (a_i^{1} \otimes \ldots \otimes a_i^d)} = \mathcal{T}$ as otherwise player $2$ can make the expected value of the objective function arbitrarily large by choosing $\oa{A}$ appropriately. If $\sum_{i}{p_i\cdot c_i\cdot (a_i^{1} \otimes \ldots \otimes a_i^d)} = \mathcal{T}$ then 
the expected value of the objective function is $\sum_{i}{p_i\cdot c_i}$, so this gives the primal program.

If player $2$ goes first, we must have that for all unit vectors $a^1,a^2,\ldots,a^d$, $\langle \oa{A}, a^1 \otimes \ldots \otimes a^d\rangle  \leq 1$ as otherwise player $1$ can make the objective function arbitrarily negative by choosing $a^1,a^2,\ldots,a^d$ and a large constant $c$. If $\forall a^1,a^2,\ldots,a^d, \langle \oa{A},a^1 \otimes \ldots \otimes a^d\rangle \leq 1$ then it is optimal for player $1$ to respond by always taking $c = 0$ in which case the objective function has value $\langle \oa{A}, \mathcal{T}  \rangle $. This gives the dual program.
\subsection{Duality for our tensor nuclear norm algorithm}
Here we show the duality between our primal program for the tensor nuclear norm and dual certificate described in Section~\ref{sec:nuclear-norm-programs}.
\begin{Optbox}
\begin{enumerate}
\item[] Primal: Minimize $tr\left(\begin{matrix}
M_{U} & \mathcal{T}\\
\mathcal{T}^T & M_{VW}
\end{matrix}\right)$ subject to $\left(\begin{matrix}
M_{U} & \mathcal{T}\\
\mathcal{T}^T & M_{VW}
\end{matrix}\right) \succeq 0$, 
\[
\forall j,k,j',k' \in [n], (M_{VW})_{jkj'k'} =  (M_{VW})_{j'kjk'} = (M_{VW})_{jk'j'k} =  (M_{VW})_{j'k'jk}
\]
\item[] Dual: Maximize $2\langle \oa{A}, \mathcal{T} \rangle$ subject to
\[
\left(\begin{matrix}
I_n & -\oa{A}\\
-\oa{A}^T & Z+I_{n^2}
\end{matrix}\right) \succeq 0, \quad Z\equiv_{poly} 0, \quad \oa{A}_{i, (j, k)} =\oa{A}_{(i, j, k)}
\]
\end{enumerate}
\end{Optbox}
To see the duality, consider the following two player zero sum game. The objective function, which player $1$ wants to minimize and player $2$ wants to maximize, is
\[
\left(\begin{matrix}
M_{U} & X\\
X^T & M_{VW}
\end{matrix}\right) \bullet \left(\begin{matrix}
I_n & -\oa{A}\\
-\oa{A}^T & Z+I_{n^2}
\end{matrix}\right) + 2\langle \oa{A}, \mathcal{T} \rangle
\]
The rules for the players are as follows.
\begin{enumerate}
\item Player $1$ chooses $M_U$, $X$, and $M_{VW}$ so that $\left(\begin{matrix}
M_{U} & \mathcal{T}\\
\mathcal{T}^T & M_{VW}
\end{matrix}\right) \succeq 0$
\item Player $2$ chooses $\oa{A}$ and $Z$ so that $Z\equiv_{poly} 0$.
\end{enumerate}
If player $1$ goes first, for player $1$'s optimal strategy, we have that 
\begin{enumerate}
\item Player $1$ must take $X = \mathcal{T}$ as otherwise player $2$ can make the objective function arbitrarily large by choosing $\oa{A}$ appropriately.
\item Player $1$ must choose $M_{VW}$ so that $\forall j,k,j',k' \in [n], (M_{VW})_{jkj'k'} =  (M_{VW})_{j'kjk'} = (M_{VW})_{jk'j'k} =  (M_{VW})_{j'k'jk}$ as otherwise player $2$ can make the objective function arbitrarily large by choosing $Z$ appropriately.
\end{enumerate}
With these choices, the objective function becomes $tr\left(\begin{matrix}
M_{U} & \mathcal{T}\\
\mathcal{T}^T & M_{VW}
\end{matrix}\right)$ as $M_{VW} \bullet Z = 0$. This gives the primal program.

If player $2$ goes first, for player $2$'s optimal strategy, we have that $\left(\begin{matrix}
I_n & -\oa{A}\\
-\oa{A}^T & Z+I_{n^2}
\end{matrix}\right) \succeq 0$ as otherwise player $1$ can make the objective function arbitrarily negative by choosing $\left(\begin{matrix}
M_{U} & X\\
X^T & M_{VW}
\end{matrix}\right)$ appropriately. If $\left(\begin{matrix}
I_n & -\oa{A}\\
-\oa{A}^T & Z+I_{n^2}
\end{matrix}\right) \succeq 0$ then $\left(\begin{matrix}
M_{U} & X\\
X^T & M_{VW}
\end{matrix}\right) = 0$ is an optimal response by player $1$ in which case the objective function has value $2\langle \oa{A}, \mathcal{T} \rangle$.

To adjust this two-player zero sum game for tensor completion, we add the additional restriction for player $2$ that $\oa{A}_{i,j, k} = 0$ whenever $(i,j,k) \notin \Omega$.

\section{Sum of Squares View}\label{sec:sumofsquares-view}
\subsection{The Sum of Squares Hierarchy}
\begin{definition}[Degree d pseudo-expectation values]
Given a set of polynomial constraints $\{s_i = 0\}$, degree $d$ pseudo-expectation values are a linear map $\tilde{E}$ from polynomials of degree at most $d$ to $\mathbb{R}$ which satisfies the following conditions:
\begin{enumerate}
    \item $\tilde{E}[1] = 1$
    \item For all $i$ and all polynomials $f$ of degree at most $deg(s_i)$, $\tilde{E}[f{s_i}] = 0$
    \item For all polynomials $g$ of degree at most $\frac{d}{2}$, $\tilde{E}[g^2] \geq 0$
\end{enumerate}
\end{definition}
This third condition can be expressed in terms of the \emph{moment matrix}.
\begin{definition}
Given degree-d pseudo-expectation values $\tilde{E}$, the moment matrix $M$ is indexed by monomials $m_1,m_2$ of degree at most $\frac{d}{2}$ and has entries
$M_{{m_1}{m_2}} = \tilde{E}[{m_1}{m_2}]$
\end{definition}
\begin{proposition}
$\tilde{E}[g^2] \geq 0$ for all polynomials $g$ of degree at most $\frac{d}{2}$ if and only if $M \succeq 0$.
\end{proposition}
\begin{proof}
Let $g$ be a polynomial of degree at most $\frac{d}{2}$. Writing $g = \sum_{\text{monomoials } m}{{g_m}m}$ and viewing $g$ as a vector with the coordinate $g_m$ for each monomial $m$, 
\[
{g^T}Mg = \sum_{\text{monomials } m_1,m_2}{g_{m_1}\tilde{E}[{m_1}{m_2}]g_{m_2}} = \tilde{E}\left[\left(\sum_{\text{monomials } m}{{g_m}m}\right)^2\right] = \tilde{E}[g^2]
\]
\end{proof}
\subsection{Sum of Squares Proofs}
\begin{definition}
Given constraints $\{s_i = 0\}$, a degree-d sum of squares proof that $p \geq c$ (over the real numbers) is an equality of the form
\[
p = c + \sum_{i}{f_i{s_i}} + \sum_{j}{g_j^2}
\]
where for all $i$, $deg(f_i) + deg(s_i) \leq d$ and for all $j$, $deg(g_j) \leq \frac{d}{2}$. Note that this is a proof that $p \geq c$ because we are given that each $s_i = 0$ and $\sum_{j}{g_j^2} \geq 0$ over the real numbers.
\end{definition}
Using the following fact, sum of squares proofs can also be viewed in terms of a PSD matrix.
\begin{definition}
Given a matrix $Q$ whose rows and columns are indexed by monomials, we say that $Q \equiv_{poly} \sum_{\text{monomials } m_1,m_2}{Q_{{m_1}{m_2}}{m_1}{m_2}}$
\end{definition}
\begin{proposition}
A polynomial $g$ is a sum of squares (i.e. $g = \sum_{j}{g_j^2}$ for some polynomials $g_j$) if and only if there exists a matrix $Q$ such that $Q \equiv_{poly} g$ and $Q \succeq 0$.
\end{proposition}
\begin{proof}
If $g = \sum_{j}{g_j^2}$ then writing each $g_j$ as $g_j = \sum_{\text{monomials } m}{g_{jm}m}$, taking each $g_j$ to be the vector with coordinates $(g_j)_m = g_{jm}$ and taking $Q = \sum_{j}{{g_j}{g_j^T}}$ we have that 
\begin{align*}
Q &\equiv_{poly} \sum_{\text{monomials } m_1,m_2}{Q_{{m_1}{m_2}}{m_1}{m_2}} \\
&= \sum_{j}{\sum_{\text{monomials } m_1,m_2}{g_{j{m_1}}g_{j{m_2}}{m_1}{m_2}}} \\
&= \sum_{j}{\left(\sum_{\text{monomials } m}{g_{jm}m}\right)^2} = \sum_{j}{g_j^2} = g
\end{align*}

Conversely, if there exists a matrix $Q$ such that $Q \equiv_{poly} g$ and $Q \succeq 0$ then writing $Q = \sum_{j}{{q_j}{q_j^T}}$ and following the same logic,
\begin{align*}
Q &\equiv_{poly} \sum_{j}{\sum_{\text{monomials } m_1,m_2}{(q_j)_{m_1}(q_j)_{m_2}{m_1}{m_2}}} \\
&= \sum_{j}{\left(\sum_{\text{monomials } m}{(q_j)_{m}m}\right)^2}
\end{align*}
\end{proof}
Thus, if $Q \equiv_{poly} p - c - \sum_{i}{f_i{s_i}}$ then we can view $Q$ as a sum of squares proof that $p \geq c$.
\subsection{An Equivalent Formulation of the Tensor Nuclear Norm Problem}
One way to view the tensor nuclear norm problem is as follows. Given a distribution on random variables $x,y,z \in \mathbb{R}^n$ such that $E[x \otimes y \otimes z] = \mathcal{T}$, what is the minimum possible value of $\frac{1}{2}E\left[\norm{x}^2 + \norm{y}^2\norm{z}^2\right]$?

To see why this is equivalent, observe that if $x = u_i$, $y = v_i$, and $z = w_i$ with probability $p_i$ then 
\begin{enumerate}
    \item $E\left[x \otimes y \otimes z\right] = 
    \sum_{i=1}^{m}{\left(p_i\norm{u_i}\norm{v_i}\norm{w_i}\right)(\hat{u}_i \otimes \hat{v}_i \otimes \hat{w}_i)}$
    \item $\frac{1}{2}E\left[\norm{x}^2 + \norm{y}^2\norm{z}^2\right] = 
    \frac{1}{2}\sum_{i=1}^{m}{p_i\left(\norm{u_i}^2 + \norm{v_i}^2\norm{w_i}^2\right)}$
\end{enumerate}
Thus, given a decomposition $\mathcal{T} = \sum_{i=1}^{m}{\lambda_i(\hat{u}_i \otimes \hat{v}_i \otimes \hat{w}_i)}$, we can take $x = \sqrt{m\lambda_i}\hat{u}_i$, $y = \sqrt[4]{m\lambda_i}\hat{v}_i$, and $z = \sqrt[4]{m\lambda_i}\hat{w}_i$ with probability $p_i = \frac{1}{m}$ and we will have that $\frac{1}{2}E\left[\norm{x}^2 + \norm{y}^2\norm{z}^2\right] = \sum_{i=1}^{m}{\lambda_i}$. This implies that  $min{\left\{\frac{1}{2}E\left[\norm{x}^2 + \norm{y}^2\norm{z}^2\right]\right\}} \leq \norm{\mathcal{T}}_{*}$.

Conversely, given a distribution where $E[x \otimes y \otimes z] = \mathcal{T}$, if $x = u_i$, $y = v_i$, and $z = w_i$ with probability $p_i$ for $i \in [m]$ then
\begin{enumerate}
    \item $E\left[x \otimes y \otimes z\right] = 
    \sum_{i=1}^{m}{\left(p_i\norm{u_i}\norm{v_i}\norm{w_i}\right)(\hat{u}_i \otimes \hat{v}_i \otimes \hat{w}_i)}$ so $\norm{\mathcal{T}}_{*} \leq \sum_{i=1}^{m}{\left(p_i\norm{u_i}\norm{v_i}\norm{w_i}\right)}$.
    \item $\frac{1}{2}E\left[\norm{x}^2 + \norm{y}^2\norm{z}^2\right] = 
    \frac{1}{2}\sum_{i=1}^{m}{p_i\left(\norm{u_i}^2 + \norm{v_i}^2\norm{w_i}^2\right)} \geq \sum_{i=1}^{m}{\left(p_i\norm{u_i}\norm{v_i}\norm{w_i}\right)}$.
\end{enumerate}
This implies that $min{\left\{\frac{1}{2}E\left[\norm{x}^2 + \norm{y}^2\norm{z}^2\right]\right\}} \geq \norm{\mathcal{T}}_{*}$.
\subsection{Primal Degree 4-SoS Program}
The degree-4 SoS relaxation of this problem is to minimize $\tilde{E}\left[\left(\sum_{a}{x_a^2}\right) + \left(\sum_{b,c}{({y_b}{z_c})^2}\right)\right]$ subject to the constraints that $\tilde{E}{{x_a}{y_b}{z_c}} = \mathcal{T}_{abc}$ and $M \succeq 0$.

Restricting our attention to the submatrix of $M$ whose rows and columns are indexed by the monomials $\{x_a\}$ and $\{{y_b}{z_c}\}$ and writing $M$ as $M = \left(\begin{matrix}
M_{U} & T\\
T^T & M_{VW}
\end{matrix}\right)$, we make the following observations
\begin{enumerate}
    \item $\tilde{E}\left[\left(\sum_{a}{x_a^2}\right) + \left(\sum_{b,c}{({y_b}{z_c})^2}\right)\right] = tr(M)$
    \item The condition that $\tilde{E}{{x_a}{y_b}{z_c}} = \mathcal{T}_{abc}$ implies that $T_{a,bc} = \mathcal{T}_{abc}$
    \item For all $b,b',c,c'$, $(M_{VW})_{bcb'c'} = (M_{VW})_{bc'b'c} = (M_{VW})_{b'cbc'} = (M_{VW})_{b'c'bc} = \tilde{E}[{y_b}{y_{b'}}{z_c}{z_{c'}}]$. 
\end{enumerate}
This gives us our primal semidefinite program.
\subsection{Dual Program}
Recall that the dual program is 

Maximize $2\langle \oa{A}, \mathcal{T} \rangle$
\begin{equation}\label{eq:opt-constraints-id}
\text{Subject to:}\quad \left(\begin{matrix}
I_n & -A\\
-A^T & Z+I_{n^2}
\end{matrix}\right) \succeq 0, \quad Z\equiv_{poly} 0, \quad A_{i, (j, k)} =\oa{A}_{(i, j, k)}.
\end{equation}
Letting $Q = \left(\begin{matrix}
I_n & -A\\
-A^T & Z+I_{n^2}
\end{matrix}\right)$, observe that 
\begin{align*}
Q &\equiv_{poly} \left(\sum_{a}{x_a^2}\right) + \left(\sum_{b}{y_b^2}\right)\left(\sum_{c}{z_c^2}\right) - 2\sum_{a,b,c}{\mathcal{A}_{abc}{x_a}{y_b}{z_c}} \\
&= \left(\sum_{a}{x_a^2}\right) + \left(\sum_{b,c}{({y_b}{z_c})^2}\right)
-2\sum_{a,b,c}{\mathcal{A}_{abc}\left({x_a}{y_b}{z_c} - \mathcal{T}_{abc}\right)} -2\langle \oa{A}, \mathcal{T} \rangle
\end{align*}
Thus, $Q$ corresponds to a sum of squares proof that $E\left[\norm{x}^2 + \norm{y}^2\norm{z}^2\right] \geq 2\langle \oa{A}, \mathcal{T} \rangle$

\section{Missing details of the tensor decomposition algorithm}\label{sec:decomposition:missing}

\subsection{The nullspace of a semidefinite program}

Let $L$ be the nuclear norm for the tensor $\mathcal{T}$ given by Eq.~\eqref{eq:ten-decomp}, which can be computed by solving~\eqref{eq:opt-constraints-id-sec2} for $m\ll n^{3/2}$. For a given space $\mathcal{X}\subseteq \mathbb{R}^{n^2}$ with orthonormal basis $\{x_i \mid i\in [k]\}$  consider the following program.
\begin{Optbox}
\begin{equation}\label{eq:opt-objective-decomp}
\text{Minimize }\quad \Tr_{\mathcal{X}}(B) = \sum\limits_{i = 1}^{k} x_i^T B x_i
\end{equation}
\begin{equation}\label{eq:opt-constraints-decomp}
\begin{gathered}
\text{Subject to:}\quad \langle \oa{A}, \mathcal{T} \rangle = L\qquad {and}\hfill \\
\left(\begin{matrix}
I_n & -A\\
-A^T & Z+B
\end{matrix}\right) \succeq 0, \quad Z\equiv_{poly} 0, \quad B\preceq I_{n^2}, \quad A_{i, (j, k)} = \oa{A}_{(i, j, k)}.
\end{gathered}
\end{equation}
\end{Optbox}

\begin{theorem}\label{thm:decomp-subspace-shrink} Let $m\ll n^{3/2}$. Consider a subspace $\mathcal{X}\subseteq \mathbb{R}^{n^2}$ with $\vspan\{v_i\otimes w_i\mid i\in [m]\} \subseteq \mathcal{X}$.  Let $(\oa{A}, B, Z)$ be an optimal solution to problem~\eqref{eq:opt-objective-decomp}-\eqref{eq:opt-constraints-decomp}. Then w.h.p. $\oa{A}$ is a dual certificate for $\mathcal{V}$ and $\dim\left( \Ker(I-B)\cap \mathcal{X}\right)\leq (m+\dim \mathcal{X})/2$.
\end{theorem}
\begin{proof} We proved in Theorem~\ref{thm:main-nuclear-norm-minimizationver2} that for any solution satisfying~\eqref{eq:opt-constraints-decomp}, $\oa{A}$ is a dual certificate for $\mathcal{V}$. Moreover, in this case $\vspan\{v_i\otimes w_i\mid i\in [m]\}\subseteq \Ker(I-B)$. Since $Z$ has zero diagonal, all diagonal entries of $B$ are non-negative, and so $\tr_{\mathcal{X}}(B)\geq \dim(\Ker(I-B)\cap \mathcal{X})$. 

 At the same time, it follows from Theorem~\ref{thm:B0-construction} and Theorem~\ref{thm:Z0-norm-bound} that w.h.p. there exists a solution to~\eqref{eq:opt-constraints-decomp} with $\Tr_{\mathcal{X}}(B) \leq m+\wt{O}\left(m/n^{3/2}\right)(\dim(\mathcal{X}) - m)$ and hence the statement of the theorem follows. Here we use that w.h.p. the vectors $v_i\otimes w_i$ are linearly independent.
\end{proof}
\begin{corollary} For $m\ll n^{3/2}$, the algorithm~\eqref{eq:decomp-subspace-alg} w.h.p. returns $\vspan\{v_i\otimes w_i\mid i\in [m]\}$ after at most $2\log(n)$ iterations.
\end{corollary}

\begin{Optbox}
\begin{equation}\label{eq:decomp-subspace-alg}
\begin{gathered}
\text{Initialize}\quad  \mathcal{X}^0 = \mathbb{R}^{n^2}, \quad t = 0\hfill\\
\text{While}\quad \dim(\mathcal{X}^t)>m:\hfill \\ \quad \quad \text{Find $B$ from~\eqref{eq:opt-objective-decomp}-\eqref{eq:opt-constraints-decomp} for }  \mathcal{X}^t \quad \text{let} \quad \mathcal{X}^{t+1} = \mathcal{X}^t\cap \Ker(I-B);\quad t\text{++};\hfill \\
\text{Return}\quad  \mathcal{X}^t \hfill
\end{gathered}
\end{equation}
\end{Optbox}
\begin{proof}
Note that $\vspan\{v_i\otimes w_i\mid i\in [m]\}\subseteq \Ker(I-B)\cap \mathcal{X}^t$ for any $t$. Hence, the claim of the corollary follows from Theorem~\ref{thm:decomp-subspace-shrink}.
\end{proof}
\begin{remark}
In our experiments an optimal solution to~\eqref{eq:opt-constraints-id-sec2} usually already satisfies $\Ker(B - I) = \vspan\{v_i\otimes w_i\mid i \in [m]\}$ after the first iteration. In some sense that's expected since opt. solutions that violate this condition have measure 0 inside the space of opt. solutions (Theorems~\ref{thm:B0-construction},~\ref{thm:Z0-norm-bound}). 
\end{remark}

\subsection{The number of sampled tensors needed to recover individual components}

In this section we justify that w.h.p. it is sufficient to sample $k = O(\log(n)^2)$ tensors $\mathcal{T}_1, \mathcal{T}_2, \ldots, \mathcal{T}_k$ uniformly from a random sphere to find individual components $u_i\otimes v_i\otimes w_i$ in $\mathcal{S}_{uvw}$. Let 
\begin{equation}
\mathcal{T}_j = \sum\limits_{i\in [m]} c_{ij} u_i\otimes v_i\otimes w_i
\end{equation}
Our algorithm starts with $\mathcal{S}_{uvw}$ and uses dual certificates to subdivide it into smaller subspaces spanned by subsets of components. As discussed in Section~\ref{sec:tensor-decomposition-algorithm} after step $j$ components $u_i\otimes v_i\otimes w_i$ and $u_{i'}\otimes v_{i'}\otimes w_{i'}$ will be separated to distinct subspaces if for at least one of the tensors $\mathcal{T}_1,  \ldots, \mathcal{T}_j$ the coefficients $c_{ij}$ and $c_{i'j}$ have opposite signs.  

For an orthonormal basis $e_1, e_2, \ldots e_{m}$ of the subspace $\mathcal{S}_{uvw}$ and for a tensor $\mathcal{T}$ sampled uniformly at random from a unit sphere in $\mathcal{S}_{uvw}$ the coefficient in front of every $e_i$ independently is +1 or -1 with probability $1/2$. Note that as can be seen from Section~\ref{sec:PS-wellbalanced}, the projection of $u_i\otimes v_i\otimes w_i$ onto $\vspan\{u_{i'}\otimes v_{i'}\otimes w_{i'}\mid i'\neq i\}$ has order $\wt{O}\left(m^{1/2}/n^{3/2}\right)\ll m^{-1/2}$. This implies that with the probability close to $1/2$ for a fixed $i, i', j$ coefficients $c_{ij}$ and $c_{ij'}$ have different signs. 

Therefore, since tensors $\mathcal{T}_j$ are sampled independently, after $k$ tensors are sampled, the probability that all of them have equal sign for $c_{ij}$ and $c_{ij'}$ is at most $(2/3)^k$. Therefore, by the union bound, for $k = O(\log(n)^2)$ w.h.p. for every $i\neq i'\in [m]$ there exists $j\in [k]$ such that $\sign(c_{ij})\neq \sign(c_{i'j})$. Therefore, after $k = O(\log(n)^2)$ steps all individual components will be in separate subspaces.

\section{Analysis for higher order tensors}\label{sec:higher-order-tensors}

In this section we show that the main theorems of this paper for order-3 tensors naturally generalize to tensors of odd order $d>3$ with esssentially no changes in the proofs (for even order $d$ the main theorems follow from the  matrix completion results). We go through the key proofs of the paper and briefly sketch how they need to be modified to get the corresponding statements for order $d$ tensors.

\subsection{Dual certificate}

Let $\mathcal{V}_d = \{a_i^t\mid i\in [m],\ t\in [d]\}$ be a collection of i.i.d. uniform random vectors on a sphere $S^{n-1}$. For $\{\lambda_i>0\mid i\in [m]\}$, consider
\[\mathcal{T} = \sum\limits_{i=1}^{m} \lambda_i a_i^{1}\otimes a_i^2\otimes \ldots \otimes a_i^d.\] 

The analog of Theorem~\ref{thm:certificate-existence} for an order $d$ tensor becomes the following.

\begin{theorem}\label{thm:ord-d-certificate-exists}
Let $m \ll n^{d/2}$. Then w.h.p. there exists a strong dual certificate $\oa{A}_d$ for $\mathcal{V}_d$.
\end{theorem}

As in the case of order 3, we start by constructing a certificate candidate. Define
\begin{equation}\label{eq:order-d-Sspace}
\begin{gathered}
 \mathcal{S}_d^t = \vspan\{a_i^1\otimes \ldots \otimes a_i^{t-1}\otimes x \otimes a_i^{t+1}\otimes \ldots \otimes a_i^d\mid x\in \mathbb{R}^n,\ i\in [m]\},\\
 \mathcal{S}_d = \vspan\{\mathcal{S}_d^t \mid t\in [d]\}.
 \end{gathered}
 \end{equation}
 For $m\ll n^{d-1}$ it is not hard to verify that w.h.p. the dimension of $\mathcal{S}_d^t$ is $mn$ and every vector in $\mathcal{S}_d^t$ is uniquely determined by specifying $m$ vectors $\wt{a}_i^t\in \mathbb{R}^n$, which are the $t$-th components of each of $i$ generators. Similarly, every $X\in \mathcal{S}_d$ is (non-uniquely) defined by $md$ vectors $\{\wt{a}_i^t\in \mathbb{R}^n\mid i\in [m],\ t\in[d]\}$.

Then the analog of Theorem~\ref{thm:candidate-exists} is 
\begin{theorem}\label{thm:ord-d-cert-cand}
With high probability over the randomness of $\mathcal{V}_d$ for $m  \ll n^{d-1}$, there exists $\oa{A}_d\in S_d$, such that for any $t\in [d]$ and any $X \in \mathcal{S}_d^t$ defined by $\{\wt{x}_i^t\mid i\in [m]\}$ we have
\begin{equation}\label{eq:ord-d-cert-cand-cond}
 \langle \oa{A}_d, X\rangle = \sum\limits_{i\in [m]} \langle \wt{x}_i^t, a_i^t \rangle. 
\end{equation} 
Moreover, there exist $md$ vectors $\{\wt{a}_i^t\in \mathbb{R}^n\mid i\in [m],\ t\in[d]\}$ defining $\left(\oa{A}_d - \sum\limits_{i\in [m]}\bigotimes_{t\in [d]} a_i^t\right)$ such that for each $t\in [d]$ the matrix $V'_t$ with columns $\{\wt{a}_i^t \mid i\in [m]\}$ satisfies
\[ \Vert V'_t\Vert = \wt{O}\left(\dfrac{m}{n^{d/2}}+\dfrac{\sqrt{m}}{n^{(d-1)/2}}\right).\]
\end{theorem}  
\begin{proof}[Sketch of the proof] As in the proof of Theorem~\ref{thm:candidate-exists} we reformulate condition~\eqref{eq:ord-d-cert-cand-cond} as a system of linear equation in the form
\[ M_d\oa{A}_d = \overrightarrow{V},\] 
where  $\overrightarrow{V}$ is the vector of length $dmn$ obtained by concatenation of $a_i^t$ and $M_d$ is an $dmn\times n^d$ matrix (constructed similarly to $M$ in Eq.~\eqref{eq:M-definition}, see also Figure~\ref{fig:ord-d-M-diagram}). Let $\overrightarrow{V'}$ be the vector of length $dmn$ obtained by concatenation of $\wt{a}_i^t$. Then, as in Eq.~\eqref{eq:Y-system}, we can write the condition on $\oa{A}_d$ in the form
\begin{equation}\label{eq:ord-d-V'-linearsystem}
 M_dM_d^T\overrightarrow{V'} = \overrightarrow{V} - \dfrac{1}{d}M_dM_d^T\overrightarrow{V}.
\end{equation}
Denote by $\{f_i^t\mid i\in [m], \ t\in [d]\}$ the orthonormal basis of $\bigoplus\limits_{t\in d} \mathbb{R}^m$.
Now,  we construct a matrix 
\[R_d= I_{dmn}+\sum\limits_{s\neq t}\sum\limits_{j\in [m]}(a_j^s\otimes f_j^s)(a_j^t\otimes f_j^t)\] 
(similarly, as $R$ is constructed in Lemma~\ref{lem:R-eigenspaces}) with three distinct eigenvalues $0$, $1$ and $d$. The eigenspaces corresponding to eigenvalues 0 and $d$ are
\[\mathcal{K}_d = \vspan\{ (a_i^{t}\otimes f_i^{t}) - (a_i^s\otimes f_i^s) \mid t, s\in [d],\ i\in [m]\}, \]
\[\mathcal{D}_d = \vspan\left\{ \sum\limits_{t\in [d]} a_i^t\otimes f_i^t \mid i\in [m]\right\}. \]
We want to verify that for $m\ll n^{d-1}$ w.h.p. 
\begin{equation}\label{eq:ord-d-M-approx}
\left\Vert M_dM_d^T - R_d\right\Vert = \wt{O}\left(\dfrac{\sqrt{m}}{n^{(d-1)/2}}\right).
\end{equation}
Similarly, as described in Section~\ref{sec:M-approx}, the matrix $M_dM_d^T - R_d$ consists of $d\times d$ blocks $\{S_{tk}\mid t, k\in [d]\}$. We verify that the norm of each block is small.
\begin{lemma}\label{lem:ord-d-M-approx-triv} For $m\ll n^{d-1}$ and any $t, k\in [d]$, w.h.p. $\Vert S_{tk}\Vert = \wt{O}\left(\dfrac{\sqrt{m}}{n^{(d-1)/2}}\right)$.
\end{lemma}
\begin{proof} Denote by $\widehat{V}_t$ the matrix with $m$ columns equal to tensor products $\bigotimes_{s \neq t}a_i^s$ for $i\in [m]$. In the case when $t=k$, $S_{tt} = I_n\otimes \left(\widehat{V}_t^T\widehat{V}_t - I_m\right)$. Applying Theorem~\ref{thm:main-diagram-tool} with $\mathcal{C} = \{\{c\}\mid c\neq t, \ c\in [d]\}$ to $\widehat{V}_t^T\widehat{V}_t$ we get that 
\[\Vert \widehat{V}_t^T\widehat{V}_t - I_m\Vert = \wt{O}\left(\dfrac{\sqrt{m}}{n^{(d-1)/2}}\right).\]
Hence, the desired bound for $S_{tt}$ holds as well.

Now assume that $t\neq k$. Then $S_{tk}$ has the matrix diagram as on Figure~\ref{fig:ord-d-M-diagram} (c) with $i\neq j$. By doing a bit more careful accounting of non-equality edges than in Theorem~\ref{thm:main-diagram-tool}, similarly as it is done in Proposition~\ref{prop:M-cross-approx}, we obtain that
  \[ 
 \left\vert\mathbb{E}\Tr\left((S_{tk}^TS_{tk})^q\right)\right\vert \leq nm(2q)^{2q+1}\wt{O}\left(\dfrac{\sqrt{m}}{n^{(d-1)/2}}\right)^{2q}.
\]
Hence, the desired norm bound holds from the trace power method (Lemma~\ref{lem:trace-power-method-norm}).
\end{proof}
Now we check that $M_dM_d^T$ and $R_d$ have the same kernels.

\begin{lemma}\label{lem:ord-d-equal-kernels} For $m\ll n^{d-1}$, w.h.p. $\Ker(M_dM_d^T) = \Ker(R_d) = \mathcal{K}_d$.
\end{lemma}
\begin{proof} It is straightforward to verify that $\mathcal{K}_d\subseteq \Ker(M_dM_d^T)$. Hence, the statement of the lemma follows from Eq.~\eqref{eq:ord-d-M-approx} and the fact that $\Ker(R_d) = \mathcal{K}_d$.
\end{proof}

\begin{figure}
\begin{subfigure}[b]{0.31\textwidth}
\begin{center}
\includegraphics[height=3cm]{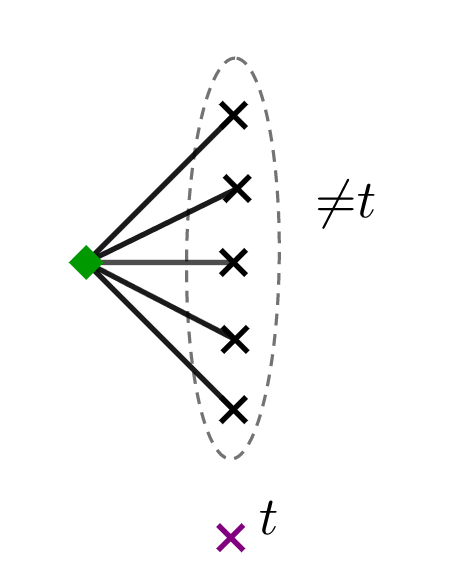}
\caption{Diagram for $(M_d)_t$.}
\end{center}
\end{subfigure}
\begin{subfigure}[b]{0.31\textwidth}
\begin{center}
\includegraphics[height=3cm]{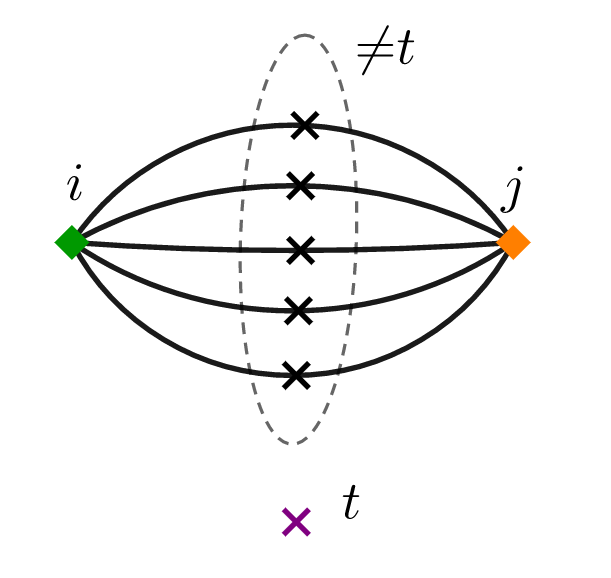}
\caption{Diagram for $F_{t, t}$.}
\end{center}
\end{subfigure}
\begin{subfigure}[b]{0.31\textwidth}
\begin{center}
\includegraphics[height=3cm]{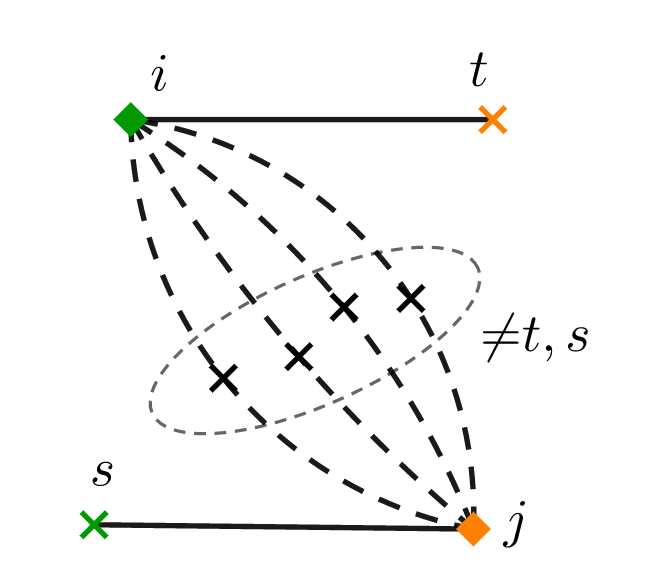}
\caption{Diagram for $F_{s,t}$, $s\neq t$}
\end{center}
\end{subfigure}
\caption{Schematic diagrams for blocks of $M_d$ and $M_dM_d^T-I$}\label{fig:ord-d-M-diagram}
\end{figure}

Finally, we can show that Eq.~\eqref{eq:ord-d-V'-linearsystem} has a well-defined solution $\overrightarrow{V'}$.
\begin{lemma} $(\overrightarrow{V} - \dfrac{1}{d}M_dM_d^T\overrightarrow{V})\perp \mathcal{K}_d$. Hence, for $m\ll n^{d-1}$, w.h.p. Eq.~\eqref{eq:ord-d-V'-linearsystem} has a well-defined solution $\overrightarrow{V'} = \left(M_dM_d^T\right)^{-1}(\overrightarrow{V} - \dfrac{1}{d}M_dM_d^T\overrightarrow{V})$. 
\end{lemma} 
\begin{proof} Observe that 
\[\dfrac{1}{d}M_dM_d^T\overrightarrow{V} = M_d\left(\sum\limits_{j\in [m]}\bigotimes\limits_{s\in [d]} a_j^s \right) = \sum\limits_{i\in [m], t\in [d]}(a_i^t\otimes f_i^t)\sum\limits_{j\in [m]}\prod\limits_{s\neq t}\langle a_i^s, a_j^s\rangle.  \]
Thus, $(\overrightarrow{V} - \dfrac{1}{d}M_dM_d^T\overrightarrow{V})\perp \mathcal{K}_d$, as the inner product
\[\left\langle (\overrightarrow{V} - \dfrac{1}{d}M_dM_d^T\overrightarrow{V}), a_i^t\otimes f_i^t \right\rangle = 1 - \sum\limits_{j\in[m]}\prod\limits_{s\in [d]}\langle a_i^s, a_j^s\rangle\]
does not depend on $t$.
\end{proof}

Next, we need to show that the desired norm bounds for $V'_t$ hold for $t\in [d]$. Similarly as in Section~\ref{sec:corr-terms} we can approximate $\left(M_dM_d^T\right)^{-1}$ up to a small error with an IP graph matrix and analyze the essential IP graph matrix components involved in $V'_t$ (as in Section~\ref{sec:explicit-approx}).

Lemma~\ref{lem:ord-d-equal-kernels} and Eq.~\eqref{eq:ord-d-M-approx} imply that for some $\mathcal{E}_{M_d}$ with norm at most $\wt{O}\left(\dfrac{\sqrt{m}}{n^{(d-1)/2}}\right)$ we have
\[ (M_dM_d^T)_{\mathcal{K}_d^{\perp}}^{-1} = (R_d+\mathcal{E}_{M_d})_{\mathcal{K}_d^{\perp}}^{-1} = (R_d)_{\mathcal{K}_d^{\perp}}^{-1}\left(I+(R_d)_{\mathcal{K}_d^{\perp}}^{-1}\mathcal{E}_{M_d}\right)^{-1} \] 

Using that $R_d$ has very simple spectrum we know that $(R_d)_{\mathcal{K}_d^{\perp}}^{-1} = P_{\mathcal{K}_d^{\perp}}- \frac{d-1}{d}P_{\mathcal{D}_d}$, where $P_{\mathcal{D}_d}$ is a projector on $\mathcal{D}_d$. Hence, by expanding the inverse of $(I+(R_d)_{\mathcal{K}_d^{\perp}}^{-1}\mathcal{E}_{M_d})$ into power series in equation above, and using the formula for $(R_d)_{\mathcal{K}_d^{\perp}}^{-1}$ we can deduce the following statement.
\begin{lemma}\label{lem:ord-d-M-approx} For $m\ll n^{d-1}$ and $p\geq 0$ w.h.p. there exists a matrix $M^{[p]}_{d, inv}$, which can be written as a polynomial of $M_dM_d^T$ and $P_{\mathcal{D}_d}$ of degree at most $3p$, such that
\[ \left\Vert  \left(M_dM_d^T\right)_{\mathcal{K}_d^{\perp}}^{-1} -  M^{[p]}_{d, inv}\right\Vert = \left(\wt{O}\left(\dfrac{\sqrt{m}}{n^{(d-1)/2}}\right)\right)^{p+1}. \]
\end{lemma}
Using this lemma, we will show below that the following vector is a linear combination of vectors with IP graph matrix structure
\begin{equation}\label{eq:ord-d-V'GM}
 \overrightarrow{V'_{GM}} = M^{[p]}_{d, inv}\left(\overrightarrow{V} - \dfrac{1}{d}M_dM_d^T\overrightarrow{V}\right)
\end{equation}
For $n\ll m^{d-1}$, taking $p$ large enough, by Lemma~\ref{lem:ord-d-M-approx}, w.h.p. we can make the difference
\[ \left\Vert \overrightarrow{V'} -  \overrightarrow{V'_{GM}}\right\Vert = \Vert \overrightarrow{V}\Vert \left(\wt{O}\left(\dfrac{\sqrt{m}}{n^{(d-1)/2}}\right)\right)^{p+1} = \wt{O}\left(\sqrt{m}\right)\left(\wt{O}\left(\dfrac{\sqrt{m}}{n^{(d-1)/2}}\right)\right)^{p+1}\]
as small as we need. 

Consider $d$ matrices $V'_{t, GM}\in \mathbb{R}^{n\times m}$ for $t\in [d]$, obtained by reshaping the $d$ blocks of $dmn$ vector $\overrightarrow{V'}_{GM}$ given by Eq.~\eqref{eq:ord-d-V'GM} for $p=1$ into $n\times m$ matrices in a natural way. Then, clearly, for each $t\in [d]$ and $m\ll n^{d/2}$ w.h.p.
\[\left\Vert V'_{t} - V'_{t, GM}\right\Vert_{F}\leq \left\Vert \overrightarrow{V'} -  \overrightarrow{V'_{GM}}\right\Vert = \wt{O}\left(\dfrac{m^{3/2}}{n^{d-1}}\right) = \wt{O}\left(\dfrac{m}{n^{d/2}}\right) = o(1).\]

\begin{lemma} Let $\mathfrak{G}$ be the collection of tuples of IP graph matrices $\left( X_{t}\in \mathbb{R}^{n\times m}\mid t\in d\right)$ such that
\begin{enumerate}
\item $type(\myscr{Ver}_L(X_t)) = \{t\}$ and $type(\myscr{Ver}_R(X_t)) = \{*\}$,
\item $X_t$ is $\mathcal{C}$-connected for every $t\in [d]$, where $\mathcal{C} = \binom{[d]}{2}$.
\end{enumerate}  
Let $\overrightarrow{X}$ be a vector of length $dmn$ obtained by concatenation of $(X_t)_{t\in [d]}$ being reshaped into vectors.
If $(X_t)_{t\in [d]}$ is in $\mathfrak{G}$, then tuples that correspond to $\left(M_dM_d^T - I\right)\overrightarrow{X}$ and $P_{\mathcal{D}_d}\overrightarrow{X}$ are linear combinations of at most $d$ tuples from $\mathfrak{G}$.
\end{lemma}
\begin{proof} Recall that $\left(M_dM_d^T - I\right)$ is an $nmd\times nmd$ matrix which naturally consist of $d\times d$ blocks. Denote these blocks $F_{s, t}$ for $s, t\in [d]$. By the definition of $M$, the entries of these matrices can be computed as
\[ (F_{s, t})_{(i, x), (j, y)} = \langle a_{j}^{s}, e_x \rangle \langle e_y, a_{i}^{t}\rangle \prod_{c\neq s, t} \langle a^{c}_{i}, a^{c}_j\rangle \quad \text{for } i, j\in [m]\ x, y\in [n] \quad \text{if } s\neq t;\]
\[ (F_{t, t})_{(i, x), (j, y)} =  \mathbf{1}[x = y]\mathbf{1}[i \neq j]\prod_{c\neq t} \langle a^{c}_{i}, a^{c}_j\rangle \quad \text{for } i, j\in [m]\ x, y\in [n].\]
See Figure~\ref{fig:ord-d-M-diagram} for the matrix diagrams of $F_{s, t}$. 

Assume that $(Y_t)_{t\in [d]}$ is the tuple of matrices that correspond to $\left(M_dM_d^T - I\right)\overrightarrow{X}$. Then
\[ \overrightarrow{Y_s} = \sum\limits_{t\in [d]} F_{s, t}\overrightarrow{X_t}. \]
Clearly, $Y_s$ has the desired type for every $s$. Note that the matrix diagram of $F_{s, t}\overrightarrow{X_t}$ has one more node than the diagram of $X_s$ (see Figure~\ref{fig:ord-d-MPD-diagram}). Moreover, if $s=t$, this node is connected by edges of $d-1$ distinct colors to a vertex of $\mathcal{MD}(X_s)$; and if $s\neq t$ it is connected by $d-1$ edges of distinct colors to at most two vertices of $\mathcal{MD}(X_s)$. Therefore, if $\mathcal{MD}(X_t)$ is $\binom{[d]}{2}$-connected for all $t$, then $Y_s$ is a linear combination of IP graph matrices with $\binom{[d]}{2}$-connected diagram.  

The schematic diagram of matrices involved in $(P_{\mathcal{D}_d})\overrightarrow{X}$ is given on Figure~\ref{fig:ord-d-MPD-diagram} (c).  It is evident that the resulting matrix diagrams have the desired type and are $\binom{[d]}{2}$-connected, if  $\mathcal{MD}(X_t)$ is $\binom{[d]}{2}$-connected for all $t$.
\begin{figure}
\begin{subfigure}[b]{0.31\textwidth}
\begin{center}
\includegraphics[height=3cm]{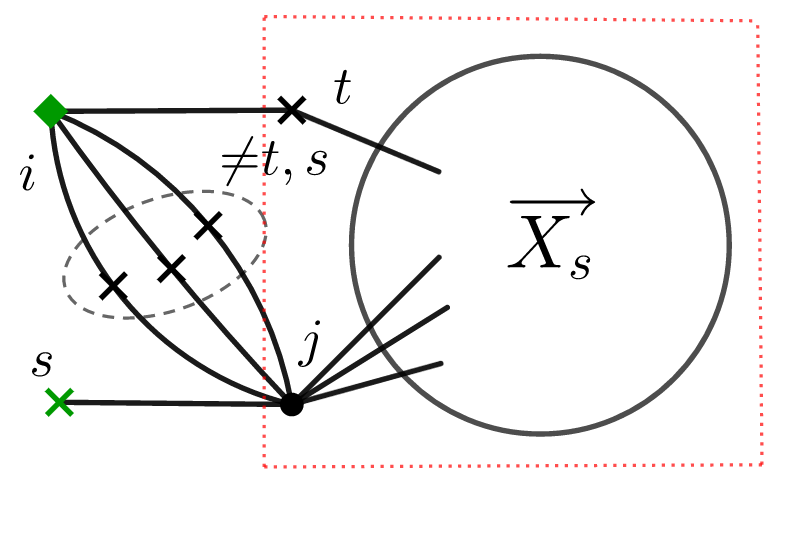}
\caption{Diagram for $F_{s, t}\overrightarrow{X_t}$, $s\neq t$.}
\end{center}
\end{subfigure}
\begin{subfigure}[b]{0.31\textwidth}
\begin{center}
\includegraphics[height=3cm]{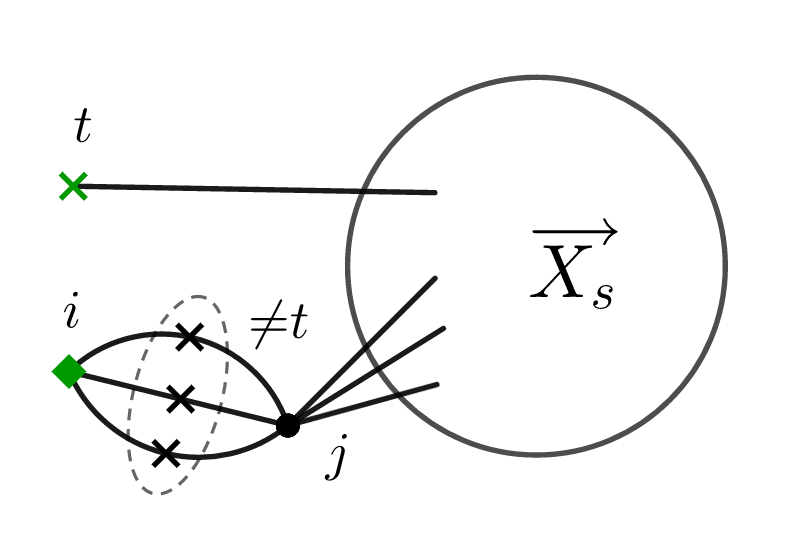}
\caption{Diagram for $F_{t, t}\overrightarrow{X_t}$.}
\end{center}
\end{subfigure}
\begin{subfigure}[b]{0.31\textwidth}
\begin{center}
\includegraphics[height=3cm]{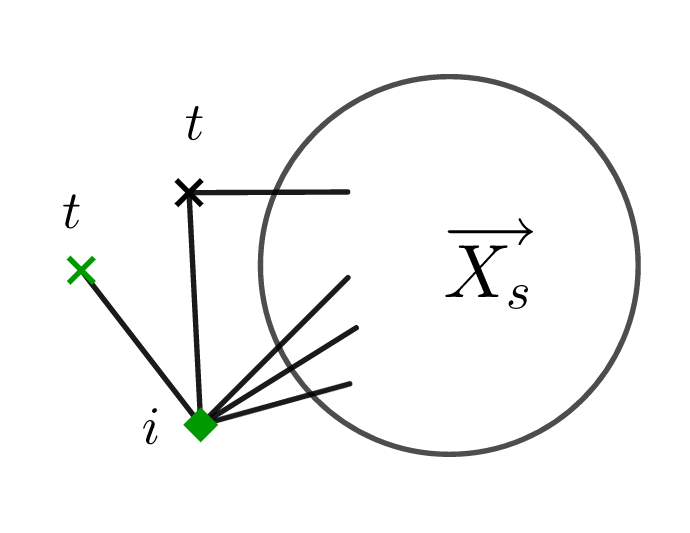}
\caption{Diagram for $(P_{\mathcal{D}_d})_t\overrightarrow{X_t}$}
\end{center}
\end{subfigure}
\caption{Diagram illustrating multiplication by $M_dM_d^T-I$ and $P_{\mathcal{D}_d}$}\label{fig:ord-d-MPD-diagram}
\end{figure}
\end{proof}
\begin{corollary}\label{cor:ord-d-GM-structure-corr} For every $t\in[d]$,  $V'_{t, GM}$ is a linear combination of $O_{p, d}(1)$ IP graph matrices with $\binom{[d]}{2}$-connected matrix diagram and $type(\myscr{Ver}_{L}) = \{t\}$ and $type(\myscr{Ver}_{R}) = \{*\}$.
\end{corollary}

Finally, to finish the proof of Theorem~\ref{thm:ord-d-cert-cand}, we establish norm bounds for $V'_{t, GM}$.

\begin{lemma} Let $(V'_{t, GM})_{t\in[d]}$ be the tuple of $d$ matrices that correspond to $\overrightarrow{V'_{GM}} = M^{[1]}_{d, inv}E$, where $\overrightarrow{E} = \left(\overrightarrow{V} - \dfrac{1}{d}M_dM_d^T\overrightarrow{V}\right)$ and  
\[M^{[1]}_{d, inv} = 2\left(I-\dfrac{d-1}{d}P_{\mathcal{D}_d}\right) + \left(I-\dfrac{d-1}{d}P_{\mathcal{D}_d}\right)M_dM_d^T\left(I-\dfrac{d-1}{d}P_{\mathcal{D}_d}\right).\]
Assume $m\ll n^{d/2}$, then w.h.p. for all $t\in [d]$,
\[ \left\Vert V'_{t, GM} \right\Vert  = \wt{O}\left(\dfrac{m}{n^{d/2}}+\dfrac{\sqrt{m}}{n^{(d-1)/2}}\right).\] 
\end{lemma}
\begin{proof} Consider the tuple of $\mathbb{R}^{n\times m}$ matrices $(E_t)_{t\in [d]}$ that corresponds to $\overrightarrow{E}$. Let $V_t$ be the $n\times m$ matrix with columns $a_{i}^t$ and $\widehat{V}_t$ as in the proof of Lemma~\ref{lem:ord-d-M-approx-triv}. Then direct computation shows that $E_t = V_t\left(\widehat{V}_t^T\widehat{V}_t - I_m\right)$ is an IP graph matrix with the diagram on Figure~\ref{fig:ord-d-V-corrections} (a). Hence 
\[\Vert E_t \Vert = \wt{O}\left(1+\dfrac{\sqrt{m}}{\sqrt{n}}\right)\cdot \wt{O}\left(\dfrac{\sqrt{m}}{n^{(d-1)/2}}\right) = \wt{O}\left(\dfrac{m}{n^{d/2}}+\dfrac{\sqrt{m}}{n^{(d-1)/2}}\right).\]
\begin{figure}
\begin{subfigure}[b]{0.3\textwidth}
\begin{center}
\includegraphics[height = 2cm]{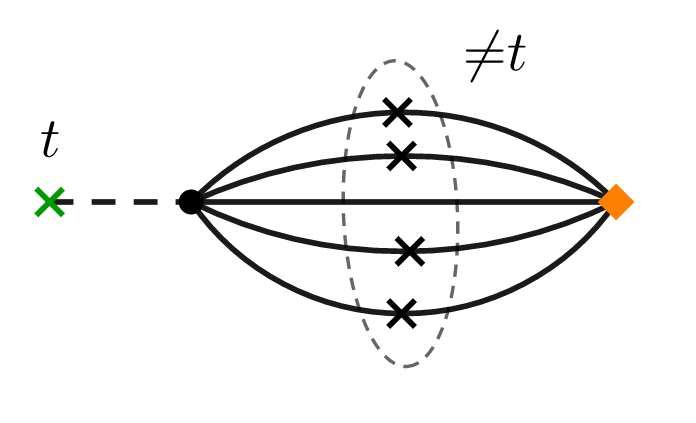}
\caption{Diagram for $E_t$}
\end{center}
\end{subfigure}
\begin{subfigure}[b]{0.3\textwidth}
\begin{center}
\includegraphics[height = 2cm]{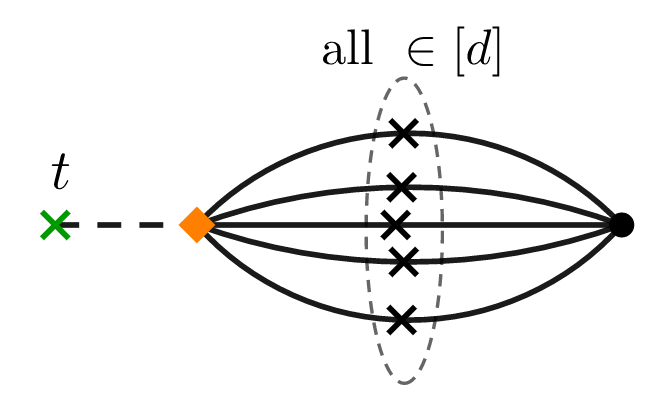}
\caption{Diagram for $(P_{\mathcal{D}_d}\overrightarrow{E})_t$}
\end{center}
\end{subfigure}
\begin{subfigure}[b]{0.3\textwidth}
\begin{center}
\includegraphics[height = 2cm]{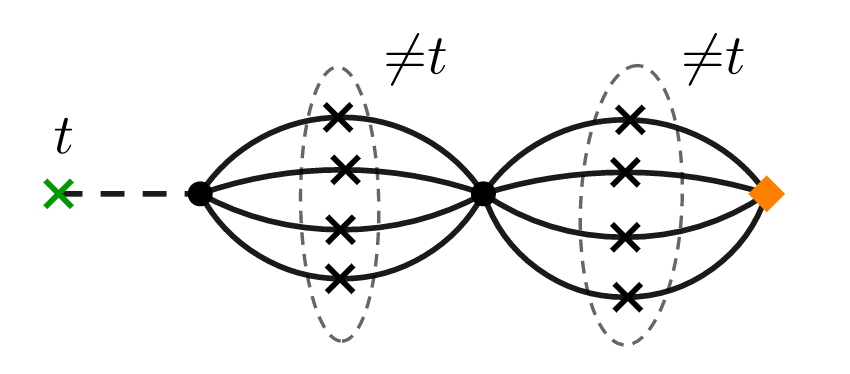}
\caption{Diagram for $F_{t,t}E_t$}
\end{center}
\end{subfigure}

\begin{subfigure}[b]{0.3\textwidth}
\begin{center}
\includegraphics[width = 0.9\textwidth]{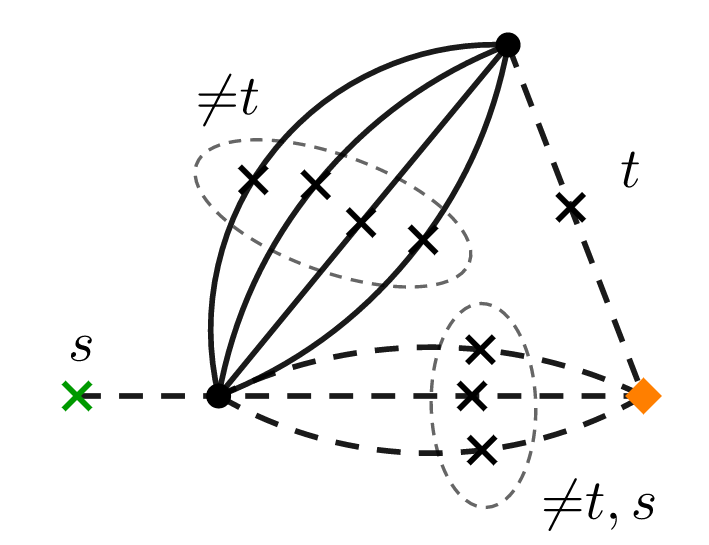}
\caption{Diagram for $F_{s,t}E_t$}
\end{center}
\end{subfigure}
\begin{subfigure}[b]{0.3\textwidth}
\begin{center}
\includegraphics[width = 0.9\textwidth]{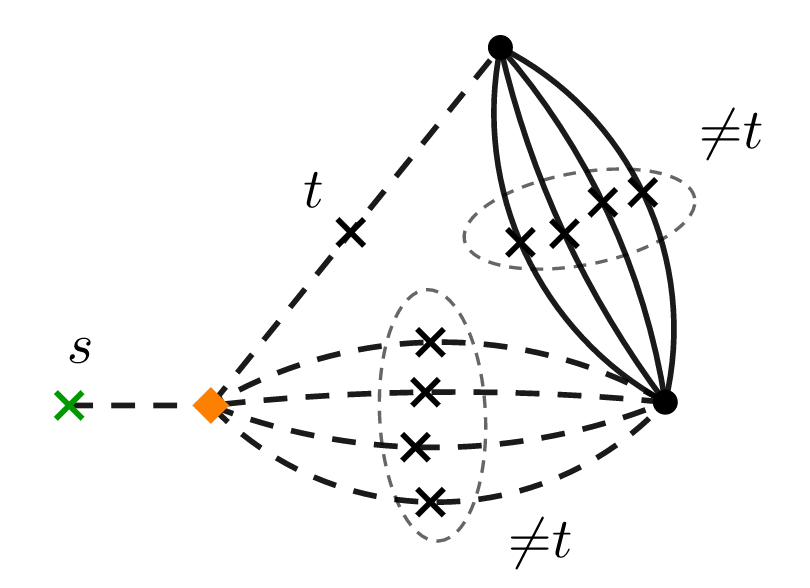}
\caption{$(P_{\mathcal{D}_d}(M_dM_d^T-I)E)_s$}
\end{center}
\end{subfigure}
\begin{subfigure}[b]{0.3\textwidth}
\begin{center}
\includegraphics[width = 0.9\textwidth]{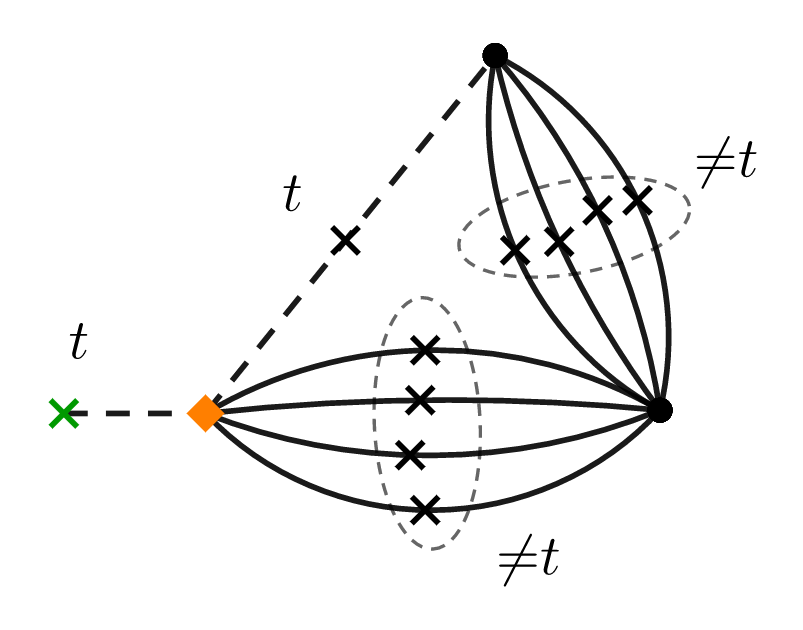}
\caption{$(P_{\mathcal{D}_d}(M_dM_d^T-I)E)_t$}
\end{center}
\end{subfigure}
\caption{Expanded matrix diagrams for matrices involved in $V'_{t, GM}$.}\label{fig:ord-d-V-corrections}
\end{figure}
Thus to get the desired norm bound for $V_{t, GM}'$ it is sufficient to prove the norm bounds for the matrices corresponding to $P_{\mathcal{D}_d}\overrightarrow{E}$, $(M_{d}M_{d}^T- I)\overrightarrow{E}$, $P_{\mathcal{D}_d}(M_{d}M_{d}^T- I)\overrightarrow{E}$ and $P_{\mathcal{D}_d}(M_{d}M_{d}^T- I)P_{\mathcal{D}_d}\overrightarrow{E}$. The expanded matrix diagrams of IP graph matrices, that are involved in these matrices are presented on Figure~\ref{fig:ord-d-V-corrections} (b)-(f).

It is not hard to see that for $m\ll n^{d/2}$
\[\Vert (P_{\mathcal{D}_d}\overrightarrow{E})_t \Vert = \Vert V_t(V^TV-I) \Vert = \wt{O}\left(\dfrac{m}{n^{(d+1)/2}}+\dfrac{\sqrt{m}}{n^{d/2}}\right), \quad \text{so} \quad \]
\[ \Vert (P_{\mathcal{D}_d}\overrightarrow{E})_t \Vert_F = \wt{O}\left(\dfrac{m}{n^{d/2}}+\dfrac{\sqrt{m}}{n^{(d-1)/2}}\right).\]
Thus, in particular, 
\[ \left\Vert \left((M_dM_d^T-I)P_{\mathcal{D}_d}\overrightarrow{E}\right)_t \right\Vert \leq \left\Vert \left((M_dM_d^T-I)P_{\mathcal{D}_d}\overrightarrow{E}\right)_t \right\Vert_F   = \wt{O}\left(\dfrac{m}{n^{d/2}}+\dfrac{\sqrt{m}}{n^{(d-1)/2}}\right)\]
Finally, observe that any matrix on Figure~\ref{fig:ord-d-V-corrections} (c) - (f) has form $V_t\cdot X$, where on diagrams (c), (e) and (f) $X$ is $(\binom{[d]}{1}\setminus \{t\})$-connected and on diagram (d) $X$ is $(\binom{[d]}{1}\setminus \{s\})$-connected. Moreover, on each of these diagrams $X$ has at least $d-1$ non-equality edges. Therefore, by Theorem~\ref{thm:main-diagram-tool}, $\Vert X\Vert = \wt{O}\left(\dfrac{\sqrt{m}}{n^{(d-1)/2}}\right)$. Hence, $\Vert V_tX \Vert = \wt{O}\left(\dfrac{m}{n^{d/2}}+\dfrac{\sqrt{m}}{n^{(d-1)/2}}\right)$.
\end{proof}
This finishes the sketch of the proof of Theorem~\ref{thm:ord-d-cert-cand}.
\end{proof} 

The statement of Theorem~\ref{thm:ord-d-certificate-exists} can be proven as in Section~\ref{sec:dual-certificate}, or can be deduced from the existence of an SOS dual certificate, proof of which we sketch below.

\subsection{SOS dual certificate}

Next we consider the higher order analog of the semidefinite program~\eqref{eq:opt-constraints-id-intr}. 

Let $d_1, d_2, d_3$ be positive integers with $d_1+d_2+d_3 = d$. We identify the index in $s\in [n^d]$ with a triple $(i, j, k)\in [n^{d_1}]\times [n^{d_2}]\times [n^{d_3}]$ in a natural way (that respects indexing of tensor product). 

Consider the following optimization problem
\begin{Optbox}
\begin{equation}\label{eq:opt-objective-ord-d}
\text{Maximize }\quad \langle \oa{A_d}, \mathcal{T} \rangle\quad \qquad \text{Subject to:}
\end{equation}
\vspace{-0.5cm}
\begin{equation}\label{eq:opt-constraints-ord-d}
\quad \left(\begin{matrix}
I_{n^{d_1}} & -A_d\\
-A_d^T & Z_d+B_d
\end{matrix}\right) \succeq 0, \quad Z_d\equiv_{poly} 0, \quad B_d\preceq I_{n^{d_2+d_3}}, \quad A_{i, (j, k)} = \oa{A}_{(i, j, k)}.
\end{equation}
\end{Optbox}
Here, $Z\equiv_{poly} 0$ means that for any $\{x^t \in \mathbb{R}^n\mid t\in [d_1+1, d]\}$ we have
\[\left(\bigotimes\limits_{t\in [d_1+1, d]} x^t\right)^TZ\left(\bigotimes\limits_{t\in [d_1+1, d]} x^t\right) = 0.\] 

Then the following analog of Theorem~\ref{thm:ABZ-exists} holds.

\begin{theorem}\label{thm:ord-d-SOS-cert} Let $m \ll n^{d/2}$. Then w.h.p. there exists a solution $(\oa{A}_d, B_d, Z_d)$ to \eqref{eq:opt-constraints-ord-d}, where $\oa{A}$ is a certificate candidate for $\mathcal{V}$ given by Theorem~\ref{thm:ord-d-cert-cand}.
\end{theorem} 

We follow the same proof strategy as before: first we construct a candidate matrix $B_{d, 0}$ which is close to the desired $B_d$; after that we construct a zero polynomial correction with small norm, that fixes $B_{d, 0}$.

\begin{definition}
Let $X, Y\in \mathbb{R}^{n^d}$ and $d_1+d_2+d_3 = d$ for integers $d_1, d_2, d_3\geq 1$. Define an $n^{d_2+d_3}\times n^{d_2+ d_3}$ matrix
\[ (\mathcal{B}_{d_1, d_2, d_3}(X, Y))_{(j, k)(j', k')} = \sum\limits_{i\in [n^{d_1}]}X_{(i, j, k')}Y_{(i, j', k)},\]
for $j, j'\in [n]^{d_2}$ and $k, k'\in [n]^{d_3}$.
\end{definition}

\begin{observation} $\Vert  \mathcal{B}_{d_1, d_2, d_3}(X, Y) \Vert_F\leq \Vert X\Vert\cdot\Vert Y\Vert$.
\end{observation}
Let $\overrightarrow{V'}$ be as above and  let $\overrightarrow{V'_{GM}}$ be given by Eq.~\eqref{eq:ord-d-V'GM}, for $p = 1$. Define 
\[\oa{A}_{d, GM} = M_d^T\overrightarrow{V'_{GM}}+\frac{1}{d}M_d^T\overrightarrow{V}.\]
Then, for $\oa{A}_{d, sm} =  \oa{A}_{d} - \oa{A}_{d, GM}$
\begin{equation}\label{eq:ord-d-Asm-bound}
\Vert \oa{A}_{d, sm} \Vert = \Vert \oa{A}_{d} - \oa{A}_{d, GM} \Vert  = \Vert M_d^T\left(\overrightarrow{V'} - \overrightarrow{V'_{GM}}\right)\Vert  = \left(\wt{O}\left(\dfrac{\sqrt{m}}{n^{(d-1)/2}}\right)\right)^{p+1}.
\end{equation}
Therefore, for $m\ll n^{d/2}$, since $\Vert \oa{A}_d \Vert = \wt{O}\left(\sqrt{m}\right)$ and since $3d/4\leq d-1$ for $d\geq 4$, we obtain 
\begin{equation}\label{eq:ord-d-BGM-approx-cor}
\begin{gathered}
 \Vert \mathcal{B}_{d_1, d_2, d_3}(\oa{A}_{d}, \oa{A}_{d}) - \mathcal{B}_{d_1, d_2, d_3}(\oa{A}_{d, GM}, \oa{A}_{d, GM}) \Vert_F \leq \\
 \leq 2\Vert \oa{A}_{d, sm} \Vert\Vert \oa{A}_{d} \Vert+\Vert \oa{A}_{d, sm} \Vert^2 =  \wt{O}\left(\dfrac{m}{n^{d-1}}\right)\cdot \wt{O}\left(\sqrt{m}\right) = \wt{O}\left(\dfrac{m^{3/2}}{n^{d-1}}\right) = o(1)
 \end{gathered}
 \end{equation} 

\begin{theorem} Let $\mathcal{A}_d$ be the certificate candidate constructed in Theorem~\ref{thm:ord-d-cert-cand}. Assume that $d_1$, $d_2$ and $d_3$ satisfy the triangle inequality. Then, for $m\ll n^{d/2}$, w.h.p.
\[ \Vert \mathcal{B}_{d_1, d_2, d_3}(\oa{A}_{d}, \oa{A}_{d}) - P_{\mathcal{L}} \Vert = o(1), \]
where $\mathcal{L} = \vspan\left\{\bigotimes\limits_{t\in [d_1+1, d]} a_i^t\mid i\in [m]\right\}$. 
\end{theorem}
\begin{proof} Define $\wt{P}_{\mathcal{L}} = \wt{P}_{d_2, d_3} = \sum\limits_{i\in [m]} \left(\bigotimes\limits_{t\in [d]} a_i^t\right)\left(\bigotimes\limits_{t\in [d]} a_i^t\right)^T$ (see Figure~\ref{fig:ord-d-B-constr}). It is easy to show that $\Vert P_{\mathcal{L}} - \wt{P}_{\mathcal{L}}\Vert = \wt{O}\left(\sqrt{m}/n^{(d-1)/2}\right)$. Using the inequality~\eqref{eq:ord-d-BGM-approx-cor} it is sufficient to verify that the lemma below is true.
\end{proof}

\begin{lemma}\label{lem:ord-d-BGMP-bound} For $m\ll n^{d/2}$, w.h.p. $\Vert \mathcal{B}_{d_1, d_2, d_3}(\oa{A}_{d, GM}, \oa{A}_{d, GM}) - \wt{P}_{\mathcal{L}} \Vert = \wt{O}\left(\dfrac{m^{1/d}}{n^{1/2}}\right) = o(1)$.
\end{lemma}
\begin{proof} It follows from the proof of Theorem~\ref{thm:ord-d-cert-cand} that $\oa{A}_{d, GM}$ is a linear combination of IP graph matrices $\oa{A}_{d, part}$ with $type(\myscr{Ver}_L(\oa{A}_{d, part})) = (\{1\}, \{2\}, \ldots, \{d\})$ and $type(\myscr{Ver}_R(\oa{A}_{d, part})) = ()$. 
\begin{figure}
\begin{center}
\includegraphics[width = 0.95\textwidth]{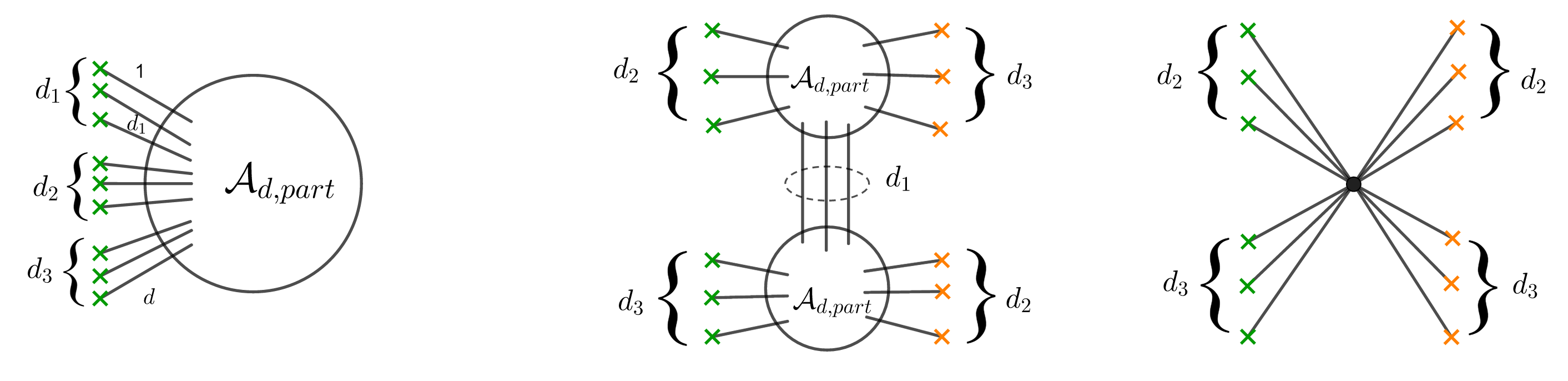}
\caption{Schematic diagrams for $\oa{A}_{d, part}$ (left), $\mathcal{B}_{d_1, d_2, d_3}(\oa{A}_{d, part1}, \oa{A}_{d, part2})$ (center) and $\wt{P}_{d_2, d_3} = \wt{P}_{\mathcal{L}}$ (right).}\label{fig:ord-d-B-constr}
\end{center}
\end{figure}
The sum of the absolute values of the coefficients in the sum is bounded by a constant $C$ that depends on $d$ and the approximation order $p$ (in Eq.~\eqref{eq:ord-d-Asm-bound}) only (see Corollary~\ref{cor:ord-d-GM-structure-corr}).  Moreover, the only summand $\oa{A}_{d, part}$ that has no non-equality edges (see Def.~\ref{def:non-equality-edges}) is $\sum\limits_{i\in [m]} \bigotimes\limits_{t\in [d]} a_i^t$.  

The matrix diagram of $B' := \mathcal{B}_{d_1, d_2, d_3}(\oa{A}_{d, part1}, \oa{A}_{d, part2})$ has form as on Figure~\ref{fig:ord-d-B-constr}. And the only term that has no non-equality edges is $\wt{P}_{d_2, d_3}$ and it appears with coefficient 1. Therefore, 
\[ \begin{gathered}
 \Vert \mathcal{B}_{d_1, d_2, d_3}(\oa{A}_{d, GM}, \oa{A}_{d, GM}) - \wt{P}_{d_2, d_3} \Vert\leq \\
 \leq C^2 \max\limits_{\oa{A}_{d, part1},\ \oa{A}_{d, part2}}\{B' = \mathcal{B}_{d_1, d_2, d_3}(\oa{A}_{d, part1}, \oa{A}_{d, part2}) \mid B'\neq \wt{P}_{d_2, d_3}\},
 \end{gathered}\]
where maximum is taken over the IP graph matrices involved in $\oa{A}_{d, GM}$.  Since each  $\oa{A}_{d, part}$ is $\displaystyle \binom{[d]}{2}$-connected, then $B'$ is $\mathcal{C}$-connected, where $\mathcal{C}$ is a 2-regular 3-partite graph with parts $[d_1], [d_1+1, d_1+d_2], [d_1+d_2+1, d]$. The existence of such $\mathcal{C}$ follows from the following observation.
\begin{observation} Let $x\leq y\leq z\leq x+y$ be positive integers. Then there exists a $3$-partite $2$-regular graph with parts of size $x$, $y$ and $z$.
\end{observation} 
\begin{proof} The existence of such graph is equivalent to a proper coloring of a cycle of length $x+y+z$ with 3 colors, with $x$, $y$ and $z$ vertices of corresponding color. We enumerate the vertices of the cycle and color vertices $1, 3, \ldots, 2z-1$ into color $3$ and color vertices between $2z$ and $x+y+z$ into colors $2$ and $1$ in alternating way. Note that at this point we colored at most $(x+y-z)/2\leq x$ vertices in color $1$, so we are able to color the rest of the vertices to satisfy the required conditions. 
\end{proof}
Note that $|\mathcal{C}| = d$ and each element belongs to precisely 2 sets. Moreover, the matrix diagram of $B'$ has at least one non-equality edge, so by Theorem~\ref{thm:main-diagram-tool}, 
\[ \Vert \mathcal{B}_{d_1, d_2, d_3}(\oa{A}_{d, part1}, \oa{A}_{d, part2}) \Vert = \wt{O}\left(\dfrac{m^{1/d}}{n^{1/2}}\right) = o(1).\] 
\end{proof} 

Next, we are looking for a zero polynomial matrix $Z_{0, d_2, d_3}$ with $\Vert Z_{0, d_2, d_3} \Vert = o(1)$ that satisfies 
\begin{equation}\label{eq:order-d-Z0-condition}
Z_{0, d_2, d_3}\left(\bigotimes\limits_{t\in [d_1+1, d]} a_t\right) = \left(\mathcal{B}_{d_1, d_2, d_3}(\oa{A}_{d}, \oa{A}_{d}) - {P}_{\mathcal{L}}\right)\left(\bigotimes\limits_{t\in [d_1+1, d]} a_t\right)\quad i\in [m].
\end{equation}

\subsection{Zero polynomial correction}

To use the construction from Section~\ref{sec:zero-poly-corr} we take $d_1 = 1$, $d_2 = d_3 = (d-1)/2$. Denote 
\begin{equation}\label{eq:a-to-u-transition}
 u_i = a_i^1,\qquad  v_i =\left( \bigotimes\limits_{t\in [2, (d+1)/2]} a_i^t\right) \quad \text{and} \quad  w_i = \left(\bigotimes\limits_{t\in [(d+3)/2, d]} a_i^t\right).
 \end{equation}

Then, by Theorem~\ref{thm:zero-poly-correction-constr}
with $m\ll \left(n^{(d-1)/2}\right)^2$,\  
 $Z_{0, d_2, d_3} :=  \mathcal{Z}\left(\mathcal{B}_{d_1, d_2, d_3}(\oa{A}_{d}, \oa{A}_{d}) - {P}_{\mathcal{L}}\right)$ satisfies Eq.~\eqref{eq:order-d-Z0-condition}. 
 
 Now we verify that $\Vert Z_{0, d_2, d_3} \Vert = o(1)$. We use $\mathcal{B}_{d_1, d_2, d_3}(\oa{A}_{d, GM}, \oa{A}_{d, GM})$ as an IP graph matrix approximation to $\mathcal{B}_{d_1, d_2, d_3}(\oa{A}_{d}, \oa{A}_{d})$ (see Eq.~\eqref{eq:ord-d-BGM-approx-cor}). Also, we use $Q^{[2p]}_{inv}$ provided by Lemma~\ref{lem:QQT-approx} for large enough $p$ (note that the statement holds with $n:= n^{(d-1)/2}$) as an approximation for $(QQ^T)^{-1}$. Then it is not hard to verify that an IP graph matrix approximation $\mathcal{Z}_{GM}\left(\mathcal{B}_{d_1, d_2, d_3}(\oa{A}_{d}, \oa{A}_{d}) - {P}_{\mathcal{L}}\right)$, constructed similarly as in Section~\ref{sec:Z-small-terms}, satisfies
 \begin{equation}\label{eq:order-d-zero-poly-approx}
  \Vert \mathcal{Z}\left(\mathcal{B}_{d_1, d_2, d_3}(\oa{A}_{d}, \oa{A}_{d}) - {P}_{\mathcal{L}}\right) - \mathcal{Z}_{GM}\left(\mathcal{B}_{d_1, d_2, d_3}(\oa{A}_{d}, \oa{A}_{d}) - P_{\mathcal{L}}\right)\Vert  = \wt{O}\left(\dfrac{m}{n^{d/2}}\right) = o(1)
  \end{equation}
 In Section~\ref{sec:zero-poly-corr-norm} we proved that if $X$ is $\mathcal{C}_{3/2}$-boundary-connected for $\mathcal{C}_{3/2} = \{\{u, v\}, \{u, w\}, \{v, w\}\}$, then $\mathcal{Z}_{GM}(X)$ is $\mathcal{C}_{3/2}$-connected. Note that after defining $u_i, v_i, w_i$ as in Eq.~\eqref{eq:a-to-u-transition} we can say that $type(\myscr{Ver}_L(X)) =   type(\myscr{Ver}_R(X)) = (v, w)$ for any IP graph matrix involved in $\mathcal{B}_{d_1, d_2, d_3}(\oa{A}_{d, GM}, \oa{A}_{d, GM})$. Note however, that $X$ itself can be not an IP graph matrix over vector collection $\{u_i, v_i, w_i\}$. Fix some collection $\mathcal{C}$ of subsets of $[d]$ which is given by a 2-regular 3-partite graph with parts $[d_1]$, $[d_1+1, d_1+d_2]$ and $[d_1+d_2+1, d_3]$. Formally, assign to each color in these groups a label $u$, $v$ and $w$ correspondingly. We showed above that $X$ is $\mathcal{C}$-boundary-connected, hence $X$ is $\mathcal{C}_{3/2}$-boundary-connected with respect to this formal labels. Then, using arguments from Section~\ref{sec:zero-poly-corr-norm}, $\mathcal{Z}_{GM}(X)$ is $\mathcal{C}_{3/2}$-connected, with respect to formal $u,v,w$-labels inside the diagram for $X$ and true $u,v,w$-labels outside the diagram for $X$. By replacing every $u$-, $v$- and $w$-edge outside of the diagram for $X$ with $d_1$, $d_2$ and $d_3$ edges of colors in $[d_1]$, $[d_1+1, d_1+d_2]$ and $[d_1+d_2+1, d_3]$, we obtain a matrix diagram for $\mathcal{Z}_{GM}(X)$ with respect to the collection of vectors $\{a_i^t\mid i\in [m],\ t\in [d]\}$ with $d$ colors. Therefore, since $X$ is $\mathcal{C}$-boundary-connected, $\mathcal{Z}_{GM}(X)$ is $\mathcal{C}_{3/2}$-connected and $type(\myscr{Ver}_L(X)) =   type(\myscr{Ver}_L(X)) = [d_1+1, d]$, the diagram for $\mathcal{Z}_{GM}(X)$ is $\mathcal{C}$-connected. 
 
 Since every IP graph matrix $X$ involved in  $\left(\mathcal{B}_{d_1, d_2, d_3}(\oa{A}_{d}, \oa{A}_{d}) - \wt{P}_{\mathcal{L}}\right)$ has matrix diagram with at least one non-equality edge, by applying Theorem~\ref{thm:main-diagram-tool} to $\mathcal{Z}_{GM}(X)$ with respect to $\mathcal{C}$ as in the previous paragraph, we obtain
 \[ \Vert \mathcal{Z}_{GM}\left(\mathcal{B}_{d_1, d_2, d_3}(\oa{A}_{d, GM}, \oa{A}_{d, GM}) - \wt{P}_{\mathcal{L}}\right)\Vert = \wt{O}\left(\dfrac{m^{1/d}}{n^{1/2}}\right) = o(1) \] 
Combining this bound with Eq.~\eqref{eq:order-d-zero-poly-approx} we get the desired norm bound for $Z_{0, d_2, d_3}$. 

Therefore, our construction of an SOS dual certificate works for order $d$ as well. Hence, the statements of Theorem~\ref{thm:ord-d-SOS-cert} and Theorem~\ref{thm:ord-d-certificate-exists} hold.  

\subsection{$\Omega$-restricted SOS dual certificate and tensor completion}

Let $\Omega\subseteq [n]^d$ be a uniformly sampled subset of size $N$. Define
\begin{equation}\label{eq:order-d-ROmega-definition}
(R_{\Omega})_{\omega, \omega} = \begin{cases}
      N/n^d & \text{if $\omega\in \Omega$} \\
      0 & \text{if $\omega \notin \Omega$}
    \end{cases}, \qquad \text{and}\qquad \ol{R}_{\Omega} = I_{n^d} - R_{\Omega}.
\end{equation}
Let $\mathcal{S}_d$ be the space defined in Eq.~\eqref{eq:order-d-Sspace} and let $P_{\mathcal{S}_d}$ be a projector on this space. Let 
\[ \wt{P}_{\mathcal{S}_d^t} = \sum\limits_{i\in [m]}\left(\bigotimes\limits_{s<t} a_i^{s}(a_i^s)^T\right)\otimes I_n \otimes \left(\bigotimes\limits_{t<s\leq d} a_i^{s}(a_i^s)^T\right) \quad \text{and} \quad \wt{P}_{int} = \sum\limits_{i\in [m]}\left(\bigotimes\limits_{s\in[d]} a_i^{s}(a_i^s)^T\right).\] 
Using similar arguments to the ones presented in Section~\ref{sec:PS-approx}, one can verify that the following bounds hold.
\begin{lemma}\label{lem:ord-d-PS-approx} For $m\ll n^{d-1}$, w.h.p. 
\[ \left\Vert P_{\mathcal{S}_d} - \left(\sum\limits_{t = 1}^{d} \wt{P}_{\mathcal{S}_d^t}\right) + (d-1)\wt{P}_{int} \right\Vert = \wt{O}\left(\dfrac{m^{1/2}}{n^{(d-1)/2}}\right).\]
More generally, for every $p\geq 1$ there exists a matrix $\wt{P}^{[p]}_{\mathcal{S}_d}$ that is a degree $p$ polynomial of $\wt{P}_{\mathcal{S}_d^t}$ for $t\in [d]$ and $\wt{P}_{int}$ and which satisfies
\[ \left\Vert P_{\mathcal{S}_d} -  \wt{P}^{[p]}_{\mathcal{S}_d}\right\Vert = \wt{O}\left(\dfrac{m^{1/2}}{n^{(d-1)/2}}\right).\] 
\end{lemma}

 As in Section~\ref{sec:Omega-dual-certificate}, for $\Omega = \Omega_1\cup \ldots \cup \Omega_k$ we look for a dual certificate in the form
\begin{equation}\label{eq:Aomega-constr-ord-d}
\oa{A}_{d, \Omega} = \sum\limits_{j=1}^{k} R_{\Omega_j}\left(\prod_{i=1}^{j-1}P_{\mathcal{S}_d}\ol{R}_{\Omega_{j-i}}\right)\oa{A}_{d}+\oa{A}_{d, \Omega, sm},
\end{equation} 

Then identically as in Theorem~\ref{thm:small-ROmega-correction},  we can show that for $N\gg nm$ and $m\ll n^{d-1}$ w.h.p. we can find $\oa{A}_{d, \Omega, sm}$ with 
\begin{equation} \Vert \oa{A}_{d, \Omega, sm} \Vert  = \wt{O}\left(\dfrac{n^{d/2}}{N^{1/2}}\right) \left\Vert \left(\prod_{i=1}^{k}P_{\mathcal{S}_d}\ol{R}_{\Omega_{k-i+1}}\right)\oa{A}_{d} \right\Vert.
\end{equation} 

By relatively straightforward repeating of the arguments of Sections~\ref{sec:PS-wellbalanced}-\ref{sec:AOmega-struct}, one can verify that by taking $k$ large enough we can make the norm $\oa{A}_{d, \Omega, sm}$ to be less than $1/n^b$ for any constant $b$. Additionally, by taking $p$ large enough (both in Eq.~\eqref{eq:ord-d-Asm-bound} and Lemma~\ref{lem:ord-d-PS-approx}) we can write $\oa{A}_{d, \Omega} = \oa{A}_{d, \Omega, GM}+\oa{A}'_{d, \Omega, sm}$, with norm of 
$\oa{A}'_{d, \Omega, sm}$ be less than $1/n^b$ for any fixed constant $b$ and with 
\[ \oa{A}_{d, \Omega, GM} \in \vspan\left\{\ol{R}_{\Omega_{k}}\left(\prod_{i=1}^{k-1}\wt{P}^{[p]}_{\mathcal{S}_d}\ol{R}_{\Omega_{k-i}}\right)\oa{A}_{d, GM},\ \left(\prod_{i=1}^{k}\wt{P}^{[p]}_{\mathcal{S}_d}\ol{R}_{\Omega_{k-i+1}}\right)\oa{A}_{d, GM}\mid 1\leq k \leq O(1)\right\}.\]
\subsection{Norm bounds for IP graph matrices {$B_{t, \ell}$} involved in {$B_{d, \Omega}$ and $Z_{d, \Omega}$}}

 To construct the matrices $B_{d, \Omega}$ and $Z_{d, \Omega}$ such that $(\oa{A}_{d, \Omega}, B_{d, \Omega}, Z_{d, \Omega})$ satisfies constraints~\eqref{eq:opt-constraints-ord-d} we apply the same construction we used to get $B_d$ and $Z_d$ from $\oa{A}_d$. 
 
 \begin{definition}\label{def:B-general-form-ord-d} Let $X, Y\in \mathbb{R}^{n^d}$ and $M\in M_{n^{d_1+2(d_2+d_3)}}(\mathbb{R})$. Define $\mathcal{B}(X, Y, M)$ to be the $n^{d_2+d_3}\times n^{d_2+d_3}$ matrix, whose $((b, c), (b', c'))$ entry for $(b, c), (b', c')\in [n]^{d_2}\times [n]^{d_3}$ is defined as
\[ \mathcal{B}(X, Y, M)_{(b, c), (b', c')} = \sum\limits_{\substack{a^U, a^L\in [n]^{d_1}, b^U, b^L\in [n]^{d_2},\\ c^U, c^L\in [n]^{d_3}}} X_{(a^U, b^U, c^U)}M_{(a^L, b, c', b^L, c^L) (a^U, b^U, c^U, b', c)}Y_{(a^L, b^L, c^L)}.\]
In the case when $M = I_{n^{d_1+2(d_2+d_3)}}$,  $\mathcal{B}_{d_1, d_2, d_3}(X, Y) = \mathcal{B}(X, Y, I_{n^{d_1+2(d_2+d_3)}})$.
\end{definition}
 
 Arguing similarly as in Section~\ref{sec:BOmega-GM}, one can check that it is sufficient to prove the norm bounds for the essential IP graph matrices involved in $B_{d, \Omega}-B_{d}$ and $Z_{d, \Omega}$. Using the description of $\oa{A}_{d, \Omega, GM}$ it is not hard to see that all these IP graph matrices are of the form

\begin{equation}\label{eq:ord-d-Btl-definition}
 B_{t, \ell} = \mathcal{B}\left(\ol{R}_{\Omega_{t}}P^U_{t-1}\ol{R}_{\Omega_{t-1}}\ldots P_1^U\ol{R}_{\Omega_{1}}X^U,\ \ol{R}_{\Omega_{\ell}}P^L_{\ell-1}\ol{R}_{\Omega_{\ell-1}}\ldots P_1^L\ol{R}_{\Omega_{1}}X^L,\ P^M\right),
 \end{equation}
where $X^U$, $X^L$,  $P_i^U$, $P_i^L$ and $P^M$ are IP graph matrices that satisfy the following properties:
\begin{enumerate}
\item $t\geq 0$, $\ell\geq 0$ and $t+\ell\geq 1$.
\item Each matrix $P_{i}^U$ and $P_{j}^L$ is some product of of $\wt{P}_{\mathcal{S}_d^t}$ for $t\in [d]$ and $\wt{P}_{int}$. So, in particular, it is weakly-$\binom{[d]}{1}$-connected.
\item $P_M$ is an $n^{d_1+2(d_2+d_3)}\times n^{d_1+2(d_2+d_3)}$ IP graph matrix with $type(\myscr{Ver}_L) = type(\myscr{Ver}_R) = (u, v, w, v, w)$, where $u = [d_1]$, $v = [d_1+1, d_1+d_2]$ and $w = [d_1+d_2+1, d_3]$.

 We assign the names to the groups of crosses (i.e., groups $u$, $v$, $w$) in $\myscr{Ver}_L$ and $\myscr{Ver}_R$ as in Def.~\ref{def:B-general-form}. 
 
 Let $\delta_u$, $\delta_v^U$, $\delta_w^U$, $\delta_v^L$ and $\delta_w^L$ be the number of crosses in the groups $a^L$, $b^U$, $c^U$, $b^L$ and $c^L$, respectively, that are incident to an edge (i.e. are not in the intersection of $\myscr{Ver}_L$ and $\myscr{Ver}_R$). These variables satisfy
 \[ \delta_u+\delta_v^U+\delta_w^U\in \{0, d-1, d\},\quad  \text{and}\quad  \delta_u+\delta_v^L+\delta_w^L\in \{0, d-1, d\} \quad \text{or}\]
 \[ \delta_v^L+\delta_w^L = d-d_1,\]
 so in any case 
\begin{equation}\label{eq:ord-d-delta-ineq}
\delta_u+\delta_v^U+\delta_w^U+\delta_v^L+\delta_w^L\notin [1, d-2].
\end{equation}
\item $P_M$ has at most 2\ $\binom{[d]}{2}$-connected components and is $\binom{[d]}{2}$-boundary-connected.
\item $X^U = X^L=\oa{A}_{d, GM}$, so in particular, they have at most 2\ $\binom{[d]}{2}$-connected components and are $\binom{[d]}{2}$-boundary-connected.
\end{enumerate}
 As in Section~\ref{sec:BOmega-GM} we study layered diagram for $B_{t, \ell}$ which schematically look like on Figure~\ref{fig:Gamma-giagram}. Let $\mathcal{C}$ be collection of 2-element sets induced by a 2-regular 3-partite graph with parts $[d_1]$, $[d_1+1, d_1+d_2]$ and $[d_1+d_2+1, d]$. Then, from the arguments above (see proof of Lemma~\ref{lem:ord-d-BGMP-bound}) and from the structure of $B_{t, \ell}$ it follows that the matrix diagram $\Gamma$ of $B_{t, \ell}$ is $\mathcal{C}$-connected and $\mathcal{C}$-boundary connected.
 
 As in Section~\ref{sec:Btl-comb-prop}, for a layered labeling $\phi$ of a trace diagram of~$\Gamma$
 \begin{itemize}
 \item $eR_i(\phi)$ denotes the number of edges incident with the crosses on $i$-th level in $\Gamma$;
 \item $h_i(\phi)$ denotes the number of distinct hyperedges in the d-uniform hypergraph $\mathcal{H}_i$, which consists of images $\phi(j, \mu(S_i^U))$ and $\phi(j, \mu(S_i^L))$ (see Def.~\ref{def:hi});
 \item $xR_i(\phi)$ the number of crosses with distinct labels (with colors taken into account) that appear at level $i$, but do not appear level $j$ for all $j>i$ (see Eq.~\eqref{eq:xRi-def}).
\end{itemize}
Consider also
\[ h(\phi) = \sum\limits_{i = 1}^{t} h_i(\phi),\qquad  eR(\phi) = \sum\limits_{i = 1}^{t} eR_i(\phi) \quad \text{and} \quad \chi R(\phi) = \sum\limits_{i=1}^{t} xR_i(\phi). \]
Alternatively, we can describe $\chi R(\phi)$ as the number of distinct labels of crosses that appear in the tuples substituted in $\ol{R}_{\Omega}$ and $eR(\phi)$ as the number of edges incident with them.

Then, using the arguments identical to Theorem~\ref{thm:val-gamma-bound}, Corollary~\ref{cor:valr-gamma-bound} and Theorem~\ref{thm:BOmega-bound-without-layers}  we get the following statement.
\begin{lemma} Let $m\ll n^{d/2}$. Then w.h.p. over the randomness of $\mathcal{V}$ and $\Omega$,
\[ \mathbb{E}\left[\Tr\left(B_{t, \ell}B_{t, \ell}^T\right)^q\right]  =  (2q|\myscr{Ver}|)^{2q|\myscr{Ver}|+2} \max\limits_{\phi\in \Phi_0^{Lay}}\left( \left(\dfrac{n^d}{N}\right)^{ 2(t+\ell)q-h(\phi)}\left(\dfrac{m^{2/d}}{n}\right)^{ eR(\phi)/2}\left(\dfrac{n}{m^{2/d}}\right)^{\chi R(\phi)}  \right)\]
\end{lemma}

Thus it is only left to establish that the analog of Theorem~\ref{thm:main-layer-analysis} holds in the order $d$ case.
\begin{theorem}
Assume $N\gg n^{d/2}m$ and $n^{d/2}\gg m$ for odd $d>3$. For every $\phi \in \Phi_0^{Lay}$, 
\[ \prod\limits_{i=1}^{t}\left(\dfrac{n^d}{N}\right)^{\displaystyle 2s_iq-h_i(\phi)}\left(\dfrac{m^{2/d}}{n}\right)^{\displaystyle eR_i(\phi)/2}\left(\dfrac{n}{m^{2/d}}\right)^{\displaystyle xR_i(\phi)} \leq \left(\dfrac{n}{m^{2/d}}\right)^{(d+1)/2}\left(\dfrac{n^{d/2}m}{N}\right)^{(t+\ell)q}. \]
\end{theorem}
\begin{proof} Fix the map $\phi\in \Phi_0^{Lay}$. To make expressions shorter we use the defined above notation without mentioning dependence on $\phi$.

Denote by $H^i_{c, c'}$, the graph induced by the hypergraph $\mathcal{H}_i$ on the vertices of colors $\{c, c'\}$ for distinct colors $c, c'\in [d]$. Denote by $x_c^i$ the number of distinct vertices of color $c$ in $\mathcal{H}_i$ with last coordinate $i$. By definition, $xR_i = \sum\limits_{c=1}^{d} x_{c}^i$.

We study the desired expression for each $1\leq i\leq t$ separately. We consider four major cases: 1) $i\notin\{t, \ell\}$, 2) $i = t = \ell$, 3) $i = \ell$ for $t\neq \ell$ and 4) $i=t$ for $t\neq \ell$.

\noindent\textbf{Case 1.} Assume that $i\notin\{t, \ell\}$.

\noindent\textbf{Case 1.a} Assume $|\myscr{Cr}_i| = d+(d-1)$. In this case, $s_i = 2$ and by Lemma~\ref{lem:eRi-set-size}, $eR_i/2 = 2(d-1)q$. Moreover, clearly, $x R_i \leq dh_i$ and $h_i\leq 2q$. Thus
\begin{equation}
\begin{gathered}
 \left(\dfrac{n^d}{N}\right)^{4q-h_i}\left(\dfrac{m^{2/d}}{n}\right)^{eR_i/2}\left(\dfrac{n}{m^{2/d}}\right)^{x R_i} \leq \left(\dfrac{n^d}{N}\right)^{4q-h_i}\left(\dfrac{m^{2/d}}{n}\right)^{2(2d-1)q}\left(\dfrac{n}{m^{2/d}}\right)^{dh_i} \leq \\
 \leq \left(\dfrac{\sqrt{n}m^{(2d-1)/d}}{N}\right)^{4q}\left(\dfrac{N}{m^2}\right)^{h_i}\leq \left(\dfrac{nm^{(2d-2)/d}}{N}\right)^{2q}.
 \end{gathered}
 \end{equation}
 
\noindent\textbf{Case 1.b} Assume $|\myscr{Cr}_i| \neq 2d-1$. Denote $k = |\myscr{Cr}_i|/s_i$. By the structure of projectors, $k\in \{d-1, d\}$. In this case, $x R_i\leq k h_i$.  Again, we have $h_i\leq s_i q$, and by Lemma~\ref{lem:eRi-set-size}, $eR_i/2 = 2|\myscr{Cr}_i|q = 2s_i kq$. Thus
\begin{equation}\label{eq:boundcase1b-interm-d}
\begin{gathered}
 \left(\dfrac{n^d}{N}\right)^{2s_iq-h_i}\left(\dfrac{m^{2/d}}{n}\right)^{eR_i/2}\left(\dfrac{n}{m^{2/d}}\right)^{x R_i} \leq \left(\dfrac{n^d}{N}\right)^{2s_iq-h_i}\left(\dfrac{m^{2/d}}{n}\right)^{2s_ikq}\left(\dfrac{n}{m^{2/d}}\right)^{kh_i} \leq \\
 \leq \left(\dfrac{n^{d-k}m^{2k/d}}{N}\right)^{2s_iq}\left(\dfrac{N}{m^{2k/d}n^{d-k}}\right)^{h_i}.
 \end{gathered}
 \end{equation}
 Since $k\in \{d-1, d\}$ in this case, and since $m<n^{d/2}$ and $n^{d/2}m<N$ we have 
 \[ 1<\dfrac{N}{n^{d/2}m}\leq \dfrac{N}{nm^{(2d-2)/d}} \leq \dfrac{N}{m^{2k/d}n^{d-k}}\leq \dfrac{N}{m^2}.\]
 Thus, the last expression in Eq.~\eqref{eq:boundcase1b-interm-d} is maximized when $h_i = s_iq$. Hence,
 \begin{equation}\label{eq:boundcase1b-d}
 \left(\dfrac{n^d}{N}\right)^{2s_iq-h_i}\left(\dfrac{m^{2/d}}{n}\right)^{eR_i/2}\left(\dfrac{n}{m^{2/d}}\right)^{x R_i} \leq \left(\dfrac{nm^{(2d-2)/d}}{N}\right)^{s_iq}.
 \end{equation}
 
\noindent\textbf{Case 2.} Assume that $i= t= \ell$. \\ In this case, $s_t = 2$ and $h_t\leq 2q$.  Using notation from Eq.~\eqref{eq:ord-d-delta-ineq}, it is easy to see that 
\[ eR_t = (2d+2\delta_u+\delta^U_v+\delta^U_w+\delta^L_v+\delta^L_w)\cdot 2q\]
By Eq.~\eqref{eq:ord-d-delta-ineq}, $2\delta_u+\delta^U_v+\delta^U_w+\delta^L_v+\delta^L_w\notin [1, d-2]$. We claim that 
\[ x R_t \leq \begin{cases}
 dh_t\quad & \text{if}\quad 2\delta_u+\delta^U_v+\delta^U_w+\delta^L_v+\delta^L_w \geq d \\
 (d-1)h_t+\min(h_t, q+1) \quad & \text{if}\quad 2\delta_u+\delta^U_v+\delta^U_w+\delta^L_v+\delta^L_w = d-1\\ 
 d(h_t+1)/2 \quad & \text{if}\quad 2\delta_u+\delta^U_v+\delta^U_w+\delta^L_v+\delta^L_w = 0
\end{cases}\]
\noindent\textbf{Case 2.a} $2\delta_u+\delta^U_v+\delta^U_w+\delta^L_v+\delta^L_w \geq d$. \\
In this case, $eR_t/2 \geq 3dq$. The bound $x R_t\leq dh_t$ is obvious. Thus
\begin{equation}\label{eq:xRi-bound-18-d}
\begin{gathered}
 \left(\dfrac{n^d}{N}\right)^{4q-h_t}\left(\dfrac{m^{2/d}}{n}\right)^{eR_t/2}\left(\dfrac{n}{m^{2/d}}\right)^{x R_t} \leq \left(\dfrac{n^d}{N}\right)^{4q-h_t}\left(\dfrac{m^{6q}}{n^{3dq}}\right)\left(\dfrac{n}{m^{2/d}}\right)^{dh_t} \leq \\
 \leq \left(\dfrac{n^{d/2}m^{3}}{N^2}\right)^{2q}\left(\dfrac{N}{m^2}\right)^{h_t}\leq \left(\dfrac{n^{d/2}m}{N}\right)^{2q}.
 \end{gathered}
 \end{equation}  

\noindent\textbf{Case 2.b} $2\delta_u+\delta^U_v+\delta^U_w+\delta^L_v+\delta^L_w = (d-1)$. \\
In this case, $eR_t/2= (3d-1)q$. Moreover, by Eq.~\eqref{eq:ord-d-delta-ineq}, in this case $\delta_u = 0$. Note also, that for some color $c\in [2, d]$, the sum $\delta_c^U+\delta_c^L\leq 1$. Then the graph $H^t_{u, c}$ has at most $q$ connected components. Thus, by Lemma~\ref{lem:induced-graphs-obv-bounds} we have $x_u^t+x_c^t\leq \min(q+1+h_t, 2h_t)$. There are at most $h_t$ vertices of any color distinct from $c$ in $\mathcal{H}_t$. 

 Thus, $x R_t\leq (d-1)h_t+\min(h_t, q+1)$ and we can bound
\begin{equation}\label{eq:xRi-bound-16-d}
\begin{gathered}
 \left(\dfrac{n^d}{N}\right)^{4q-h_t}\left(\dfrac{m^{2/d}}{n}\right)^{eR_t/2}\left(\dfrac{n}{m^{2/d}}\right)^{x R_t} \leq \left(\dfrac{n^d}{N}\right)^{4q-h_t}\left(\dfrac{m^{2(3d-1)q/d}}{n^{(3d-1)q}}\right)\left(\dfrac{n}{m^{2/d}}\right)^{x R_t} \leq \\
 \leq \left(\dfrac{n^{(d+1)/2}m^{(3d-1)/d}}{N^2}\right)^{2q}\left(\dfrac{N}{nm^{2(d-1)/d}}\right)^{h_t}\left(\dfrac{n}{m^{2/d}}\right)^{\min(h_t, q+1)} \leq \\
 \leq   \left(\dfrac{n^{(d-1)/2}m^{(d+1)/d}}{N}\right)^{2q}\left(\dfrac{n}{m^{2/d}}\right)^{q+1} 
 \leq \left(\dfrac{n^{d/2}m}{N}\right)^{2q}\left(\dfrac{n}{m^{2/d}}\right),
 \end{gathered}
 \end{equation}  
 where to deduce the inequality transition to the last line we use that $N>m^2$, so $h_t\geq q+1$ is optimal, and then, since $N>nm^{2(d-1)/d}$ and $h_t\leq 2q$, $h_t = 2q$ is optimal.

\noindent\textbf{Case 2.c} $2\delta_u+\delta^U_v+\delta^U_w+\delta^L_v+\delta^L_w = 0$. \\
In this case, $eR_t/2= 2dq$ and for each $c\in [2, (d+1)/2], c'\in [(d+3)/2, d]$ the graphs $H^t_{u, c}$, $H^t_{u, c'}$ and $H^t_{c, c'}$ are connected. Thus, by Lemma~\ref{lem:induced-graphs-obv-bounds}, 
\[x^t_u+x^t_c\leq h_t+1, \qquad x^t_u+x^t_{c'}\leq h_t+1, \qquad x^t_{c}+x^t_{c'}\leq h_t+1,\]
and so the bound $x R_t\leq d(h_t+1)/2$ holds (by summing over pairs in $\mathcal{C}$). Hence,
\begin{equation}
\begin{gathered}
 \left(\dfrac{n^d}{N}\right)^{4q-h_t}\left(\dfrac{m^{2/d}}{n}\right)^{eR_t/2}\left(\dfrac{n}{m^{2/d}}\right)^{x R_t} \leq \left(\dfrac{n^d}{N}\right)^{4q-h_t}\left(\dfrac{m^{4q}}{n^{2dq}}\right)\left(\dfrac{n^{d/2}}{m}\right)^{h_t+1} \leq \\
 \leq \left(\dfrac{n^{d/2}m}{N}\right)^{4q}\left(\dfrac{N}{n^{d/2}m}\right)^{h_t}\left(\dfrac{n^{d/2}}{m}\right)\leq \left(\dfrac{n^{d/2}m}{N}\right)^{2q}\left(\dfrac{n^{d/2}}{m}\right).
 \end{gathered}
 \end{equation} 
 
 \noindent\textbf{Case 3.} Assume $i= \ell$ and $t\neq \ell$. In this case, $s_{\ell} = 2$ and $h_{\ell}\leq 2q$.
 
 Let $\delta_P$ be the indicator of the event that $P^U_{\ell}$ has edges of all three colors. Then
 \[ eR_{\ell} = \left(2(d-1+\delta_P)+(d-1)+2\delta_u+\delta_w^L+\delta_v^L\right)\cdot 2q\]
Recall, $2\delta_u+\delta_w^L+\delta_v^L \notin [1, (d-3)/2]$, so for $d>3$, $2\delta_P+2\delta_u+\delta_w^L+\delta_v^L\neq 1$. We claim that
 \[ xR_{\ell} \leq \begin{cases}
 dh_{\ell}\quad & \text{if}\quad 2\delta_P+2\delta_u+\delta_w^L+\delta_v^L\geq 3 \\
 (d-1)h_{\ell}+\min(h_{\ell}, q) \quad & \text{if}\quad 2\delta_P+2\delta_u+\delta_w^L+\delta_v^L = 2\\ 
 (d-2)(h_{\ell}+1)+\min(h_{\ell}, q) \quad & \text{if}\quad 2\delta_P+2\delta_u+\delta_w^L+\delta_v^L = 0
\end{cases}\]

\noindent\textbf{Case 3.a} $2\delta_P+2\delta_u+\delta_w^L+\delta_v^L\geq 3$. \\
In this case $eR_{\ell}/2\geq 3dq$, and clearly $xR_{\ell}\leq dh_{\ell}$. So as in Eq.~\eqref{eq:xRi-bound-18-d},
\[ \left(\dfrac{n^d}{N}\right)^{4q-h_{\ell}}\left(\dfrac{m^{2/d}}{n}\right)^{eR_{\ell}/2}\left(\dfrac{n}{m^{2/d}}\right)^{xR_{\ell}} \leq \left(\dfrac{n^{d/2}m}{N}\right)^{2q}\]

\noindent\textbf{Case 3.b} $2\delta_P+2\delta_u+\delta_w^L+\delta_v^L= 2$. In this case $eR_{\ell}/2= (3d-1)q$. 

If $\delta_u = 0$ or $P_{\ell}^U$ is $u$-diagonal, then $x_u^{\ell} \leq \min(h_{\ell}, q)$, as every vertex is covered by an hyperedge of $\mathcal{H}_{\ell}$ at least twice. Clearly, $x^{\ell}_c\leq h_{\ell}$ for any $c$, so $xR_{\ell}\leq (d-1)h_{\ell}+\min(h_{\ell}, q)$.

Next, assume $\delta_u = 1$, $\delta_P = 0$ and $P_{\ell}^{U}$ is $c$-diagonal for some $c\in [2, d]$. Take $c'\in [2, d]$ such that precisely one of $c$ and $c'$ is in group ``$v$''. Then the graph $H^{\ell}_{c, c'}$ restricted to vertices of layer $\ell$ (i.e. with last coordinate being $\ell$) is an image of a cycle of length $2q$ and $2q$ isolated vertices. Therefore, by Lemma~\ref{lem:induced-graphs-obv-bounds}, there are at most $q+1$ distinct vertices in a cycle of $H^{\ell}_{c, c'}$. Since, every hyperedge in $\mathcal{H}_{\ell}$ appears at least twice, there are at most $q$ isolated vertices in $H^{\ell}_{c, c'}$. Thus, $x_c^{\ell}+x_{c'}^{\ell}\leq \min(2h_{\ell}, 2q+1)$. For any other color, there are at most $h_{\ell}$ crosses. Therefore, in this case, 
\[xR_{\ell}\leq (d-2)h_{\ell}+\min(2h_{\ell}, 2q+1)\leq (d-1)h_{\ell}+\min(h_{\ell}, q),\]
where the last inequality can be checked by considering cases $q\geq h_{\ell}$ and $q<h_{\ell}$.
So, similarly as in Eq.~\eqref{eq:xRi-bound-16-d},
\[
 \left(\dfrac{n^d}{N}\right)^{4q-h_{\ell}}\left(\dfrac{m^{2/d}}{n}\right)^{eR_{\ell}/2}\left(\dfrac{n}{m^{2/d}}\right)^{xR_{\ell}} \leq  \left(\dfrac{n^{d/2}m}{N}\right)^{2q}
\] 

\noindent\textbf{Case 3.c} $2\delta_P+2\delta_u+\delta_w^L+\delta_v^L= 0$. In this case $eR_{\ell}/2= (3d-3)q$ and $\delta_P = \delta_u = \delta_v^L = \delta_w^L = 0$.

If $P^U_{\ell}$ is $u$-diagonal, then $x^{\ell}_u = 0$ and by Lemma~\ref{lem:induced-graphs-obv-bounds}, for any $c\in [2, (d+1)/2]$ and $c'\in [(d+3)/2, d]$ the inequality $x^{\ell}_c+x^{\ell}_{c'}\leq \min(2h_{\ell}, h_{\ell}+q+1)$ holds. Therefore, \[xR_{\ell}\leq \left(\dfrac{d-1}{2}\right)\min(2h_{\ell}, h_{\ell}+q+1) = \left(\dfrac{d-1}{2}\right)h_{\ell}+\left(\dfrac{d-1}{2}\right)\min(h_{\ell}, q+1)\].

If $P^U_{\ell}$ is $c^*$-diagonal for some $c^*\in [2, d]$. Consider any $c\in [2, (d+1)/2]$ and $c'\in [(d+3)/2, d]$ distinct from $c^*$. Then, by Lemma~\ref{lem:induced-graphs-obv-bounds}, the inequality $x^{\ell}_c+x^{\ell}_{c'}\leq \min(2h_{\ell}, h_{\ell}+q+1)$ holds. Note also that every hyperedge contributes at most one new $u$ or $c^*$ vertex. Thus, $x^{\ell}_u+x^{\ell}_{c^*}\leq h_{\ell}$. Therefore, by forming $(d-3)/2$ disjoint pairs of colors between $[2, (d+1)/2]$ and $[(d+3)/2, d]$ which do not contain $c^*$ we conclude the inequality 
\[xR_{\ell}\leq \left(\dfrac{d-3}{2}\right)\min(2h_{\ell}, h_{\ell}+q+1) +2h_{\ell}= \left(\dfrac{d+1}{2}\right)h_{\ell}+\left(\dfrac{d-3}{2}\right)\min(h_{\ell}, q+1)\]

 In any case, for odd $d>3$ we have $xR_{\ell}\leq (d-2)h_{\ell}+\min(h_{\ell}, q+1)$. 

\begin{equation}
\begin{gathered}
\left(\dfrac{n^d}{N}\right)^{4q-h_{\ell}}\left(\dfrac{m^{2/d}}{n}\right)^{eR_{\ell}/2}\left(\dfrac{n}{m^{2/d}}\right)^{xR_{\ell}} \leq \\ 
\leq \left(\dfrac{n^d}{N}\right)^{4q-h_{\ell}}\left(\dfrac{m^{(3d-3)q/d}}{n^{(3d-3)q}}\right)\left(\dfrac{n}{m^{2/d}}\right)^{(d-2)h_{\ell}+\min(h_{\ell}, q+1)} \leq \\
\leq \left(\dfrac{n^{(d+3)/2}m^{(3d-3)/d}}{N^2}\right)^{2q}\left(\dfrac{N}{n^2m^{2(d-2)/d}}\right)^{h_{\ell}}\left(\dfrac{n}{m^{2/d}}\right)^{\min(h_{\ell}, q+1)}\leq \\
\leq \left(\dfrac{n}{m^{2/d}}\right)\left(\dfrac{n^{d/2}m}{N}\right)^{2q} 
\end{gathered}
\end{equation}
where in the last line we first use that it is optimal to have $h_{\ell}\geq q+1$, as $N>nm^{2(d-1)/d}$. And then, for $h_{\ell}\geq q+1$, we use that for $d>3$ we have $N>n^2m^{2(d-2)/d}$, so $h_{\ell}= 2q$ is optimal.

\noindent\textbf{Case 4.} Assume $i= t$ and $t\neq \ell$. In this case, $s_t = 1$ and $h_t\leq q$.
 Then
 \[ eR_t = (d+1+\delta_v^U+\delta_w^U)\cdot 2q.\]
\noindent\textbf{Case 4.a} Assume that $\delta_v^U+\delta_w^U\geq (d-1)/2$. Then $eR_t\geq (3d+1)q$.

 Clearly $x R_t\leq dh_t$, so
\begin{equation}
\begin{gathered}
 \left(\dfrac{n^d}{N}\right)^{2q-h_t}\left(\dfrac{m^{2/d}}{n}\right)^{eR_t/2}\left(\dfrac{n}{m^{2/d}}\right)^{x R_t} \leq \left(\dfrac{n^d}{N}\right)^{2q-h_t}\left(\dfrac{m^{(3d+1)q/d}}{n^{(3d+1)q/2}}\right)\left(\dfrac{n}{m^{2/d}}\right)^{dh_t} \leq \\
 \leq \left(\dfrac{n^{(d-1)/2}m^{(3d+1)/d}}{N^2}\right)^{q}\left(\dfrac{N}{m^2}\right)^{h_t}\leq \left(\dfrac{n^{(d-1)/2}m^{(d+1)/d}}{N}\right)^{q}
 \end{gathered}
 \end{equation}

\noindent\textbf{Case 4.b} Assume that $\delta_v^U=\delta_w^U=0$. Then $eR_t/2 = (d+1)q$.

Clearly, $x^t_u\leq h_t$ and graph $H^{t}_{v, w}$ is connected. So, by Lemma~\ref{lem:induced-graphs-obv-bounds}, $x^t_c+x^t_{c'}\leq h_t+1$ for any $c\in [2, (d+1)/2]$ and $c'\in [(d+3)/2, d]$. Thus $xR_t\leq (d+1)h_t/2+(d-1)/2$ and

\begin{equation}
\begin{gathered}
 \left(\dfrac{n^d}{N}\right)^{2q-h_t}\left(\dfrac{m^{2/d}}{n}\right)^{eR_t/2}\left(\dfrac{n}{m^{2/d}}\right)^{x R_t} \leq \\
 \leq  \left(\dfrac{n^d}{N}\right)^{2q-h_t}\left(\dfrac{m^{2(d+1)q/d}}{n^{(d+1)q}}\right)\left(\dfrac{n^{1/2}}{m^{1/d}}\right)^{(d+1)h_t+(d-1)} \leq \\
 \leq \left(\dfrac{n}{m^{2/d}}\right)^{(d-1)/2}\left(\dfrac{n^{(d-1)/2}m^{(d+1)/d}}{N}\right)^{2q}\left(\dfrac{N}{n^{(d-1)/2}m^{(d+1)/d}}\right)^{h_t}\leq \\
 \leq  \left(\dfrac{n}{m^{2/d}}\right)^{(d-1)/2}\left(\dfrac{n^{(d-1)/2}m^{(d+1)/d}}{N}\right)^{q}.
 \end{gathered}
 \end{equation}
 
 Finally, we can collect the analysis in this cases together. Denote
\[ X = \prod\limits_{i=1}^{t}\left(\dfrac{n^d}{N}\right)^{\displaystyle 2s_iq-h_i(\phi)}\left(\dfrac{m^{2/d}}{n}\right)^{\displaystyle eR_i(\phi)/2}\left(\dfrac{n}{m^{2/d}}\right)^{\displaystyle xR_i(\phi)}  \]
Recall that $n>m^{d/2}$ and $N>n^{d/2}m>nm^{2(d-1)/d}>m^2$. Thus
\begin{enumerate}
\item $X\leq \left(\dfrac{n^{d/2}}{m}\right)\left(\dfrac{n^{d/2}m}{N}\right)^{2q}\left(\dfrac{nm^{2(d-1)/d}}{N}\right)^{2q(t-1)}$\quad  if $t= \ell\geq 1$;
\item $X\leq \left(\dfrac{n}{m^{2/d}}\right)^{(d-1)/2}\left(\dfrac{n^{(d-1)/2}m^{(d+1)/d}}{N}\right)^{q}\left(\dfrac{nm^{2(d-1)/d}}{N}\right)^{(t-1)q}$ \quad if $t>\ell = 0$.
\item $X\leq \left(\dfrac{n}{m^{2/d}}\right)^{(d+1)/2}\left(\dfrac{n^{(3d-1)/2}m^{(3d+1)/d}}{N^3}\right)^{q}\left(\dfrac{nm^{2(d-1)/d}}{N}\right)^{(t+\ell-3)q}$, if $t>\ell\geq 1$.
\end{enumerate}
Hence, in any of these cases, the inequality from the statement of the theorem holds.
\end{proof}

\section{Norm Bounds Using Graph Matrices}\label{graphmatrixappendix}
In this appendix, we describe how the graph matrices studied in \cite{graph-matrix-bounds} can be used to give alternate proofs for many of the specific norm bounds which we needed in our analysis.
\subsection{Graph Matrix Definitions}
As noted in Remark \ref{graphmatrixremark}, for our setting, graph matrices are essentially equivalent to expanded matrix diagrams where all of the nodes and crosses must be distinct from each other.

For completeness, we give the definitions of graph matrices for our setting where we are analyzing order $3$ tensors. If we are instead considering higher order tensors, the definitions are the same except that there are more types of crosses.
\begin{definition}[Matrix Indices]
We define a matrix index $A = \{A_{node},A_{cross,u}, A_{cross,v}, A_{cross,w}\}$ to be a tuple $A_{node}$ of nodes with distinct values in $[m]$ together with a tuple $A_{cross,u}$ of $u$ crosses with distinct values in $[n]$, a tuple $A_{cross,v}$ of $v$ crosses with distinct values in $[n]$, and a tuple $A_{cross,w}$ of $w$ crosses with distinct values in $[n]$.
\end{definition}
\begin{definition}
Given an edge $e$ with label $l$ between a $u$ cross with value $i \in [n]$ and node with value $j \in [m]$, we define $\chi_{e}$ to be $\chi_{e} = h_{l}((u_{j})_i)$ where $h_l$ is the lth Hermite polynomial normalized so that $E_{x \sim N(0,1)}[h_l^2(x)] = 1$.

We have a similar definition for edges between nodes and $v$ or $w$ crosses.
\end{definition}
\begin{definition}[Ribbons]
A ribbon $R = (A_R,B_R,C_R,E(R))$ consists of the following:
\begin{enumerate}
    \item Matrix indices $A_R$ and $B_R$.
    \item An additional set $C_R$ of nodes, $u$ crosses, $v$ crosses, and $w$ crosses with values in $[m]$ and $[n]$ respectively. These nodes, $u$ crosses, $v$ crosses, and $w$ crosses must be distinct from the ones in $A_R$ and $B_R$. 
    \item A set of labeled edges $E_R$ where each edge is between a node and a cross.
\end{enumerate}
\end{definition}
\begin{definition}
Given a ribbon $R$, we define $\chi_R = \prod_{e \in E_R}{\chi_e}$ and we define $M_R$ to be the matrix indexed by matrix indices such that $M_R(A,B) = \chi_R$ if $A = A_R$ and $B = B_R$ and $M_R(A,B) = 0$ otherwise.
\end{definition}
\begin{definition}[Index Shapes]
We define an index shape $U = \{U_{node},U_{cross,u}, U_{cross,v}, U_{cross,w}\}$ to be a tuple $U_{node}$ of nodes, together with a tuple $U_{cross,u}$ of $u$ crosses, a tuple $U_{cross,v}$ of $v$ crosses, and a tuple $U_{cross,w}$ of $w$ crosses. Instead of having values in $[m]$ or $[n]$, these nodes and crosses are specified by distinct unspecified variables.

In other words, an index shape is a matrix index where all of the nodes and crosses are distinct unspecified variables rather than having values in $[n]$ or $[m]$.
\end{definition}
\begin{definition}[Shapes]
A ribbon $\alpha = (U_{\alpha},V_{\alpha},W_{\alpha},E(\alpha))$ consists of the following:
\begin{enumerate}
    \item Index shapes $U_{\alpha}$ and $V_{\alpha}$ which may intersect each other arbitrarily.
    \item An additional set $W_{\alpha}$ of nodes, $u$ crosses, $v$ crosses, and $w$ crosses with values in $[m]$ and $[n]$ respectively. These nodes, $u$ crosses, $v$ crosses, and $w$ crosses are described by unspecified variables which are distinct from the nodes and crosses in $U_{\alpha}$ and $V_{\alpha}$.
    \item A set of labeled edges $E(\alpha)$ where each edge is between a node and a cross.
\end{enumerate}
In other words, a shape is a ribbon where the indices in $[m]$ and $[n]$ are replaced by unspecified variables.
\end{definition}
\begin{definition}
Given a shape $\alpha$, we define $V(\alpha)$ to be the set of nodes and crosses which appear in $U_{\alpha}$, $V_{\alpha}$, or $W_{\alpha}$.
\end{definition}
With these definitions, we can finally define graph matrices.
\begin{definition}[Graph Matrices]
Given a shape $\alpha$, the graph matrix $M_{\alpha}$ is the matrix indexed by matrix indices which is given by
\[
M_{\alpha} = \sum_{\pi:V(\alpha) \to [n] \cup [m], \pi \text{ is injective on nodes, injective on u crosses, } \atop \text{ injective on v crosses, and injective on w crosses.}}{M_{\pi(\alpha)}}
\]
\end{definition}
\subsection{Norm Bounds on Graph Matrices}
We now describe the norm bounds proved by \cite{graph-matrix-bounds} for these graph matrices.
\begin{definition}
Given a shape $\alpha$, we say that a set $S \subseteq V(\alpha)$ of nodes and crosses is a vertex separator if every path (including paths of length $0$) from $U_{\alpha}$ to $V_{\alpha}$ intersects $S$
\end{definition}
\begin{definition}
We define the weight of a separator $S$ to be 
$w(S) = log_{n}(m)(\# \text{ of nodes in } S) + (\# \text{ of crosses in } S)$
\end{definition}
For our setting, the graph matrix norm bounds shown by \cite{graph-matrix-bounds} can be stated as follows. 
\begin{theorem}[Graph Matrix Norm Bounds]
Given a shape $\alpha$, letting $S$ be a minimum weight vertex separator of $\alpha$, with high probability 
$\norm{M_{\alpha}}$ is $\tilde{O}(\sqrt{m}^{\#\text{ of nodes in } V(\alpha) \setminus S}\sqrt{n}^{\#\text{ of crosses in } V(\alpha) \setminus S})$
\end{theorem}
\begin{remark}
It should be noted that these norm bounds were proved for the case when each $u$, $v$, and $w$ vector has independent Gaussian coordinates. That said, these norm bounds are also true for our setting where $u$, $v$, and $w$ are on the sphere. Roughly speaking, the reason for this is that the distributions for the entries of  $u$, $v$, and $w$ when these vectors are on the sphere is very close to being independent Gaussian entries.
\end{remark}
\subsection{Alternate Proofs via the Graph Matrix Norm Bounds}
We now show how the graph matrix norm bounds can be used to give alternative proofs of the norm bounds for several of the matrices which we have analyzed. We start with an alternate proof for the norm bound on $S_{12}$, which appears in the proof of Proposition \ref{prop:M-cross-approx}. For convenience, we repeat the diagram for $S_{12}$ here (see Figure \ref{S12repeat}).
\begin{figure}[h]
\begin{subfigure}[b]{0.5\textwidth}
\begin{center}
\includegraphics[height = 2.5cm]{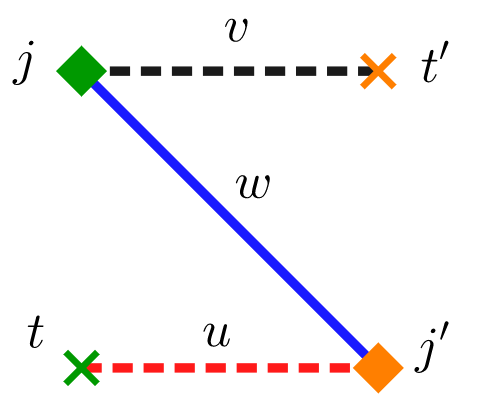}
\end{center}
\caption{The matrix diagram for $S_{12}$}\label{S12repeat}
\end{subfigure}
\begin{subfigure}[b]{0.5\textwidth}
\begin{center}
\includegraphics[height = 3cm]{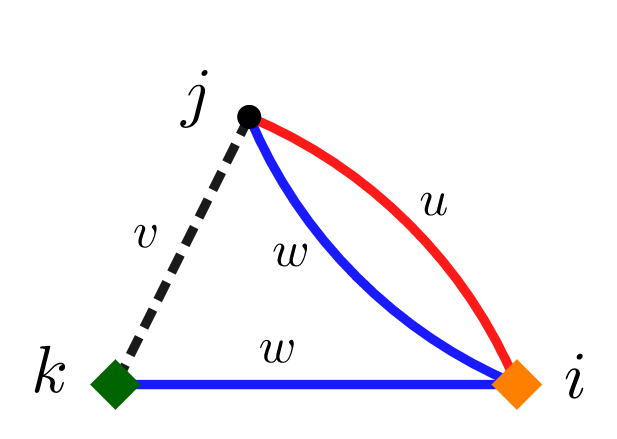}
\caption{The matrix diagram for $S$}\label{Srepeated}
\end{center}
\end{subfigure}
\caption{Matrix diagrams for $S_{12}$ and $S$}
\end{figure}
\begin{lemma}
If $m \gg n$ then $\norm{S_{12}}$ is $\wt{O}\left(\dfrac{\sqrt{m}}{n}\right)$. 
\end{lemma}
\begin{proof}
Here the two nodes must be distinct and all of the crosses are of different types, so no two crosses can be equal to each other. Thus, this is already a graph matrix so we can apply the norm bound directly. Here there are $2$ nodes and $3$ crosses, there are $4$ edges between a cross and a node, and the minimum vertex separator consists of one cross and one node. Thus, the norm bound is 
\[
\wt{O}\left(\frac{1}{n^2}m^{\frac{2-1}{2}}n^{\frac{3-1}{2}}\right) = \wt{O}\left(\frac{\sqrt{m}}{n}\right)
\]
\end{proof}
We now give an alternate proof for the norm bound on the matrix $S$ which appears in the proof of Lemma \ref{lem:MMTE-inner-prod-bound}. For convenience, we repeat the diagram for $S$ here (see Figure \ref{Srepeated}).
 
\begin{lemma}
If $n \ll m \ll n^2$ then $\norm{S}$ is $\wt{O}\left(\dfrac{m}{n^2}\right)$.
\end{lemma}
\begin{proof}
Here there are several cases to consider based on whether $j = k$ and whether the two $w$ crosses are equal.
\begin{enumerate}
\item If $j = k$ and the two $w$ crosses are equal then there are two nodes, one isolated $w$ cross, one isolated $v$ cross, and one $u$ cross. There are $8$ edges between a node and a cross and the minimum vertex separator consists of the $u$ cross. Thus, the norm bound for this case is 
\[
\wt{O}\left(\frac{1}{n^4}m^{\frac{2}{2}}n^{\frac{3+2-1}{2}}\right) = \wt{O}\left(\frac{m}{n^2}\right)
\]
\item If $j = k$ and the two $w$ crosses are not equal then there are two nodes, two $w$ crosses, one isolated $v$ cross, and one $u$ cross. There are $8$ edges between a node and a cross and the minimum vertex separator consists of one node. Thus, the norm bound for this case is 
\[
\wt{O}\left(\frac{1}{n^4}m^{\frac{2-1}{2}}n^{\frac{4+1}{2}}\right) = \wt{O}\left(\frac{\sqrt{m}}{n^\frac{3}{2}}\right)
\]
\item If $j \neq k$ and the two $w$ crosses are equal then there are two nodes, one $w$ cross, one $v$ cross, and one $u$ cross. There are $8$ edges between a node and a cross and the minimum vertex separator consists of the $u$ cross. Thus, the norm bound for this case is 
\[
\wt{O}\left(\frac{1}{n^4}m^{\frac{3}{2}}n^{\frac{3-1}{2}}\right) = \wt{O}\left(\frac{m^{\frac{3}{2}}}{n^3}\right)
\]
\item If $j \neq k$ and the two $w$ crosses are not equal then there are two nodes, two $w$ crosses, one $v$ cross, and one $u$ cross. There are $8$ edges between a node and a cross and the minimum vertex separator consists of one node. Thus, the norm bound for this case is 
\[
\wt{O}\left(\frac{1}{n^4}m^{\frac{3-1}{2}}n^{\frac{4}{2}}\right) = \wt{O}\left(\frac{m}{n^2}\right)
\]
\end{enumerate}
\end{proof}
\subsubsection{Alternate proofs of bounds in Section \ref{sec:explicit-approx}}
We now use the graph matrix norm bounds to give alternate proofs for the bounds on many of the terms in the explicit approximations to the correction terms $\alpha_i, \beta_i, \gamma_i$ (see Section \ref{sec:explicit-approx}).

The Frobenius norm bounds all follow from the following theorem, which can be shown using the techniques in Appendix A of~\cite{graph-matrix-bounds}.
\begin{theorem}
For any shape $\alpha$ such that every vertex of $\alpha$ is either connected to $U_{\alpha}$ or $V_{\alpha}$, with high probability $\norm{M_{\alpha}}_{F}$ is $\Theta\left(m^{\frac{|Nod(\alpha)|}{2}}n^{\frac{|Cross(\alpha)|}{2}}\right)$
\end{theorem}
For the other bounds, we go level by level, though we only cover level $0$ and most of level $1$. For level $0$, we have $U_E$ and $U_{PE}$. For convenience, we repeat the matrix diagrams of $U_E$ and $U_{PE}$ here (see Figure \ref{UEUPErepeat}).
\begin{figure}[h]
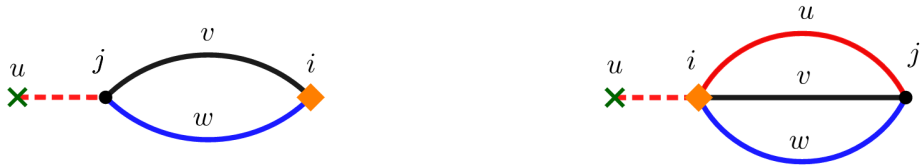

\begin{subfigure}[b]{0.5\textwidth}
\begin{center}
\includegraphics[width = 0.6\textwidth]{gadget-E.png}
\end{center}
\end{subfigure}
\begin{subfigure}[b]{0.5\textwidth}
\begin{center}
\includegraphics[width = 0.6\textwidth]{gadget-PE.png}
\end{center}
\end{subfigure}
\caption{Matrix diagrams of $U_E$ (left) and $U_{PE}$ (right).}\label{UEUPErepeat}
\end{figure} 
\begin{lemma}
If $m \gg n$ then $\norm{U_E}$ is $\wt{O}\left(\dfrac{m}{n^{\frac{3}{2}}}\right)$ and the contribution of $E$ to $\alpha_i$ is $\wt{O}\left(\dfrac{\sqrt{m}}{n}\right)$.
\end{lemma}
\begin{proof}
For $U_{E}$, the nodes cannot be equal to each other and all of the crosses are of different types, so this is already a graph matrix. Here there are $2$ nodes and $3$ crosses, there are $5$ edges between a cross and a node, and the minimum vertex separator consists of one cross. Thus, the norm bound is 
\[
\wt{O}\left(\frac{1}{n^\frac{5}{2}}m^{\frac{2}{2}}n^{\frac{3-1}{2}}\right) = \wt{O}\left(\frac{m}{n^{\frac{3}{2}}}\right)
\]
For the contribution of $E$ to $\alpha_i$, we fix $i$ which has the effect of making $i$ the separator. This gives a norm bound of 
\[
\wt{O}\left(\frac{1}{n^\frac{5}{2}}m^{\frac{2-1}{2}}n^{\frac{3}{2}}\right) = \wt{O}\left(\frac{\sqrt{m}}{n}\right)
\]
\end{proof}
\begin{figure}

\end{figure}
\begin{lemma}\label{PElemma}
If $m \gg n$ then $\norm{U_{PE}}$ is $\wt{O}\left(\dfrac{m}{n^{\frac{3}{2}}}\right)$ and the contribution of $PE$ to $\alpha_i$ is $\wt{O}\left(\dfrac{\sqrt{m}}{n}\right)$.
\end{lemma}
\begin{proof}
For $U_{PE}$, the only crosses which can be equal to each other are the $u$ crosses. Making the $u$ crosses equal to each other does not make any vertex isolated or affect the minimum vertex separator, so it is sufficient to consider the case where the $u$ crosses are distinct. In this case, there are $2$ nodes and $4$ crosses, there are $7$ edges between a cross and a node, and the minimum vertex separator consists of one cross. Thus, the norm bound is 
\[
\wt{O}\left(\frac{1}{n^\frac{7}{2}}m^{\frac{2}{2}}n^{\frac{4-1}{2}}\right) = \wt{O}\left(\frac{m}{n^{2}}\right)
\]
For the contribution of $PE$ to $\alpha_i$, we fix $i$ which has the effect of making $i$ the separator. This gives a norm bound of 
\[
\wt{O}\left(\frac{1}{n^\frac{7}{2}}m^{\frac{2-1}{2}}n^{\frac{4}{2}}\right) = \wt{O}\left(\frac{\sqrt{m}}{n^{\frac{3}{2}}}\right)
\]
\end{proof}
For level $1$, we analyze $U_{ME}$ and $U_{PME}$.
\begin{figure}
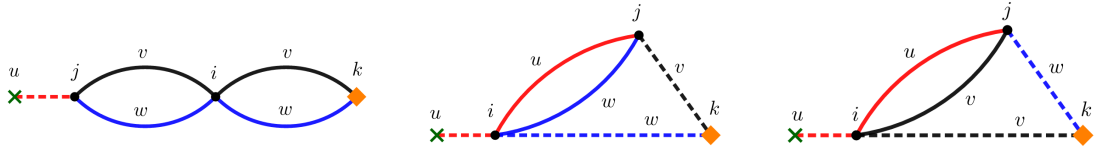

\begin{subfigure}[b]{0.35\textwidth}
\begin{center}
\includegraphics[width = 0.9\textwidth]{gadget-ME1.png}
\end{center}
\end{subfigure}
\begin{subfigure}[b]{0.3\textwidth}
\begin{center}
\includegraphics[width = 0.9\textwidth]{gadget-ME2.png}
\end{center}
\end{subfigure}
\begin{subfigure}[b]{0.3\textwidth}
\begin{center}
\includegraphics[width = 0.9\textwidth]{gadget-ME3.png}
\end{center}
\end{subfigure}
\caption{Matrix diagrams for $U_{ME}$.}\label{UMErepeated}
\end{figure} 
\begin{lemma}
If $n \ll m \ll n^2$ then $\norm{U_{ME}}$ is $\wt{O}\left(\dfrac{m^{\frac{3}{2}}}{n^{\frac{5}{2}}}\right)$ and the contribution of $ME$ to $\alpha_i$ is $\wt{O}\left(\dfrac{m}{n^2}\right)$ (See Figure~\ref{UMErepeated}).
\end{lemma}
\begin{proof}
For the diagram on the left, there are two cases to consider. Either all of the nodes are distinct or $j = k$.
\begin{enumerate}
\item If all of the nodes are distinct then the only way making crosses equal can increase the bound is if a vertex becomes isolated (as the minimum vertex separator will not be affected).
\begin{enumerate}
\item If the $v$ crosses are equal to each other and the $w$ crosses are equal to each other then there are $2$ non-isolated nodes, one isolated node, and $3$ crosses, there are $9$ edges between a cross and a node, and the minimum vertex separator consists of one cross. Thus, the norm bound is 
\[
\wt{O}\left(\frac{1}{n^\frac{9}{2}}m^{\frac{3+1}{2}}n^{\frac{3-1}{2}}\right) = \wt{O}\left(\frac{m^2}{n^{\frac{7}{2}}}\right)
\]
For the contribution to $\alpha_k$, we fix $k$ which has the effect of making $k$ the separator. This gives a norm bound of 
\[
\wt{O}\left(\frac{1}{n^\frac{9}{2}}m^{\frac{3+1-1}{2}}n^{\frac{3}{2}}\right) = \wt{O}\left(\frac{m^{\frac{3}{2}}}{n^3}\right)
\]
\item If the $v$ crosses are not equal to each other and the $w$ crosses are not equal to each other then there are $3$ nodes and $5$ crosses, there are $9$ edges between a cross and a node, and the minimum vertex separator consists of one cross. Thus, the norm bound is 
\[
\wt{O}\left(\frac{1}{n^\frac{9}{2}}m^{\frac{3}{2}}n^{\frac{5-1}{2}}\right) = \wt{O}\left(\frac{m^{\frac{3}{2}}}{n^{\frac{5}{2}}}\right)
\]
For the contribution to $\alpha_k$, we fix $k$ which has the effect of making $k$ the separator. This gives a norm bound of 
\[
\wt{O}\left(\frac{1}{n^\frac{9}{2}}m^{\frac{3-1}{2}}n^{\frac{5}{2}}\right) = \wt{O}\left(\frac{m}{n^2}\right)
\]
\end{enumerate}
\item If $j = k$ then this behaves like $\dfrac{m}{n^2}$ times the diagram with a single cross, a single node, and an edge between them. For this diagram, there is $1$ node and $1$ cross, there is $1$ edge between a cross and a node, and the minimum vertex separator consists of one cross. Thus, the norm bound is 
\[
\wt{O}\left(\frac{m}{n^2}\cdot\frac{1}{\sqrt{n}}m^{\frac{1}{2}}n^{\frac{1-1}{2}}\right) = \wt{O}\left(\frac{m^{\frac{3}{2}}}{n^{\frac{5}{2}}}\right)
\]
For the contribution to $\alpha_k$, we fix $k$ which has the effect of making $k$ the separator. This gives a norm bound of 
\[
\wt{O}\left(\frac{m}{n^2}\cdot\frac{1}{\sqrt{n}}m^{\frac{1-1}{2}}n^{\frac{1}{2}}\right) = \wt{O}\left(\frac{m}{n^2}\right)
\]
\end{enumerate}
The diagram in the middle and the diagram on the right are the same except for switching $v$ and $w$, so we will only analyze one of them. For the diagram in the middle, there are several cases to consider based on which nodes are equal to each other.
\begin{enumerate}
\item If all of the nodes are distinct, making the $w$ crosses equal will not make any vertices isolated or affect the minimum vertex separator, so it is sufficient to consider the case where the crosses are all distinct as well. In this case, there are $3$ nodes and $5$ crosses, there are $9$ edges between a cross and a node, and the minimum vertex separator consists of one cross. Thus, the norm bound is 
\[
\wt{O}\left(\frac{1}{n^\frac{9}{2}}m^{\frac{3}{2}}n^{\frac{5-1}{2}}\right) = \wt{O}\left(\frac{m^{\frac{3}{2}}}{n^{\frac{5}{2}}}\right)
\]
For the contribution of $PE$ to $\alpha_k$, we fix $k$ which has the effect of making $k$ the separator. This gives a norm bound of 
\[
\wt{O}\left(\frac{1}{n^\frac{9}{2}}m^{\frac{3-1}{2}}n^{\frac{5}{2}}\right) = \wt{O}\left(\frac{m}{n^{2}}\right)
\]
\item If $k =j$, this behaves like $\dfrac{1}{n}$ times the diagram which only has the red edges. For this diagram, there are $2$ nodes and $2$ crosses, there are $3$ edges between a cross and a node, and the minimum vertex separator consists of one cross. Thus, the norm bound is 
\[
\wt{O}\left(\frac{1}{n}\cdot\frac{1}{n^\frac{3}{2}}m^{\frac{2}{2}}n^{\frac{2-1}{2}}\right) = \wt{O}\left(\frac{m}{n^2}\right)
\]
For the contribution of $PE$ to $\alpha_k$, we fix $k$ which has the effect of making $k$ the separator. This gives a norm bound of 
\[
\wt{O}\left(\frac{1}{\frac{1}{n} \cdot n^\frac{3}{2}}m^{\frac{2-1}{2}}n^{\frac{2}{2}}\right) = \wt{O}\left(\frac{\sqrt{m}}{n^{\frac{3}{2}}}\right)
\]
\item If $k = i$, this behaves like the diagram for $U_{PE}$ which we analyzed in Lemma \ref{PElemma}
\end{enumerate}
\end{proof}
\begin{figure}[h]
 \begin{subfigure}[b]{0.3\textwidth}
\begin{center}
\includegraphics[width = 0.9\textwidth]{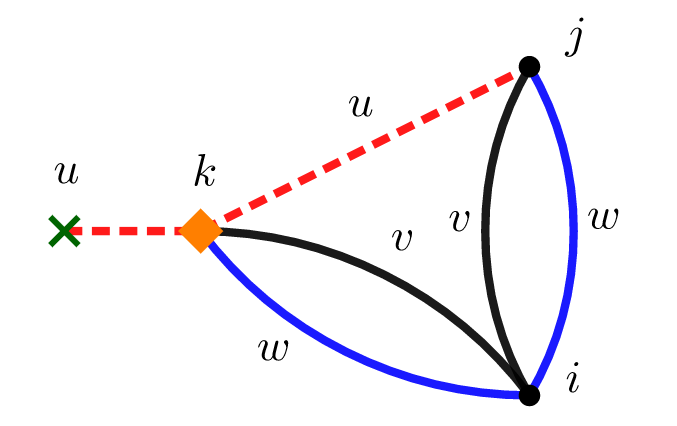}
\end{center}
\end{subfigure}
\begin{subfigure}[b]{0.3\textwidth}
\begin{center}
\includegraphics[width = 0.9\textwidth]{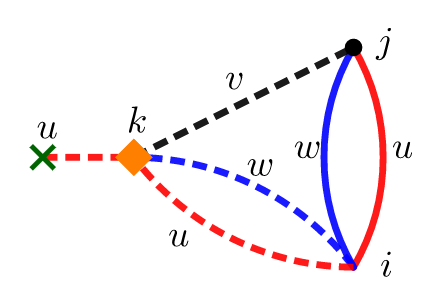}
\end{center}
\end{subfigure}
\begin{subfigure}[b]{0.3\textwidth}
\begin{center}
\includegraphics[width = 0.9\textwidth]{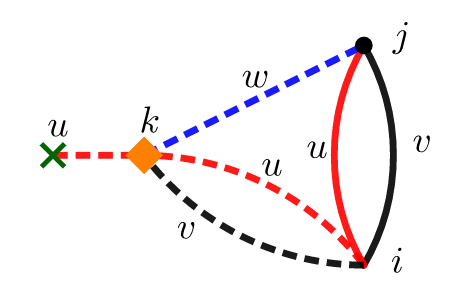}
\end{center}
\end{subfigure}
\caption{Matrix diagrams for $U_{PME}$.}\label{UPMErepeated}
\end{figure} 
\begin{lemma}
If $n \ll m \ll n^2$ then $\norm{U_{PME}}$ is $\wt{O}\left(\dfrac{m^{\frac{3}{2}}}{n^{\frac{5}{2}}}\right)$ and the contribution of $PME$ to $\alpha_i$ is $\wt{O}\left(\dfrac{m}{n^2}\right)$ (See Figure~\ref{UPMErepeated}).
\end{lemma}
\begin{proof}
For the diagram on the left, there are two cases to consider. Either all of the nodes are distinct or $j = k$
\begin{enumerate}
\item If all of the nodes are distinct then the only way making crosses equal can increase the bound is if a vertex becomes isolated (as the minimum vertex separator will not be affected).
\begin{enumerate}
\item If the $v$ crosses are equal to each other and the $w$ crosses are equal to each other then there are $2$ non-isolated nodes, one isolated node, and $4$ crosses, there are $11$ edges between a cross and a node, and the minimum vertex separator consists of one cross. Thus, the norm bound is 
\[
\wt{O}\left(\frac{1}{n^\frac{11}{2}}m^{\frac{3+1}{2}}n^{\frac{4-1}{2}}\right) = \wt{O}\left(\frac{m^2}{n^4}\right)
\]
For the contribution to $\alpha_k$, we fix $k$ which has the effect of making $k$ the separator. This gives a norm bound of 
\[
\wt{O}\left(\frac{1}{n^\frac{11}{2}}m^{\frac{3+1-1}{2}}n^{\frac{4}{2}}\right) = \wt{O}\left(\frac{m^{\frac{3}{2}}}{n^{\frac{7}{2}}}\right)
\]
\item If the $v$ crosses are not equal to each other and the $w$ crosses are not equal to each other then there are $3$ nodes and $6$ crosses, there are $11$ edges between a cross and a node, and the minimum vertex separator consists of one cross. Thus, the norm bound is 
\[
\wt{O}\left(\frac{1}{n^\frac{11}{2}}m^{\frac{3}{2}}n^{\frac{6-1}{2}}\right) = \wt{O}\left(\frac{m^{\frac{3}{2}}}{n^3}\right)
\]
For the contribution to $\alpha_k$, we fix $k$ which has the effect of making $k$ the separator. This gives a norm bound of 
\[
\wt{O}\left(\frac{1}{n^\frac{11}{2}}m^{\frac{3-1}{2}}n^{\frac{6}{2}}\right) = \wt{O}\left(\frac{m}{n^{\frac{5}{2}}}\right)
\]
\end{enumerate}
\item If $j = k$ then this again behaves like $\dfrac{m}{n^2}$ times the diagram with a single cross, a single node, and an edge between them.
\end{enumerate}
The diagram in the middle and the diagram on the right are the same except for switching $v$ and $w$, so we will only analyze one of them. For the diagram in the middle, there are several cases to consider based on which nodes are equal to each other.
\begin{enumerate}
\item If all of the nodes are distinct, the analysis is the same as the analysis for the shape on the left when the nodes are distinct.
\item If $k = j$ then this again behaves like $\dfrac{m}{n^2}$ times the diagram with a single cross, a single node, and an edge between them.
\item If $k = i$, this behaves like the diagram for $U_{PE}$ which we analyzed in Lemma \ref{PElemma}.
\end{enumerate}
\end{proof}
\subsubsection{Miscellaneous Example}
Here we give one more example.
\begin{lemma}
The matrix in equation~\eqref{eq:certificate-core-term-cross-prod} has norm $\wt{O}\left(\dfrac{m}{n^{\frac{3}{2}}}\right)$.
\end{lemma}
\begin{proof}
The diagram for this matrix is as follows. There are two distinct nodes $i$ and $j$ in the middle. Between these nodes, there is a path of length $2$ with a $w$ cross in the middle. On the left, there are two $u$ crosses adjacent to $i$ and $j$ respectively. On the right, there are two $v$ crosses adjacent to $i$ and $j$ respectively. 

If all of the crosses are distinct, there are $2$ nodes and $5$ crosses, there are $6$ edges between a cross and a node, and the minimum vertex separator consists of two crosses. Thus, the norm bound is 
\[
\wt{O}\left(\frac{1}{n^3}m^{\frac{2}{2}}n^{\frac{5-2}{2}}\right) = \wt{O}\left(\frac{m}{n^{\frac{3}{2}}}\right)
\]
If some of the crosses are equal to each other, this may reduce the size of the minimum vertex separator by $1$ cross, but since it also reduces the total number of crosses and cannot make any vertices isolated, this will not give a larger bound.
\end{proof}
Note: The same argument works for the matrix in Theorem~\ref{thm:A0-twist-bound}, except that the separator size cannot decrease by making crosses equal to each other.

\bibliographystyle{alpha}
\bibliography{ten_comp_bib}

\end{document}